\documentclass[sigconf,nonacm]{acmart}

\AtBeginDocument{%
  }

\usepackage[T1]{fontenc}

\usepackage{xcolor}
\usepackage[utf8]{inputenc} 
\usepackage[T1]{fontenc}    
\usepackage{hyperref}       
\usepackage{url}            
\usepackage{booktabs}       
\usepackage{amsfonts}       
\usepackage{nicefrac}       
\usepackage{microtype}      
\usepackage{lipsum}		
\usepackage{graphicx}
\usepackage{doi}
\usepackage{amsmath}
\usepackage{amsthm}
\usepackage{algorithm}
\usepackage[noend]{algpseudocode}
\usepackage{mdframed}
\usepackage{multirow}
\usepackage{todonotes}

\usepackage{filecontents}

\usepackage{amssymb}
\usepackage{textcomp}
\usepackage{xcolor}
\usepackage{hyperref}
\usepackage{xspace}
\usepackage{multirow}
\usepackage{caption}
\usepackage{xspace}
\usepackage{subcaption}
\usepackage{soul}

\usepackage{enumitem}
\usepackage{breqn}
\usepackage{thmtools}
\usepackage{amsthm}
\usepackage{xurl}

\mdfsetup{frametitlealignment=\center, innerleftmargin=5pt, innerrightmargin=3pt}
\newcommand{\mdframeSkip}{15pt}

\newenvironment{functionality}[1]
{\begin{mdframed}[
    frametitle={#1},
    frametitlerule=true,
    frametitlebackgroundcolor=gray!20,
    splittopskip=\mdframeSkip,
    splitbottomskip=\mdframeSkip,
    startinnercode={
        \begingroup
        \scriptsize
        \baselineskip=0cm
    },
    endinnercode={\endgroup}
]}
{\end{mdframed}}

\newtheorem{lemma}{Lemma}
\newtheorem{definition}{Definition}

\newcommand{\cmt}[1]{{\color{blue}  \textbackslash\textbackslash#1}}

\newcommand{\pcur}{\textsf{pid\textsubscript{cur}}\xspace}
\newcommand{\scur}{\textsf{sid\textsubscript{cur}}\xspace}

\newcommand{\xr}{\textsf{FRoll}\xspace}

\begin{document}
\title{A Security Framework for General Blockchain Layer~2 Protocols}

\author{Zeta Avarikioti}
\affiliation{%
  \institution{TU Wien \& Common Prefix}
  \country{}
}

\author{Matteo Maffei}
\affiliation{%
  \institution{TU Wien}
  \country{}
}

\author{Yuheng Wang}
\affiliation{%
  \institution{TU Wien}
  \country{}
}

\begin{abstract}
Layer 2 (L2) protocols, such as payment channels, sidechains, and rollups, are central to blockchain scalability, enabling off-chain execution while preserving on-chain security guarantees. Despite their growing deployment, existing security models remain largely protocol-specific and monolithic, hindering compositional reasoning and preventing principled comparisons of assumptions and requirements, both at the level of functionality and security.

We present a general security framework for L2 protocols based on the IITM-style Universal Composability (iUC) model. Our framework introduces a modular ideal functionality $\mathcal{F}_{\text{layer2}} $ that abstracts away mechanism-specific details while capturing the essential structure of L2 systems through composable subroutines. These subroutines formalize the core phases of joining, submission, updating, reading, and settlement under adversarial conditions, enabling uniform definitions of key security and performance properties, including safety, liveness, and data availability, across a broad class of L2 protocols.

To demonstrate generality, we instantiate the framework for three representative constructions: the Brick payment channel, the Liquid sidechain, and the Arbitrum Nitro rollup. Each case study yields a protocol-specific ideal functionality derived from $\mathcal{F}_{\text{layer2}}$, tailored to its underlying assumptions. Our analysis (i) establishes security via simulation-based proofs, (ii) exposes inherent trade-offs among safety, liveness, and data-availability, and (iii) derives lower bounds that characterize fundamental limitations of each design class.

Finally, we illustrate the framework’s utility as a design tool by presenting \xr, the first optimistic rollup protocol with fast-finality guarantees, together with a security analysis carried out within our model. This case study demonstrates how our framework supports requirement-driven design of L2 protocols.
\end{abstract}




\maketitle




\section{Introduction}
\label{sec:intro}

Blockchain (also known as a distributed ledger) has attracted significant attention in recent years as a novel decentralized system maintained by nodes such as miners~\cite{bitcoin_devguide_mining} or validators~\cite{ethereum_validator_phase0}, whose core functionality for user clients can be summarized into two types of requests: \emph{submit} transactions to the blockchain and \emph{reading} previously executed on-chain transactions~\cite{graf2021security}. Despite their potential to transform existing centralized infrastructures, blockchains suffer from inherent limitations in scalability, transaction costs, and latency. To mitigate these limitations, a broad class of protocols, commonly referred to as Layer~2 (L2) protocols, has been proposed. L2 protocols act as an auxiliary interface layer that processes these requests off-chain, offloading the workload from the Layer~1 (L1) blockchain and interacting with L1 only when necessary~\cite{gudgeon2020sok}.

L2 protocols have emerged as a fundamental building block for blockchain applications and as the primary interaction interface for blockchain clients. Recent studies~\cite{a16z2024state,wilson2023layer2} report that transaction volume on L2s has already surpassed that of L1. Representative L2 protocol examples include payment channels (e.g., the Bitcoin Lightning Network~\cite{poon2016bitcoin}), rollups (e.g., Optimism~\cite{optimism} and Arbitrum~\cite{bousfield2022arbitrum}), and sidechains (e.g., the Liquid Network~\cite{nick2020liquid} and Polygon~\cite{polygon}). However, this rapid and largely informalized emergence and growth of L2s has introduced two key challenges for blockchain application designers seeking to leverage them.

First, what fundamental security properties should a generic L2 protocol satisfy? Without a clear understanding of these properties, the design of new protocols or the integration of existing ones into larger systems can inadvertently introduce subtle vulnerabilities, stemming from oversights caused by insufficient knowledge. Existing security definitions are typically tailored to specific constructions. As a result, L2 protocols are often proven in isolation, without accounting for their interaction with the underlying blockchain. This gap has contributed to several design flaws in recently proposed L2 protocols~\cite{egger2019atomic,DBLP:journals/iacr/KiayiasL21,DBLP:conf/cans/JourenkoLT20,DBLP:conf/ndss/AumayrMKM23}.

Second, since the primary and shared purpose of L2 protocols is to handle submit and read requests from clients to the L1 blockchain more efficiently, designers face two intertwined questions: how to compare different types of L2 protocols and determine which is most suitable for a given requirement, and how to design new L2 protocols (or modify existing ones) when no current solution meets an application's needs. Although formal security analysis of some L2 protocols exists~\cite{kiayias2020composable,aumayr2021generalized}, these works are largely ad hoc and focus on individual protocols, often entangling abstract security guarantees with real-world implementation details. As a result, the differences between distinct L2 paradigms such as sidechains and rollups are rarely examined within a unified framework, and modifying existing protocols or designing new ones to meet application-specific requirements remains complex and error-prone.

In this work, we close these two gaps by introducing \emph{the first security framework for general type L2 protocols}. 
At the core of our framework lies the \textbf{ideal functionality $\mathcal{F}_{\text{layer2}}$}, which formalizes the key security guarantees that L2 protocols should satisfy ($\mathcal{F}^{\Lambda}_{\text{layer2}}$ for protocol $\Lambda$).
The functionality abstracts over concrete implementation details to capture the essential workflow from the client side common to all L2 constructions:  joining the L2, submitting off-chain requests, updating the state, reading executed requests and L2 states, and exiting the L2.
The requirements for each phase of secure L2 interaction are captured by the checks performed within a dedicated subroutine ($\mathcal{F}_{\text{join}}$, $\mathcal{F}_{\text{submit}}$, $\mathcal{F}_{\text{update}}$, $\mathcal{F}_{\text{read}}$, $\mathcal{F}_{\text{settlement}}$), invoked through a central interface machine $\mathcal{F}_{\text{client}}$ (see Figure~\ref{fig:structure}). 
This decomposition enables formal reasoning about security, performance, and trust assumptions across a wide spectrum of designs, within a composable setting.

To address the first challenge, we formally define four security properties: correct L2 initialization, correct L2 settlement, safety, and liveness; and one efficiency property, data availability, based on our framework. The framework's contribution is not only in providing a shell around protocol-specific subroutine checks that can be trivially realized by the real world protocols, but in setting constraints that no realization can escape: how the subroutines are influenced with each other through $\mathcal{F}_{\text{client}}$'s $\mathsf{internalState}$ and how they interact with the blockchain $\mathcal{F}_{\text{ledger}}$ to satisfy the security properties defined by the framework. Within this fixed structure, the acceptance checks supply the paradigm-specific work, while the framework provides the paradigm-independent reasoning that turns them into the four security properties. For instance, all three case studies realize safety with their different consistency checks, e.g., Brick's state sequence number, Liquid's block predecessor pointer. 

For the second challenge, our framework makes precise the trade-offs among safety, liveness, and data availability across L2 paradigms: payment channels provide instant finality but are sensitive to corruption; rollups inherit L1 safety at the cost of latency and on-chain storage; sidechains trade stronger trust assumptions for performance via internal consensus. While these distinctions have long been recognized informally, we formalize them and show they stem from inherent design limits (Tables~\ref{tab:comparison}–\ref{tab:difference2}, Fig.~\ref{fig:Difference}), proving lower bounds (Theorems~\ref{thm:livenessPS}–\ref{thm:datar}) that characterise the constraints each paradigm faces. Building on these insights, we propose \xr, the first optimistic-rollup protocol to achieve fast finality while preserving rollup safety, demonstrating how the framework's modular structure directly guides the design of L2 protocols tailored to application-specific requirements.

In this work, we focus on payment-style transaction requests, which is the core L2 functionality common to payment channels, sidechains, and rollups, and leave arbitrary smart-contract execution as future work.

\textbf{Summary of Contributions.} In summary, our contributions are as follows:

\begin{itemize}[topsep=2pt, partopsep=0pt, parsep=2pt, itemsep=3pt]

     \item We propose the first general security framework for L2 protocols, formalized as an ideal functionality $\mathcal{F}_{\text{layer2}}$ that captures core L2 protocol interactions and supports composable reasoning with modular subroutines (Sec.~\ref{sec:functionality}).

    \item We apply our framework to simplified core models of three qualitatively distinct L2 protocols: Brick (payment channel), Liquid (sidechain), and Arbitrum Nitro (rollup), demonstrating its expressiveness and generality across paradigms (Sec.~\ref{sec:casestudy}).

    \item We leverage our framework to formalize the design differences across these protocols, analyze and prove fundamental trade-offs among safety, liveness, and data availability, showing that these trade-offs arise inherently from specific subroutine instantiations (Sec.~\ref{sec:insight}).

    \item Building on the insights gained from the case studies and comparative analysis, we present and formally analyze a new optimistic-rollup protocol, \xr, which is the first to achieve fast finality while preserving safety (Sec.~\ref{sec:design}).
\end{itemize}

\section{Preliminaries}

\subsection{System Model and Assumptions.}

\textbf{Protocol Participants.} In practice, despite varying terminology across protocols, L2 protocols generally involve three types of core roles (a single participant may take over multiple roles):
\begin{itemize}
\item \textbf{Client:} joins the L2 protocol, submits execution requests to update the off-chain state, settles the L2 state back to L1, and queries the current L2 state.
\item \textbf{Operator:} (also referred to as channel owner, sequencer, or maintainer) orders and processes client requests and manages off-chain state transitions. Also, help the client with the specific steps like protocol joining and state settlement.
\item \textbf{Third party:} auxiliary participants such as \emph{watchtowers}, which monitor L1 on behalf of clients, and \emph{verifiers}, which validate state transitions executed off-chain.
\end{itemize}

We assume static, publicly known sets of operators and third parties as global protocol parameters. For each protocol, we assume at least one honest client, so that the protocol is meaningful, and assume its operator/verifier corruption thresholds hold throughout the corresponding security analysis. All participants are modeled as PPT interactive Turing machines.

\textbf{Cryptographic and Communication.} We assume the existence of ideal functionalities for EUF-CMA signatures, denoted $\mathcal{F}_{\text{sig}}$, and for communication and time-related behavior, denoted $\mathcal{F}_{\text{com}}$, both defined in Appendix~\ref{apdx:sigcom}. We note that different protocols may rely on different off-chain communication network assumptions, which are reflected in the corresponding definition of $\mathcal{F}_{\text{com}}$.

\textbf{L1 blockchain.} The underlying L1 blockchain is modeled as a global ideal functionality $\mathcal{F}_{\text{ledger}}$, instantiated following the general ledger definition of Graf et al.~\cite{graf2021security}. This functionality maintains a totally ordered list of transactions and captures both UTXO-based blockchains such as Bitcoin and account-based blockchains with smart-contract support such as Ethereum. It is parameterized by subroutines that instantiate protocol-specific security properties such as consistency and liveness, thereby subsuming existing blockchain functionalities and their realizations. $\mathcal{F}_{\text{ledger}}$ exposes an L1 client interface $\mathcal{F}_{\text{client\textsubscript{L1}}}$ with \emph{SubmitL1} to submit a transaction and \emph{ReadL1} to read committed transactions and their execution state. An excerpt of the functionality and further explanation are provided in Appendix~\ref{appdx:ledger}.

\textbf{The iUC Framework.} Our framework is built on the iUC framework~\cite{camenisch2019iuc}, which is based on IITM~\cite{kusters2020iitm} and extends Canetti's UC~\cite{canetti2001universally} with a modular structure well suited to modeling real-world protocols, ideal functionalities, and hybrid settings. A protocol $\mathcal{P}$ or ideal functionality $\mathcal{F}$ is specified as a system of interactive Turing machines (ITMs), e.g., \(\mathcal{P}=\{\mathcal{M}_1 | \ldots,\mathcal{M}_n\}\), \(\mathcal{F}=\{\mathcal{M}'_1 | \ldots,\mathcal{M}'_m\}\), where each machine implements one or more \emph{roles} declared as $(\text{role}_{\mathrm{pub},1},\ldots \mid \text{role}_{\mathrm{priv},1},\ldots)$. Every machine exposes \emph{I/O interfaces} for caller–subroutine interaction with other protocol machines and \emph{NET connections} for interaction with the adversary $\mathcal{A}$ (or the simulator $\mathcal{S}$ in the ideal world). I/O interfaces of \emph{public} roles are accessible to other protocols and the environment $\mathcal{E}$, while those of \emph{private} roles are restricted to internal use; protocols can thus be composed only through matching public roles. We write $\mathcal{P} \leq \mathcal{F}$ to denote that $\mathcal{P}$ iUC-realizes $\mathcal{F}$, i.e., for every adversary $\mathcal{A}$ there exists a simulator $\mathcal{S}$ such that $\{\mathcal{E},\mathcal{A},\mathcal{P}\}$ and $\{\mathcal{E},\mathcal{S},\mathcal{F}\}$ are indistinguishable for all $\mathcal{E}$. The composition theorem then permits replacing $\mathcal{F}$ with $\mathcal{P}$ in any higher-level protocol. In our proofs, we let $\mathcal{E}$ subsume the network adversary, yielding an equivalent but simpler formulation.

In this paper, each ideal functionality instance models a single session (\textit{sid}) and manages all parties ($\textit{pid}$) within it through a shared $\mathsf{internalState}$  that is initialized with protocol-specific predefined content, whereas in the real protocol each machine instance corresponds to a single party, reflecting distributed execution. We write $(\textit{pid}_{\text{cur}}, \textit{sid}_{\text{cur}}, \text{role}_{\text{cur}})$ to denote the currently active entity. By default, a machine accepts all incoming messages carrying the same \textit{sid}, as determined by its \textbf{CheckID} function. Upon corruption, the functionality releases only the leakage specified in its \textbf{Corruption behavior} block. A brief introduction to iUC is provided in Appendix~\ref{apdx:iUC}.

\section{The Ideal Functionality for L2 Protocols}
\label{sec:functionality}
In this section, we formalize the core contribution of this work: an ideal functionality $\mathcal{F}_{\text{layer2}}$ that captures the essential behavior and security properties of general L2 protocols. This functionality is defined within the iUC framework and is composed of modular subroutines that correspond to the phases of a L2 protocol’s lifecycle. The overall structure is depicted in Fig.~\ref{fig:structure}.


\begin{figure}[h]  
    \centering
    \includegraphics[width=\columnwidth]{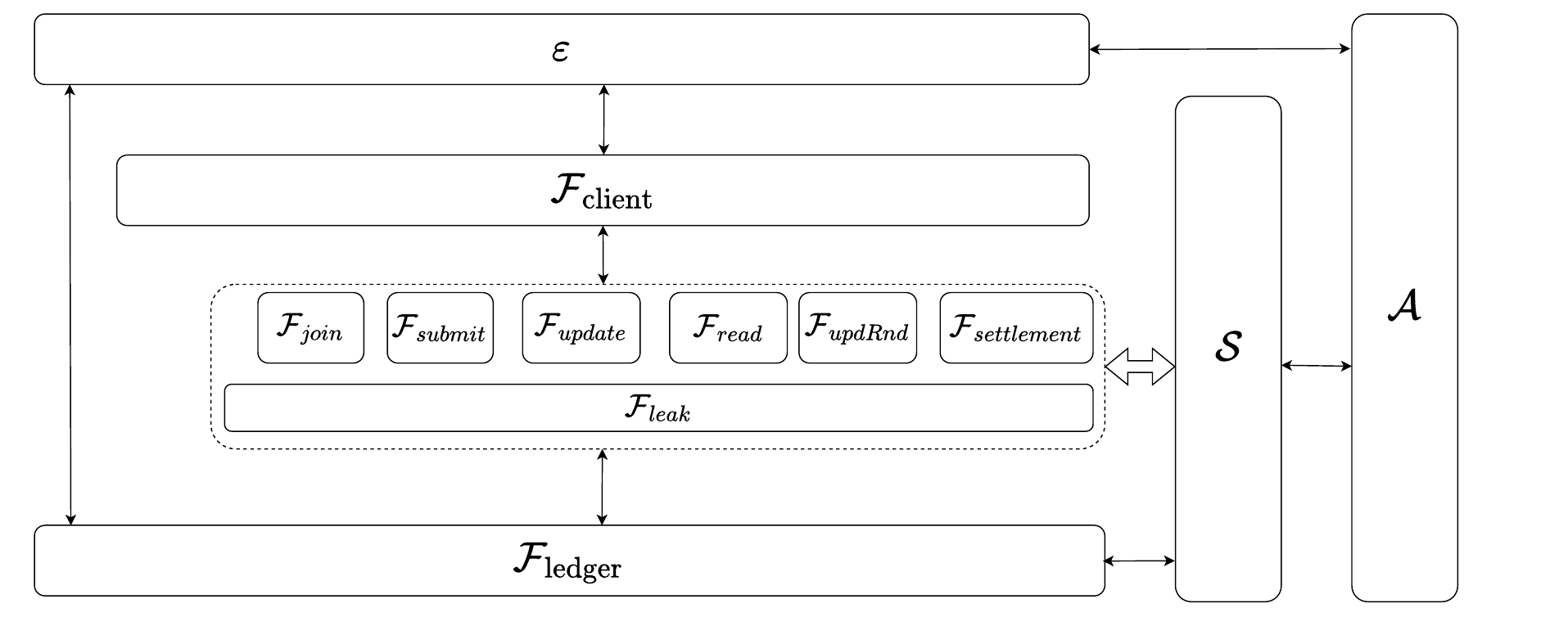}
    \caption{Structure of the ideal functionality $\mathcal{F}_{\text{layer2}}$. $\mathcal{E}$ refers to the environment, $\mathcal{S}$ refers to the simulator and $\mathcal{A}$ refers to the adversary.}
    \label{fig:structure}
\end{figure}

{\textbf{Design Rationale.}}
The ideal functionality $\mathcal{F}_{\text{layer2}}:=(\mathcal{F}_{\text{client}},\\\mathcal{F}_{\text{ledger}}\mid \mathcal{F}_{\text{submit}},\allowbreak\mathcal{F}_{\text{join}},\allowbreak\mathcal{F}_{\text{update}},\allowbreak\mathcal{F}_{\text{read}},\allowbreak\mathcal{F}_{\text{settlement}},\allowbreak\mathcal{F}_{\text{updRnd}},\allowbreak\mathcal{F}_{\text{leak}})$ is designed to capture the core phases of interaction in a broad class of L2 protocols, from payment channels (factories) to sidechains and rollups. Its structure reflects a minimal and modular decomposition of protocol behavior that abstracts away implementation-specific details while preserving the key security and performance considerations common to L2 designs. Each subroutine models a fundamental aspect of L2 operation: $\mathcal{F}_{\text{join}}$ models the act of joining the L2 protocol, such as opening a channel or joining a rollup or sidechain; $\mathcal{F}_{\text{submit}}$ captures what off-chain requests or transactions can be issued by clients; $\mathcal{F}_{\text{update}}$ formalizes how state transitions are agreed upon and executed: whether via unanimous agreement (channels), consensus (sidechains), or sequencer output (rollups); $\mathcal{F}_{\text{read}}$ reflects the accessibility of the L2 executed requests and states, modeling different data availability assumptions; $\mathcal{F}_{\text{settlement}}$ abstracts the final anchoring of L2 state back to the L1 blockchain; $\mathcal{F}_{\text{updRnd}}$ simulates the internal protocol clock to capture communication network; $\mathcal{F}_{\text{leak}}$ defines the information leakage during each protocol phase resulting from corruption.

\subsection{$\mathcal{F}_{\text{client}}$ description and subroutines overview}

At its core, the functionality \(\mathcal{F}_{\text{client}}\) specifies the client interface machine that connects to the external environment \(\mathcal{E}\) and mediates how the protocol returns outputs to higher-level protocols via parameterized subroutines. Importantly, the client interface does not generate outputs on its own, it delegates all request handling and output generation to the appropriate subroutines and replies to \(\mathcal{E}\). In this paper, we focus on the L2 client role, though additional roles can be incorporated as needed. The core logic of the client interface machine in \(\mathcal{F}_{\text{client}}\) is shown in Figure~\ref{fig:interface}.
\vspace{.5em}

\begin{functionality}{Description of functionality $\mathcal{F}_{\text{client}} = (\text{client})$}{

\textbf{Participating roles:} \{client\}

\noindent \textbf{Corruption model:} dynamic corruption, at least one honest
}\end{functionality}
\vspace{.5em}
\begin{functionality}{Description of $\mathcal{M}_{\text{client}}$ of functionality $\mathcal{F}_{\text{client}}$}

{

\textbf{Implemented role(s):} \{client\}

\noindent\textbf{Subroutines:} $\mathcal{F}_{\text{submit}}$: submit, $\mathcal{F}_{\text{join}}$: join, $\mathcal{F}_{\text{update}}$: update, $\mathcal{F}_{\text{read}}$: read, $\mathcal{F}_{\text{settlement}}$: settlement, $\mathcal{F}_{\text{updRnd}}$: updRnd, $\mathcal{F}_{\text{leak}}$: leak

\noindent \textbf{Internal state:} (initial values are protocol-specific)
\begin{itemize}
    \item $\mathsf{round} \in \mathbb{N}_{\geq 0}$ \hfill\cmt{Inner clock round}
    \item $\mathsf{requestQueue} \subset \{0,1\}^{*}$ \hfill\cmt{Submitted unexecuted requests}
    \item $\mathsf{executedRequest} \subset \{0,1\}^{*}$ \hfill\cmt{Executed requests}
    \item $\mathsf{stateList} \subset \{0,1\}^{*}$ \hfill\cmt{L2 state list}
    \item $\mathsf{onchainState} \subset \{0,1\}^{*}$ \hfill\cmt{L1 committed state}
    \item $\mathsf{identities} \subset \{0,1\}^{*}$ \hfill\cmt{Joined identities}
\end{itemize}


\noindent \textbf{CheckID}(\emph{pid}, \emph{sid}, \emph{role}): Accept all messages with the same \emph{sid}.

\noindent \noindent \textbf{Corruption behavior:}\hfill\cmt{Can be further defined with $\mathcal{F}_{\text{leak}}$}
\begin{itemize}
    \item \textbf{DetermineCorrStatus}(\emph{pid}, \emph{sid}, \emph{role}): Return $\mathsf{corr}$.
    \item \textbf{LeakedData}(\emph{pid}, \emph{sid}, \emph{role}): Return $\mathsf{internalState}$. \cmt{Under the plaintext assumption, all $\mathsf{internalState}$ is leaked to $\mathcal{S}$}
\end{itemize}


\vspace{0.5em}
\noindent \textbf{Main:}

\vspace{0.5em}

\textbf{recv} \{Submit, $\mathit{request}$\} \textbf{from} I/O:\hfill\cmt{Input requests from $\mathcal{E}$}

\begin{enumerate}[itemsep=0.5em]

\item \textbf{send} \{Submit, $\mathit{request}$, $\mathsf{internalState}$\} \textbf{to} $(\pcur, \scur, \mathcal{F}_{\text{submit}}:\text{submit})$;
\textbf{wait for} \{Submit, $\mathit{response}$\} s.t. $\mathit{response} \in \{\text{true}, \text{false}\}$;

\item \textbf{if} $\mathit{response} =$ true: $\mathsf{requestQueue}.\text{add}(\mathit{request})$, 
\textbf{send} $\mathit{request}$ \textbf{to} $\mathcal{S}$ via NET;
\end{enumerate}

\hrule
\vspace{0.5em}

\textbf{recv} \{Join, $\mathit{Attachment}$\} \textbf{from} NET:\hfill\cmt{Join notification from $\mathcal{S}$ for subroutine checks}

\begin{enumerate}[itemsep=0.5em]

\item \textbf{send} \{Join, $\mathit{Attachment}$, $\mathsf{internalState}$\} \textbf{to} $(\pcur, \scur, \mathcal{F}_{\text{join}}:\text{join})$;
\textbf{wait for} \{Join, $\mathit{response}$, $s_{\mathit{init}}$\} s.t. $\mathit{response} \in \{\text{true}, \text{false}\}$;

\item \textbf{if} $\mathit{response} =$ true: update $\mathsf{internalState}$ according to $\mathit{Attachment}$, 
\textbf{reply} \{Join, $s_{\mathit{init}}$\} to $\mathcal{E}$ via I/O;

\end{enumerate}

\hrule
\vspace{0.5em}

\textbf{recv} \{Update, $\mathit{Attachment}$\} \textbf{from} NET:\hfill\cmt{Update notification from $\mathcal{S}$ for subroutine checks}

\begin{enumerate}[itemsep=0.5em]
    \item \textbf{send} \{Update, $\mathit{Attachment}$, $\mathsf{internalState}$\} \textbf{to} $(\pcur, \scur, \mathcal{F}_{\text{update}}:\text{update})$,
    \textbf{wait for} \{Update, $\mathit{response}$, $\mathit{newState}$, $\mathit{executedReq}$\};

    \item \textbf{if} $\mathit{response} =$ true: update $\mathsf{internalState}$ according to the reply and $\mathit{Attachment}$;
\end{enumerate}

\hrule
\vspace{0.5em}

\textbf{recv} \{Read\} \textbf{from} I/O:\hfill\cmt{Read request from $\mathcal{E}$}

\begin{enumerate}[itemsep=0.5em]

\item \textbf{send} \{Read, $\mathsf{internalState}$\} \textbf{to} $(\pcur, \scur, \mathcal{F}_{\text{read}}:\text{read})$,
\textbf{wait for} $\mathit{ReadResult}$;

\item \textbf{if} $\mathit{ReadResult} \neq \bot$: \textbf{reply} \{Read, $\mathit{ReadResult}$\} to $\mathcal{E}$ via I/O;

\end{enumerate}

\hrule
\vspace{0.5em}

\textbf{recv} \{Settlement, $\mathit{Attachment}$\} \textbf{from} NET:\hfill\cmt{Settlement notification from $\mathcal{S}$ for subroutine checks}

\begin{enumerate}[itemsep=0.5em]

\item \textbf{send} \{Settlement, $\mathit{Attachment}$, $\mathsf{internalState}$\} \textbf{to} $(\pcur, \scur,\\ \mathcal{F}_{\text{settlement}}:\text{settlement})$,
\textbf{wait for} \{Settlement, $\mathit{response}$, $s_{\mathit{settle}}$\} s.t. $\mathit{response} \in \{\text{true}, \text{false}\}$;

\item \textbf{if} $\mathit{response} =$ true: $\mathsf{onchainState} \leftarrow s_{\mathit{settle}}$,
\textbf{reply} \{Settlement, $s_{\mathit{settle}}$\} to $\mathcal{E}$ via I/O;
\end{enumerate}

\hrule
\vspace{0.5em}

\textbf{recv} \{UpdateRound\} \textbf{from} NET:\hfill\cmt{Round update notification from $\mathcal{S}$ for subroutine checks}

\begin{enumerate}[itemsep=0.5em]

\item \textbf{send} \{UpdateRound, $\mathsf{internalState}$\} \textbf{to} $(\pcur, \scur, \mathcal{F}_{\text{updRnd}}:\text{updRnd})$;
\textbf{wait for} \{UpdateRound, $\mathit{response}$\} s.t. $\mathit{response} \in \{\text{true}, \text{false}\}$;

\item \textbf{if} $\mathit{response} =$ true: $\mathsf{round} \leftarrow \mathsf{round} + 1$;

\end{enumerate}

\hrule
\vspace{0.5em}

\textbf{recv} \{GetCurRound\} \textbf{from} I/O:\hfill\cmt{Round read request from $\mathcal{E}$}

\begin{enumerate}[itemsep=0.5em]
    \item \textbf{reply} \{GetCurRound, $\mathsf{round}$\} to $\mathcal{E}$ via I/O;;
\end{enumerate}




}\end{functionality}
\begin{minipage}{\columnwidth}
    \captionof{figure}{The main logic of $\mathcal{F}_{\text{client}}$ for requests handling. The update of $\mathsf{internalState}$ according to $\mathit{Attachment}$ is protocol specific.} 
    \label{fig:interface}
\end{minipage}
\vspace{1em}

Unlike $\mathcal{F}_{\text{client}}$, which captures the interaction logic common to all L2 protocol types, the subroutine functionalities may vary depending on the underlying scheme employed by a specific protocol. In what follows, we provide a general description of how these subroutines operate based on the requests they receive. We note that all ideal subroutines are assumed to be incorruptible.

\textbf{Input request communication.} The environment $\mathcal{E}$ may instruct an honest party to propose different types of execution requests during its participation in the protocol. Three types of requests are accepted: (1) protocol joining (channel opening), (2) state transition, and (3) state settlement. Upon receiving such a request, $\mathcal{F}_{\text{client}}$ forwards it to the subroutine $\mathcal{F}_{\text{submit}}$, which verifies whether the submitted request is valid and falls within one of the three accepted categories. Based on these checks, $\mathcal{F}_{\text{submit}}$ returns a boolean value indicating whether the submission is accepted. $\mathcal{F}_{\text{client}}$ then updates $\mathsf{requestQueue}$ in $\mathsf{internalState}$ accordingly.

\textbf{Protocol joining (channel opening).} The parameter subroutine $\mathcal{F}_{\text{join}}$ checks with the underlying L1 blockchain $\mathcal{F}_{\text{ledger}}$ whether the requirements for joining a protocol or opening a channel are satisfied. If the checks pass, $\mathcal{F}_{\text{client}}$ updates $\mathsf{requestQueue}$, $\mathsf{executedReq}\\\mathsf{uest}$, $\mathsf{onchainState}$, and $\mathsf{identities}$ in $\mathsf{internalState}$ accordingly, and the corresponding output is delivered to the environment $\mathcal{E}$. To capture adversarial behavior such as ignoring requests or tampering with message delivery, we let the simulator initiate the checking procedure by sending a request to $\mathcal{F}_{\text{client}}$ via the NET connection. This request includes an \emph{Attachment} containing the initial state and other auxiliary data required for the subroutine checks and for producing outputs, such as transactions to be committed on-chain.

\textbf{Protocol state update.} Similar to the joining procedure, upon receiving an update request from the simulator through the NET connection, $\mathcal{F}_{\text{client}}$ forwards it to the subroutine $\mathcal{F}_{\text{update}}$ along with the current $\mathsf{internalState}$. $\mathcal{F}_{\text{client}}$ then waits for a response from $\mathcal{F}_{\text{update}}$, which includes the executed requests and the updated state after checking. $\mathcal{F}_{\text{client}}$ subsequently updates $\mathsf{internalState}$: the $\mathsf{stateList}$ is updated according to the execution result, $\mathsf{requestQueue}$ is updated to reflect which previously submitted requests have been executed, and the evidence for execution is recorded in $\mathsf{executedRequest}$. By customizing the checks within $\mathcal{F}_{\text{update}}$, our framework accommodates a wide range of mechanisms used by different L2 protocols, such as consensus-based agreement in sidechain protocols or L1 interactions in rollup protocols.

\textbf{L2 state settlement.} The environment could instruct a client to exit or close the L2 protocol, which requires the final settlement of the L2 state on the L1 blockchain. After the settlement request is accepted through checks in $\mathcal{F}_{\text{submit}}$, the functionality waits for the simulator to trigger the settlement subroutine $\mathcal{F}_{\text{settlement}}$. This subroutine accesses $\mathcal{F}_{\text{ledger}}$ to verify whether the on-chain state published by participants is consistent with $\mathsf{internalState}$. If the settlement is successful, $\mathcal{F}_{\text{client}}$ updates $\mathsf{onchainState}$ in $\mathsf{internalState}$ and generates the corresponding output to the environment $\mathcal{E}$.

\textbf{Reading L2 information.} a L2 protocol should support read requests that allow clients to query the necessary L2 information, such as participants' states and the executed requests related to state transitions. Upon receiving a read request, $\mathcal{F}_{\text{client}}$ forwards it to the subroutine $\mathcal{F}_{\text{read}}$ along with the identity of the requester and the current $\mathsf{internalState}$. $\mathcal{F}_{\text{read}}$ then determines which information from $\mathsf{onchainState}$, $\mathsf{stateList}$ and $\mathsf{executedRequest}$ the requester is able to access, based on the scheme employed by the protocol and the adversarial influence specified in $\mathcal{F}_{\text{read}}$.

\textbf{L2 clock simulation.} To capture time-related assumptions in L2 protocols, such as bounded message delivery in synchronous off-chain communication or periodic L1 publication by rollup operators, our framework introduces $\mathcal{F}_{\text{updRnd}}$, which simulates an internal clock for the L2 protocol. We note that $\mathcal{F}_{\text{updRnd}}$ can also be defined to access $\mathcal{F}_{\text{ledger}}$ when certain timing requirements are related to L1 blockchain like block number.

\textbf{Information leakage.} For completeness, we introduce the subroutine $\mathcal{F}_{\text{leak}}$ to capture information leaked to the adversary upon corruption, as defined in \textbf{Corruption behavior}. This subroutine is primarily relevant for analyzing protocols with specific privacy properties, which typically rely on dedicated cryptographic primitives rather than the functional logic across L2 designs that is our focus. We therefore assume that parties \emph{have no secret inputs} and receive commands only from the higher-level protocol or the environment, so that the simulation challenge reduces to handling malicious deviations rather than hidden information. Under this assumption, all protocol messages and requests are transmitted in plaintext, and $\mathcal{F}_{\text{leak}}$ simply leaks all information to the simulator; it is therefore not specifically instantiated in our case studies.

\subsection{Security Properties}
\label{sec:def}

While the subroutine functionalities of our framework can be customized to cover the diversity of L2 protocols, any ideal functionality realized by an L2 protocol must enforce a set of fundamental security properties through these subroutines. In what follows, we formally define these properties with the ideal functionality subroutines from the client's perspective, independent of the real-world underlying scheme employed by the protocol.

\textbf{Correct L2 initialization.} A fundamental prerequisite for the subsequent execution of an L2 protocol is the correct initialization of the L2 protocol state after a client requests to join the protocol. In particular, the checks within the subroutine $\mathcal{F}_{\text{join}}$ should guarantee two properties: (1) the initial state is proposed by all relevant honest clients, and (2) a corresponding transaction and state have been committed on the L1 blockchain. Assuming the environment instructs an honest client $c_h$ to join the protocol with an initial state $s_{\text{init}}$, we define the following security property to formally capture the guarantees required for $\mathcal{F}_{\text{join}}$:

\begin{definition}[Correct L2 Initialization]
An ideal functionality $\mathcal{F}^{\Lambda}_{\text{layer2}}$ for an L2 protocol $\Lambda$ built on a secure L1 blockchain $\mathcal{F}_{\text{ledger}}$ satisfies \emph{correct L2 initialization} if its subroutine $\mathcal{F}^{\Lambda}_{\text{join}}$ outputs $\{\text{true}, s_{\text{init}}\}$ to an honest client $c_h$ only when:
\begin{itemize}
  \item The initial L2 state $s_{\text{init}}$ was previously submitted as a request by an honest client $c_h$.
  \item The initial L2 state $s_{\text{init}}$ is committed as a valid state in the output of $\mathcal{F}_{\text{ledger}}$ through the \emph{ReadL1} request.
\end{itemize}
\end{definition}

\textbf{$f$-safety.}
The \emph{safety} of an L2 protocol ensures that, under certain corruption, honest clients observe a consistent order of correctly executed valid requests. The corruption tolerance $f$ is intrinsically a parameter of the real-world realization $\Lambda$ and hard to capture in ideal functionality. However, our framework compositionally separates safety into two sources, using the iUC composition theorem rather than per-protocol arguments. \emph{L1-based safety} is inherited from blockchain $\mathcal{F}_{\text{ledger}}$ under L1 participants corruption, denoted as $f_{L_1}$. \emph{L2-internal safety} is enforced by $\mathcal{F}^{\Lambda}_{\text{update}}$, which ensures each new state and executed-request is consistent with the previous one (e.g.,~via sequence numbers or block pointers) before $\mathcal{F}^{\Lambda}_{\text{read}}$ derives the read result; this holds under L2-side corruption $f_{L_2}$. Letting $\mathcal{R}^{c_h}_r$ denote the read output of correctly executed requests for honest client $c_h$ at round $r$, we define $f$-safety as follows.

\begin{definition}[$f$-Safety]
\label{def:safety}
A real-world L2 protocol $\Lambda$ that iUC-realises an ideal functionality $\mathcal{F}^{\Lambda}_{\text{layer2}}$ over $\mathcal{F}_{\text{ledger}}$ satisfies \emph{$f$-safety} if, under at most $f$ corrupted participants (whether at L1 or L2), for every honest client $c_h$ the correct read output $\mathcal{R}^{c_h}_{r}$ returned at round $r$ by $\mathcal{F}^{\Lambda}_{\text{read}}$ satisfies the following common-prefix property, denoted by $\preceq$:
\begin{itemize}
  \item \textbf{Self-consistency:} For any honest client $c_h$ and any rounds $r_1 \le r_2$, $\mathcal{R}^{c_h}_{r_1} \preceq \mathcal{R}^{c_h}_{r_2}$.
  \item \textbf{View-consistency:} For any two honest clients $c_{h1}, c_{h2}$ and any round $r$, $\mathcal{R}^{c_{h1}}_{r} \preceq \mathcal{R}^{c_{h2}}_{r}$ or $\mathcal{R}^{c_{h2}}_{r} \preceq \mathcal{R}^{c_{h1}}_{r}$.
\end{itemize}
\end{definition}

 \textbf{Correct L2 settlement.} As the final step of a client's life cycle in an L2 protocol, \emph{correct L2 settlement} must be ensured to guarantee that the latest L2 state is eventually committed to the L1 blockchain. In our framework, the completion of the settlement procedure is indicated by the output $\{\text{true}, s^{c_h}_n\}$ from the subroutine $\mathcal{F}_{\text{settlement}}$, where $s^{c_h}_n$ denotes the latest state of an honest client $c_h$ as recorded in $\mathsf{stateList}$ at the time of settlement. We define this property as follows:
\begin{definition}[Correct L2 Settlement]
\label{def:LSC}
An ideal functionality $\mathcal{F}^{\Lambda}_{\text{layer2}}$ for an L2 protocol $\Lambda$ built on a secure L1 blockchain $\mathcal{F}_{\text{ledger}}$ satisfies \emph{correct L2 settlement} if its subroutine $\mathcal{F}^{\Lambda}_{\text{settlement}}$ outputs $\{\text{true}, s^{c_h}_n\}$ to an honest client $c_h$ only when:
\begin{itemize}
  \item The state $s^{c_h}_n$ is the latest state for $c_h$ recorded in $\mathsf{stateList}$.
  \item The state $s^{c_h}_n$ is committed as a valid state in the output of $\mathcal{F}_{\text{ledger}}$ through a \emph{ReadL1} request.
\end{itemize}
\end{definition}

\textbf{$\{f, T\}$-liveness.}
The \emph{liveness} of an L2 protocol ensures that, under bounded corruption, every valid request submitted by an honest client is executed and made visible via a read within a latency bound $T$. As with safety, $f$ cannot be fully captured by the ideal functionality, and the additional difficulty for liveness is that the L2 inner clock update and message delivery are also adversary-scheduled. Our framework distinguishes two sources of the latency bound. \emph{L2-internal latency} $T_{L_2}$ is the aggregate latency of L2 execution and off-chain communication. We use it to represent both a bounded delay under synchrony, or fully controlled by the adversary under asynchrony. \emph{L1-based latency} $T_{L_1}$, the aggregate latency of submitting and committing a transaction on $\mathcal{F}_{\text{ledger}}$ and having it observed by other honest participants. If a fraud-proof scheme is used as the L1 commitment rule for certain transactions, the challenge period \(T_{\text{challenge}}\) is added to \(T_{L_1}\). A single operation may incur one or more such interactions, so we write $T$ as a sum.  $\mathcal{F}^{\Lambda}_{\text{read}}$ and $\mathcal{F}^{\Lambda}_{\text{updRnd}}$ ensure timely visibility of executed requests once these bounds are met. We define $\{f, T\}$-liveness as follows.

\begin{definition}[$\{f, T\}$-Liveness]
\label{def:liveness}
A real-world L2 protocol $\Lambda$ that iUC-realises an ideal functionality $\mathcal{F}^{\Lambda}_{\text{layer2}}$ over $\mathcal{F}_{\text{ledger}}$ satisfies $\{f, T\}$-\emph{liveness}, with $T$ expressed in terms of  $T_{L_1}$, $T_{L_2}$, and optionally $T_{\text{challenge}}$, if, under at most $f$ corrupted participants (whether at L1 or L2), for any valid request $q$ submitted by an honest client $c_h$ at time $t$, the subroutines $\mathcal{F}^{\Lambda}_{\text{updRnd}}$ and $\mathcal{F}^{\Lambda}_{\text{read}}$ jointly guarantee that at any time $t' \geq t + T$, $q \in \mathsf{executedRequest}$ or corresponding output is returned to $c_h$.
\end{definition}

After formally defining the security properties, we can propose the definition for a secure L2 protocol:

\begin{definition}[Secure L2 protocol]
    a L2 protocol $\Lambda$ is secure if and only if its real-world protocol $\mathcal{P}^{\Lambda}$ iUC-realizes the ideal functionality 
$\mathcal{F}^{\Lambda}_{\text{layer2}} = (\mathcal{F}_{\text{client}} \allowbreak, \mathcal{F}_{\text{ledger}} \mid \allowbreak
\mathcal{F}^{\Lambda}_{\text{join}}, \allowbreak
\mathcal{F}^{\Lambda}_{\text{submit}}, \allowbreak
\mathcal{F}^{\Lambda}_{\text{update}}, \allowbreak
\mathcal{F}^{\Lambda}_{\text{read}}, \allowbreak
\mathcal{F}^{\Lambda}_{\text{settlement}}, \allowbreak
\mathcal{F}^{\Lambda}_{\text{updRnd}})$, whose subroutines satisfy \emph{correct initialization}, \emph{safety}, \emph{correct L2 settlement} and \emph{liveness}. 
\end{definition}

\subsection{Efficiency Property}




While different protocols may realize the same security properties, they can differ significantly in efficiency: storage, communication complexity, etc. We focus on \emph{$(G_{L_2}, G_{L_1})$-data availability}, a metric capturing the \emph{minimal} L2 and L1 storage required to maintain security properties. Let $f_{st}(s, e) = s'$ denote the deterministic state-transition function deriving $s'$ from a prior state $s$ and executed requests $e$. Under data-availability constraints, $\mathcal{F}^{\Lambda}_{\text{read}}$ must return enough information to reconstruct the current state: the initial L1-committed state $s_{L1}$, the current state $s'$ at round $r$, and executed requests with evidence $e$ such that $f_{st}(s_{L1}, e) = s'$. We formalize this property as follows:

\begin{definition}[$(G_{L_2}, G_{L_1})$-Data Availability]
An L2 protocol $\Lambda$ built on a secure L1 blockchain $\mathcal{F}_{\text{ledger}}$ realizes \emph{$(G_{L_2}, G_{L_1})$-data availability} if, for every honest client $c_h$ at every round $r$, the data set $\{e, s_{L1}, s'\}$ is accessible via $\mathcal{F}^{\Lambda}_{\text{read}}$ and satisfies $f_{st}(s_{L1}, e) = s'$, while the storage required on L2 and L1 is lower-bounded by $G_{L_2}$ and $G_{L_1}$, respectively.
\end{definition}




\section{Case Studies in L2 Protocols}
\label{sec:casestudy}

After defining secure L2 protocols and their properties, we demonstrate how our framework captures three major L2 paradigms, through case studies of Brick, Liquid, and Arbitrum Nitro (full proofs in Appendix~\ref{apdx:brick}--\ref{apdx:arbitrum}). Each proof has two steps: first, we prove that the according framework-based ideal functionality realizes all four security properties via the checks of its parameter subroutines; second, the real-world protocol iUC-realizes this ideal functionality through seven successive game hops, starting from a dummy functionality that merely forwards messages and progressively migrating each subroutine's behaviour from the simulator into the functionality. Under the plaintext-leakage assumption, indistinguishability for $\mathcal{E}$ at each game hop reduces to showing that the adversary cannot induce divergent observable outputs through dynamic corruption or message-delivery scheduling.

\subsection{Case Study: The Brick Channel}

\subsubsection{The Real Brick Protocol $\mathcal{P}^{\text{Brick}}$}

The Brick channel~\cite{avarikioti2021b} is a two-party payment channel designed for \emph{asynchronous} off-chain networks. Two parties transact through the channel and serve as both \emph{clients} and \emph{operators} (used interchangeably), with at least one assumed honest. A set of third-party \emph{wardens} verifies and stores state updates and assists in unilateral settlement on L1, with more than $2/3$ assumed honest. The protocol operates as follows:

\textbf{Channel opening.} The two clients agree on the initial state, send the agreement to preknown wardens, and altogether deposit collateral on L1. The channel opens once the required deposits are committed on-chain.

\textbf{Channel update.} The clients exchange and sign the new state together with an incremented sequence number, then send the signed state to the wardens. A client commits the update after receiving signatures from at least $2/3$ of the wardens.

\textbf{Channel settlement.} A client can settle on L1 in two ways. \emph{Collaborative:} both clients sign a settlement transaction containing the final state and publish it on L1. \emph{Unilateral:} upon observing an on-chain request, wardens publish their stored latest states; once at least $2/3$ have done so, the channel closes at the state with the highest sequence number.

We define the real-world Brick protocol as $\mathcal{P}^{\text{Brick}} := (\mathcal{P}^{\text{Brick}}_{\text{client}},\\ \mathcal{F}_{\text{ledger}} \mid \mathcal{P}^{\text{Brick}}_{\text{warden}}, \mathcal{F}_{\text{sig}}, \mathcal{F}^{\text{Brick}}_{\text{com}})$, where $\mathcal{P}^{\text{Brick}}_{\text{client}}$ implements client behavior and connects to the environment, $\mathcal{P}^{\text{Brick}}_{\text{warden}}$ implements warden behavior without environment access, and $\mathcal{F}^{\text{Brick}}_{\text{com}}$ captures the assumed asynchronous network.
\subsubsection{The Ideal Functionality $\mathcal{F}^{\text{Brick}}_{\text{layer2}}$}

After showing the real-world protocol of the Brick channel, we then propose the formal definition for the ideal functionality $\mathcal{F}^{\text{Brick}}_{\text{layer2}} := (\mathcal{F}_{\text{client}} \allowbreak, \mathcal{F}_{\text{ledger}} \mid \allowbreak \space 
\mathcal{F}^{\text{Brick}}_{\text{submit}}, \allowbreak 
\mathcal{F}^{\text{Brick}}_{\text{join}}, \allowbreak
\mathcal{F}^{\text{Brick}}_{\text{update}}, \allowbreak
\mathcal{F}^{\text{Brick}}_{\text{read}}, \allowbreak
\mathcal{F}^{\text{Brick}}_{\text{settlement}}, \allowbreak
\mathcal{F}^{\text{Brick}}_{\text{updRnd}})$. To capture the intended security properties of the Brick payment channel, the subroutines are instantiated as follows: $\mathcal{F}^{\text{Brick}}_{\text{submit}}$ accepts only three types of valid requests, namely channel opening, state update, and settlement; $\mathcal{F}^{\text{Brick}}_{\text{join}}$ checks that an opening request from every honest client has been recorded in $\mathsf{requestQueue}$ and that the collateral deposit transaction and initial state has been committed on L1; $\mathcal{F}^{\text{Brick}}_{\text{update}}$ verifies that each state update proposed by $\mathcal{S}$ corresponds to requests from all honest clients recorded in $\mathsf{requestQueue}$ and that the proposed state value and sequence number preserve consistency; $\mathcal{F}^{\text{Brick}}_{\text{read}}$ queries $\mathcal{S}$ via NET to simulate asynchronous delivery and returns read results derived from $\mathsf{internalState}$, including $\mathsf{executedRequest}$, $\mathsf{stateList}$, and $\mathsf{onchainState}$; $\mathcal{F}^{\text{Brick}}_{\text{settlement}}$ produces a successful output only if either (1) the honest clients have proposed the settlement state (as recorded in $\mathsf{requestQueue}$) and the collaborative settlement transaction has been committed on L1, or (2) the unilateral settlement transaction and the latest state has been committed; and $\mathcal{F}^{\text{Brick}}_{\text{updRnd}}$ imposes no additional checks on round updates under the asynchronous communication, except during unilateral closing, which is tied to the L1 blockchain. The following conclusion holds for the ideal functionality $\mathcal{F}^{\text{Brick}}_{\text{layer2}}$:

\begin{restatable}{theorem}
{ThmidealBrick}
    The ideal functionality $\mathcal{F}^{\text{Brick}}_{\text{layer2}}$ guarantees all the security properties of a secure L2 protocol.
\end{restatable}

The security of the real Brick payment channel protocol can then be demonstrated by proving that the real protocol iUC-realizes the ideal functionality. This can be formally defined as follows:

\begin{restatable}{theorem}
{ThmrealizeBrick}
    With existence of ideal functionalities $\mathcal{F}_{\text{ledger}}$, $\mathcal{F}_{\text{sig}}$ and $\mathcal{F}^{\text{Brick}}_{\text{com}}$. The real Brick payment channel protocol $\mathcal{P}^{\text{Brick}}$ iUC-realizes the ideal Brick payment channel functionality $\mathcal{F}^{\text{Brick}}_{\text{layer2}}$.
\end{restatable}

\subsection{Case Study: The Liquid Sidechain} 

\subsubsection{The Real Liquid Protocol $\mathcal{P}^{\text{Liquid}}$}

The Liquid sidechain~\cite{nick2020liquid} is a Bitcoin sidechain operating over a \emph{synchronous} network with two participant classes: \emph{clients}, who receive environment instructions and submit requests, and \emph{operators}. The original design separates operators into block signers (producing sidechain blocks) and watchmen (creating peg-out transactions on L1). For simplicity, we unify them as operators, of which more than $2/3$ are assumed honest. We also abstract away complementary features of the original protocol, such as privacy preservation, and focus on its core functional logic, described from the client's perspective as follows:

\textbf{Sidechain joining.} The client first sends collateral to a public deposit address controlled by the operator federation on L1. Once the collateral is committed, the client submits a peg-in transaction with its initial state to the operators. The client is considered joined once this peg-in is recorded on the sidechain.

\textbf{Sidechain update.} The client submits transactions to the operators, who run a three-phase BFT protocol. A block, which also references the previously committed sidechain block, becomes valid once it obtains a quorum certificate, i.e., signatures from more than $2/3$ of the operators. Once the client's transaction is included in such a certified block, it is executed and the client's state is updated accordingly.

\textbf{Sidechain leaving.} The client submits a peg-out request to the operators. Once this request is recorded on the sidechain, the operators post the corresponding settlement transaction on L1 and execute the transfer according to the latest state, thereby completing the exit procedure.

We define the real-world protocol as $\mathcal{P}^{\text{Liquid}} := (\mathcal{P}^{\text{Liquid}}_{\text{client}}, \mathcal{F}_{\text{ledger}} \mid \mathcal{P}^{\text{Liquid}}_{\text{operator}}, \mathcal{F}_{\text{sig}}, \mathcal{F}^{\text{Liquid}}_{\text{com}})$, where $\mathcal{P}^{\text{Liquid}}_{\text{client}}$ interfaces with the environment and processes L2 requests, $\mathcal{P}^{\text{Liquid}}_{\text{operator}}$ implements the consensus and settlement logic, and $\mathcal{F}^{\text{Liquid}}_{\text{com}}$ captures the assumed synchronous off-chain network.
\subsubsection{The Ideal Functionality $\mathcal{F}^{\text{Liquid}}_{\text{layer2}}$}

After presenting the real-world protocol of the Liquid sidechain, we propose the ideal functionality based on framework $\mathcal{F}^{\text{Liquid}}_{\text{layer2}} := (\mathcal{F}_{\text{client}}, \mathcal{F}_{\text{ledger}} \allowbreak \mid \allowbreak \mathcal{F}^{\text{Liquid}}_{\text{submit}}, \allowbreak \mathcal{F}^{\text{Liquid}}_{\text{join}}, \allowbreak \mathcal{F}^{\text{Liquid}}_{\text{update}}, \allowbreak \mathcal{F}^{\text{Liquid}}_{\text{read}}, \allowbreak \mathcal{F}^{\text{Liquid}}_{\text{settlement}}, \allowbreak \mathcal{F}^{\text{Liquid}}_{\text{updRnd}})$. The subroutines are instantiated as follows: $\mathcal{F}^{\text{Liquid}}_{\text{submit}}$ accepts only three types of valid requests: join, state update (transaction), and settlement; $\mathcal{F}^{\text{Liquid}}_{\text{join}}$ approves a join request only if a corresponding peg-in transaction appears in $\mathsf{executedRequest}$ and the collateral deposit transaction has been committed on L1; $\mathcal{F}^{\text{Liquid}}_{\text{update}}$ accepts a newly proposed updated block only if the requests and state included in the block are executed correctly, the requests from honest clients are recorded in $\mathsf{requestQueue}$, and proposed block includes the previously committed block that is recorded in $\mathsf{executedRequest}$, thereby guaranteeing consistency; $\mathcal{F}^{\text{Liquid}}_{\text{read}}$ returns the corresponding state and executed requests derived from $\mathsf{internalState}$, including $\mathsf{executedRequest}$, $\mathsf{stateList}$, and $\mathsf{onchainState}$, in accordance with the behavior of a synchronous network; $\mathcal{F}^{\text{Liquid}}_{\text{settlement}}$ produces a successful output only if a peg-out request carrying the latest state, according to $\mathsf{stateList}$, has been recorded in $\mathsf{executedRequest}$ and the corresponding settlement transaction has been committed on L1; Finally, $\mathcal{F}^{\text{Liquid}}_{\text{updRnd}}$ refuses to advance the round whenever a request in $\mathsf{requestQueue}$ has been pending for more than the $T_{L_2}$ type latency $(3f{+}4)\delta$, where $\delta$ is the per-round delivery bound for synchronous off-chain communication. The following result holds for the ideal functionality:

\begin{restatable}{theorem}
{ThmidealLiquid}
        The ideal functionality $\mathcal{F}^{\text{Liquid}}_{\text{layer2}}$ guarantees all the security properties of a secure L2 protocol.
\end{restatable}

The security of the Liquid sidechain protocol can be demonstrated by proving that the real protocol iUC-realizes the ideal functionality. This can be formally defined as follows:

\begin{restatable}{theorem}
{ThmrealizeLiquid}
    With existence of ideal functionalities $\mathcal{F}_{\text{ledger}}$, $\mathcal{F}_{\text{sig}}$ and $\mathcal{F}^{\text{Liquid}}_{\text{com}}$. The real Liquid sidechain protocol $\mathcal{P}^{\text{Liquid}}$ iUC-realizes the ideal Liquid sidechain functionality $\mathcal{F}^{\text{Liquid}}_{\text{layer2}}$.
\end{restatable}

\subsection{Case Study: The Arbitrum Nitro Rollup}
\label{sec:arbitrum}







Arbitrum Nitro~\cite{bousfield2022arbitrum} is an optimistic rollup protocol, here we focus on its core protocol roles. \emph{Clients} submit transactions to the operator for execution. To mitigate censorship by the operator, clients may alternatively publish transactions directly to the L1 blockchain to force their inclusion. \emph{Operators} order and post transaction batches together with the resulting execution state to L1, without requiring re-execution by the L1 blockchain itself. \emph{Verifiers} check the operator-published state on L1, and at least one honest verifier is assumed. In the analysis, we let clients also act as verifiers; therefore, at least one honest verifier is always holds. The off-chain communication is assumed to be \emph{asynchronous}. The core logic of the Arbitrum Nitro protocol works as follows:

\textbf{Rollup Joining.}
The client first deposits coins on L1 through a deposit transaction and then submits a peg-in transaction to the operator. Upon verifying that the deposit transaction has been committed on the L1 blockchain, the operator includes the peg-in transaction in the next batch to be published and generates the initial state. The client is considered to have joined once this request is included in the operator's published transaction batch and no valid fraud-proof transaction from the verifier is submitted during the challenge period.

\textbf{Rollup Update.}
At regular time period, the operator batches and executes pending client transactions, then posts both the transactions and the resulting state to L1. The resulting state is also published in the transaction data space (e.g., as a blob) and is not re-executed by L1. Verifiers independently re-execute the batch to check its correctness. A valid fraud proof transaction will be published on L1 within the challenge period to invalidate the update, whereas if no such proof appears, the batch and the resulting state are considered committed. Clients may also submit transactions directly to L1 to force their execution when censored by the operator, which also needs to be verified by the verifier.

\textbf{Rollup Leaving.}
The client submits a peg-out transaction to the operator. Once the peg-out transaction has been executed through the update procedure, and the client's rollup state is committed on L1, the client is considered as leaving and can withdraw the coin whenever it wants. Alternatively, a client may leave the protocol independently via the \emph{escape hatch} mechanism, by publishing the peg-out transaction directly on L1 with the latest state, which must also be validated by the verifiers.

Based on the protocol description above, we can define the real-world protocol as 
$\mathcal{P}^{\text{Arbitrum}} := (\mathcal{P}^{\text{Arbitrum}}_{\text{client}}: \text{client},\mathcal{F}_{\text{ledger}}\allowbreak \space \mid \allowbreak \space  \mathcal{P}^{\text{Arbitrum}}_{\text{operator}},\allowbreak \space \mathcal{P}^{\text{Arbitrum}}_{\text{client}}: \text{verifier}, \allowbreak\mathcal{F}_{\text{sig}},\allowbreak \mathcal{F}^{\text{Arbitrum}}_{\text{com}})$, where $\mathcal{P}^{\text{Arbitrum}}_{\text{client}}$ represents the interaction interface with the environment and processes requests to the rollup protocol. $\mathcal{P}^{\text{Arbitrum}}_{\text{operator}}$ includes the machine code that implements the core logic for the operator. $\mathcal{F}^{\text{Arbitrum}}_{\text{com}}$ captures the asynchronous off-chain network assumed in protocol and also models the periodic behavior of operators and verifiers, including the batch publishing on the L1 blockchain by operators and the subsequent verification of these updates by verifiers defined by $\mathcal{P}^{\text{Arbitrum}}_{\text{verifier}}$.

\subsubsection{The Ideal Functionality $\mathcal{F}^{\text{Arbitrum}}_{\text{layer2}}$}

After showing the real-world protocol of the Arbitrum Nitro rollup, we then propose the ideal functionality based on framework $\mathcal{F}^{\text{Arbitrum}}_{\text{layer2}} := (\mathcal{F}_{\text{client}},\mathcal{F}_{\text{ledger}} \allowbreak \space \mid \allowbreak \space 
\mathcal{F}^{\text{Arbitrum}}_{\text{submit}}, \allowbreak 
\mathcal{F}^{\text{Arbitrum}}_{\text{join}}, \allowbreak
\mathcal{F}^{\text{Arbitrum}}_{\text{update}}, \allowbreak
\mathcal{F}^{\text{Arbitrum}}_{\text{read}}, \allowbreak
\mathcal{F}^{\text{Arbitrum}}_{\text{settlement}}, \allowbreak
\mathcal{F}^{\text{Arbitrum}}_{\text{updRnd}})$. 

 For Arbitrum Nitro, $\mathcal{F}^{\text{Arbitrum}}_{\text{submit}}$ accepts only three types of valid requests, namely join, state update (transaction), and settlement. $\mathcal{F}^{\text{Arbitrum}}_{\text{join}}$ approves a join request only if an operator-published batch on L1 includes the corresponding peg-in transaction and the collateral deposit transaction has been committed on L1. $\mathcal{F}^{\text{Arbitrum}}_{\text{update}}$ validates that the executed transaction batch and the resulting state committed on L1 are correctly computed, and that the transactions from honest clients included in the batch are recorded in $\mathsf{requestQueue}$. $\mathcal{F}^{\text{Arbitrum}}_{\text{read}}$ returns the portion of $\mathsf{internalState}$ that is consistent with the read result obtained directly from $\mathcal{F}_{\text{ledger}}$, thereby ensuring the consistency with L1. $\mathcal{F}^{\text{Arbitrum}}_{\text{settlement}}$ produces a successful output only if the peg-out transaction has been executed and committed on L1 through the operator-published batch, and a corresponding settlement transaction has also been committed on L1. Both transactions must be consistent with the client's latest state recorded in $\mathsf{stateList}$. Finally, $\mathcal{F}^{\text{Arbitrum}}_{\text{updRnd}}$ imposes no additional checks for the L2 communication network, but it checks with L1 for force-inclusion client requests such as the escape hatch. The following conclusion holds for the ideal functionality:

\begin{restatable}{theorem}
{ThmidealArbitrum}
        The ideal functionality $\mathcal{F}^{\text{Arbitrum}}_{\text{layer2}}$ guarantees all the security properties of a secure L2 protocol.
\end{restatable}

The security of the Arbitrum Nitro rollup protocol can be demonstrated by proving that the real protocol iUC-realizes the ideal functionality. This can be formally defined as follows:

\begin{restatable}{theorem}
{ThmrealizeArbitrum}
    With the existence of ideal functionalities $\mathcal{F}_{\text{ledger}}$, $\mathcal{F}_{\text{sig}}$ and $\mathcal{F}^{\text{Arbitrum}}_{\text{com}}$. Then, the real Arbitrum protocol $\mathcal{P}^{\text{Arbitrum}}$ iUC-realizes the ideal Arbitrum Nitro rollup functionality $\mathcal{F}^{\text{Arbitrum}}_{\text{layer2}}$.
\end{restatable}

\section{Comparative Analysis of L2 Designs}
\label{sec:insight}

In this section, we analyze the design differences and trade-offs among L2 protocols as revealed by instantiating our framework for representative constructions. From the case studies showing how each real-world protocol realizes the corresponding ideal functionality based on our framework, we identify the fundamental differences in parameter subroutines that influence protocol behavior and demonstrate how these differences impact both security guarantees and efficiency properties. Throughout this section, we generalize payment channels to payment channel factories (PCFs), which follow the same operational logic based on unanimous agreement but support a larger set of clients~\cite{burchert2018scalable,avarikioti2018towards}. Omitted proofs are provided in Appendix~\ref{apdx:proof}.

\subsection{Subroutine Comparison}

As observed in the case studies above, the core mechanism by which the framework captures the properties of L2 protocols is through triggering subroutines to determine whether the protocol state, represented by the $\mathsf{internalState}$, should be updated based on protocol-specific requirements. Here we analyze how these subroutines define the structural differences among various types of L2 protocols, thereby revealing the key design choices that differentiate them. Specifically, we focus on five subroutines: $\mathcal{F}_{\text{join}}$, $\mathcal{F}_{\text{submit}}$, $\mathcal{F}_{\text{update}}$, $\mathcal{F}_{\text{read}}$, and $\mathcal{F}_{\text{settlement}}$, as these encapsulate the essential procedures involved in the lifecycle of a L2 protocol. The other two subroutines, $\mathcal{F}_{\text{leak}}$ and $\mathcal{F}_{\text{updRnd}}$, may also vary between protocols. However, their differences primarily stem from underlying assumptions, such as the communication model or cryptographic schemes used. 
These aspects are not directly tied to the fundamental classification of L2 protocol types. 


\begin{table*}[htp]
\caption{Content differences in the main parameter subroutines across L2 protocol types.}
\label{tab:comparison}
\resizebox{\textwidth}{!}{%
\begin{tabular}{|c|c|c|c|c|c|}
\hline
 &
  $\mathcal{F}_{\text{submit}}$ &
  $\mathcal{F}_{\text{join}}$ &
  $\mathcal{F}_{\text{update}}$ &
  $\mathcal{F}_{\text{read}}$ &
  $\mathcal{F}_{\text{settlement}}$ \\ \hline
PCF &
  \begin{tabular}[c]{@{}c@{}}Request state is\\ joint state of all clients\end{tabular} &
  \begin{tabular}[c]{@{}c@{}}All honest clients' requests\\ in $\mathsf{requestQueue}$\\ +\\ L1 commitment check\end{tabular} &
  \begin{tabular}[c]{@{}c@{}}All honest clients' request\\ proposals in $\mathsf{requestQueue}$\\ +\\ off-chain consistency\\ (e.g., sequence number)\end{tabular} &
  \multirow{2}{*}{\begin{tabular}[c]{@{}c@{}}Derived from $\mathsf{internalState}$\\ +\\ According with off-chain communication\\ network delay\end{tabular}} &
  \begin{tabular}[c]{@{}c@{}}Collaborative: all honest clients'\\ requests in $\mathsf{requestQueue}$\\ Unilateral: proposer's request\\ in $\mathsf{requestQueue}$\\ +\\ L1 commitment check\end{tabular} \\ \cline{1-4} \cline{6-6} 
Sidechain &
  \multirow{2}{*}{\begin{tabular}[c]{@{}c@{}}Request state is\\ proposer clients's own state\end{tabular}} &
  \multirow{2}{*}{\begin{tabular}[c]{@{}c@{}}Proposer's request\\ in $\mathsf{executedRequest}$\\ +\\ L1 commitment check\end{tabular}} &
  \begin{tabular}[c]{@{}c@{}}Proposer client's individual request\\ proposals in $\mathsf{requestQueue}$\\ +\\ off-chain consistency\\ (e.g., block  predecessor link)\end{tabular} &
   &
  \multirow{2}{*}{\begin{tabular}[c]{@{}c@{}}Proposer client's request\\ in $\mathsf{executedRequest}$\\ +\\ L1 commitment check\end{tabular}} \\ \cline{1-1} \cline{4-5}
Rollup &
   &
   &
  \begin{tabular}[c]{@{}c@{}}Proposer client's individual request\\ proposals in $\mathsf{requestQueue}$\\ +\\ on-chain consistency\\ (e.g., read from L1)\end{tabular} &
  \begin{tabular}[c]{@{}c@{}}Derived from $\mathsf{internalState}$\\ +\\ According with L1 read result\end{tabular} &
   \\ \hline
\end{tabular}%
}
\end{table*}

\subsubsection{The Subroutine Content} The conditions under which each subroutine generates a positive output are what distinguish L2 protocols from one another, as summarized in Table~\ref{tab:comparison}. In what follows, we describe how these conditions vary across the three types of L2 protocols.

\textbf{Request submission $\mathcal{F}_{\text{submit}}$.} The subroutine $\mathcal{F}_{\text{submit}}$ filters valid requests to be accepted by the ideal functionality and forwards the corresponding leakage to the simulator. All three types of L2 protocols accept the same three categories of requests: protocol joining, transaction submission, and settlement. However, differences arise in the content of the input requests themselves. As revealed by our case studies, protocol joining and transaction submission requests typically contain either a proposed state $s$ or a transaction $TX$ encoding such a state. In PCF protocols such as Brick, the state $s$ encompasses the state of all clients. In contrast, in sidechain and rollup protocols such as Liquid and Arbitrum, $s$ refers only to the proposing client's own state. This difference stems from the underlying design that in PCFs, each client also acts as an operator, whereas in the other two protocol types the operator is realized by a separate, dedicated role.

\textbf{Protocol joining $\mathcal{F}_{\text{join}}$.} In the ideal functionalities of all three L2 protocol types, the subroutine $\mathcal{F}_{\text{join}}$ validates successful protocol joining requests sent from the simulator to the ideal functionality via the network connection NET, and generates the corresponding output if the checks pass. The common check across all three types verifies that the protocol joining transaction encoding the initial state $s_{\text{init}}$ has been committed on the L1 blockchain. However, the remaining conditions differ across the three types. In PCFs, since clients also assume the role of operators, the opening procedure requires explicit agreement from all honest clients, conveyed through their submitted requests and tracked in $\mathsf{requestQueue}$, which is verified by the PCF's ideal functionality. In contrast, in sidechain and rollup protocols, clients only need to apply for joining with their own individual state,  therefore, only the proposing client's request needs to be checked in $\mathsf{executedRequest}$, which is changed based on $\mathsf{requestQueue}$ by $\mathcal{F}_{\text{update}}$.

\textbf{Request execution $\mathcal{F}_{\text{update}}$.} The subroutine $\mathcal{F}_{\text{update}}$ specifies how the ideal functionality's $\mathsf{internalState}$ is updated to guarantee the \emph{safety} properties of view-consistency and cross-consistency. All three paradigms locally verify the correctness of executed requests; the primary distinction lies in how they establish consistency. As revealed by our case studies, PCFs such as the Brick payment channel enforce consistency by checking the sum of each client's new state together with a monotonically increasing sequence number. In sidechains such as Liquid, consistency is instead ensured by requiring each new block to point to the previous block already committed in $\mathsf{internalState}$. Both approaches achieve \emph{safety} without relying on the L1 blockchain. Rollups, by contrast, perform the same local execution check but additionally require executed transactions and resulting state transitions to be published on the L1 blockchain, inheriting consistency directly from $\mathcal{F}_{\text{ledger}}$; without this interaction, or if the L1 blockchain's integrity is compromised, rollup \emph{safety} cannot be guaranteed. Access to $\mathcal{F}_{\text{ledger}}$ therefore constitutes a \emph{necessary condition} for rollup safety. A distinction similar to that in $\mathcal{F}_{\text{join}}$ also arises here: in PCFs, each new state proposed by the simulator must be verified against the proposals of all honest clients in $\mathsf{requestQueue}$, since the proposed state encodes the joint state of all clients; in sidechains and rollups, only the individual state of each honest client needs to be checked.


\textbf{Data reading $\mathcal{F}_{\text{read}}$.} Although the subroutine $\mathcal{F}_{\text{read}}$ maintains the same core logic across all three L2 protocol types, that is deciding the read result based on $\mathsf{internalState}$, the differences still arise because $\mathsf{internalState}$ is updated differently in each type, as discussed above for $\mathcal{F}_{\text{update}}$. In PCFs and sidechains, where the update procedure is realized off-chain, the read result derived from $\mathsf{internalState}$ is influenced solely by off-chain factors such as the underlying communication network. For instance, the Brick payment channel is designed for an asynchronous network. Therefore the read result could lag behind the latest $\mathsf{internalState}$ due to network delays introduced by the adversary. In rollup protocols such as Arbitrum, by contrast, updates to $\mathsf{internalState}$ are tied to the L1 blockchain. Consequently, the corresponding $\mathcal{F}_{\text{read}}$ must additionally check with the L1 blockchain when deriving the read result, as each client's view of the L1 blockchain may itself differ due to varying network connections to L1.

\textbf{L2 State settlement $\mathcal{F}_{\text{settlement}}$.} The subroutine $\mathcal{F}_{\text{settlement}}$ across all three L2 protocol types shares the common task of verifying that the latest state recorded in $\mathsf{internalState}$ has been committed on the L1 blockchain. However, notable differences remain. PCF protocols such as the Brick payment channel typically support two modes of settlement: \emph{collaborative} and \emph{unilateral}. In the collaborative mode, the settlement requests of all honest clients must be present in $\mathsf{requestQueue}$; this reflects the fact that other clients can influence an honest client's collaborative settlement request. In the unilateral mode, by contrast, only the proposing client's request must be checked in $\mathsf{requestQueue}$. Sidechain and rollup protocols follow a checking logic similar to PCF's unilateral mode, except that verification is performed against $\mathsf{executedRequest}$ rather than $\mathsf{requestQueue}$, since $\mathsf{executedRequest}$ is derived from $\mathsf{requestQueue}$ through the protocol's execution. Moreover, the distinction in how $\mathsf{executedRequest}$ is maintained by $\mathcal{F}_{\text{update}}$ in sidechains and rollups, as discussed earlier, carries over to the settlement procedure, and so does the differing reliance on the L1 blockchain.

\subsubsection{The Subroutine Connection} Additionally to the content differences of individual subroutines, our analysis also reveals insights into the connection relationships among subroutines across different types of L2 protocols. These connection differences are illustrated in Figure~\ref{fig:Difference}.

We begin with \emph{direct I/O caller-subroutine connections}, which are the more visible form of caller-subroutine connection. Rollup protocols are distinct from the other two L2 types in that their subroutines $\mathcal{F}_{\text{update}}$ and $\mathcal{F}_{\text{read}}$ require direct interaction with the underlying L1 blockchain functionality $\mathcal{F}_{\text{ledger}}$. Recall that $\mathcal{F}_{\text{update}}$ specifies the requirements for generating a new valid state. In PCFs and sidechains, consistency is established through off-chain checks, such as monotonically increasing sequence numbers in PCFs or hash-based linkage between consecutive sidechain blocks. Rollups, on the contrary, achieve consistency by requiring the executed requests and the resulting state to be published on the L1 blockchain that has consistency itself. Consequently, participants learn the executed requests and the latest state through $\mathcal{F}_{\text{read}}$ by reading from $\mathcal{F}_{\text{ledger}}$ as well.

Beyond direct connections, subroutines also exhibit \emph{indirect dependencies} through the shared $\mathsf{internalState}$. As noted earlier, the three protocol types differ in how their $\mathcal{F}_{\text{join}}$ and $\mathcal{F}_{\text{settlement}}$ subroutines check $\mathsf{requestQueue}$ and $\mathsf{executedRequest}$, which are $\mathsf{internalState}$ variables maintained and updated by different subroutines $\mathcal{F}_{\text{submit}}$ and $\mathcal{F}_{\text{update}}$ accordingly. This mechanism creates an indirect connection between subroutines that do not communicate through I/O links.

\begin{figure*}[t] 
    \centering
    \begin{subfigure}{0.32\textwidth} 
        \centering
        \includegraphics[width=\linewidth]{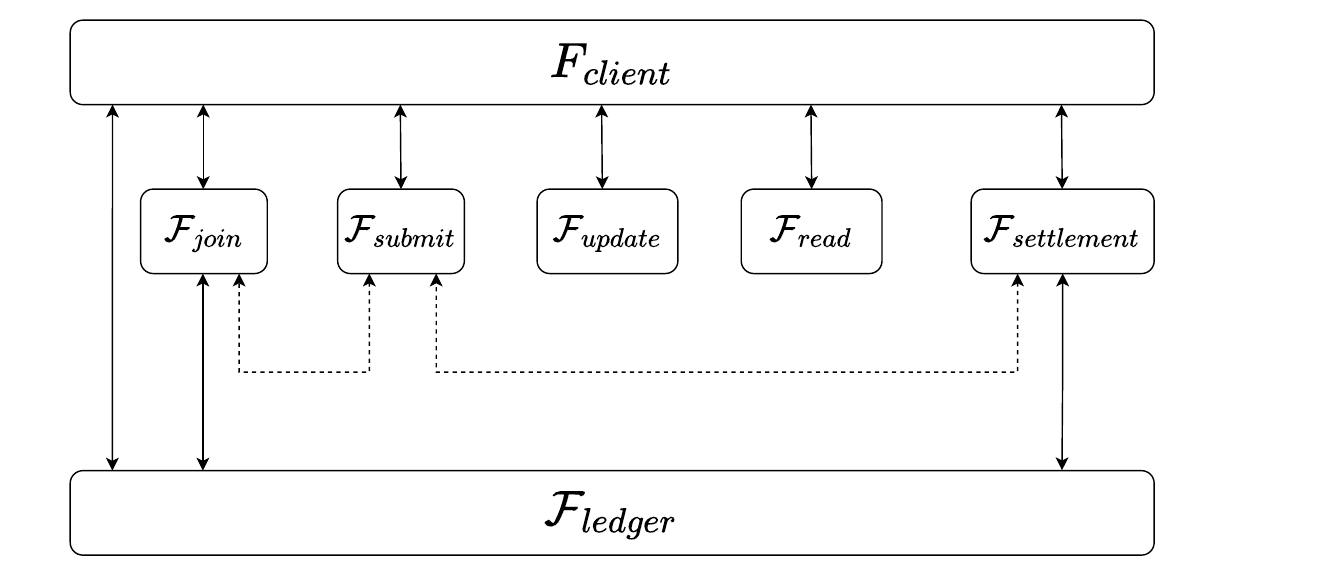} 
        \caption{PCF}
        \label{fig:PCF}
    \end{subfigure}
    \hfill
    \begin{subfigure}{0.32\textwidth}
        \centering
        \includegraphics[width=\linewidth]{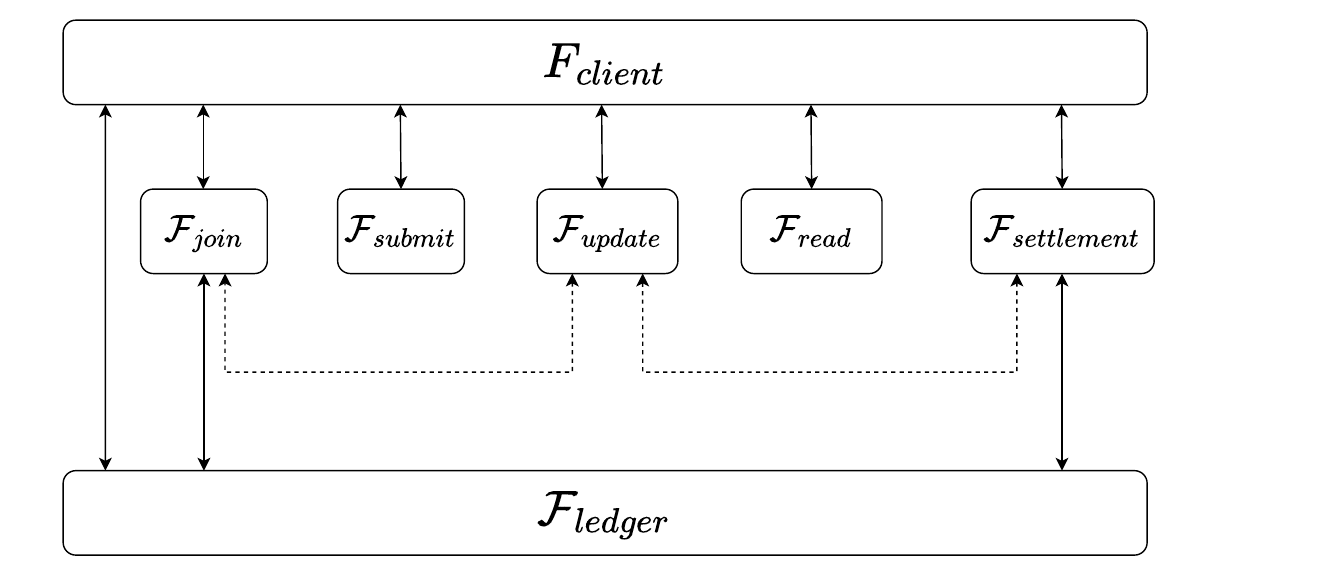}
        \caption{Sidechain}
        \label{fig:sidechain}
    \end{subfigure}
    \hfill
    \begin{subfigure}{0.32\textwidth}
        \centering
        \includegraphics[width=\linewidth]{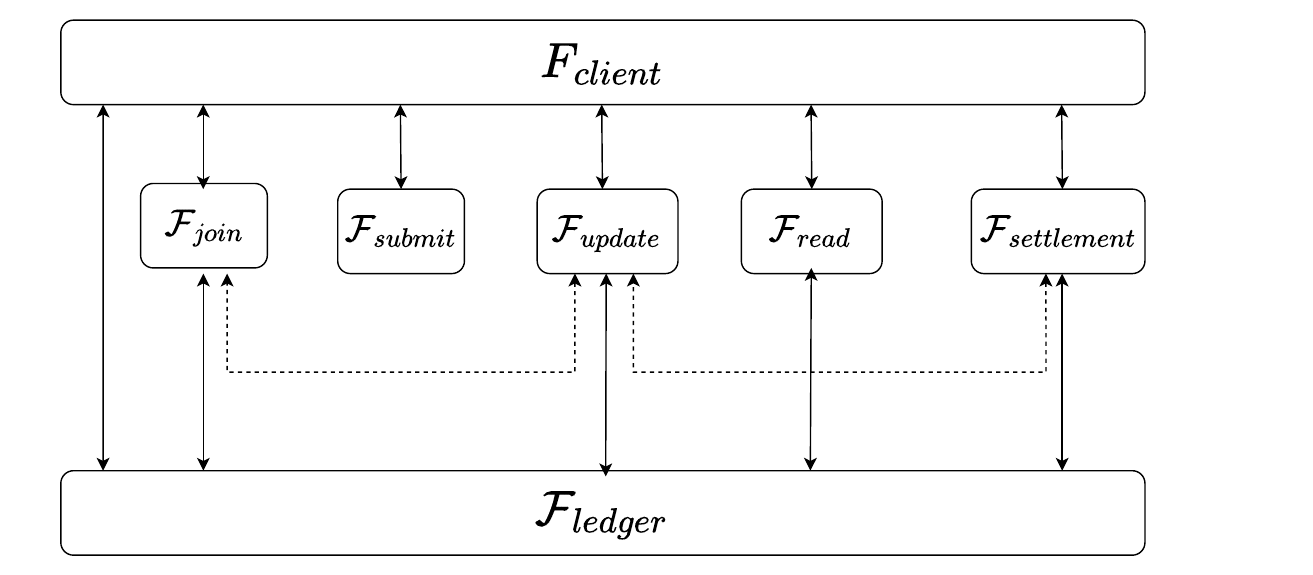}
        \caption{Rollup}
        \label{fig:Rollup}
    \end{subfigure}

    \caption{Connections among the main parameter subroutines. Solid lines indicate direct I/O caller-subroutine connections; dotted lines indicate indirect connections caused by $\mathsf{internalState}$ of $\mathcal{F}_{\text{client}}$.}
    \label{fig:Difference}
\end{figure*}

\subsubsection{Classification and Definition}
According to the analysis above, the primary factor that differentiates L2 protocols from each other is the behavior and structure of their subroutines. Based on the characteristics of these subroutines, we propose the following definitions for the three primary types of L2 protocols:

\begin{definition}[PCF]
\label{def:PCF}
A L2 protocol $\Lambda$ is a \emph{payment channel (factory) protocol} if the ideal functionality $\mathcal{F}^{\Lambda}_{\text{layer2}}$ realised by $\Lambda$ satisfies:
\begin{itemize}
    \item $\mathcal{F}^{\Lambda}_{\text{join}}$ and $\mathcal{F}^{\Lambda}_{\text{settlement}}$ have indirect connection with $\mathcal{F}_{\text{submit}}$.
        \item $\mathcal{F}_{\text{update}}$ and $\mathcal{F}_{\text{read}}$ do not require I/O connection with $\mathcal{F}_{\text{ledger}}$ to guarantee \emph{safety}.
\end{itemize}
\end{definition}

\begin{definition}[Sidechain]
\label{def:sidechain}
A L2 protocol $\Lambda$ is a \emph{sidechain protocol} if the ideal functionality $\mathcal{F}^{\Lambda}_{\text{layer2}}$ realised by $\Lambda$ satisfies:
\begin{itemize}
    \item $\mathcal{F}^{\Lambda}_{\text{join}}$ and $\mathcal{F}^{\Lambda}_{\text{settlement}}$ have indirect connection with $\mathcal{F}_{\text{update}}$.
        \item The subroutine $\mathcal{F}_{\text{update}}$ and $\mathcal{F}_{\text{read}}$ do not require I/O connection with $\mathcal{F}_{\text{ledger}}$ to guarantee \emph{safety}.
\end{itemize}
\end{definition}

\begin{definition}[Rollup]
\label{def:rollup}
A L2 protocol $\Lambda$ is a \emph{rollup protocol} if the ideal functionality $\mathcal{F}^{\Lambda}_{\text{layer2}}$ realised by $\Lambda$ satisfies:
\begin{itemize}
    \item $\mathcal{F}^{\Lambda}_{\text{join}}$ and $\mathcal{F}^{\Lambda}_{\text{settlement}}$ have indirect connection with $\mathcal{F}_{\text{update}}$.
        \item I/O connection with $\mathcal{F}_{\text{ledger}}$ is the \emph{necessary condition} for $\mathcal{F}_{\text{update}}$ and $\mathcal{F}_{\text{read}}$ to guarantee \emph{safety}.
\end{itemize}
\end{definition}

\subsection{Security and Performance Properties}

\subsubsection{Comparison of Properties}
The key distinction between L2 protocol types lies in the structure of their subroutines, which in the end affects how they satisfy security properties. While the realizations of \emph{correct L2 initialization} and \emph{correct L2 settlement} do not differ across protocol types, the performance and assumptions required to satisfy the remaining security and efficiency properties vary due to differing design choices. In what follows, we compare how PCF, sidechain, and rollup protocols realize three key properties: $f$-safety, $(f, T)$-liveness, and data availability. The results are summarized in Table~\ref{tab:difference2}.

\begin{table*}[!htp]
\caption{Comparison across L2 paradigms. We assume $n_P = n_C + n_{OP}$ L2 participants ($n_C$ clients, $n_{OP}$ operators). $f_{L_1}$ denotes L1 corruption tolerated by $\mathcal{F}_{\text{ledger}}$, and $f_{L_2}$ denotes L2 corruption, with $f_{OP}$ the maximum number of corrupted operators tolerated by the sidechain consensus; $m$ is the number of state-update requests executed since protocol start (no clients leave). $T_{L_1}$ and $T_{L_2}$ denote latency sources from L1 and L2 interactions; an additional $T_{\text{challenge}}$ term applies when the protocol enforces correct execution via an L1 dispute mechanism (e.g., a fraud-proof challenge window).}
\label{tab:difference2}
\resizebox{\textwidth}{!}{%
\begin{tabular}{|c|c|cccccc|c|}
\hline
\multicolumn{1}{|l|}{\multirow{2}{*}{}} &
  \multirow{2}{*}{$f$-safety} &
  \multicolumn{6}{c|}{$(f, T)$-Liveness} &
  \multirow{2}{*}{Data availability} \\ \cline{3-8}
\multicolumn{1}{|l|}{} &
   &
  \multicolumn{3}{c|}{$f$} &
  \multicolumn{3}{c|}{$T$ (add $T_{\text{challenge}}$ for fraud-proof designs)} &
   \\ \hline
PCF &
  $(n_{OP}-1)$ &
  \multicolumn{1}{c|}{$0_{OP} + $ $f_{L_1}$} &
  \multicolumn{1}{c|}{$0_{OP}$} &
  \multicolumn{1}{c|}{$0_{OP} + $ $f_{L_1}$ or $=$ $f_{L_1}$} &
  \multicolumn{1}{c|}{$T_{L_1} + T_{L_2}$} &
  \multicolumn{1}{c|}{$T_{L_2}$} &
  $T_{L_1} + T_{L_2}$ or $T_{L_1}$ &
  $\{\Omega(m), \Omega(n_P)\}$ \\ \hline
Sidechain &
  $f_{OP}$ &
  \multicolumn{1}{c|}{$\lfloor \frac{n - (f_{OP}+1)}{2} \rfloor + $ $f_{L_1}$} &
  \multicolumn{1}{c|}{$\lfloor \frac{n - (f_{OP}+1)}{2} \rfloor$} &
  \multicolumn{1}{c|}{$\lfloor \frac{n - (f_{OP}+1)}{2} \rfloor + $ $f_{L_1}$} &
  \multicolumn{1}{c|}{$T_{L_1} + T_{L_2}$} &
  \multicolumn{1}{c|}{$T_{L_2}$} &
  $T_{L_1} + T_{L_2}$ &
  $\{\Omega(m), \Omega(n_P)\}$ \\ \hline
Rollup &
  $f_{L_2}+f_{L_1}$ &
  \multicolumn{3}{c|}{$(n_{OP}-1) + $ $f_{L_1}$ or $f_{L_1}$} &
  \multicolumn{3}{c|}{$T_{L_1} + T_{L_2}$ or $T_{L_1}$} &
  $\{\Omega(1), \Omega(m) + \Omega(n_P)\}$ \\ \hline
\end{tabular}%
}
\end{table*}
\textbf{$f$-Safety} characterizes a protocol's resilience to adversarial corruption while maintaining consistency of the L2 executed requests and states, as captured jointly by $\mathcal{F}_{\text{update}}$ and $\mathcal{F}_{\text{read}}$. Note that it is initially assumed that there is at least one honest client to make the L2 protocol meaningful. As discussed in the analysis of subroutine content differences, safety is strongest in PCFs: each state update must be proposed by \emph{all} honest clients (who also act as operators), so the protocol remains safe as long as at least one client is honest, which is the initial assumption for the L2 protocols. In sidechains, safety is achieved through a consensus protocol executed among the operators independently of the clients, typically tolerating up to $f_{OP}$ operator threshold. In rollups, safety's consistency requirement is anchored in the underlying L1 blockchain via the corruption budget $f_{L_1}$. Additionally, the L2 corruption $f_{L_2}$ must hold to guarantee correct execution, since it does not directly read the L1 state executed by L1 participants, which, in optimistic rollups, coincides with the at least one honest client assumption.


\textbf{$(f, T)$-Liveness} characterizes the protocol's resilience to adversarial blocking of honest requests, as well as the latency that is caused by off-chain interaction or on-chain interaction. In sidechains and rollups, the update logic is reused across three different request types via $\mathcal{F}_{\text{update}}$, resulting in uniform liveness behavior. For sidechains, $T_{L_2}$-bounded liveness is guaranteed for off-chain updates as long as the consensus threshold is met (e.g., a majority of honest operators). However, operations that require anchoring on L1, such as joining or exiting, also depend on the liveness of the L1 blockchain, yielding a total latency of $T_{L_2} + T_{L_1}$. 

Rollups exhibit uniform liveness across all requests, as every update must be committed on L1. When at least one operator is honest and the L1 blockchain remains live, the protocol progress is guaranteed. Rollups are therefore less sensitive to L2 synchrony and achieve optimal fault tolerance. However, they are entirely dependent on L1 throughput and liveness, resulting in a latency of $T_{L_2} + T_{L_1}$ for all operations. For the special case in which censored clients submit requests directly to the L1 blockchain, liveness performance coincides with that of L1 itself with the latency $T_{L_1}$.

In contrast, PCFs exhibit request-specific liveness behavior. The joining and update procedures require active participation from all honest clients; as captured by the ideal functionality, the adversary can influence when these procedures complete either through corruption (under the assumption of at least one honest client) or through communication delays. Consequently, any corrupted or delayed participant can censor progress. Since joining involves interaction with the L1 blockchain while updates are committed entirely off-chain, the corresponding latencies are $T_{L_1} + T_{L_2}$ and $T_{L_2}$, respectively. The \emph{settlement} procedure, by contrast, supports two modes: the collaborative mode, whose liveness can still be affected by other clients, and the unilateral mode, whose liveness depends solely on L1 and thus incurs a latency dominated by $T_{L_1}$. Overall, PCF protocols are highly sensitive to participant behavior, with liveness varying significantly across subroutines.

\textbf{Data Availability.} The \(\mathcal{F}_{\text{read}}\) subroutine reflects where clients can access the state and executed requests. In PCFs and sidechains, data availability relies on the assumption that a sufficient number of honest operators remain online and responsive. Clients reconstruct the current state by combining the initial state that is committed on the L1 blockchain, with off-chain logs of executed requests and states. As a result, both designs require $\Omega(n_p)$ storage on L1, where $n_p$ denotes the number of participants, and $\Omega(m)$ storage off-chain, where $m$ denotes the number of executed requests representing the L2 ledger size.

Rollups, by contrast, enforce data availability through on-chain publication, which are usually stored in BLOB~\cite{buterin2022eip4844}. All executed requests and resulting states are posted to the L1 blockchain, allowing the full state to be reconstructed directly from on-chain data without relying on operator availability. This design yields an L1 storage complexity of \(\Omega(n_p) +\Omega(m)\), while removing the need for persistent off-chain storage to guarantee availability. Accordingly, no lower bound on L2 storage is required, since all state-relevant data is on L1.
We emphasize, however, that this is a design choice. In more decentralized rollup architectures, data availability could instead be ensured by relying on existential honesty among L2 archival nodes: i.e., requiring that at least one node faithfully stores the L2 ledger~\cite{milkomeda,capretto2024fast}. In such designs, the on-chain storage burden may be reduced at the cost of additional trust assumptions on L2 participants. 
Our framework captures this via $\mathcal{F}_{\text{read}}$ depending on local storage (PCFs, sidechains) or $\mathcal{F}_{\text{ledger}}$ (rollups).

\subsubsection{Properties Tradeoffs}
After outlining how L2 protocols differ in security and efficiency properties, we now examine the trade-offs these designs entail. Since the core goal of a L2 protocol is to execute state updates efficiently while maintaining consistency, this objective is captured by the \emph{safety} and \emph{liveness} properties. We show that a fundamental trade-off exists between them for state update requests. Omitted proofs appear in Appendix~\ref{apdx:proof}.


\begin{restatable}{theorem}{ThmlivenessPS}\label{thm:livenessPS}
Any secure PCF or sidechain protocol provides at most \(f_{OP}\)-safety and cannot achieve liveness better than \(\{\lfloor\frac{n_{OP}-(f_{OP}+1)}{2}\rfloor,\\ T_{L_2}\}\) for state-update requests.
\end{restatable}

\begin{restatable}{theorem}{ThmlivenessA}\label{thm:livenessA}
Any secure rollup protocol provides ($f_{L_2} + f_{L_1}$)-safety and cannot achieve liveness better than $\{(n_{OP}-1)+f_{L_1},\, T_{L_1} + T_{L_2}\}$ or $\{f_{L_1},\, T_{L_1}\}$ for state-update requests.
\end{restatable}

Beyond the relationship between security properties, our framework also reveals a trade-off between security and efficiency. In particular, recall that under our definition, \emph{liveness} for a given state update request holds if the resulting state can be retrieved via a corresponding read request. As demonstrated in the case studies, different L2 protocols implement the \(\mathcal{F}_{\text{read}}\) subroutine using distinct mechanisms to determine the final read result. These design choices directly affect how L1 and L2 storage are utilized in practice. We summarize this connection below.

\begin{restatable}{theorem}
{ThmdataPS}\label{thm:dataPS}
    Assume $m$ state update requests are executed after the protocol starts. To guarantee \emph{liveness} for secure PCF protocol and sidechain protocol, the \emph{data availability} needs to have $\{\Omega(m), \Omega(n_p)\}$ efficiency.
\end{restatable} 

\begin{restatable}{theorem}
{ThmdataA}\label{thm:datar}
    Assume $m$ state update requests are executed after the protocol starts. To guarantee \emph{liveness} for secure rollup protocol, the \emph{data availability} needs to have $\{\Omega(1), \Omega(m) + \Omega(n_p)\}$ efficiency.
\end{restatable} 

\section{Insights for Designing Suitable Layer 2 Protocols}
\label{sec:design}

As argued in Section~\ref{sec:intro}, real-world scenarios often impose requirements beyond those addressed by the L2 protocols shown in our case studies, for instance. In this section, we show how to leverage insights from our framework to derive a suitable L2 design from existing constructions. Specifically, we present a novel optimistic-rollup protocol.

\subsection{Problem Statement}  

Despite widespread deployment of optimistic rollups in cryptocurrency markets, a primary limitation is high confirmation latency before finality~\cite{Buterin2021IncompleteRollups}. As discussed in Section~\ref{sec:casestudy} and \ref{sec:insight}, this latency also includes the dispute window $T_{\text{challenge}}$, which is usually set to be long enough, for example, Arbitrum uses a dispute period of roughly one week~\cite{ArbitrumDocs2025GentleBoLD}. Such a design is unsuitable for time-sensitive applications. This raises a natural question: how should the rollup design be modified to meet the latency requirements?

\subsection{\xr Overview} 

Our case studies of the Brick payment channel and the Liquid sidechain in Section~\ref{sec:casestudy} provide insights for solving the problem. Both protocols show that introducing a committee of additional parties like BFT operators in Liquid allows the system to trade a stronger trust assumption on that committee for faster, committee-certified finality, with the committee's sole responsibility being to produce a quorum certificate. Guided by this insight, and aiming to preserve the L1-based safety of rollup protocols according to Definition~\ref{def:rollup} without relying on the \emph{off-chain consensus}, we propose \xr, a clean extension that adds \emph{fast-finality} through a third-party watchtower committee. Concretely, the watchtowers issue a quorum certificate with majority signatures guaranteeing timely execution under an honest majority. A single operator still periodically posts all transactions to L1 to preserve consistency between fast-finality and regular rollup transactions, and misbehavior is penalizable using the issued certificates.

\textbf{System Model.}
We consider four roles: (1) \emph{clients}, who submit regular transactions and those requiring liveness guarantees; (2) \emph{operators}, who validate transactions, batch and publish them on-chain, and provide liveness guarantees when required; (3) \emph{verifiers}, who check published batches and submit fraud-proofs (realized by clients here); and (4) \emph{watchtowers}, who monitor operator posts and issue \emph{fast-finality} certificates (signatures) for transactions with special liveness needs. The off-chain communication is assumed to be \emph{asynchronous}.

\textbf{Protocol Joining.}
As in standard rollups, clients deposit funds on L1 and are admitted once the L2 state reflects these deposits. Operators also post collateral to ensure safety. We assume an L1-registered committee of \(2f+1\) watchtowers, with at most \(f\) Byzantine (i.e., at least \(f+1\) honest).

\textbf{Transaction Execution.}
Transactions are considered as committed in two different situations. A \emph{regular} transaction becomes final once it is included in an L1 batch, published by the operator or clients themselves, and remains unchallenged for the full dispute window. A \emph{fast-finality} transaction becomes final when a watchtower committee that continuously monitors the blockchain issues a quorum certificate and posts it to L1. The certificate consists of $f+1$ watchtower signatures over the transaction. Clients can also send transactions by themselves to L1 when they get censored by the operator.

\textbf{Protocol Leaving.}
A client issues a leave request claiming its latest L2 state. Once processed as a regular transaction, the client finalizes by submitting the corresponding L1 withdrawal. A client can also leave unilaterally through \emph{escape-hatch}.

\subsection{Protocol Analysis}


Here we give the conclusion for analysis, the detailed protocol pseudocode and analysis can be found in Appendix~\ref{apdx:cross}. We formally define the real world protocol of the \xr as $\mathcal{P}^{\text{FRoll}} := (\mathcal{P}^{\text{FRoll}}_{\text{client}}: \text{client},\mathcal{F}_{\text{ledger}} \space \mid \allowbreak \space  \mathcal{P}^{\text{FRoll}}_{\text{operator}},\allowbreak \space \mathcal{P}^{\text{FRoll}}_{\text{watchtower}},\allowbreak \space \mathcal{P}^{\text{FRoll}}_{\text{client}}: \text{verifier}, \allowbreak\mathcal{F}_{\text{sig}},\allowbreak \mathcal{F}^{\text{FRoll}}_{\text{com}})$. Based on our framework, we can define the ideal functionality for our \xr protocol as $\mathcal{F}^{\text{FRoll}}_{\text{layer2}} := (\mathcal{F}_{\text{client}},\mathcal{F}_{\text{ledger}} \allowbreak \space \mid \allowbreak \space 
\mathcal{F}^{\text{FRoll}}_{\text{submit}}, \allowbreak 
\mathcal{F}^{\text{FRoll}}_{\text{join}}, \allowbreak
\mathcal{F}^{\text{FRoll}}_{\text{update}}, \allowbreak
\mathcal{F}^{\text{FRoll}}_{\text{read}}, \allowbreak
\mathcal{F}^{\text{FRoll}}_{\text{settlement}}, \allowbreak
\mathcal{F}^{\text{FRoll}}_{\text{updRnd}})$ and we have the following conclusion for the ideal functionality $\mathcal{F}^{\text{FRoll}}_{\text{layer2}}$:

\begin{restatable}{theorem}
{ThmidealCross}
    The ideal functionality $\mathcal{F}^{\text{FRoll}}_{\text{layer2}}$ guarantees all the security properties of a secure L2 protocol.
\end{restatable}

The security of the real \xr protocol can then be demonstrated by proving that the real protocol iUC-realizes the ideal functionality. This can be formally defined as follows:

\begin{restatable}{theorem}
{ThmrealizeCross}
    With the existence of ideal functionalities $\mathcal{F}_{\text{ledger}}$, $\mathcal{F}_{\text{sig}}$ and $\mathcal{F}^{\text{FRoll}}_{\text{com}}$. The real \xr protocol $\mathcal{P}^{\text{FRoll}}$ iUC-realizes the ideal \xr functionality $\mathcal{F}^{\text{FRoll}}_{\text{layer2}}$.
\end{restatable}

\section{Related Work}

\textbf{Payment channels}~\cite{poon2016bitcoin,avarikioti2020cerberus,avarikioti2021b,avarikioti2025thunderdome,aumayr2021blitz,decker2015fast,burchert2018scalable,aumayr2022sleepy,miller2019sprites} enable two parties to conduct off-chain transactions by locking funds on-chain and exchanging signed messages that update their balance, with only the final state submitted to the blockchain for settlement. Security analyses for these protocols typically focus on the closing phase, capturing guarantees such as \emph{balance security}, which ensures honest parties can reclaim their rightful funds. These analyses are protocol-specific and do not generalize to broader L2 architectures. Our framework instead captures payment channels as a special case of a general L2 functionality, enabling uniform reasoning across diverse protocol designs.

\textbf{Sidechains}~\cite{nick2020liquid,gavzi2019proof,polygon} operate as separate blockchains with their own consensus mechanisms, often maintained by federated operators; clients interact with the sidechain while funds remain escrowed on the L1. Existing work defines security in terms of \emph{ledger safety} and \emph{liveness}, drawing from traditional blockchain theory~\cite{garay2017bitcoin}, but typically analyzes the sidechain in isolation, without supporting cross-paradigm comparisons or composable reasoning. Our framework treats sidechains within the same ideal functionality as other L2 protocols, allowing their assumptions and trade-offs to be systematically compared.

\textbf{Rollups}~\cite{bousfield2022arbitrum,zksyncdocs,fisch2024permissionless,gudgeon2020sok,capretto2024fast} outsource execution to an off-chain sequencer that posts commitments and proofs to the L1. Optimistic rollups rely on fraud proofs and timeout-based dispute windows, while zk-rollups use succinct validity proofs. A central challenge is \emph{data availability}, ensuring honest parties can reconstruct and verify protocol state. Although addressed in protocol design and empirical analysis~\cite{fisch2024permissionless,capretto2024fast}, formal models are lacking and existing analyses are protocol-specific. Our framework unifies these models through a common interface, exposing trade-offs between availability, latency, and trust assumptions.

\textbf{Security frameworks} for L2 protocols have focused on specific constructions. Aumayr et al.~\cite{aumayr2021generalized} formalize state channel networks using UC; Kiayias et al.~\cite{kiayias2020composable} study composability in Lightning-style channels. These offer strong guarantees but do not extend to sidechains or rollups. In the L1 setting, Graf et al.~\cite{graf2021security} introduce a general framework for distributed ledgers based on trust and consistency assumptions. We build on this line by presenting the first general security framework for L2 protocols, supporting composable reasoning across paradigms within a single ideal functionality. Jain et al.~\cite{jain2022tiramisu} formalize L2 protocols by modeling L1 as responsible solely for participant registration, an abstraction that does not accurately capture rollups, which rely on L1 for more than registration. Chaliasos et al.~\cite{chaliasos2024towards} focus on additional zk-rollup functionalities such as censorship resistance and protocol upgrades, rather than core execution. \cite{gudgeon2020sok} discusses L2 design trade-offs but addresses only implementation differences without formally analyzing the corresponding security guarantees. Weintraub et al.~\cite{weintraub2024payout} adopt a rigorous formal-methods approach for channel-protocol analysis but remain protocol-specific, lacking the generality of our iUC-based framework.

\section{Conclusion}

This paper introduces the first general security framework for Layer 2 (L2) blockchain protocols, formalized as a modular ideal functionality in the iUC framework. By structuring protocol behavior through well-defined subroutines covering joining, submission, state updates, reading, and settlement, our framework abstracts away concrete implementation mechanisms while preserving the key dimensions of L2 security.

This abstraction enables, for the first time, composable reasoning across structurally diverse L2 designs (payment channels, sidechains, and rollups) and provides a formal basis for comparing their safety, liveness, and data-availability guarantees under a common umbrella. Through detailed case studies and comparative analysis, we demonstrate the framework's generality and both validate known trade-offs and surface implicit design constraints that were previously understood only informally. We further illustrate the framework's design guidance through a blueprint for an optimistic-rollup protocol offering fast-finality.

More broadly, our framework lays the groundwork for the formal study of new challenges in L2 design. Its modular structure supports clean extensions to other settings, such as incorporating rational adversaries via Rational Protocol Design~\cite{garay2013rational} or modeling cross-chain interoperability and cross-protocol L2 composition. We view this work as a foundational step toward security- and incentive-aware design of a scalable blockchain infrastructure.

\section{Acknowledgments}
The work was partially supported by CoBloX Labs, by the European Research Council (ERC) under the European Union’s Horizon 2020 research (grant agreement 771527-BROWSEC), by the Austrian Science Fund (FWF) through the SFB SpyCode project F8510-N and F8512-N, and by the WWTF through the projects 10.47379/ICT22045 and 10.47379/ICT25056.





\bibliographystyle{plain}
\bibliography{reference}

\appendix

\section{The iUC Framework}
\label{apdx:iUC}

In this part, we give a short introduction to the iUC framework, which underlies the specifications and proofs in this paper. We refer to~\cite{camenisch2019iuc} for the formal specification and to~\cite{kusters2020iitm} for the underlying IITM (Inexhaustible Interactive Turing Machine) model. Prior work argues that iUC is particularly suitable for complex, stateful systems such as blockchain protocols~\cite{graf2021security}.

\paragraph{Protocols and machines.}
In iUC, both real-world protocols and ideal functionalities are modeled as systems of Interactive Turing Machines (ITMs) that exchange messages through designated interfaces. A protocol is written as
$\mathcal{P} = (\mathcal{M}_1 \mid \mathcal{M}_2, \ldots, \mathcal{M}_n)$,
where each machine $\mathcal{M}_i$ implements one or more \emph{roles}. A role describes a piece of code that performs a specific task within the protocol (e.g., a client submitting transactions, an operator producing blocks, or a watchtower monitoring disputes). An ideal functionality $\mathcal{F}$ is also a system of ITMs and is written analogously. Protocols and ideal functionalities can equivalently be written by replacing machine names with the functionality or protocol that contains them.

Each machine declares its roles as either public or private, using the notation
$(\text{public roles} \mid \text{private roles})$.
Public roles can be invoked via I/O by external callers, including higher-level protocols and the environment, whereas private roles are internal and accessible only from within the system. Communication among machine instances is realized through two types of connections:
\begin{itemize}
  \item \textbf{I/O connections} link callers to subroutines according to the protocol's role graph. These are the typed, structured connections that implement modular composition, allowing higher-level protocols to invoke operations of their subroutines.
  \item \textbf{NET connections} are controlled by the adversary $\mathcal{A}$ (or by the simulator $\mathcal{S}$). The adversary uses this interface to schedule message delivery, to notify entities of corruption, to send and receive messages on behalf of corrupted entities, and to receive the specified leakage from the protocol.
\end{itemize}
This separation allows iUC to model idealized subroutine calls (I/O) alongside adversarially controlled network behavior (NET).

\paragraph{Instances, entities, and identifiers.}
During a run, multiple instances of a given machine may be created. At any time step, only one instance is \emph{active} (the instance that has most recently received a message); all others are idle. Each instance may manage one or more \emph{entities}, identified by a triple
$(\mathrm{pid},\; \mathrm{sid},\; \mathrm{role})$,
where $\mathrm{pid}$ is a party identifier, $\mathrm{sid}$ is a session identifier, and $\mathrm{role}$ specifies which role the entity executes. When a message arrives addressed to an entity whose role is implemented by machine $\mathcal{M}$, the framework runs a user-specified \emph{CheckID} procedure over the existing instances of $\mathcal{M}$ (in creation order) to determine which instance manages that entity; if none does, CheckID can authorize spawning a new instance. A common convention for ideal functionalities is that a single instance manages all entities for a given session (CheckID accepts all entities with the same $\mathrm{sid}$), whereas real-world protocols typically instantiate one instance per party per session (CheckID accepts entities with a single $\mathrm{pid}$ and $\mathrm{sid}$).

The special variable $(\pcur, \scur, \mathrm{role}_{\mathrm{cur}})$ refers to the currently active entity of the current machine instance. When sending a message to a subroutine, the notation $(\epsilon, \scur, \mathrm{role})$ indicates that the party identifier is left unspecified; this is used when the recipient is a session-wide subroutine (such as $\mathcal{F}_{\text{update}}$) that is not associated with any specific party but rather manages shared state for the entire session. Similarly, the notation $(\_, \mathrm{sid}, \mathrm{role})$ denotes all entities with a given session and role, without restricting to a particular party.

\paragraph{$\mathsf{internalState}$ initialization.}
When an instance accepts its first entity, it executes an \textbf{Initialization} procedure to set up its $\mathsf{internalState}$, including global parameters, policy flags, and initial data structures. In our specifications, we explicitly let the $\mathsf{internalState}$ be initiated with predefined contents. Additionally, per-entity initialization logic can be defined within the initialization procedure when necessary.

\paragraph{Corruption model.}
Corruption in iUC is modeled at the entity level and is highly customizable. An instance records the corruption status of the entities it manages (e.g., honest, statically corrupted, dynamically corrupted) and may alter its behavior accordingly. The protocol designer specifies a \emph{leakage policy} that defines what information is released to the adversary upon corruption. Unless stated otherwise, all machines, environments, and adversaries are required to be probabilistic polynomial-time (PPT).

\paragraph{iUC realization.}
Security is defined by showing a real protocol $\mathcal{P}$ realizing an ideal functionality $\mathcal{F}$. We write
$\mathcal{P} \leq \mathcal{F}$
and say that ``$\mathcal{P}$ iUC-realizes $\mathcal{F}$'' if for every real-world adversary $\mathcal{A}$ there exists a simulator $\mathcal{S}$ (interacting with $\mathcal{F}$ via the NET connection) such that, for all environments $\mathcal{E}$, the executions $\{\mathcal{E}, \mathcal{A}, \mathcal{P}\}$ and $\{\mathcal{E}, \mathcal{S}, \mathcal{F}\}$ are computationally indistinguishable. Equivalently, in the formulation where the environment subsumes the network attacker, $\{\mathcal{E}, \mathcal{P}\}$ is indistinguishable from $\{\mathcal{E}, \mathcal{S}, \mathcal{F}\}$.

\paragraph{Composition.}
A major advantage of iUC is its support for modular composition. Suppose a higher-level protocol $\mathcal{Q}$ uses an ideal subroutine $\mathcal{F}$ through its public-role I/O interface. If a real protocol $\mathcal{P}$ has been proven to realize $\mathcal{F}$ (i.e., $\mathcal{P} \leq \mathcal{F}$), then $\mathcal{F}$ can be replaced by $\mathcal{P}$ inside $\mathcal{Q}$ without re-proving the security of the entire system. Formally, if $\mathcal{R} \in \mathrm{Comb}(\mathcal{Q}, \mathcal{P})$ and $\mathcal{I} \in \mathrm{Comb}(\mathcal{Q}, \mathcal{F})$ are the composed systems with matching public-role interfaces, then $\mathcal{R} \leq \mathcal{I}$. This underpins the modular ``plug-and-prove'' workflow used throughout our case studies: we first prove that each L2 protocol realizes its respective ideal L2 functionality assuming an ideal L1 ledger, and then replace the ideal L1 ledger with a concrete blockchain implementation.

\paragraph{Conventions used in this paper.}
We adopt the following conventions for our specifications:
\begin{itemize}
  \item \emph{Ideal functionalities} use a single machine instance per session (identified by $\mathrm{sid}$), which internally manages all entities and the shared state for that session.
  \item \emph{Real-world protocols} instantiate one machine instance per entity, reflecting that real participants hold independent local state and communicate only through the network.
\end{itemize}

\paragraph{Responsive environments.}
We leverage an iUC feature that forces the adversary to respond immediately to designated network messages. Such messages are indicated by the phrase ``\textbf{send responsively}'' in our specifications. Once the adversary receives such a message, it must supply the requested response in its very next activation before interacting with any other part of the system. This mechanism is used, for example, to leak information to the adversary or to let the adversary determine the corruption status of a new entity without disrupting the intended execution flow of the protocol.

\paragraph{Notation summary.}
We use $(\mathrm{pid}, \mathrm{sid}, \mathrm{role})$ to identify entities; $(\text{public} \mid \text{private})$ to declare role visibility; I/O for caller--subroutine connections; and NET for adversarially controlled network delivery. The variable $(\pcur, \scur, \mathrm{role}_{\mathrm{cur}})$ refers to the currently active entity. We use $\epsilon$ to denote an empty or unspecified field, indicating that the corresponding component is intentionally left empty. The notation $(\_, \mathrm{sid}, \mathrm{role})$ denotes any entities sharing a given session and role, without restricting to a particular party identifier.


\section{Ideal Functionality $\mathcal{F}_{\text{ledger}}$}
\label{appdx:ledger}

We briefly introduce the ideal ledger functionality used to model the underlying L1 blockchain; for a detailed treatment, see~\cite{graf2021security}. The functionality $\mathcal{F}_{\text{ledger}}$ is specified in the iUC framework~\cite{camenisch2019iuc} and follows a modular design: a single client machine $\mathcal{F}_{\text{client}}$ implements the core logic for handling read and write requests, while a set of parameterized subroutines captures the specific security properties of the ledger. Formally,
$\mathcal{F}_{\text{ledger}} = (\mathcal{F}_{\text{client}_{L1}} \mid \mathcal{F}_{\text{submit}_{L1}},\; \mathcal{F}_{\text{update}_{L1}},\; \mathcal{F}_{\text{read}_{L1}},\\\; \mathcal{F}_{\text{updRnd}_{L1}},\; \mathcal{F}_{\text{init}_{L1}},\; \mathcal{F}_{\text{leak}_{L1}}).$

At a high level, $\mathcal{F}_{\text{submit}_{L1}}$ validates incoming transactions and determines what information leaks to the adversary upon submission; $\mathcal{F}_{\text{update}_{L1}}$ governs how the adversary may extend the global transaction list $\mathsf{msglist}$, enforcing properties such as double-spending protection and transaction validity; $\mathcal{F}_{\text{read}_{L1}}$ determines what each party may observe when reading from the ledger, thereby capturing consistency and privacy guarantees; $\mathcal{F}_{\text{updRnd}_{L1}}$ controls the advancement of the built-in clock, enabling the formalization of time-dependent properties such as liveness; $\mathcal{F}_{\text{init}_{L1}}$ specifies the initial state of the ledger (e.g., a genesis block); and $\mathcal{F}_{\text{leak}_{L1}}$ determines what information is disclosed upon corruption of a party.

Higher-level protocols interact with $\mathcal{F}_{\text{ledger}_{L1}}$ exclusively through \textbf{SubmitL1} and \textbf{ReadL1} operations. A Submit request from an honest party is forwarded to $\mathcal{F}_{\text{submit}_{L1}}$, which decides whether the transaction is accepted and, if so, places it in a buffer $\mathsf{requestQueue}$. The adversary may then propose updates via $\mathcal{F}_{\text{update}_{L1}}$ to append accepted transactions to the global list $\mathsf{msglist}$, subject to the validity checks enforced by $\mathcal{F}_{\text{update}_{L1}}$. A Read request is processed by $\mathcal{F}_{\text{read}_{L1}}$, which returns an appropriate view of $\mathsf{msglist}$ to the requesting party. In this paper, we define the output of an L1 read request limited to L2 protocol related information and consist of three components: the transaction set $\{TX\}_{\text{L1}}$, representing committed L1 transactions, which also includes the transactions contained in both block and BLOB; the state set $\{state\}_{\text{L1}}$, representing committed states on the L1 ledger that is correctly executed based on the $\{TX\}_{\text{L1}}$; and the identity set $\{pid\}$, representing the identities of registered L2 participants.

\paragraph{Transaction format.} For simplicity and uniformity across both account-based blockchains (e.g., Ethereum) and UTXO-based block-chains (e.g., Bitcoin), we define a transaction as the tuple
$TX = (\mathit{sender},\; \mathit{receiver},\; \mathit{value},\; \mathit{data}),$
where $\mathit{sender}$ and $\mathit{receiver}$ denote the originating and destination parties (corresponding to sending and receiving accounts in account-based models, or inputs and outputs in UTXO-based models), $\mathit{value}$ denotes the resulting state after $TX$ is committed, and $\mathit{data}$ encodes any additional payload attached to the transaction.

\paragraph{Smart contract support.} The framework $\mathcal{F}_{\text{ledger}}$ fully supports the modeling of smart contracts. Following~\cite{graf2021security}, a smart contract programming language is fixed as a parameter, and smart contracts are represented as bit strings interpreted by the subroutines of $\mathcal{F}_{\text{ledger}}$ according to that language. In particular, contract logic is enforced primarily through $\mathcal{F}_{\text{submit}}$ and $\mathcal{F}_{\text{update}}$: the former validates that submitted transactions satisfy the conditions imposed by referenced contracts, while the latter ensures that only state transitions consistent with correct contract execution are appended to $\mathsf{msglist}$. This approach treats the contract language as an arbitrary but fixed parameter, making the security analysis independent of any specific smart contract language. Since the L2 constructions studied in this paper rely on smart contracts deployed on L1 (e.g., rollup contracts, payment channel contract), we discuss the concrete instantiation of $\mathcal{F}_{\text{submit}}$ and $\mathcal{F}_{\text{update}}$ to capture the relevant contract logic in each case study.
\vspace{1em}
\vspace{1em}\begin{functionality}{Excerpt description of $\mathcal{M}_{\text{client}_{\text{L1}}}$ of $\mathcal{F}_{\text{ledger}}$}{

\textbf{Implemented role(s):} \{client\textsubscript{L1}\}

\vspace{0.5em}
\noindent \textbf{Main (excerpt):}

\vspace{0.5em}

\textbf{recv} \{SubmitL1, $\mathit{msg}$\} \textbf{from} I/O:

\begin{enumerate}[itemsep=0.5em]
  \item \textbf{send} \{Submit, $\mathit{msg}$, $\mathsf{internalState}$\} \textbf{to} $(\pcur, \scur, \mathcal{F}_{\text{submit}_{\text{L1}}}:\text{submit})$,
  \textbf{wait for} \{Submit, $\mathit{response}$\} s.t. $\mathit{response} \in \{\text{true}, \text{false}\}$;

  \item \textbf{if} $\mathit{response} = \text{true}$:
    \begin{itemize}
      \item $\mathsf{reqCtr} \leftarrow \mathsf{reqCtr} + 1$;
      \item $\mathsf{requestQueue}.\text{add}(\mathsf{reqCtr}, \mathsf{round}, \mathit{msg})$;
    \end{itemize}
\end{enumerate}

\hrule
\vspace{0.5em}

\textbf{recv} \{ReadL1, $\mathit{msg}$\} \textbf{from} I/O:

\begin{enumerate}[itemsep=0.5em]
  \item \textbf{send} \{InitRead, $\mathit{msg}$, $\mathsf{internalState}$\} \textbf{to} $(\pcur, \scur, \mathcal{F}_{\text{read}_{\text{L1}}}:\text{read})$,
  \textbf{wait for} \{InitRead, $\mathit{local}$\} s.t. $\mathit{local} \in \{\text{true}, \text{false}\}$;

  \item \textbf{if} $\mathit{local} = \text{true}$:
    \begin{itemize}
      \item[$\star$] \textbf{send responsively} \{InitRead\} \textbf{to} $\mathcal{S}$ via NET;
      \textbf{wait for} \{InitRead,\\ $\mathit{suggestedOutput}$\} from NET;

      \item \textbf{send} \{FinishRead, $\mathit{msg}$, $\mathit{suggestedOutput}$, $\mathsf{internalState}$\} \textbf{to} $(\pcur, \scur, \\ \mathcal{F}_{\text{read}_{\text{L1}}}:\text{read})$;
      \textbf{wait for} \{FinishRead, $\mathit{output} = \{\{\mathit{TX}\}_{\text{L1}}, \{\mathit{state}\}_{\text{L1}}, \{\mathit{pid}\}\}$\};

      \item \textbf{if} $\mathit{output} = \bot$: go to $\star$ and repeat; \cmt{Reset local variables and retry}

      \item \textbf{send responsively} \{FinishRead\} \textbf{to} $\mathcal{S}$ via NET;
      \textbf{wait for} ack;

      \item \textbf{reply} \{Read, $\mathit{output}$\} via I/O;
    \end{itemize}

  \item \textbf{else} ($\mathit{local} = \text{false}$):
    \begin{itemize}
      \item $\mathsf{readCtr} \leftarrow \mathsf{readCtr} + 1$;
      \item $\mathsf{readQueue}.\text{add}(\mathit{pid}, \mathsf{readCtr}, \mathsf{round}, \mathit{msg})$; \cmt{Queue for later delivery}
    \end{itemize}
\end{enumerate}

\hrule
\vspace{0.5em}

\textbf{recv} \{DeliverRead, $\mathsf{readCtr}$, $\mathit{suggestedOutput}$\} \textbf{from} NET, \textbf{s.t.} $(\mathit{pid}, \mathsf{readCtr}, r, \mathit{msg}) \in \mathsf{readQueue}$:

\begin{enumerate}[itemsep=0.5em]
  \item \textbf{send} \{FinishRead, $\mathit{msg}$, $\mathit{suggestedOutput}$, $\mathsf{internalState}$\} \textbf{to} $(\pcur, \scur, \mathcal{F}_{\text{read}_{\text{L1}}}:\text{read})$,
  \textbf{wait for} \{FinishRead, $\mathit{output}$\};

  \item \textbf{if} $\mathit{output} \neq \bot$:
    \begin{itemize}
      \item \textbf{send responsively} \{FinishRead, $\mathsf{readCtr}$\} \textbf{to} $\mathcal{S}$ via NET;
      \textbf{wait for} ack;
      \item $\mathsf{readQueue}.\text{remove}(\mathit{pid}, \mathsf{readCtr}, r, \mathit{msg})$;
      \item \textbf{send} \{Read, $\mathit{output}$\} \textbf{to} $(\mathit{pid}, \scur, \text{I/O})$;
    \end{itemize}

  \item \textbf{else}: \textbf{send} \{nack\} \textbf{to} $\mathcal{S}$ via NET; \cmt{Suggested output rejected by subroutine}
\end{enumerate}

}\end{functionality}\vspace{1em}

\section{Ideal Functionality $\mathcal{F}_{
\text{sig}
}$, $\mathcal{F}_{\text{com}}$}
\label{apdx:sigcom}

Here we formally define the ideal functionality based on the iUC framework for the ideal signature scheme $\mathcal{F}_{\text{sig}}$ that realizes existentially unforgeable under chosen-message attack (EUF-CMA) security, time-related functionalities $\mathcal{F}_{\text{com}}$. 

$\mathcal{F}_{\text{sig}}$ denotes the ideal functionality for a signature scheme with security parameter $k\in \mathbb{N}$, which can be accessed to generate or verify signatures. Two roles participate in this functionality: the \emph{signer} and the \emph{verifier}. The definition of $\mathcal{F}_{\text{sig}}$ is as follows. Since all protocols mentioned in this paper use this functionality as subroutines, we do not specify the subroutine here.

\vspace{1em}\begin{functionality}{Description of protocol $\mathcal{F}_{\text{sig}} = (\text{signer}, \text{verifier})$}{

\textbf{Participating roles:} \{signer, verifier\}

\noindent \textbf{Corruption model:} incorruptible

\noindent \textbf{Protocol parameters:}
\begin{itemize}
    \item $p \in \mathbb{Z}[x]$: polynomial bounding the runtime of adversary-provided algorithms
    \item $\eta \in \mathbb{N}$: security parameter
    \item $\mathsf{gen}$: key generation algorithm; outputs $(\mathit{pk}, \mathit{sk})$ on input $1^{\eta}$
    \item $\mathsf{sig}$: signing algorithm; outputs a signature $\sigma$ of length $\eta$ bits on input $(\mathit{msg}, \mathit{sk})$
    \item $\mathsf{ver}$: signature verification algorithm; outputs a verification result on input $(\mathit{msg}, \sigma, \mathit{pk})$
\end{itemize}

}\end{functionality}
\vspace{.5em}

\begin{functionality}{Description of $\mathcal{M}_{\text{signer,verifier}}$ of $\mathcal{F}_{\text{sig}}$}{

\textbf{Implemented role(s):} \{signer, verifier\}

\noindent \textbf{Internal state:}
\begin{itemize}
    \item $(\mathit{pk}, \mathit{sk}) \in (\{0,1\}^{*} \cup \{\bot\})^{2}$, $(\mathit{pk}, \mathit{sk}) = (\bot, \bot)$ \hfill\cmt{Key pair}
    \item $\mathit{pid}_{\mathit{owner}} \in \{0,1\}^{*} \cup \{\bot\}$, $\mathit{pid}_{\mathit{owner}} = \bot$ \hfill\cmt{Party ID of the key owner}
    \item $\mathsf{msgList} \subset \{0,1\}^{*}$, $\mathsf{msgList} = \emptyset$ \hfill\cmt{Set of recorded signed messages}
    \item $\mathsf{corr} \in \{\text{true}, \text{false}\}$, $\mathsf{corr} = \text{false}$ \hfill\cmt{Signing key corruption status}
\end{itemize}

\noindent \textbf{CheckID}(\emph{pid}, \emph{sid}, \emph{role}):
\begin{itemize}
    \item Check that $\mathit{sid} = (\mathit{pid}', \mathit{sid}')$. If this check fails, output \textbf{reject}.
    \item Otherwise, accept all entities with the same \emph{sid}. \cmt{A single instance manages all parties and roles in a single session, modeling one signature key pair belonging to party $\mathit{pid}'$}
\end{itemize}

\noindent \textbf{Corruption behavior:}
\begin{itemize}
    \item \textbf{DetermineCorrStatus}(\emph{pid}, \emph{sid}, \emph{role}): Return $\mathsf{corr}$.
\end{itemize}

\noindent \textbf{Initialization:}
\begin{itemize}
    \item $(\mathit{pk}, \mathit{sk}) \overset{\$}{\leftarrow} \mathsf{gen}(1^{\eta})$; \cmt{Generate public/secret key pair}
    \item Parse $\scur$ as $(\mathit{pid}, \mathit{sid})$; $\mathit{pid}_{\mathit{owner}} \leftarrow \mathit{pid}$;
\end{itemize}

\vspace{0.5em}
\noindent \textbf{Main:}

\vspace{0.5em}

\textbf{recv} \{Sign, $\mathit{msg}$\} \textbf{from} I/O \textbf{to} $(\mathit{pid}_{\mathit{owner}}, \_, \text{signer})$:

\begin{enumerate}[itemsep=0.5em]
    \item Let $\sigma \leftarrow \mathsf{sig}^{(p)}(\mathit{msg}, \mathit{sk})$; \cmt{Sign the message with the secret key under runtime bound $p$}
    \item $\mathsf{msgList}.\text{add}(\mathit{msg})$; \cmt{Record the signed message for later verification}
    \item \textbf{reply} \{Signature, $\sigma$\};
\end{enumerate}

\hrule
\vspace{0.5em}

\textbf{recv} \{Verify, $\mathit{msg}$, $\sigma$\} \textbf{from} I/O \textbf{to} $(\_, \_, \text{verifier})$:

\begin{enumerate}[itemsep=0.5em]
    \item Let $b \leftarrow \mathsf{ver}^{(p)}(\mathit{msg}, \sigma, \mathit{pk})$; \cmt{Verify the signature with the public key under runtime bound $p$}
    \item \textbf{if} $b = \text{true}$ and $\mathit{msg} \notin \mathsf{msgList}$ and $\mathsf{corr} = \text{false}$:
    \textbf{reply} \{VerResult, false\}; \cmt{Prevent forgery: reject signatures on never-signed messages when the key is not corrupted}
    \item \textbf{else}: \textbf{reply} \{VerResult, $b$\}; \cmt{Return the underlying verification result}
\end{enumerate}

\hrule
\vspace{0.5em}

\textbf{recv} \{CorruptSigKey\} \textbf{from} NET: \cmt{Allow network attacker to corrupt signature keys}

\begin{enumerate}[itemsep=0.5em]
    \item $\mathsf{corr} \leftarrow \text{true}$;
    \item \textbf{reply} \{CorruptSigKey, ok\};
\end{enumerate}

}\end{functionality}
\vspace{1em}
$\mathcal{F}_{\text{com}}$ denotes the ideal functionality that captures the communication network and the inner clock for time-related behavior. 

First, in the Brick channel, the communication network is assumed to be adversary-controlled. Accordingly, we define the functionality \(\mathcal{F}^{\text{Brick}}_{\text{com}}\) whose round can be advanced by the adversary without additional checks. To ensure the wardens are monitoring the blockchain and react to the unilateral closing request, each time the round is updated, it will trigger the wardens to check the blockchain.
\vspace{1em}
\vspace{1em}\begin{functionality}{Description of protocol $\mathcal{F}^{\text{Brick}}_{\text{com}} = (\text{com})$}{

\textbf{Participating roles:} \{com\}

\noindent \textbf{Corruption model:} incorruptible

}\end{functionality}
\vspace{.5em}
\begin{functionality}{Description of $\mathcal{M}^{\text{Brick}}_{\text{com}}$ of $\mathcal{F}^{\text{Brick}}_{\text{com}}$}{

\textbf{Implemented role(s):} \{com\}

\noindent \textbf{Subroutines:} $\mathcal{P}^{\text{Brick}}_{\text{client}}$: \{client, operator\}, $\mathcal{P}^{\text{Brick}}_{\text{warden}}$: warden

\noindent \textbf{Internal state:}
\begin{itemize}
    \item $\mathsf{bufferMsg} \subset \{0,1\}^{*}$, $\mathsf{bufferMsg} = \emptyset$ \hfill\cmt{Buffered messages pending delivery}
    \item $\mathsf{labels} \subset \mathbb{N}$, $\mathsf{labels} = \emptyset$ \hfill\cmt{Labels assigned to buffered messages}
    \item $\mathsf{round} \in \mathbb{N}_{\geq 0}$, $\mathsf{round} = 0$ \hfill\cmt{Current protocol round}
\end{itemize}

\noindent \textbf{CheckID}(\emph{pid}, \emph{sid}, \emph{role}): Accept all messages with the same \emph{sid}.

\vspace{0.5em}
\noindent \textbf{Main:}

\vspace{0.5em}

\textbf{recv} \{Message, $\mathit{msg}$\} \textbf{from} I/O:

\begin{enumerate}[itemsep=0.5em]
    \item Sample $\mathit{label} \overset{\$}{\leftarrow} \mathbb{N}$; $\mathsf{labels}.\text{add}(\mathit{label})$; \cmt{Assign a fresh label to the message}

    \item Parse $\mathit{msg} = (\mathit{round}_{\mathit{send}}, \mathit{pid}_{\mathit{call}}, \mathit{receiver}, m)$;

    \item $\mathsf{bufferMsg}.\text{add}(\mathit{label}, \mathit{round}_{\mathit{send}}, \mathit{pid}_{\mathit{call}}, \mathit{receiver}, m)$; \cmt{Buffer the labeled message}
\end{enumerate}

\hrule
\vspace{0.5em}

\textbf{recv} \{Deliver, $\mathit{label}$\} \textbf{from} NET:

\begin{enumerate}[itemsep=0.5em]
    \item Fetch $(\mathit{label}, \mathit{round}_{\mathit{send}}, \mathit{pid}_{\mathit{call}}, \mathit{receiver}, m) \in \mathsf{bufferMsg}$;

    \item $\mathsf{bufferMsg}.\text{remove}(\mathit{label}, \ldots)$; \cmt{Remove the delivered message from buffer}

    \item \textbf{send} $(m, \mathit{pid}_{\mathit{call}})$ \textbf{to} $\mathit{receiver}$;
\end{enumerate}

\hrule
\vspace{0.5em}

\textbf{recv} \{UpdateRound\} \textbf{from} NET:

\begin{enumerate}[itemsep=0.5em]
    \item $\forall\, \mathit{pid}_W$ running the warden role:
    \textbf{send} \{SettlementCheck\} \textbf{to} $(\mathit{pid}_W, \scur, \mathcal{P}^{\text{Brick}}_{\text{warden}}:\text{warden})$; \cmt{Trigger all wardens to perform settlement checks}

    \item $\mathsf{round} \leftarrow \mathsf{round} + 1$;
    \textbf{reply} \{UpdateRound, true\};
\end{enumerate}

\hrule
\vspace{0.5em}

\textbf{recv} \{GetCurRound\} \textbf{from} NET or I/O:

\begin{enumerate}[itemsep=0.5em]
    \item \textbf{reply} \{GetCurRound, $\mathsf{round}$\};
\end{enumerate}

}\end{functionality}\vspace{1em}

 Secondly, in the Liquid Sidechain, the communication is assumed to be synchronous. Therefore, the adversary must notify whether a message can be delivered within a bounded delay~$\delta$; otherwise, the protocol is not allowed to advance its round. Additionally, we also use $\mathcal{F}_{\text{com}}$ to notify according operator to propose a new block every three rounds, which is decided by a leader selection algorithm \emph{SelectL}(). The functionality is defined as follows:

\vspace{1em}\begin{functionality}{Description of protocol $\mathcal{F}^{\text{Liquid}}_{\text{com}} = (\text{com})$}{

\textbf{Participating roles:} \{com\}

\noindent \textbf{Corruption model:} incorruptible

\noindent \textbf{Protocol parameters:}
\begin{itemize}
    \item $\delta \in \mathbb{N}_{\geq 0}$: message delay bound \cmt{synchronous communication bound}
    \item $\mathsf{SelectL}()$: leader selection algorithm
\end{itemize}

}\end{functionality}

\vspace{.5em}\begin{functionality}{Description of $\mathcal{M}^{\text{Liquid}}_{\text{com}}$ of $\mathcal{F}^{\text{Liquid}}_{\text{com}}$}{

\textbf{Implemented role(s):} \{com\}

\noindent \textbf{Subroutines:} $\mathcal{P}^{\text{Liquid}}_{\text{client}}$: client, $\mathcal{P}^{\text{Liquid}}_{\text{operator}}$: operator

\noindent \textbf{Internal state:}
\begin{itemize}
    \item $\mathsf{bufferMsg} \subset \{0,1\}^{*}$, $\mathsf{bufferMsg} = \emptyset$ \hfill\cmt{Buffered messages pending delivery}
    \item $\mathsf{labels} \subset \mathbb{N}$, $\mathsf{labels} = \emptyset$ \hfill\cmt{Labels assigned to buffered messages}
    \item $\mathsf{round} \in \mathbb{N}_{\geq 0}$, $\mathsf{round} = 0$ \hfill\cmt{Current protocol round}
\end{itemize}

\noindent \textbf{CheckID}(\emph{pid}, \emph{sid}, \emph{role}): Accept all messages with the same \emph{sid}.

\vspace{0.5em}
\noindent \textbf{Main:}

\vspace{0.5em}

\textbf{recv} \{Message, $\mathit{msg}$\} \textbf{from} I/O:

\begin{enumerate}[itemsep=0.5em]
    \item Sample $\mathit{label} \overset{\$}{\leftarrow} \mathbb{N}$; $\mathsf{labels}.\text{add}(\mathit{label})$; \cmt{Assign a fresh label to the message}

    \item Parse $\mathit{msg} = (\mathit{round}_{\mathit{send}}, \mathit{pid}_{\mathit{call}}, \mathit{receiver}, m)$;

    \item $\mathsf{bufferMsg}.\text{add}(\mathit{label}, \mathit{round}_{\mathit{send}}, \mathit{pid}_{\mathit{call}}, \mathit{receiver}, m)$; \cmt{Buffer the labeled message}
\end{enumerate}

\hrule
\vspace{0.5em}

\textbf{recv} \{Deliver, $\mathit{label}$\} \textbf{from} NET:

\begin{enumerate}[itemsep=0.5em]
    \item Fetch $(\mathit{label}, \mathit{round}_{\mathit{send}}, \mathit{pid}_{\mathit{call}}, \mathit{receiver}, m) \in \mathsf{bufferMsg}$;

    \item $\mathsf{bufferMsg}.\text{remove}(\mathit{label}, \ldots)$; \cmt{Remove the delivered message from buffer}

    \item \textbf{send} $(m, \mathit{pid}_{\mathit{call}})$ \textbf{to} $\mathit{receiver}$;
\end{enumerate}

\hrule
\vspace{0.5em}

\textbf{recv} \{UpdateRound\} \textbf{from} NET:

\begin{enumerate}[itemsep=0.5em]
    \item \textbf{if} $\exists\, (\_, \mathit{round}_{\mathit{send}}, \_, \_, \_) \in \mathsf{bufferMsg}$ \textbf{s.t.} $\mathit{round}_{\mathit{send}} + \delta \geq \mathsf{round}$:
    \textbf{reply} \{UpdateRound, false\}; \cmt{Message still within delay bound; round cannot advance}

    \item \textbf{else}:
    \begin{itemize}
        \item \textbf{if} $\mathsf{round} \bmod 3 = 0$:
        let $\mathit{pid}_{\mathit{leader}} \leftarrow \mathsf{SelectL}()$;
        \textbf{send} \{UpdateSideChain\} \textbf{to} $(\mathit{pid}_{\mathit{leader}}, \scur, \mathcal{P}^{\text{Liquid}}_{\text{operator}}:\text{operator})$; \cmt{Trigger the newly selected leader to propose a block}

        \item $\mathsf{round} \leftarrow \mathsf{round} + 1$;
        \textbf{reply} \{UpdateRound, true\};
    \end{itemize}
\end{enumerate}

\hrule
\vspace{0.5em}

\textbf{recv} \{GetCurRound\} \textbf{from} NET or I/O:

\begin{enumerate}[itemsep=0.5em]
    \item \textbf{reply} \{GetCurRound, $\mathsf{round}$\};
\end{enumerate}

}\end{functionality}\vspace{1em}

Thirdly, we also use it for the case study of the Arbitrum Nitro Rollup protocol, which enforces that operators periodically publish data on the L1 blockchain by notifying all operator machines whenever the round is updated and the period $T_{\text{period}}$ is reached. Verifiers will be notified during the $T_{\text{challenge}}$ period. But the communication network is not guaranteed to be synchronous.

\vspace{1em}\begin{functionality}{Description of protocol $\mathcal{F}^{\text{Arbitrum}}_{\text{com}} = (\text{com})$}{

\textbf{Participating roles:} \{com\}

\noindent \textbf{Corruption model:} incorruptible

\noindent \textbf{Protocol parameters:}
\begin{itemize}
    \item $T_{\text{period}} \in \mathbb{N}_{\geq 1}$: batch publication period
    \item $T_{\text{challenge}} \in \mathbb{N}_{\geq 1}$: fraud-proof challenge period
\end{itemize}

}\end{functionality}\vspace{.5em}

\begin{functionality}{Description of $\mathcal{M}^{\text{Arbitrum}}_{\text{com}}$ of $\mathcal{F}^{\text{Arbitrum}}_{\text{com}}$}{

\textbf{Implemented role(s):} \{com\}

\noindent \textbf{Subroutines:} $\mathcal{P}^{\text{Arbitrum}}_{\text{operator}}$: operator, $\mathcal{P}^{\text{Arbitrum}}_{\text{verifier}}$: verifier

\noindent \textbf{Internal state:}
\begin{itemize}
    \item $\mathsf{bufferMsg} \subset \{0,1\}^{*}$, $\mathsf{bufferMsg} = \emptyset$ \hfill\cmt{Buffered messages pending delivery}
    \item $\mathsf{labels} \subset \mathbb{N}$, $\mathsf{labels} = \emptyset$ \hfill\cmt{Labels assigned to buffered messages}
    \item $\mathsf{round} \in \mathbb{N}_{\geq 0}$, $\mathsf{round} = 0$ \hfill\cmt{Current protocol round}
\end{itemize}

\noindent \textbf{CheckID}(\emph{pid}, \emph{sid}, \emph{role}): Accept all messages with the same \emph{sid}.

\vspace{0.5em}
\noindent \textbf{Main:}

\vspace{0.5em}

\textbf{recv} \{Message, $\mathit{msg}$\} \textbf{from} I/O:

\begin{enumerate}[itemsep=0.5em]
    \item Sample $\mathit{label} \overset{\$}{\leftarrow} \mathbb{N}$; $\mathsf{labels}.\text{add}(\mathit{label})$; \cmt{Assign a fresh label to the message}

    \item Parse $\mathit{msg} = (\mathit{round}_{\mathit{send}}, \mathit{pid}_{\mathit{call}}, \mathit{receiver}, m)$;

    \item $\mathsf{bufferMsg}.\text{add}(\mathit{label}, \mathit{round}_{\mathit{send}}, \mathit{pid}_{\mathit{call}}, \mathit{receiver}, m)$; \cmt{Buffer the labeled message}
\end{enumerate}

\hrule
\vspace{0.5em}

\textbf{recv} \{Deliver, $\mathit{label}$\} \textbf{from} NET:

\begin{enumerate}[itemsep=0.5em]
    \item Fetch $(\mathit{label}, \mathit{round}_{\mathit{send}}, \mathit{pid}_{\mathit{call}}, \mathit{receiver}, m) \in \mathsf{bufferMsg}$;

    \item $\mathsf{bufferMsg}.\text{remove}(\mathit{label}, \ldots)$; \cmt{Remove the delivered message from buffer}

    \item \textbf{send} $(m, \mathit{pid}_{\mathit{call}})$ \textbf{to} $\mathit{receiver}$;
\end{enumerate}

\hrule
\vspace{0.5em}

\textbf{recv} \{UpdateRound\} \textbf{from} NET:

\begin{enumerate}[itemsep=0.5em]
    \item \textbf{if} $(\mathsf{round} + 1) \bmod T_{\text{period}} = 0$:
    $\forall\, \mathit{pid}_O$ running the operator role:
    \textbf{send} \{UpdateRequest\} \textbf{to} $(\mathit{pid}_O, \scur, \mathcal{P}^{\text{Arbitrum}}_{\text{operator}}:\text{operator})$; \cmt{Trigger operators to publish a new batch}

    \item \textbf{if} $(\mathsf{round} + 1) \bmod (T_{\text{period}} + T_{\text{challenge}}) \neq 0$:
    $\forall\, \mathit{pid}_V$ running the verifier role:
    \textbf{send} \{UpdateCheck\} \textbf{to} $(\mathit{pid}_V, \scur, \mathcal{P}^{\text{Arbitrum}}_{\text{verifier}}:\text{verifier})$; \cmt{Trigger verifiers to check published batches}

    \item $\mathsf{round} \leftarrow \mathsf{round} + 1$;
    \textbf{reply} \{UpdateRound, true\};
\end{enumerate}

\hrule
\vspace{0.5em}

\textbf{recv} \{GetCurRound\} \textbf{from} NET or I/O:

\begin{enumerate}[itemsep=0.5em]
    \item \textbf{reply} \{GetCurRound, $\mathsf{round}$\};
\end{enumerate}

}\end{functionality}\vspace{1em}

Finally, in addition to the functionality used in Arbitrum Nitro, we further employ the functionality to trigger the watchtowers in our \xr protocol to publish fast-finality guarantees to the L1 blockchain. We denote this functionality by \(\mathcal{F}^{\text{FRoll}}_{\text{com}}\), defined as follows:

\vspace{1em}\begin{functionality}{Description of protocol $\mathcal{F}^{\text{FRoll}}_{\text{com}} = (\text{com})$}{

\textbf{Participating roles:} \{com\}

\noindent \textbf{Corruption model:} incorruptible

\noindent \textbf{Protocol parameters:}
\begin{itemize}
    \item $T_{\text{period}} \in \mathbb{N}_{\geq 1}$: batch publication period
\end{itemize}

}\end{functionality}
\vspace{.5em}

\begin{functionality}{Description of $\mathcal{M}^{\text{FRoll}}_{\text{com}}$ of $\mathcal{F}^{\text{FRoll}}_{\text{com}}$}{

\textbf{Implemented role(s):} \{com\}

\noindent \textbf{Subroutines:} $\mathcal{P}^{\text{FRoll}}_{\text{operator}}$: operator, $\mathcal{P}^{\text{FRoll}}_{\text{client}}$: verifier, $\mathcal{P}^{\text{FRoll}}_{\text{watchtower}}$: watchtower

\noindent \textbf{Internal state:}
\begin{itemize}
    \item $\mathsf{bufferMsg} \subset \{0,1\}^{*}$, $\mathsf{bufferMsg} = \emptyset$ \hfill\cmt{Buffered messages pending delivery}
    \item $\mathsf{labels} \subset \mathbb{N}$, $\mathsf{labels} = \emptyset$ \hfill\cmt{Labels assigned to buffered messages}
    \item $\mathsf{round} \in \mathbb{N}_{\geq 0}$, $\mathsf{round} = 0$ \hfill\cmt{Current protocol round}
\end{itemize}

\noindent \textbf{CheckID}(\emph{pid}, \emph{sid}, \emph{role}): Accept all messages with the same \emph{sid}.

\vspace{0.5em}
\noindent \textbf{Main:}

\vspace{0.5em}

\textbf{recv} \{Message, $\mathit{msg}$\} \textbf{from} I/O:

\begin{enumerate}[itemsep=0.5em]
    \item Sample $\mathit{label} \overset{\$}{\leftarrow} \mathbb{N} \setminus \mathsf{labels}$; $\mathsf{labels}.\text{add}(\mathit{label})$; \cmt{Assign a fresh unique label}

    \item Parse $\mathit{msg} = (\mathit{round}_{\mathit{send}}, \mathit{pid}_{\mathit{call}}, \mathit{receiver}, m)$;

    \item $\mathsf{bufferMsg}.\text{add}(\mathit{label}, \mathit{round}_{\mathit{send}}, \mathit{pid}_{\mathit{call}}, \mathit{receiver}, m)$; \cmt{Buffer the labeled message}
\end{enumerate}

\hrule
\vspace{0.5em}

\textbf{recv} \{Deliver, $\mathit{label}$\} \textbf{from} NET, \textbf{s.t.} $\mathit{label} \in \mathsf{labels}$:

\begin{enumerate}[itemsep=0.5em]
    \item Fetch $(\mathit{label}, \mathit{round}_{\mathit{send}}, \mathit{pid}_{\mathit{call}}, \mathit{receiver}, m) \in \mathsf{bufferMsg}$;

    \item $\mathsf{bufferMsg}.\text{remove}(\mathit{label}, \ldots)$; \cmt{Remove the delivered message from buffer}

    \item \textbf{send} $(m, \mathit{pid}_{\mathit{call}})$ \textbf{to} $\mathit{receiver}$;
\end{enumerate}

\hrule
\vspace{0.5em}

\textbf{recv} \{UpdateRound\} \textbf{from} NET:

\begin{enumerate}[itemsep=0.5em]
    \item \textbf{if} $(\mathsf{round} + 1) \bmod T_{\text{period}} = 0$:
    \begin{itemize}
        \item $\forall\, \mathit{pid}_O$ running the operator role:
        \textbf{send} \{UpdateRequest\} \textbf{to} $(\mathit{pid}_O, \scur,\\ \mathcal{P}^{\text{FRoll}}_{\text{operator}}:\text{operator})$; \cmt{Trigger operators to publish a new batch}

        \item $\forall\, \mathit{pid}_V$ running the verifier role:
        \textbf{send} \{UpdateCheck\} \textbf{to} $(\mathit{pid}_V, \scur,\\ \mathcal{P}^{\text{FRoll}}_{\text{client}}:\text{verifier})$; \cmt{Trigger verifiers to check published batches}

        \item $\forall\, \mathit{pid}_W$ running the watchtower role:
        \textbf{send} \{CheckFastL1\} \textbf{to} $(\mathit{pid}_W, \scur,\\ \mathcal{P}^{\text{FRoll}}_{\text{watchtower}}:\text{watchtower})$; \cmt{Trigger watchtowers to check fast-finality transactions on L1}
    \end{itemize}

    \item $\mathsf{round} \leftarrow \mathsf{round} + 1$;
    \textbf{reply} \{UpdateRound, true\};
\end{enumerate}

\hrule
\vspace{0.5em}

\textbf{recv} \{GetCurRound\} \textbf{from} NET or I/O:

\begin{enumerate}[itemsep=0.5em]
    \item \textbf{reply} \{GetCurRound, $\mathsf{round}$\};
\end{enumerate}

}\end{functionality}\vspace{1em}

\section{Case Study: The Brick Channel}
\label{apdx:brick}

\subsection{$\mathcal{F}_{\text{ledger}}$ instantiation for Brick channel}
\label{apd:BrickLedger}

\subsubsection*{Submission functionality $\mathcal{F}_{\text{submit}_{\text{L1}}}$}

The submission subroutine accepts a request to submit a transaction to L1 if and only if the transaction conforms to one of the Brick transaction types. For the Brick channel, following transactions that are sent to the contract address: channel opening ($\mathit{TX}_{\mathit{open}}$), warden collateral deposits ($\mathit{TX}_{\mathit{collateral}}$), collaborative closing ($\mathit{TX}_{\mathit{close}}$), unilateral closing ($\mathit{TX}_{\mathit{unilateral}}$), warden settlement transactions ($\mathit{TX}_{\mathit{settle}}$), and fraud proofs ($\mathit{TX}_{\mathit{fraud}}$) are accepted by $\mathcal{F}^{\text{Brick}}_{\text{submit}_{\text{L1}}}$.

\subsubsection*{Update functionality $\mathcal{F}_{\text{update}_{\text{L1}}}$}

The update subroutine refines the base $\mathcal{F}^{\text{Brick}}_{\text{update}_{\text{L1}}}$ by enforcing Brick-specific state-commitment semantics when resolving a closing event. Specifically, upon resolving a settlement trigger (i.e., once $\mathit{TX}_{\mathit{close}}$ or $\mathit{TX}_{\mathit{unilateral}}$ reaches finality on L1), $\mathcal{F}^{\text{Brick}}_{\text{update}_{\text{L1}}}$ determines the committed state as follows:
\begin{itemize}
    \item \textbf{Collaborative closing.} If $\mathit{TX}_{\mathit{close}}$ is committed on L1, the committed state $\mathit{State}_{\text{L1}}$ is set to the state $s_L$ carried in $\mathit{TX}_{\mathit{close}}$ (validated by the two client signatures $\sigma, \sigma'$).
    \item \textbf{Unilateral closing.} If $\mathit{TX}_{\mathit{unilateral}}$ is committed on L1,\\ $\mathcal{F}^{\text{Brick}}_{\text{update}_{\text{L1}}}$ collects all warden settlement transactions $\{\mathit{TX}_{\mathit{settle}}^{(j)}\}$ committed on L1 within the settlement window and verifies the attached warden signature on each. The according committed state is then set to $s^{*}$ where $(s^{*}, i^{*})$ is the tuple with the \emph{highest sequence number $i^{*}$} among at least $2f{+}1$ distinct, signature-valid warden settlement transactions. 
    \item \textbf{Fraud proof.} If any $\mathit{TX}_{\mathit{fraud}}$ committed on L1 demonstrates that a warden $W_j$ signed two distinct states with overlapping sequence numbers, the conflicting $\mathit{TX}_{\mathit{settle}}^{(j)}$ is excluded from the settlement aggregation above, and $W_j$'s collateral $c$ is forfeited.
\end{itemize}

\subsubsection*{Preservation of L1 safety and liveness.}

The instantiation preserves the original security guarantees of $\mathcal{F}_{\text{ledger}}$. All Brick-specific logic is additive: new transaction types are accepted as valid L1 payloads, and the settlement-resolution rule is a deterministic function of the transactions already committed on L1. No existing L1 transaction is rejected or reordered by the instantiation, so the underlying ledger's self-consistency and view-consistency (L1 safety) are inherited without change. Similarly, because each Brick transaction is submitted through the unchanged $\text{SubmitL1}$ interface and included in $\{\mathit{TX}\}_{\text{L1}}$ under the same conditions as any other L1 transaction, the liveness bound of the underlying ledger carries over: any honest submission is reflected on L1 within $T_{L_1}$. The highest-sequence-number rule for unilateral settlement is enforced on data that is already final on L1, which means the settlement outcome cannot diverge across honest views so long as L1 itself does not. Consequently, $\mathcal{F}^{\text{Brick}}_{\text{ledger}}$ still realizes the same $(f_{L_1})$-safety and liveness based on the specification above.

\subsection{Brick Channel Real Protocol}
\label{Apd:Brickreal}

We formally define the real-world protocol implementation $\mathcal{P}^{\text{Brick}}$ as
$\mathcal{P}^{\text{Brick}}:=(\mathcal{P}^{\text{Brick}}_{\text{client}} \mid \mathcal{P}^{\text{Brick}}_{\text{warden}},\mathcal{F}_{\text{sig}}, \mathcal{F}^{\text{Brick}}_{\text{com}})$.
The client and warden are the two participating roles. Clients are the main parties instructed by the environment to proceed with the protocol; wardens are service parties that prevent incorrect state settlement to the L1 blockchain. We assume $n=3f+1$ wardens in total, of which the adversary may corrupt at most~$f$. The structure is shown in Figure~\ref{fig:realbrick}.

\subsubsection{Client Protocol}

The client machine $\mathcal{M}^{\text{Brick}}_{\text{client}}$ defines the code run by the two main parties of the Brick channel that perform off-chain transactions.

\vspace{1em}\begin{functionality}{Description of protocol $\mathcal{P}^{\text{Brick}}_{\text{client}} = (\text{client})$}{

\textbf{Participating roles:} \{client\}

\noindent \textbf{Corruption model:} dynamic corruption, at least one honest

\noindent \textbf{Protocol parameters:}
\begin{itemize}
    \item $n = 3f+1$: warden committee size
    \item $f$: corruption threshold for wardens
\end{itemize}

}\end{functionality}\vspace{.5em}
\begin{functionality}{Description of $\mathcal{M}^{\text{Brick}}_{\text{client}}$}{

\textbf{Implemented role(s):} \{client\}

\noindent \textbf{Subroutines:} $\mathcal{F}_{\text{sig}}$: \{signer, verifier\}, $\mathcal{F}^{\text{Brick}}_{\text{com}}$: com, $\mathcal{F}_{\text{ledger}}$: client\textsubscript{L1}, $\mathcal{P}^{\text{Brick}}_{\text{warden}}$: warden

\noindent \textbf{Internal state:}
\begin{itemize}
    \item $\mathsf{round} \in \mathbb{N}_{\geq 0}$, $\mathsf{round} = 0$ \hfill\cmt{Current round}
    \item $\mathsf{requestQueue} \subset \{0,1\}^{*}$, $\mathsf{requestQueue} = \emptyset$ \hfill\cmt{Queued unexecuted requests}
    \item $\mathsf{executedRequest} \subset \{0,1\}^{*}$, $\mathsf{executedRequest} = \emptyset$ \hfill\cmt{Executed requests with certificates}
    \item $\mathsf{stateList} \subset \{0,1\}^{*} \times \mathbb{N}$, $\mathsf{stateList} = \emptyset$ \hfill\cmt{L2 state list of form $(s, i)$}
    \item $\mathsf{onchainState} \subset \{0,1\}^{*}$, $\mathsf{onchainState} = \emptyset$ \hfill\cmt{L1 committed state}
    \item $\mathsf{identities} \subset \{0,1\}^{*}$, $\mathsf{identities} = \{W_1, \ldots, W_n, \emph{ContractAddr}\}$ \hfill\cmt{Registered participants}
    \item $\mathsf{lastReadPointer} \in \mathbb{N}_{\geq 0}$, $\mathsf{lastReadPointer} = 0$ \hfill\cmt{Read location pointer}
    \item $\mathsf{wardenCounter} \in \mathbb{N}_{\geq 0}$, $\mathsf{wardenCounter} = 0$ \hfill\cmt{Warden signature counter}
    \item $\mathsf{wardenSig} \subset \mathbb{N} \times \{0,1\}^{*}$, $\mathsf{wardenSig} = \emptyset$ \hfill\cmt{Collected warden signatures}
\end{itemize}

\noindent \textbf{CheckID}(\emph{pid}, \emph{sid}, \emph{role}): Accept all messages with the same \emph{sid}.

\noindent \textbf{Corruption behavior:}
\begin{itemize}
    \item \textbf{DetermineCorrStatus}(\emph{pid}, \emph{sid}, \emph{role}): Return $\mathsf{corr}$.
    \item \textbf{LeakedData}(\emph{pid}, \emph{sid}, \emph{role}): Return \texttt{Internal State}.
\end{itemize}

\vspace{0.5em}
\noindent \textbf{Main:}

\vspace{0.5em}

\textbf{recv} \{Join, $s_{\mathit{init}}=\{s_{\mathit{int}_A},s_{\mathit{int}_B}\}$, \emph{counterIdentity}\} \textbf{from} I/O, \textbf{s.t.} $\mathsf{stateList} = \emptyset$:

\begin{enumerate}[itemsep=0.5em]

\item \textbf{send} \{Sign, $\{s_{\mathit{init}},1\}$\} \textbf{to} $(\pcur, \scur, \mathcal{F}_{\text{sig}}:\text{signer})$, and
\textbf{wait for} \{Signature, $\sigma$\};

\item \textbf{send} \{Message, \{Join, $s_{\mathit{init}}$, $\sigma$, \emph{receiver}\}\} \textbf{to} $(\pcur, \scur, \mathcal{F}^{\text{Brick}}_{\text{com}}:\text{com})$, where $\mathit{receiver} = (\mathit{counterIdentity}, \scur, \mathcal{P}^{\text{Brick}}_{\text{client}}:\text{client})$; \cmt{Send to counterparty}

\item \texttt{reqeustQueue}.add(\{Join,$s_{\mathit{init}}$\});

\end{enumerate}

\hrule
\vspace{0.5em}

\textbf{recv} \{Join, $s_{\mathit{init}}$, $\sigma'$, $\mathit{pid}_{\mathit{call}}$\} \textbf{from} $(\pcur, \scur, \mathcal{F}^{\text{Brick}}_{\text{com}}: \text{com})$:

\begin{enumerate}[itemsep=0.5em]

\item \textbf{send} \{Verify, $\{s_{\mathit{init}},1\}$, $\sigma'$\} \textbf{to} $(\pcur, \scur, \mathcal{F}_{\text{sig}}:\text{verifier})$,
\textbf{wait for} \{VerResult, $b$\};

\item \textbf{if} $b =$ true and $\exists$ \{Join,\{$s_{\mathit{init}}$\}\}$\in$ \texttt{requestQueue}:\\ \textbf{send} \{Message, \{$s_{\mathit{init}}$, $\sigma$, $\sigma'$, \emph{receiver}\}\} \textbf{to} $(\pcur, \scur, \mathcal{F}^{\text{Brick}}_{\text{com}}:\text{com})$, where $\mathit{receiver} \\= (\_, \scur, \mathcal{P}^{\text{Brick}}_{\text{warden}}: \text{warden})$; \cmt{Send to all wardens}

\end{enumerate}

\hrule
\vspace{0.5em}

\textbf{recv} \{Join, $s_{\mathit{init}}$, $\sigma_{W_i}$, $\mathit{pid}_{\mathit{call}}$\} \textbf{from} $(\pcur, \scur, \mathcal{F}^{\text{Brick}}_{\text{com}}: \text{com})$:

\begin{enumerate}[itemsep=0.5em]

\item \textbf{send} \{Verify, \{$s_{int}$, 1\}, $\sigma_{W_i}$\} \textbf{to} $(\pcur, \scur, \mathcal{F}_{\text{sig}}:\text{verifier})$ and
\textbf{wait for} \{VerResult, $b$\};

\item \textbf{if} $b =$ true:
$\mathsf{wardenCounter} \leftarrow \mathsf{wardenCounter} + 1$;
$\mathsf{wardenSig}.\text{add}(i, \sigma_{W_i})$;

\item \textbf{if} $\mathsf{wardenCounter} = 2f+1$:
\begin{itemize}
    \item \textbf{send} \{SubmitL1, $\mathit{TX}_{\mathit{open}}=\{\pcur, \emph{ContractAddr}, s_{\mathit{int}_A}, \{W_1, \ldots, W_n, f\}\}$\} \textbf{to} $(\pcur,\scur, \mathcal{F}_{\text{ledger}}: \text{client}_{\text{L1}})$;

    \item \textbf{send} \{ReadL1\} \textbf{to} $(\pcur,\scur, \mathcal{F}_{\text{ledger}}: \text{client}_{\text{L1}})$,
    \textbf{wait for} \{ReadL1, $\mathit{output}=\{\{\mathit{TX}_{\text{L1}}\}, \{state_{\text{L1}}\}, \{\mathit{pid}\}\}$\};

    \item \textbf{if} $\mathit{TX}_{\mathit{open}} \in \{\mathit{TX}_{\text{L1}}\}$ and $s_{\mathit{init}} \in \{state_{\text{L1}}\}$:
    $\mathsf{stateList}\leftarrow \{s_{\mathit{init}},1\}$;
    \\$\mathsf{onchainState}\leftarrow s_{\mathit{init}}$;
    $\mathsf{identities}.\text{add}(\mathit{counterIdentity})$;
    $\mathsf{wardenCounter} \leftarrow 0$;
    \textbf{reply} \{Join, $s_{\mathit{init}}$\} via I/O;
\end{itemize}

\end{enumerate}

\hrule
\vspace{0.5em}

\textbf{recv} \{Submit, $\{s=\{s_A,s_B\},i\}$\} \textbf{from} I/O, \textbf{s.t.} $s_A + s_B = s_{\mathit{int}_A} + s_{\mathit{int}_B}$ and $\nexists \{s', j\} \in \mathsf{stateList}$ with $j \geq i$:

\begin{enumerate}[itemsep=0.5em]

\item \textbf{send} \{Sign, $\{s,i\}$\} \textbf{to} $(\pcur, \scur, \mathcal{F}_{\text{sig}}: \text{signer})$,
\textbf{wait for} \{Signature, $\sigma$\};

\item $\mathsf{requestQueue}.\text{add}(\{s,i\})$;

\item \textbf{send} \{Message, \{Submit, $\{s,i\}$, $\sigma$, \emph{receiver}\}\} \textbf{to} $(\pcur, \scur, \mathcal{F}^{\text{Brick}}_{\text{com}}:\text{com})$, where $\mathit{receiver} = (\mathit{pid}_{\mathit{counter}}, \scur, \mathcal{P}^{\text{Brick}}_{\text{client}}:\text{client})$; \cmt{Send to counterparty}
\end{enumerate}

\hrule
\vspace{0.5em}

\textbf{recv} \{Submit, $\{s,i\}$, $\sigma'$, $\mathit{pid}_{\mathit{call}}$\} \textbf{from} $(\pcur, \scur, \mathcal{F}^{\text{Brick}}_{\text{com}}: \text{com})$:

\begin{enumerate}[itemsep=0.5em]
    \item \textbf{send} \{Verify, $\{s,i\}$, $\sigma'$\} \textbf{to} $(\pcur, \scur, \mathcal{F}_{\text{sig}}: \text{verifier})$,
    \textbf{wait for} \{VerResult, $b$\};
    
    \item \textbf{if} $b =$ true: \textbf{send} \{Message, \{$\{s,i\}$, $\sigma$, $\sigma'$, \emph{receiver}\}\} \textbf{to} $(\pcur, \scur, \mathcal{F}^{\text{Brick}}_{\text{com}}:\text{com})$, where $\mathit{receiver} = (\_, \scur, \mathcal{P}^{\text{Brick}}_{\text{warden}}:\text{warden})$; \cmt{Send to all wardens}

\end{enumerate}

\hrule
\vspace{0.5em}

\textbf{recv} \{Update, $\{s, i\}$, $\sigma_{W_j}$, $\mathit{pid}_{\mathit{call}}$\} \textbf{from} $(\_, \scur, \mathcal{F}^{\text{Brick}}_{\text{com}}: \text{com})$:

\begin{enumerate}[itemsep=0.5em]

\item \textbf{send} \{Verify, \{s,i\},$\sigma_{W_j}$\} \textbf{to} $(\pcur, \scur, \mathcal{F}_{\text{sig}}:\text{verifier})$,
\textbf{wait for} \{VerResult, $b$\};

\item \textbf{if} $b =$ true:
$\mathsf{wardenCounter}\leftarrow \mathsf{wardenCounter} + 1$;
$\mathsf{wardenSig}[j]\leftarrow({s,i}, \sigma_{W_j})$;

\item \textbf{if} $\mathsf{wardenCounter} = 2f+1$:
$\mathsf{stateList}\leftarrow \{s, i\} $;
$\mathsf{executedRequest}.\text{add}(\{\{s, i\}, \sigma, \sigma',\\ \{\sigma_W\}_{2f+1}\})$;
$\mathsf{requestQueue}.\text{remove}(\{s, i\})$;

\end{enumerate}

\hrule
\vspace{0.5em}

\textbf{recv} \{Settlement, \emph{collaborate}\} \textbf{from} I/O:

\begin{enumerate}[itemsep=0.5em]

\item Let
$
\{s_L,i_L\} = \mathsf{stateList}
$
be such that
$
\forall \{s,i\}\in \mathsf{stateList},\ i\leq i_L.
$
\Comment{$s_L$ is the latest state in $\mathsf{stateList}$}.

\item $\mathsf{requestQueue}$.add(\{Settlement, \texttt{collaborate}, $\{s_L,i_L\}$\});

\item \textbf{send} \{Sign, $s_L$\} \textbf{to} $(\pcur, \scur, \mathcal{F}_{\text{sig}}: \text{signer})$,
\textbf{wait for} \{Signature, $\sigma$\};

\item \textbf{send} \{Message, \{Settlement, $s_L$, $\sigma$, \emph{receiver}\}\} \textbf{to} $(\pcur, \scur, \mathcal{F}^{\text{Brick}}_{\text{com}}: \text{com})$, where $\mathit{receiver} = (\mathit{pid}_{\mathit{counter}}, \scur, \mathcal{P}^{\text{Brick}}_{\text{client}}: \text{client})$; \cmt{Send to counterparty}
\end{enumerate}

\hrule
\vspace{0.5em}

\textbf{recv} \{Settlement, $s_L$, $\sigma'$, $\mathit{pid}_{\mathit{call}}$\} \textbf{from} $(\_, \scur, \mathcal{F}^{\text{Brick}}_{\text{com}}: \text{com})$:

\begin{enumerate}
    \item \textbf{send} \{Verify, $\sigma'$\} \textbf{to} $(\pcur, \scur, \mathcal{F}_{\text{sig}}:\text{verifier})$,
\textbf{wait for} \{VerResult, $b$\};

    \item If $b$ = True, prepare transaction $\mathit{TX}_{\mathit{close}}=\{\pcur, \emph{ContractAddr}, \epsilon, \{Settlement, s_L,\\ \sigma, \sigma'\}\}$; \cmt{Prepare the channel closing transaction}
    
    \item \textbf{send} \{SubmitL1, $\mathit{TX}_{\mathit{close}}=\{s_L, \sigma, \sigma'\}$\} \textbf{to} $(\pcur, \scur, \mathcal{F}_{\text{ledger}}:\text{client}_{\text{L1}})$;
    
    \item \textbf{send} \{ReadL1\} \textbf{to} $(\pcur, \scur, \mathcal{F}_{\text{ledger}}:\text{client}_{\text{L1}})$, \textbf{wait for} \{ReadL1, $\mathit{output}=\{\{\mathit{TX}_{\text{L1}}\}, \{state_{\text{L1}}\}, \{\mathit{pid}\}\}$\};

    \item \textbf{if} $\mathit{TX}_{\mathit{close}} \in \{\mathit{TX}_{\text{L1}}\}$, then \textbf{reply} \{Settlement, \texttt{collaborate}, $s_L$\} via I/O;
\end{enumerate}

\hrule
\vspace{0.5em}

\textbf{recv} \{Settlement, \emph{unilateral}\} \textbf{from} I/O:

\begin{enumerate}[itemsep=0.5em]

\item Let
$
\{s_L,i_L\} = \mathsf{stateList}
$
be such that
$
\forall \{s,i\}\in \mathsf{stateList},\ i\leq i_L.
$;

\item \textbf{send} \{SubmitL1, $\mathit{TX}_{\mathit{unilateral}}$\} \textbf{to} $(\pcur, \scur, \mathcal{F}_{\text{ledger}}:\text{client}_{\text{L1}})$;

\item \textbf{send} \{ReadL1\} \textbf{to} $(\pcur, \scur, \mathcal{F}_{\text{ledger}}:\text{client}_{\text{L1}})$,
\textbf{wait for} \{ReadL1, $\mathit{output}=\{\{\mathit{TX}_{\text{L1}}\}, \{state_{\text{L1}}\}, \{\mathit{pid}\}\}$\};

\item \textbf{if} $\exists$ at least $2f{+}1$ transactions $\{\mathit{TX}_{\mathit{settle}}^{(j)}\}_{j \in W}$ from $2f{+}1$ distinct wardens $W \subseteq \{W_1, \ldots, W_n\}$ in $\{\mathit{TX}_{\text{L1}}\}$, \textbf{s.t.} $s_L$ is included in each $\mathit{TX}_{\mathit{settle}}^{(j)}$:

\begin{itemize}
    \item \textbf{for all} $j \in W$ \textbf{do}:
\begin{itemize}
    \item \textbf{if} $\exists\, (\sigma'_j, \{s', i'\}) \in \mathsf{wardenSig}[j]$ \textbf{s.t.} $i' > i^{(j)}$:
    \begin{itemize}
        \item Create $\mathit{TX}_{\mathit{fraud}} = \{\pcur, \emph{ContractAddr}, \epsilon,  \{\{s', i'\}, \sigma'_j\}\}$; \cmt{Fraud proof: warden $j$ signed a higher sequence number}
        \item \textbf{send} \{SubmitL1, $\mathit{TX}_{\mathit{fraud}}$\} \textbf{to} $(\pcur, \scur, \mathcal{F}_{\text{ledger}}:\text{client}_{\text{L1}})$;

        \item \textbf{reply} \{Settlement, \texttt{unilateral}, $s_L$\} via I/O;
    \end{itemize}
\end{itemize}
\end{itemize}

\end{enumerate}

\hrule
\vspace{0.5em}

\textbf{recv} \{Read\} \textbf{from} I/O:

\begin{enumerate}[itemsep=0.5em]
\item \textbf{reply} \{Read, $\mathsf{executedRequest}$, $\mathsf{stateList}$, $\mathsf{onchainState}$\} via I/O;
\end{enumerate}

\hrule
\vspace{0.5em}

\textbf{recv} \{GetCurRound\} \textbf{from} I/O:
\begin{enumerate}[itemsep=0.5em]
\item \textbf{send} \{GetCurRound\} \textbf{to} $(\pcur, \scur, \mathcal{F}^{\text{Brick}}_{\text{com}}: \text{clock})$,
\textbf{wait for} \{GetCurRound, $\mathit{round}$\};
\item \textbf{reply} \{GetCurRound, $\mathit{round}$\} via I/O;
\end{enumerate}

}\end{functionality}\vspace{1em}

\begin{figure}[!h]
    \centering
    \includegraphics[width=\columnwidth]{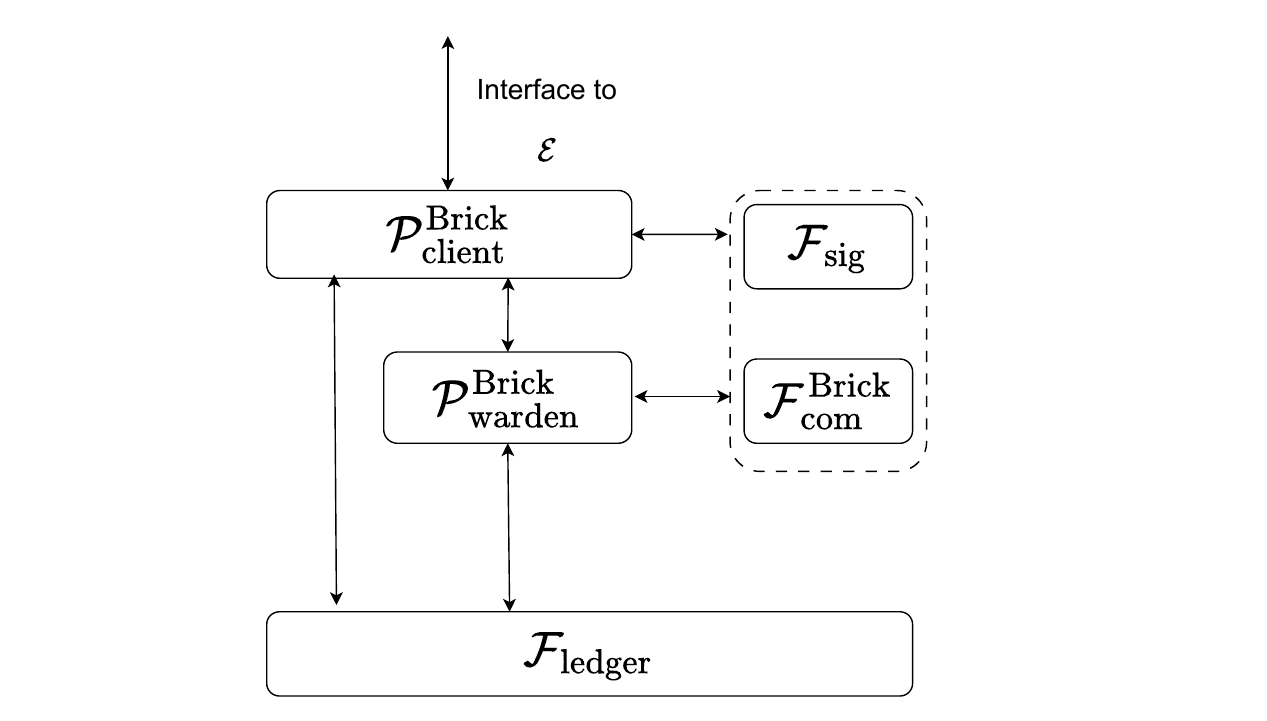}
    \caption{The Brick payment channel protocol}
    \label{fig:realbrick}
\end{figure}

\subsubsection{Warden Protocol}

The warden machine $\mathcal{M}^{\text{Brick}}_{\text{warden}}$ defines the code run by the third-party wardens. Wardens do not receive inputs from the environment directly but only communicate with clients through the asynchronous communication channel $\mathcal{F}^{\text{Brick}}_{\text{com}}$, in which message delivery is controlled by the adversary.

\vspace{1em}\begin{functionality}{Description of protocol $\mathcal{P}^{\text{Brick}}_{\text{warden}} = (\text{warden})$}{

\textbf{Participating roles:} \{warden\}

\noindent \textbf{Corruption model:} corrupt no more than $f$

\noindent \textbf{Protocol parameters:}
\begin{itemize}
    \item $n = 3f+1$: warden committee size
    \item $f$: corruption threshold

    \item $c$: warden collateral
\end{itemize}

}\end{functionality}\vspace{.5em}

\begin{functionality}{Description of $\mathcal{M}^{\text{Brick}}_{\text{warden}}$}{

\textbf{Implemented role(s):} \{warden\}

\noindent \textbf{Subroutines:} $\mathcal{F}_{\text{sig}}$: \{signer, verifier\}, $\mathcal{F}_{\text{ledger}}$: client\textsubscript{L1}, $\mathcal{P}^{\text{Brick}}_{\text{client}}$: client, $\mathcal{F}^{\text{Brick}}_{\text{com}}$: com

\noindent \textbf{Internal state:}
\begin{itemize}
    \item $\mathsf{executedRequest} \subset \{0,1\}^{*}$ 
    \item $\mathsf{stateList} \subset \{0,1\}^{*} \times \mathbb{N}$, $\mathsf{stateList} = \emptyset$ \hfill\cmt{L2 state list of form $(s, i)$}

    \item $\mathsf{identities} \subset \{0,1\}^{*}$, $\mathsf{identities} = \{client_A, client_B, \emph{ContractAddr}\}$ \hfill\cmt{Two clients and contract address}
\end{itemize}

\noindent \textbf{CheckID}(\emph{pid}, \emph{sid}, \emph{role}): Accept all messages with the same \emph{sid}. Only accept \emph{role} = warden.

\vspace{0.5em}
\noindent \textbf{Main:}

\vspace{0.5em}

\textbf{recv} \{Join, $s_{\mathit{init}}$, $\sigma$, $\sigma'$, $\mathit{pid}_{\mathit{call}}$\} \textbf{from} $(\pcur, \scur, \mathcal{F}^{\text{Brick}}_{\text{com}}: \text{com})$:

\begin{enumerate}[itemsep=0.5em]

\item \textbf{send} \{Verify, $\{s_{\mathit{init}},1\}$, $\{\sigma, \sigma'\}$\} \textbf{to} $(\pcur, \scur, \mathcal{F}_{\text{sig}}: \text{verifier})$,
\textbf{wait for} \{VerResult, $b$\};

\item \textbf{if} $b =$ true: \textbf{send} \{Sign, $\{s_{\mathit{init}},1\}$\} \textbf{to} $(\pcur, \scur, \mathcal{F}_{\text{sig}}: \text{signer})$,
\textbf{wait for} \{Signature, $\sigma_W$\};

\item \textbf{send} \{Message, \{$\sigma_W$, \emph{receiver}\}\} \textbf{to} $(\pcur, \scur, \mathcal{F}^{\text{Brick}}_{\text{com}}: \text{com})$, where $\mathit{receiver} = (\mathsf{pid_{sender}}, \scur, \mathcal{P}^{\text{Brick}}_{\text{client}}: \text{client})$; \cmt{Reply to the client}

\item Prepare deposit transaction $\mathit{TX}_{\mathit{collateral}}=\{\pcur, \emph{ContractAddr}, c, \_\}$; 

\item \textbf{send} \{SubmitL1, $\mathit{TX}_{\mathit{collateral}}$\} \textbf{to} $(\pcur, \scur, \mathcal{F}_{\text{ledger}}:\text{client}_{\text{L1}})$; \cmt{Send deposit to blockchain}

\end{enumerate}

\hrule
\vspace{0.5em}

\textbf{recv} \{Update, $\{s,i\}$, $\sigma$, $\sigma'$, $\mathit{pid}_{\mathit{call}}$\} \textbf{from} $(\_, \scur, \mathcal{F}^{\text{Brick}}_{\text{com}}: \text{com})$, \textbf{s.t.} $s_A + s_B = s_{\mathit{int}_A} + s_{\mathit{int}_B}$ and $\nexists \{s', j\} \in \mathsf{stateList}$ with $j \geq i$:

\begin{enumerate}[itemsep=0.5em]

\item \textbf{send} \{Verify, $\{s,i\}$, $\{\sigma, \sigma'\}$\} \textbf{to} $(\pcur, \scur, \mathcal{F}_{\text{sig}}: \text{verifier})$,
\textbf{wait for} \{VerResult, $b$\};

\item \textbf{if} $b =$ true: \textbf{send} \{Sign, $\{s,i\}$\} \textbf{to} $(\pcur, \scur, \mathcal{F}_{\text{sig}}: \text{signer})$,
\textbf{wait for} \{Signature, $\sigma_W$\};

\item \textbf{send} \{Message, \{$\sigma_W$, \emph{receiver}\}\} \textbf{to} $(\pcur, \scur, \mathcal{F}^{\text{Brick}}_{\text{com}}: \text{com})$, where $\mathit{receiver} = (\mathsf{pid_{sender}}, \scur, \mathcal{P}^{\text{Brick}}_{\text{client}}: \text{client})$; \cmt{Reply to the client}

\item $\mathsf{executedRequest}.\text{add}(\{\{s,i\}, \sigma, \sigma', \sigma_W\})$, $\mathsf{stateList}\leftarrow \{s,i\}$;

\end{enumerate}

\hrule
\vspace{0.5em}

\textbf{recv} \{SettlementCheck, $\mathit{pid}_{\mathit{call}}$\} \textbf{from} $(\_, \scur, \mathcal{F}^{\text{Brick}}_{\text{com}}: \text{com})$:

\begin{enumerate}[itemsep=0.5em]

\item \textbf{send} \{ReadL1\} \textbf{to} $(\pcur, \scur, \mathcal{F}_{\text{ledger}}:\text{client}_{\text{L1}})$,
\textbf{wait for} \{ReadL1, $\mathit{output}=\{\{\mathit{TX}_{\text{L1}}\}, \{state_{\text{L1}}\}, \{\mathit{pid}\}\}$\};

\item \textbf{if} $\mathit{TX}_{\mathit{unilateral}}\in \{\mathit{TX}_{\text{L1}}\}$:

\begin{itemize}

\item From $\mathsf{executedRequest}$, choose $q_l=\{\{s,i\}, \sigma, \sigma', \sigma_W\}$ with highest $i$;

\item Prepare $\mathit{TX}_{\mathit{settle}}=\{\pcur, \emph{ContractAddr}, \_, q_l\}$

\item \textbf{send} \{SubmitL1, $\mathit{TX}_{\mathit{settle}}$\} \textbf{to} $(\pcur, \scur, \mathcal{F}_{\text{ledger}}:\text{client}_{\text{L1}})$;
\end{itemize}
\end{enumerate}

}\end{functionality}\vspace{1em}

\subsection{Brick Channel Ideal Functionality}
\label{apd:BrickIdeal}

We define the ideal functionality for the Brick channel as
$\mathcal{F}^{\text{Brick}}_{\text{layer2}}=(\mathcal{F}_{\text{client}},\mathcal{F}_{\text{ledger}} \mid \mathcal{F}^{\text{Brick}}_{\text{join}}, \mathcal{F}^{\text{Brick}}_{\text{submit}}, \mathcal{F}^{\text{Brick}}_{\text{update}}, \mathcal{F}^{\text{Brick}}_{\text{read}}, \mathcal{F}^{\text{Brick}}_{\text{settlement}}, \mathcal{F}^{\text{Brick}}_{\text{updRnd}})$. The subroutines' ideal functionalities are defined as follows.

\subsubsection{Submit Functionality}
\label{apd:Bricksubmit}

The submit subroutine $\mathcal{F}^{\text{Brick}}_{\text{submit}}$ handles requests that change the protocol status. It verifies that incoming requests are valid requests within one of three types: \textit{(i)} open requests to initiate the payment channel, \textit{(ii)} update requests to update the channel state, or \textit{(iii)} settlement requests. The subroutine ensures that no state-update request is accepted before the open request is executed or after a settlement request is executed.

\vspace{1em}\begin{functionality}{Description of subroutine $\mathcal{F}^{\text{Brick}}_{\text{submit}} = (\text{submit})$}{

\textbf{Participating roles:} \{submit\}

\noindent \textbf{Corruption model:} incorruptible

}\end{functionality}\vspace{.5em}

\begin{functionality}{Description of $\mathcal{M}^{\text{Brick}}_{\text{submit}}$}{

\textbf{Implemented role(s):} \{submit\}

\noindent \textbf{CheckID}(\emph{pid}, \emph{sid}, \emph{role}): Accept all messages with the same \emph{sid}.

\vspace{0.5em}
\noindent \textbf{Main:}

\vspace{0.5em}

\textbf{recv} \{Submit, $\mathit{request}$, $\mathsf{internalState}$\} \textbf{from} I/O:

\begin{enumerate}[itemsep=0.5em]
    
\item Check that $\mathit{request}$ belongs to one of the following valid types:
\begin{itemize}
    \item Open request: $\mathit{request} = (\text{Join}, s_{\mathit{init}}, \mathit{counterIdentity})$, \textbf{s.t.}:
    \begin{itemize}
        \item $s_{\mathit{init}} = \{s_{\mathit{int}_A}, s_{\mathit{int}_B}\}$ is a well-formed initial state;
        \item $\mathit{counterIdentity} \in \{0,1\}^{*}$; \cmt{Valid counterparty identifier}
    \end{itemize}
    \item State update request: $\mathit{request} = (\text{Update}, \{s, i\})$ with $s = \{s_A, s_B\}$, \textbf{s.t.}:
    \begin{itemize}
        \item Check that $\mathsf{stateList} \neq \emptyset$; \cmt{The channel has been opened (initial state committed)}

\item Check that $\nexists\, q \in \mathsf{requestQueue}$ \textbf{s.t.} $\mathtt{getType}(q) \in \{\text{settlement-collaborate},\\ \text{settlement-unilateral}\}$; \cmt{the channel is still open}

        \item $s_A + s_B = s_{\mathit{int}_A} + s_{\mathit{int}_B}$; \cmt{Total balance is conserved w.r.t.\ the initial state in $\mathsf{stateList}$}
        \item $i \in \mathbb{N}$ and $i > 0$; \cmt{Valid sequence number}
    \end{itemize}
    \item Settlement request: $\mathit{request} \in \{(\text{Settlement}, \texttt{collaborate}),\\ (\text{Settlement}, \texttt{unilateral})\}$;
\end{itemize}

\item \textbf{if} all checks pass, \textbf{reply} \{Submit, true\};
\item \textbf{else}, \textbf{reply} \{Submit, false\};
\end{enumerate}

}\end{functionality}\vspace{1em}

\subsubsection{Join Functionality}
\label{apd:Brickjoin}

The join subroutine $\mathcal{F}^{\text{Brick}}_{\text{join}}$ validates the channel-opening procedure. It checks that \textit{(i)} all honest clients have submitted matching open requests, \textit{(ii)} the opening transaction is consistent with $s_{\mathit{init}}$, and \textit{(iii)} the opening transaction is committed on the L1 blockchain.

\vspace{1em}\begin{functionality}{Description of subroutine $\mathcal{F}^{\text{Brick}}_{\text{join}} = (\text{join})$}{

\textbf{Participating roles:} \{join\}

\noindent \textbf{Corruption model:} incorruptible

\noindent \textbf{Protocol parameters:} 
\begin{itemize}
    \item $n = 3f+1$: warden committee size
    \item $f$: corruption threshold
\end{itemize}

}\end{functionality}\vspace{.5em}

\begin{functionality}{Description of $\mathcal{M}^{\text{Brick}}_{\text{join}}$}{

\textbf{Implemented role(s):} \{join\}

\noindent \textbf{CheckID}(\emph{pid}, \emph{sid}, \emph{role}): Accept all messages with the same \emph{sid}.

\vspace{0.5em}
\noindent \textbf{Main:}

\vspace{0.5em}

\textbf{recv} \{Join, $\mathit{Attachment}=\{\{\mathit{TX}_{\mathit{open}}\}, s_{\mathit{init}}\}$, $\mathsf{internalState}$\} \textbf{from} I/O:

\begin{enumerate}[itemsep=0.5em]
    
\item Check that $\forall\, \mathit{pid} \in \mathsf{identities}$ with $(\mathit{pid}, \scur, \text{client}) \notin \mathsf{CorruptionSet}$:
\begin{itemize}
    \item $\exists\, (\text{Join}, s_{\mathit{init}}^{(\mathit{pid})}, \mathit{Identities}^{(\mathit{pid})}) \in \mathsf{requestQueue}$ submitted by $\mathit{pid}$;
    \item $\forall\, \mathit{pid}, \mathit{pid'} \notin \mathsf{CorruptionSet}$: $s_{\mathit{init}}^{(\mathit{pid})} = s_{\mathit{init}}^{(\mathit{pid'})}$; \cmt{All honest clients agree on the initial state}
\end{itemize}


\item Check that $\mathit{TX}_{\mathit{open}} = (\mathit{pid}_{\mathit{init}}, \emph{ContractAddr}, s_{\mathit{int}_A'}, \{W_1', \ldots, W_n', f\})$ \textbf{s.t.}:
\begin{itemize}
    \item $s_{\mathit{int}_A'} = s_{\mathit{int}_A}$; \cmt{Opening transaction carries the agreed initial state}
    \item $\mathit{pid}_{\mathit{init}} \in \mathsf{identities}$; \cmt{Initiator is a registered participant}
    \item $\{W_1', \ldots, W_n'\} \subseteq \mathsf{identities}$ and $|\{W_1', \ldots, W_n'\}| = 3f+1$; \cmt{Warden set matches registered wardens}
    \item $\emph{ContractAddr}$ matches the contract identifier recorded in $\mathsf{identities}$; \cmt{Correct contract reference}
\end{itemize}

\item \textbf{send} \{ReadL1\} \textbf{to} $(\pcur, \scur, \mathcal{F}_{\text{ledger}}: \text{client}_{\text{L1}})$,
\textbf{wait for} \{ReadL1, $\mathit{output} = \{\{\mathit{TX}_{\text{L1}}\}, \{\mathit{state}_{\text{L1}}\}, \{\mathit{pid}\}\}$\};

\item Check the following conditions on $\mathit{output}$:
\begin{itemize}
    \item $\mathit{TX}_{\mathit{open}} \in \{\mathit{TX}_{\text{L1}}\}$; \cmt{Opening transaction is committed on L1}
    \item $s_{\mathit{init}} \in \{\mathit{state}_{\text{L1}}\}$; \cmt{Initial state is recorded on L1}
    \item $\forall\, W_j \in \{W_1, \ldots, W_n\}$: $\exists\, \mathit{TX}_{\mathit{collateral}}^{(j)} \in \{\mathit{TX}_{\text{L1}}\}$ \textbf{s.t.} $\mathit{TX}_{\mathit{collateral}}^{(j)}$ is a collateral deposit from warden $W_j$; \cmt{All wardens have deposited collateral on L1}
\end{itemize}

\item \textbf{if} all checks pass: \textbf{reply} \{Join, True, $s_{\mathit{init}}$\};

\end{enumerate}

}\end{functionality}\vspace{1em}

\begin{lemma}
\label{lem:OpenB}
    The subroutine $\mathcal{F}^{\text{Brick}}_{\text{join}}$ guarantees \emph{correct L2 initialization}.
\end{lemma}

\begin{proof}
    We prove \emph{correct L2 initialization} by contradiction. Suppose it is violated, then there could be two situations according to the definition in Section~\ref{sec:def}: (i) some honest participant receives a successful initialization output whose committed initial state differs from the proposed state; (ii) the proposed state is not evetually committed on L1. However, $\mathcal{F}^{\text{Brick}}_{\text{join}}$ outputs success only if (i) the initial state matches all honest clients' proposals recorded in $\mathsf{internalState}$'s $\mathsf{requestQueue}$ via $\mathcal{F}^{\text{Brick}}_{\text{submit}}$, verified in Step~1, and (ii) the initial state is committed on L1 by reading from the L1 blockchain functionality, verified in Step~3 \& 4. As long as $\mathcal{F}_{\text{ledger}}$ provides a secure L1 ledger, no adversary can cause $\mathcal{F}^{\text{Brick}}_{\text{join}}$ to accept a mismatched state. This contradicts the assumption.
\end{proof}

\subsubsection{Update Functionality}

The update subroutine $\mathcal{F}^{\text{Brick}}_{\text{update}}$ verifies that any newly proposed state update from the simulator: \textit{(i)} porposed by all the honest clients, \textit{(ii)} no conflict with previously executed requests, and \textit{(iii)} a sequence number that is successor to the latest one. Notably, $\mathcal{F}^{\text{Brick}}_{\text{update}}$ checks structural well-formedness rather than cryptographic signature validity.

\vspace{1em}\begin{functionality}{Description of subroutine $\mathcal{F}^{\text{Brick}}_{\text{update}} = (\text{update})$}{

\textbf{Participating roles:} \{update\}

\noindent\textbf{Corruption model:} incorruptible


}\end{functionality}\vspace{.5em}

\begin{functionality}{Description of $\mathcal{M}^{\text{Brick}}_{\text{update}}$}{

\textbf{Implemented role(s):} \{update\}

\noindent \textbf{CheckID}(\emph{pid}, \emph{sid}, \emph{role}): Accept all messages with the same \emph{sid}.

\vspace{0.5em}
\noindent \textbf{Main:}

\vspace{0.5em}

\textbf{recv} \{Update, $\mathit{Attachment}=\{\mathit{executedReq}, \mathit{newState}\}$, $\mathsf{internalState}$\} \textbf{from} I/O:

\begin{enumerate}[itemsep=0.5em]

\item Parse $\mathit{executedReq} = \{s, i\}$ from $\mathit{Attachment}$. Check that $\forall\, \mathit{pid} \in \mathsf{identities}$ with $(\mathit{pid}, \scur, \text{client}) \notin \mathsf{CorruptionSet}$:
\begin{itemize}
    \item $\exists\, (\text{Update}, \{s', i'\}) \in \mathsf{requestQueue}$ submitted by $\mathit{pid}$, \textbf{s.t.} $s' = s$ and $i' = i$; \cmt{Each honest client has a matching request in the queue}
    \item $\forall\, \mathit{pid}, \mathit{pid'}$ honest: the request submitted by $\mathit{pid}$ equals the request submitted by $\mathit{pid'}$; \cmt{All honest clients proposes the same state update}
\end{itemize}


\item Check that $\mathit{executedReq} = \{s, i\}$ with $s = \{s_A, s_B\}$ does not conflict with\\ $\mathsf{executedRequest}$, \textbf{s.t.}:
\begin{itemize}
    \item $\nexists\, \{s', i'\} \in \mathsf{executedRequest}$ with $i' = i$ and $s' \neq s$; \cmt{No conflicting entry at the same sequence number}
    \item $\nexists\, \{s', i'\} \in \mathsf{executedRequest}$ with $i' \geq i$; \cmt{No already-executed entry at or beyond this sequence number}
    \item $s_A + s_B = s_{\mathit{int}_A} + s_{\mathit{int}_B}$; \cmt{Total balance is conserved w.r.t.\ the initial state}
\end{itemize}

\item Let $i_{\max} = \max\{i' \mid \{s', i'\} \in \mathsf{executedRequest}\}$. Check:
\begin{itemize}
    \item $i = i_{\max} + 1$; \cmt{Sequence number is the immediate successor of the latest executed request}
\end{itemize}

\item \textbf{if} all checks pass: \textbf{reply} \{Update, true, $\mathit{newState}$, $\mathit{executedReq}$\};

\end{enumerate}

}\end{functionality}\vspace{1em}



\subsubsection{Read Functionality}

The read subroutine $\mathcal{F}^{\text{Brick}}_{\text{read}}$ determines what information a participant can learn from $\mathsf{internalState}$. Due to asynchronous communication, the adversary can influence the read result only by delaying message delivery. Consequently, the read result is either the latest read result or the previously read state.

\vspace{1em}\begin{functionality}{Description of subroutine $\mathcal{F}^{\text{Brick}}_{\text{read}} = (\text{read})$}{

\textbf{Participating roles:} \{read\} \\ \textbf{Corruption model:} incorruptible

}\end{functionality}\vspace{.5em}

\begin{functionality}{Description of $\mathcal{M}^{\text{Brick}}_{\text{read}}$}{

\textbf{Implemented role(s):} \{read\}

\noindent \textbf{CheckID}(\emph{pid}, \emph{sid}, \emph{role}): Accept all messages with the same \emph{sid}.

\vspace{0.5em}
\noindent \textbf{Main:}

\vspace{0.5em}

\textbf{recv} \{Read, $\mathsf{internalState}$\} \textbf{from} I/O:

\begin{enumerate}[itemsep=0.5em]

\item Let $i_{\max} = \max\{i \mid (s, i) \in \mathsf{stateList}\}$ and let $i_{\mathit{ptr}} \leftarrow \mathsf{lastReadPointer}$;

\item \textbf{send responsively} \{ReadDelivery, $i_{\mathit{ptr}}$, $i_{\max}$\} \textbf{to} $\mathcal{S}$ via NET, 
\textbf{wait for} \{ReadDelivery, $\mathit{delivered}$\} \textbf{s.t.} $\mathit{delivered} \in \{\text{true}, \text{false}\}$; \cmt{Query $\mathcal{S}$ for the delivery decision on off-chain updates}

\item \textbf{if} $i_{\mathit{ptr}} = i_{\max}$ or $\mathit{delivered} = \text{true}$: \cmt{No delay}
\begin{itemize}
    \item Let $\mathit{execView} \leftarrow \{e \in \mathsf{executedRequest} \mid e.\mathit{seq} \leq i_{\max}\}$;
    \item Let $\mathit{stateView} \leftarrow  \mathsf{stateList}$;
    \item $\mathsf{lastReadPointer} \leftarrow i_{\max}$; \cmt{Advance pointer to the latest delivered state}
\end{itemize}

\item \textbf{else}: \cmt{$\mathcal{S}$ withholds new off-chain updates; return the view at the current pointer}
\begin{itemize}
    \item Let $\mathit{execView} \leftarrow \{e \in \mathsf{executedRequest} \mid e.\mathit{seq} \leq i_{\mathit{ptr}}\}$; \cmt{Executed requests up to the last delivered pointer}
    \item Let $\mathit{stateView} \leftarrow e_{\mathit{ptr}}.\mathit{state}$ where $e_{\mathit{ptr}} \in \mathsf{executedRequest}$ \textbf{s.t.} $e_{\mathit{ptr}}.\mathit{seq} = i_{\mathit{ptr}}$; \cmt{The L2 state produced by the executed request at sequence $i_{\mathit{ptr}}$; stale w.r.t.\ the latest \textsf{stateList} entry}
\end{itemize}

\item Let $\mathit{ReadResult} = \{\mathit{execView},\; \mathit{stateView},\; \mathsf{onchainState}\}$;

\item \textbf{reply} \{Read, $\mathit{ReadResult}$\};

\end{enumerate}

}\end{functionality}\vspace{1em}

\begin{lemma}
    The subroutines $\mathcal{F}^{\text{Brick}}_{\text{read}}$ and $\mathcal{F}^{\text{Brick}}_{\text{update}}$ jointly guarantee $f_{L_2}$-safety.
\end{lemma}

\begin{proof}
    Suppose \emph{Safety} is violated. By the definition in Section~\ref{sec:def}, this means that either self-consistency or view-consistency fails. For self-consistency, the variable $\mathsf{lastReadPointer}$ ensures, through the checks in Steps~1--4, that every earlier read result is a prefix of every later read result, even when the adversary is influenceing the message delivery. Hence, self-consistency cannot be violated.
    
    For view-consistency, assume that two honest clients receive outputs that are not prefix-related. By $\mathcal{F}^{\text{Brick}}_{\text{update}}$, every valid state update requires agreement from both clients (Step~1). Therefore, as long as at least one client is honest, no inconsistent state can be recorded in the $\mathsf{internalState}$, which prevents inconsistency. Hence, the read outputs returned by $\mathcal{F}^{\text{Brick}}_{\text{read}}$ cannot violate view-consistency.
    
    Therefore, neither self-consistency nor view-consistency can be broken, which contradicts the assumption that \emph{Safety} is violated.
\end{proof}

\subsubsection{Settlement Functionality}

The settlement subroutine $\mathcal{F}^{\text{Brick}}_{\text{settlement}}$ handles two types of settlement. For \emph{collaborative} settlement, it checks that all honest clients propose to close the channel with the settlement state and that the closing transaction and the state is committed on L1. For \emph{unilateral} settlement, it verifies that the unilateral settlement transaction and eventually the latest channel state is committed on L1 blockchain.

\vspace{1em}\begin{functionality}{Description of subroutine $\mathcal{F}^{\text{Brick}}_{\text{settlement}} = (\text{settlement})$}{

\textbf{Participating roles:} \{settlement\} 

\noindent\textbf{Corruption model:} incorruptible

}\end{functionality}\vspace{.5em}

\begin{functionality}{Description of $\mathcal{M}^{\text{Brick}}_{\text{settlement}}$}{

\textbf{Implemented role(s):} \{settlement\}

\noindent \textbf{CheckID}(\emph{pid}, \emph{sid}, \emph{role}): Accept all messages with the same \emph{sid}.

\vspace{0.5em}
\noindent \textbf{Main:}

\vspace{0.5em}

\textbf{recv} \{Settlement, $\mathit{settlementType}$, $\mathit{Attachment} = \{\mathit{TX}_{\mathit{close}}, \mathit{TX}_{\mathit{unilateral}}\}$, $\mathsf{internalState}$\} \textbf{from} I/O:

\begin{enumerate}[itemsep=0.5em]

\item Let $i_{\max} = \max\{i \mid (s, i) \in \mathsf{stateList}\}$ and let $s_{\mathit{settle}} = \mathsf{stateList} = (s^{*}, i_{\max})$; \cmt{The final settlement state is the one with the highest sequence number}

\item \textbf{send} \{ReadL1\} \textbf{to} $(\pcur, \scur, \mathcal{F}_{\text{ledger}}: \text{client}_{\text{L1}})$;
\textbf{wait for} \{ReadL1, $\mathit{L1ReadResult} = \{\{\mathit{TX}\}_{\text{L1}}, \mathit{State}_{\text{L1}}, \{\mathit{pid}\}\}$\};

\item \textbf{if} $\mathit{settlementType} = \texttt{collaborate}$:
\begin{itemize}[itemsep=0.3em]
    \item Check that $\forall\, \mathit{pid} \in \mathsf{identities}$ with $(\mathit{pid}, \scur, \text{client}) \notin \mathsf{CorruptionSet}$:
    \begin{itemize}
        \item $\exists\, (\text{Settlement}, \texttt{collaborate}, s_{\mathit{settle}}) \in \mathsf{requestQueue}$ submitted by $\mathit{pid}$; \cmt{Each honest client has requested collaborative settlement on the latest state}
    \end{itemize}
    \item Parse $\mathit{TX}_{\mathit{close}} = (\mathit{pid}, \mathit{ContractAddr}, \epsilon, \{\text{Settlement}, s_{\mathit{close}}, \sigma, \sigma'\})$ from\\ $\mathit{Attachment}$. Check:
    \begin{itemize}
        \item $s_{\mathit{close}} = s^{*}$; \cmt{Closing transaction carries the latest agreed state}
    \end{itemize}
    \item Check the following conditions on $\mathit{L1ReadResult}$:
    \begin{itemize}
        \item $\mathit{TX}_{\mathit{close}} \in \{\mathit{TX}\}_{\text{L1}}$; \cmt{Closing transaction is committed on L1}
        \item $s^{*} \in \mathit{State}_{\text{L1}}$; \cmt{The settlement state is recorded on L1}
    \end{itemize}
\end{itemize}

\item \textbf{if} $\mathit{settlementType} = \texttt{unilateral}$:
\begin{itemize}[itemsep=0.3em]
    \item Check the following conditions on $\mathit{L1ReadResult}$:
    \begin{itemize}
        \item $\mathit{TX}_{\mathit{unilateral}} \in \{\mathit{TX}\}_{\text{L1}}$; \cmt{Unilateral settlement transaction is committed on L1}
        \item $s^{*} \in \mathit{State}_{\text{L1}}$;
    \end{itemize}
\end{itemize}

\item \textbf{if} all checks pass: \textbf{reply} \{Settlement, true, $s_{\mathit{settle}}$\};

\end{enumerate}

}\end{functionality}\vspace{1em}

\begin{lemma}
\label{lem:SettleB}
    The subroutine $\mathcal{F}^{\text{Brick}}_{\text{settlement}}$ guarantees \emph{correct L2 settlement}.
\end{lemma}

\begin{proof}
    According to the definition in Section~\ref{sec:def}, \emph{correct L2 settlement} is violated if either (i) an honest client receives a successful settlement output even though the output state is neither the latest state nor a state agreed upon by honest clients, or (ii) the output state is not committed on the L1 blockchain. However, for both types of channel closing, $\mathcal{F}^{\text{Brick}}_{\text{settlement}}$ outputs success and the closing state only after verifying that the committed state matches $\mathsf{stateList}$ and that the corresponding L1 transactions are committed (Steps~1--3). This contradicts the assumption that correct L2 settlement is violated.
    
\end{proof}

\subsubsection{Update Round Functionality}

Since the Brick channel operates under asynchronous off-chain 
communication, the clock subroutine 
$\mathcal{F}^{\text{Brick}}_{\text{updRnd}}$ does not impose 
constraints on round-update requests for join, state update, 
or collaborative settlement operations, as these require 
cooperation from all parties and the asynchronous network 
provides no finite delivery bound. However, 
$\mathcal{F}^{\text{Brick}}_{\text{updRnd}}$ enforces inclusion 
liveness for \emph{unilateral} settlement requests: an honest 
party can post the latest attested state to the L1 ledger 
without counterparty cooperation, and hence 
$\mathcal{F}^{\text{Brick}}_{\text{updRnd}}$ rejects any 
round advancement that would cause a pending unilateral 
settlement from an honest party to exceed its delay of $T_{L_1}$ 
rounds. This captures the fundamental L2 guarantee that an 
honest party can always exit the channel via L1, even if all 
other participants are corrupted.

\vspace{1em}\begin{functionality}{Description of subroutine $\mathcal{F}^{\text{Brick}}_{\text{updRnd}} = (\text{updRnd})$}{

\textbf{Participating roles:} \{updRnd\} 

\noindent\textbf{Corruption model:} incorruptible

}\end{functionality}\vspace{.5em}

\begin{functionality}{Description of 
  $\mathcal{M}^{\text{Brick}}_{\text{updRnd}}$}{

\textbf{Implemented role(s):} \{updRnd\}

\noindent \textbf{Parameters:} 
$T_{commit}$: L1 blockchain liveness bound.

\vspace{0.5em}
\noindent \textbf{CheckID}(\emph{pid}, \emph{sid}, 
\emph{role}): Accept all messages with the same \emph{sid}.

\vspace{0.5em}
\noindent \textbf{Main:}
\vspace{0.5em}

\textbf{recv} \{UpdateRound, $\mathsf{internalState}$\} 
\textbf{from} I/O:

Parse $\mathsf{round}$, $\mathsf{requestQueue}$, 
and $\mathsf{CorruptionSet}$ 
from $\mathsf{internalState}$.

\begin{enumerate}[itemsep=0.5em]

  \item $\forall$ pending request $q \in \mathsf{requestQueue}$ submitted by some $\mathit{pid}$ at round $t_{\mathit{sub}}$ \textbf{s.t.} $(\mathit{pid}, \scur, \text{client}) \notin \mathsf{CorruptionSet}$ and $\mathtt{getType}(q) = \text{settle-unilateral}$:
    \begin{enumerate}
      \item If $\mathsf{round} + 1 > 
        t_{\mathrm{sub}} + 2*T_{\mathrm{commit}}$:\cmt{can be changed for specific implementation requirement}
        \begin{enumerate}
          \item \textbf{reply} 
            \{UpdateRound, false, $\bot$\};
            \hfill $\triangleright$ 
            \emph{Unilateral settlement liveness 
            violated}
        \end{enumerate}
    \end{enumerate}

  \item \textbf{reply} \{UpdateRound, true, $\bot$\};
    \hfill $\triangleright$ 
    \emph{No overdue unilateral settlements; 
    all other requests unconstrained}

\end{enumerate}
}\end{functionality}\vspace{1em}

\begin{lemma}
\label{lem:LiveB}
The ideal functionality $\mathcal{F}^{\text{Brick}}_{\text{layer2}}$ 
guarantees the following liveness properties:
\begin{enumerate}
    \item \textbf{(Channel opening.)}
    $(f_{L_2} + f_{L_1},\, T_{L_2} + T_{L_1})$-liveness: under $f_{L_2}$ and $f_{L_1},$ corruption, an accepted Join request results in output $\{\text{Join}, s_{\mathit{init}}\}$ within $T_{L_2} + T_{L_1}$. The off-chain latency $T_{L_2}$ is not enforced by $\mathcal{F}^{\text{Brick}}_{\text{updRnd}}$ due to asynchronous communication; the L1-anchoring phase $T_{L_1}$ is latency inherited from $\mathcal{F}_{\text{ledger}}$.

    \item \textbf{(State update.)}
    $(f_{L_2},\, T_{L_2})$-liveness: under $f_{L_2}$ corruption, an accepted Update request results in the corresponding entry being added to $\mathsf{executedRequest}$ and changes in $\mathsf{stateList}$ within $T_{L_2}$. The off-chain latency $T_{L_2}$ is not enforced by $\mathcal{F}^{\text{Brick}}_{\text{updRnd}}$ due to asynchronous communication.

    \item \textbf{(Collaborative settlement.)}
    $(f_{L_2},\, T_{L_2} + T_{L_1})$-liveness: under $f_{L_2}$ corruption, an accepted collaborative Settlement request results in $\mathsf{onchainState}$ reflecting $s_{\mathit{settle}}$ and $\{\text{Settlement}, s_{\mathit{settle}}\}$ output within $T_{L_2} + T_{L_1}$. The off-chain latency $T_{L_2}$ is not enforced by $\mathcal{F}^{\text{Brick}}_{\text{updRnd}}$ due to asynchronous communication; the L1 latency $T_{L_1}$ is inherited from $\mathcal{F}_{\text{ledger}}$.

    \item \textbf{(Unilateral settlement.)}
    $(f_{L_1},\, T_{L_1})$-liveness: under $f_{L_1}$ an accepted unilateral Settlement request results in $\mathsf{onchainState}$ reflecting $s_{\mathit{settle}}$ and $\{\text{Settlement}, s_{\mathit{settle}}\}$ output within $T_{L_1}$. The L1 latency $T_{L_1}$ is inherited from $\mathcal{F}_{\text{ledger}}$
\end{enumerate}
\end{lemma}

\begin{proof}
We prove each part seperatively.

\medskip
\noindent\textbf{(1) Channel opening.}
$\mathcal{F}^{\text{Brick}}_{\text{join}}$ outputs 
success only if (i)~matching open requests from all 
honest clients exist in $\mathsf{RequestQueue}$, 
(ii)~$2f+1$ warden agreements are present, and 
(iii)~the funding transactions are committed on L1. If 
no client is corrupted and the corruption threshold 
holds for wardens, the adversary can at most delay 
off-chain message delivery but cannot prevent the 
opening procedure from completing. The off-chain 
negotiation phase takes at most $T_{L_2}$ (conditional 
on delivery), and the on-chain funding adds at most 
$T_{L_1}$, yielding a total conditional bound of 
$T_{L_2} + T_{L_1}$. 

\medskip
\noindent\textbf{(2) State update.}
$\mathcal{F}^{\text{Brick}}_{\text{update}}$ outputs 
success when the proposed update includes such request 
from all honest clients. In that case, the simulator 
can always construct a valid update once the off-chain 
messages are delivered as decided by the adversary. Since state updates are purely 
off-chain and require no L1 interaction, therefore the latency is 
$T_{L_2}$. However, the asynchronous network provides 
no finite delivery guarantee, so there is no bound for $T_{L_2}$. To simulate such influece, $\mathcal{F}^{\text{Brick}}_{\text{updRnd}}$ 
does not enforce any time constraint for state updates.

\medskip
\noindent\textbf{(3) Collaborative settlement.}
Suppose both clients are honest and have submitted 
settlement requests recorded in $\mathsf{RequestQueue}$. 
$\mathcal{F}^{\text{Brick}}_{\text{settlement}}$ checks 
that settlement requests exist from all honest clients 
and that the closing transaction 
$\mathit{TX}_{\mathit{close}}$ is consistent with the 
latest state and committed on L1. If both clients 
cooperate and messages are delivered within the 
asynchronous network, the off-chain agreement phase 
completes within $T_{L_2}$, and L1 inclusion adds at 
most $T_{L_1}$, yielding a total bound of 
$T_{L_2} + T_{L_1}$. However, since the off-chain 
communication is asynchronous, the adversary can delay 
messages indefinitely without violating any protocol 
invariant. Therefore, 
$\mathcal{F}^{\text{Brick}}_{\text{updRnd}}$ does not 
enforce this bound; the property holds only conditionally 
on message delivery.

\medskip
\noindent\textbf{(4) Unilateral settlement.}
By definition, $\mathcal{F}^{\text{Brick}}_{\text{updRnd}}$ 
rejects any request to advance time beyond 
$t + 2*T_{commit}$, where $T_{commit}$ is related to L1 interaction types of latency $T_{L_1}$, whenever there exists a pending unilateral 
settlement request from an honest client submitted at 
time~$t$ that has not yet been reflected in 
$\mathsf{onchainState}$. This forces the simulator to 
trigger the settlement before 
time can advance past the deadline.


\end{proof}


    



\subsection{ Security Proof}
\label{apd:BrickProof}

After proposing the ideal functionality and real-world implementation, we now show the security of the Brick channel protocol. To start with we first show the ideal functionality for Brick captures all the security properties with the following conclusion:

\ThmidealBrick*

\begin{proof}
    According to Lemma~\ref{lem:OpenB}-~\ref{lem:LiveB}, the ideal functionality $\mathcal{F}^{\text{Brick}}_{\text{layer2}}$ guarantees all security properties.
\end{proof}

After defining the ideal functionality $\mathcal{F}^{\text{Brick}}_{\text{layer2}}$, we prove that the real Brick protocol iUC-realizes it. The proof is done in 7 steps of successive game replacement. We first define a simulator $\mathcal{S}_{\mathcal{P}^{\text{Brick}}}$ that internally simulates a full run of $\mathcal{P}^{\text{Brick}}$, and a dummy functionality $\mathcal{F}^{\text{Brick}}_{\text{dummy}}$ that relays messages between $\mathcal{E}$ and $\mathcal{S}_{\mathcal{P}^{\text{Brick}}}$. This base ideal execution yields the same distribution of messages to $\mathcal{E}$ as the real execution. We use the execution ensemble $\mathsf{EXEC}$ to denote the messages observed by $\mathcal{E}$, including the output for the input request and the leakage to adversary, when interacting with adversary $\mathcal{A}$, real protocol $\mathcal{P}$, ideal functionality $\mathcal{F}$ and simulator $\mathcal{S}$ in the proofs that follow.

In each subsequent step, we incrementally add interaction between the simulator and the ideal functionality and extend the functionality, thereby forming the corresponding subroutines, while keeping the changes transparent to both $\mathcal{E}$ and $\mathcal{A}$. We continue until we obtain the target functionality $\mathcal{F}^{\text{Brick}}_{\text{layer2}}$ defined by our framework. At every step, the simulator is adjusted so that the new ideal execution is indistinguishable from the previous one. For each transition, we discuss the differences relative to the prior step and prove that, given the same inputs from $\mathcal{E}$ and $\mathcal{A}$, the resulting outputs remain the same up to computational indistinguishability under any adversarial influence strategy.

We begin by defining the dummy ideal functionality
$\mathcal{F}^{\text{Brick}}_{\text{dummy}}=(\mathcal{F}_{\text{client-dummy}}, \mathcal{F}_{\text{ledger}} \mid \perp)$
and the simulator $\mathcal{S}_{\text{Brick}}$ as follows. The dummy functionality that forwards every request from $\mathcal{E}$ to the simulator and returns the simulator's response unchanged. The simulator $\mathcal{S}_{\text{Brick}}$ runs $\mathcal{P'}^{\text{Brick}}$ internally and produces identical outputs.

\vspace{1em}\begin{functionality}{Description of $\mathcal{M}_{\text{client-dummy}}$ of $\mathcal{F}^{\text{Brick}}_{\text{dummy}}$}{

\textbf{Implemented role(s):} \{client-dummy\}

\noindent \textbf{Main:}

\textbf{recv} any request \textbf{from} I/O:
\begin{enumerate}[itemsep=0.5em]
    \item Forward request to $\mathcal{S}$ through NET;
\end{enumerate}

\hrule
\vspace{1em}

\textbf{recv} any message \textbf{from} NET:

\begin{enumerate}
    \item Output the message to $\mathcal{E}$ through I/O;
\end{enumerate}

}\end{functionality}\vspace{.5em}

\begin{functionality}{Description of simulator $\mathcal{S}_{\text{Brick}}$}{

$\mathcal{S}_{\text{Brick}}$ internally simulates $\mathcal{P'}^{\text{Brick}}$, a copy of the real protocol $\mathcal{P}^{\text{Brick}}$ as defined in Section~\ref{Apd:Brickreal}.

\vspace{0.5em}
\textbf{Real protocol simulation:}
\begin{itemize}[itemsep=0.3em]
    \item $\mathcal{S}_{\text{Brick}}$ simulates honest clients' and wardens' actions inside $\mathcal{P'}^{\text{Brick}}$ according to the real protocol.
    \item If participants are corrupted, $\mathcal{S}_{\text{Brick}}$ leaks the corresponding messages sent to corrupted entities to the adversary $\mathcal{A}$ and continues simulating honest parties based on $\mathcal{A}$'s instructions.
\end{itemize}

\textbf{Network communication from/to the environment:}
\begin{itemize}[itemsep=0.3em]
    \item Messages that $\mathcal{S}_{\text{Brick}}$ receives on the network interface (from $\mathcal{E}$/$\mathcal{A}$) are forwarded to $\mathcal{P'}^{\text{Brick}}$.
    \item Messages sent by $\mathcal{P'}^{\text{Brick}}$ on its network interface (to $\mathcal{E}$/$\mathcal{A}$) are forwarded to the environment.
\end{itemize}

\textbf{Input requests and outputs:}
\begin{itemize}[itemsep=0.3em]
    \item Unlike $\mathcal{P}^{\text{Brick}}$, which receives inputs directly from $\mathcal{E}$, the simulation $\mathcal{P'}^{\text{Brick}}$ receives requests forwarded from $\mathcal{F}^{\text{Brick}}_{\text{layer2}}$. Instead of sending outputs directly to $\mathcal{E}$, $\mathcal{S}_{\text{Brick}}$ sends them to $\mathcal{F}^{\text{Brick}}_{\text{layer2}}$.
\end{itemize}

\textbf{Message delivery:}
\begin{itemize}[itemsep=0.3em]
    \item The Brick channel assumes asynchronous communication via $\mathcal{F}^{\text{Brick}}_{\text{com}}$. The simulator bookkeeps all messages in $\mathcal{P'}^{\text{Brick}}$ and triggers delivery according to the adversary's scheduling decisions, mirroring the real-world protocol.
\end{itemize}

\textbf{Corruption handling:}
\begin{itemize}[itemsep=0.3em]
    \item $\mathcal{S}_{\text{Brick}}$ keeps the corruption status of entities in $\mathcal{P}^{\text{Brick}}$, $\mathcal{P'}^{\text{Brick}}$ and $\mathcal{F}^{\text{Brick}}_{\text{layer2}}$ synchronized. When an entity in $\mathcal{P'}^{\text{Brick}}$ becomes corrupted, $\mathcal{S}_{\text{Brick}}$ corrupts the corresponding entity in $\mathcal{F}^{\text{Brick}}_{\text{layer2}}$ before continuing.
    \item Adversarial commands for corrupted participants (e.g., publishing on $\mathcal{F}_{\text{ledger}}$, sending via $\mathcal{F}^{\text{Brick}}_{\text{com}}$) are forwarded to $\mathcal{P'}^{\text{Brick}}$.
    \item When a corrupted participant in $\mathcal{P'}^{\text{Brick}}$ wants to output to $\mathcal{E}$, $\mathcal{S}_{\text{Brick}}$ instructs the corresponding entity in $\mathcal{F}^{\text{Brick}}_{\text{layer2}}$ to output.
\end{itemize}
}\end{functionality}\vspace{1em}

\begin{lemma}
\label{lem:Brick1}
    For all PPT adversaries $\mathcal{A}$, there exists a PPT simulator $\mathcal{S}_{\text{Brick}}$ such that for all PPT environments $\mathcal{E}$ and all security parameters $k\in\mathbb{N}$,
    $\mathsf{EXEC}^{\mathcal{P}^{\text{Brick}}}_{\mathcal{A},\mathcal{E}}(k)\ \stackrel{c}{\approx}\
    \mathsf{EXEC}^{\mathcal{F}^{\text{Brick}}_{\text{dummy}}}_{\mathcal{S}_{\text{Brick}},\mathcal{E}}(k)$,
    where $\stackrel{c}{\approx}$ denotes computational indistinguishability.
\end{lemma}

\begin{proof}
    Fix an arbitrary PPT environment $\mathcal{E}$ and adversary $\mathcal{A}$. We argue that the execution ensembles in the real and ideal worlds are computationally indistinguishable by analyzing the three components observable by $\mathcal{E}$.

    \medskip
    \noindent\textbf{Observable components.} In the real world, the execution ensemble $\mathsf{EXEC}^{\mathcal{P}^{\text{Brick}}}_{\mathcal{A},\mathcal{E}}(k)$ consists of:
    \begin{enumerate}
        \item \emph{I/O outputs} delivered to $\mathcal{E}$ by $\mathcal{P}^{\text{Brick}}_{\text{client}}$:
        $\{\text{Join}, s_{\mathit{init}}\}$,
        $\{\text{Settlement},\\\texttt{collaborate}, s_L\}$,
        $\{\text{Settlement}, \texttt{unilateral}, s_L\}$,\\
        $\{\text{Read}, \mathsf{executedRequest}, \mathsf{stateList}\}$,
        $\{\text{GetCurRound}, \mathit{round}\}$.
        \item \emph{On-chain transactions} committed on $\mathcal{F}_{\text{ledger}}$ during execution:
        $\mathit{TX}_{\mathit{open}}$, $\mathit{TX}_{\mathit{close}}$, $\mathit{TX}_{\mathit{collateral}}$, $\mathit{TX}_{\mathit{fraud}}$, $\mathit{TX}_{\mathit{settle}}$.
        \item \emph{Adversarial leakage}, comprising messages received by corrupted parties, their internal state, their outputs, and the transactions they publish on L1.
    \end{enumerate}

    In the ideal world, the simulator $\mathcal{S}_{\text{Brick}}$ internally runs $\mathcal{P'}^{\text{Brick}}$ and interacts with the dummy functionality $\mathcal{F}^{\text{Brick}}_{\text{dummy}}$, which by definition forwards every request from $\mathcal{E}$ to $\mathcal{S}_{\text{Brick}}$ unchanged and relays the simulator's responses back to $\mathcal{E}$. We show that each component is computationally indistinguishable across the two worlds.

    \medskip
    \noindent\textbf{(1) I/O outputs.} Since $\mathcal{F}^{\text{Brick}}_{\text{dummy}}$ acts as a transparent relay, $\mathcal{S}_{\text{Brick}}$ receives exactly the same sequence of requests as $\mathcal{P}^{\text{Brick}}_{\text{client}}$ would in the real world. By construction, $\mathcal{S}_{\text{Brick}}$ executes the same client and warden logic inside $\mathcal{P'}^{\text{Brick}}$ under the same adversarial scheduling, producing the same I/O outputs. The only potential difference arises from the randomness of $\mathcal{F}_{\text{sig}}$: signature strings in read results and settlement outputs may differ between the two worlds because fresh randomness is sampled independently. However, since $\mathcal{F}_{\text{sig}}$ realizes EUF-CMA security (as defined in Appendix~\ref{apdx:sigcom}), signatures generated on the same messages are computationally indistinguishable. Hence the I/O outputs are computationally indistinguishable.

    \medskip
    \noindent\textbf{(2) On-chain transactions.} Since $\mathcal{F}^{\text{Brick}}_{\text{dummy}}$ forwards all requests to $\mathcal{S}_{\text{Brick}}$, the simulated protocol $\mathcal{P'}^{\text{Brick}}$ generates and publishes the same set of transactions to $\mathcal{F}_{\text{ledger}}$ as $\mathcal{P}^{\text{Brick}}$ would in the real world. Transactions may contain different signature values due to independent randomness in $\mathcal{F}_{\text{sig}}$, but by the EUF-CMA security of the signature scheme, the transaction distributions are computationally indistinguishable.

    \medskip
    \noindent\textbf{(3) Adversarial leakage.} By the definition of $\mathcal{S}_{\text{Brick}}$, the corruption status of all entities is kept synchronized between $\mathcal{P'}^{\text{Brick}}$ and $\mathcal{F}^{\text{Brick}}_{\text{dummy}}$. Since $\mathcal{F}^{\text{Brick}}_{\text{dummy}}$ forwards all $\mathcal{E}$-requests to the simulator, $\mathcal{S}_{\text{Brick}}$ can reconstruct the same internal state and message history for corrupted parties as in the real execution. Consequently, the leakage delivered to $\mathcal{A}$ (and hence observable by $\mathcal{E}$, if $\mathcal{E}$ subsumes $\mathcal{A}$) is identical up to signature randomness, which is again computationally indistinguishable by EUF-CMA security.

    \medskip
    Conclusively, we have:
    $\mathsf{EXEC}^{\mathcal{P}^{\text{Brick}}}_{\mathcal{A}, \mathcal{E}}(k)\stackrel{c}{\approx}\mathsf{EXEC}^{\mathcal{F}^{\text{Brick}}_{\text{dummy}}}_{\mathcal{S}_{\text{Brick}}, \mathcal{E}}(k)$.
\end{proof}

Next, we extend the dummy functionality $\mathcal{F}^{\text{Brick}}_{\text{dummy}}$ with the submission subroutine, yielding
$\mathcal{F}^{\text{Brick}}_{\text{layer2-submit}} = (\mathcal{F}_{\text{client-submit}}, \mathcal{F}_{\text{ledger}} \mid \mathcal{F}^{\text{Brick}}_{\text{submit}})$.
The subroutine $\mathcal{F}^{\text{Brick}}_{\text{submit}}$ is defined in Appendix~\ref{apd:Bricksubmit}, while $\mathcal{F}_{\text{client-submit}}$ is given below. We also define the corresponding simulator $\mathcal{S}_{\text{Brick-submit}}$ as follows:

\vspace{1em}\begin{functionality}{Description of $\mathcal{M}_{\text{client-submit}}$ of $\mathcal{F}^{\text{Brick}}_{\text{layer2-submit}}$}{

\textbf{Implemented role(s):} \{client-submit\}

\noindent \textbf{Main:}

\vspace{0.5em}

\textbf{recv} \{Submit, $\mathit{request}$\} \textbf{from} I/O: \label{request:submit}

\begin{enumerate}[itemsep=0.5em]
\item \textbf{send} \{Submit, $\mathit{request}$, $\mathsf{internalState}$\} \textbf{to} $(\pcur, \scur, \mathcal{F}^{\text{Brick}}_{\text{submit}}:\text{submit})$,
\textbf{wait for} \{Submit, $\mathit{response}$\} s.t. $\mathit{response} \in \{\text{true}, \text{false}\}$;
\item \textbf{if} $\mathit{response} =$ true: $\mathsf{requestQueue}.\text{add}(\mathit{request})$;
\textbf{send} $\mathit{request}$ \textbf{to} $\mathcal{S}$ via NET;
\end{enumerate}

\hrule
\vspace{0.5em}

\textbf{recv} any other request \textbf{from} I/O:
\begin{enumerate}[itemsep=0.5em]
    \item Forward request to $\mathcal{S}$ through NET;
\end{enumerate}

\hrule
\vspace{1em}

\textbf{recv} any other request \textbf{from} NET:

\begin{enumerate}
    \item Output the message to $\mathcal{E}$ through I/O;
\end{enumerate}

}\end{functionality}\vspace{.5em}\begin{functionality}{Description of simulator $\mathcal{S}_{\text{Brick-submit}}$}{

The simulator $\mathcal{S}_{\text{Brick-submit}}$ behaves the same as $\mathcal{S}_{\text{Brick}}$.

}\end{functionality}\vspace{1em}

\begin{lemma}
\label{lem:Brick2}
    For all PPT adversaries $\mathcal{A}$, there exists a PPT simulator $\mathcal{S}_{\text{Brick-submit}}$ such that for all PPT environments $\mathcal{E}$ and all security parameters $k\in\mathbb{N}$,
    \[
    \mathsf{EXEC}^{\mathcal{F}^{\text{Brick}}_{\text{dummy}}}_{\mathcal{S}_{\text{Brick}},\mathcal{E}}(k)
    \ \stackrel{c}{\approx}\
    \mathsf{EXEC}^{\mathcal{F}^{\text{Brick}}_{\text{layer2-submit}}}_{\mathcal{S}_{\text{Brick-submit}},\mathcal{E}}(k),
    \]
    where $\stackrel{c}{\approx}$ denotes computational indistinguishability.
\end{lemma}

\begin{proof}
    Fix an arbitrary PPT environment $\mathcal{E}$. The simulator logic is identical in both executions, the only difference is the ideal functionality through which requests are forwarded to the simulator. We analyze the three observable components of the execution ensemble. In $\mathcal{F}^{\text{Brick}}_{\text{dummy}}$, every request from $\mathcal{E}$ is forwarded directly to $\mathcal{S}_{\text{Brick}}$. In $\mathcal{F}^{\text{Brick}}_{\text{layer2-submit}}$, the subroutine $\mathcal{F}^{\text{Brick}}_{\text{submit}}$ intercepts each request and checks semantic validity before forwarding to $\mathcal{S}_{\text{Brick-submit}}$. Concretely, $\mathcal{F}^{\text{Brick}}_{\text{submit}}$ verifies:
    \begin{itemize}
        \item The request belongs to a valid type (Join, Update, or Settlement);
        \item For submit requests: the balance conservation invariant $s_A + s_B = s_{\mathit{int}_A} + s_{\mathit{int}_B}$ holds;
        \item No state-update request is accepted before the join is executed or after a settlement request is pending.
    \end{itemize}
    Requests that fail these checks are rejected by $\mathcal{F}^{\text{Brick}}_{\text{submit}}$ and never reach the simulator.

    \medskip
    \noindent\textbf{(1) I/O outputs.} In the real protocol $\mathcal{P}^{\text{Brick}}$ (and hence in the simulated $\mathcal{P'}^{\text{Brick}}$), the client machine already enforces the same validity checks when receiving the request: malformed requests, requests violating balance conservation, and requests issued in an invalid protocol phase produce no output. Therefore, any request rejected by $\mathcal{F}^{\text{Brick}}_{\text{submit}}$ would also produce no output inside the simulation, since $\mathcal{P'}^{\text{Brick}}$ would silently discard it. For accepted requests, both simulators execute the same client logic and produce the same I/O outputs. Hence, the I/O outputs are identical.

    \medskip
    \noindent\textbf{(2) On-chain transactions.} Since only accepted requests trigger protocol actions in $\mathcal{P'}^{\text{Brick}}$, and the set of accepted requests is the same in both executions, the transactions published to $\mathcal{F}_{\text{ledger}}$ are identical except the signature part. Since the ideal signature functionality $\mathcal{F}_{\text{sig}}$ guarantees the EUF-CMA security, the on-chain transactions are computationally indistinguishable. 

    \medskip
    \noindent\textbf{(3) Adversarial leakage.} In $\mathsf{EXEC}^{\mathcal{F}^{\text{Brick}}_{\text{dummy}}}_{\mathcal{S}_{\text{Brick}},\mathcal{E}}(k)$, the simulator receives all requests (including invalid ones) but produces no observable effect for invalid requests. In $\mathsf{EXEC}^{\mathcal{F}^{\text{Brick}}_{\text{layer2-submit}}}_{\mathcal{S}_{\text{Brick-submit}},\mathcal{E}}(k)$, invalid requests are filtered before reaching the simulator. Since invalid requests generate no leakage in either world (no messages are sent to corrupted parties, no state changes occur, and no transactions are published).  Additionally, the corruption status is synchronized, and the existence of ideal signature functionality, the leakage delivered to $\mathcal{A}$ is computationally indistinguishable.

    \medskip
    Conclusively, the execution ensemble observed by $\mathcal{E}$ are computationally indistinguishable in both executions:
    $\mathsf{EXEC}^{\mathcal{F}^{\text{Brick}}_{\text{dummy}}}_{\mathcal{S}_{\text{Brick}}, \mathcal{E}}(k) \stackrel{c}{\approx} \mathsf{EXEC}^{\mathcal{F}^{\text{Brick}}_{\text{layer2-submit}}}_{\mathcal{S}_{\text{Brick-submit}}, \mathcal{E}}(k)$.
\end{proof}

Next step, we extend the ideal functionality $\mathcal{F}^{\text{Brick}}_{\text{layer2-submit}}$ with the subroutine $\mathcal{F}^{\text{Brick}}_{\text{join}}$ and defined as $\mathcal{F}^{\text{Brick}}_{\text{layer2,join}} = (\mathcal{F}_{\text{client-join}}, \mathcal{F}_{\text{ledger}} \mid \mathcal{F}^{\text{Brick}}_{\text{submit}}, \mathcal{F}^{\text{Brick}}_{\text{join}})$, the according simulator $\mathcal{S}_{\text{Brick-join}}$ defined as follow:

\vspace{1em}\begin{functionality}{Description of $\mathcal{M}_{\text{client-join}}$ of $\mathcal{F}^{\text{Brick}}_{\text{layer2-join}}$}{

\textbf{Implemented role(s):} \{client-join\}

\noindent \textbf{Main:}

\vspace{0.5em}

\textbf{recv} \{Submit, $\mathit{request}$\} \textbf{from} I/O: \label{request:submit}

\begin{enumerate}[itemsep=0.5em]
\item \textbf{send} \{Submit, $\mathit{request}$, $\mathsf{internalState}$\} \textbf{to} $(\pcur, \scur, \mathcal{F}^{\text{Brick}}_{\text{submit}}:\text{submit})$,
\textbf{wait for} \{Submit, $\mathit{response}$\} s.t. $\mathit{response} \in \{\text{true}, \text{false}\}$;
\item \textbf{if} $\mathit{response} =$ true: $\mathsf{requestQueue}.\text{add}(\mathit{request})$;
\textbf{send} $\mathit{request}$ \textbf{to} $\mathcal{S}$ via NET;
\end{enumerate}

\hrule
\vspace{0.5em}

\textbf{recv} \{Join, $\mathit{Attachment}$\} \textbf{from} NET: \label{request:join}

\begin{enumerate}[itemsep=0.5em]
\item \textbf{send} \{Join, $\mathit{Attachment}$, $\mathsf{internalState}$\} \textbf{to} $(\pcur, \scur, \mathcal{F}^{\text{Brick}}_{\text{join}}:\text{join})$,
\textbf{wait for} \{Join, $\mathit{response}$\} s.t. $\mathit{response} \in \{\text{true}, \text{false}\}$;
\item \textbf{if} $\mathit{response} =$ true: update $\mathsf{internalState}$ according to $\mathit{Attachment}$;
\textbf{reply} \{Join, $s_{\mathit{init}}$\} via I/O;
\end{enumerate}

\hrule
\vspace{1em}

\textbf{recv} any other request \textbf{from} I/O:
\begin{enumerate}[itemsep=0.5em]
    \item Forward request to $\mathcal{S}$ through NET;
\end{enumerate}

\hrule
\vspace{1em}

\textbf{recv} any other request \textbf{from} NET:

\begin{enumerate}
    \item Output the message to $\mathcal{E}$ through I/O;
\end{enumerate}

}\end{functionality}\vspace{.5em}

\begin{functionality}{Description of simulator $\mathcal{S}_{\text{Brick-join}}$}{

The simulator $\mathcal{S}_{\text{Brick-join}}$ behaves identically to $\mathcal{S}_{\text{Brick-submit}}$, except for the following additional behavior upon detecting a completed channel-opening procedure:

\vspace{0.5em}
\textbf{Channel opening interaction with $\mathcal{F}^{\text{Brick}}_{\text{layer2-join}}$:}

\begin{enumerate}[itemsep=0.5em]

\item $\mathcal{S}_{\text{Brick-join}}$ monitors the simulated protocol $\mathcal{P'}^{\text{Brick}}$. When it detects that a client entity is about to produce the I/O output $\{\text{Join}, s_{\mathit{init}}\}$, $\mathcal{S}_{\text{Brick-join}}$ intercepts this output and proceeds as follows.

\item $\mathcal{S}_{\text{Brick-join}}$ prepares $\mathit{Attachment}$ by extracting from the simulation state:
\begin{itemize}
    \item $s_{\mathit{init}}$: the initial state included in the simulated client's Join output;
    \item $\mathit{TX}_{\mathit{open}}$: the channel-opening transaction formed by the simulated client and committed on $\mathcal{F}_{\text{ledger}}$ in $\mathcal{P'}^{\text{Brick}}$;
\end{itemize}

\item $\mathcal{S}_{\text{Brick-join}}$ sends $\{\text{Join}, \mathit{Attachment}\}$ to $\mathcal{F}^{\text{Brick}}_{\text{layer2-join}}$ via NET.

\end{enumerate}

}\end{functionality}\vspace{1em}

\begin{lemma}
\label{lem:Brick3}
    For all PPT adversaries $\mathcal{A}$, there exists a PPT simulator $\mathcal{S}_{\text{Brick-join}}$ such that for all PPT environments $\mathcal{E}$ and all security parameters $k\in\mathbb{N}$,
    \[
    \mathsf{EXEC}^{\mathcal{F}^{\text{Brick}}_{\text{layer2-submit}}}_{\mathcal{S}_{\text{Brick-submit}},\mathcal{E}}(k)
    \ \stackrel{c}{\approx}\
    \mathsf{EXEC}^{\mathcal{F}^{\text{Brick}}_{\text{layer2-join}}}_{\mathcal{S}_{\text{Brick-join}},\mathcal{E}}(k),
    \]
    where $\stackrel{c}{\approx}$ denotes computational indistinguishability.
\end{lemma}

\begin{proof}
    Fix an arbitrary PPT environment $\mathcal{E}$. The only difference between the two execution is that $\mathcal{F}^{\text{Brick}}_{\text{layer2-join}}$ forward the Join request (from simulator thtough NET) to the subroutine $\mathcal{F}^{\text{Brick}}_{\text{join}}$ to decide output, whereas in $\mathcal{F}^{\text{Brick}}_{\text{layer2-submit}}$ the join output is produced entirely by the simulator. We analyze the three observable components, considering the two adversarial strategies that could create a difference: \emph{(i)}~influencing message delivery (the off-chain communication is asynchronous), and \emph{(ii)}~corrupt parties to deviate from the protocol.

    \medskip
    \noindent\textbf{(1) I/O outputs.} The only I/O output affected by this game hop is $\{\text{Join}, s_{\mathit{init}}\}$. We consider two cases.

    \emph{Case~1: Successful join.} In $\mathsf{EXEC}^{\mathcal{F}^{\text{Brick}}_{\text{layer2-submit}}}_{\mathcal{S}_{\text{Brick-submit}},\mathcal{E}}(k)$, the simulator $\mathcal{S}_{\text{Brick-submit}}$ produces a join output to $\mathcal{E}$ when the simulated $\mathcal{P'}^{\text{Brick}}$ completes the opening procedure, and the output is directly forwarded to $\mathcal{E}$. In $\mathsf{EXEC}^{\mathcal{F}^{\text{Brick}}_{\text{layer2-join}}}_{\mathcal{S}_{\text{Brick-join}},\mathcal{E}}(k)$, the simulator $\mathcal{S}_{\text{Brick-join}}$ instead prepares an $\mathit{Attachment}$ and sends it to $\mathcal{F}^{\text{Brick}}_{\text{join}}$, which outputs to $\mathcal{E}$ only if all defined checks pass. Since $\mathcal{S}_{\text{Brick-join}}$ sends the $\mathit{Attachment}$ precisely when notice $\mathcal{P'}^{\text{Brick}}$ completes the opening, and a successful opening in $\mathcal{P'}^{\text{Brick}}$ implies:
    \begin{itemize}
        \item $\forall$ honest client $\mathit{pid}$: a matching Join request exists in $\mathsf{requestQueue}$ (validated by $\mathcal{F}^{\text{Brick}}_{\text{submit}}$ in both games);
        \item $\mathit{TX}_{\mathit{open}}$ is consistent with $s_{\mathit{init}}$ and committed on L1;
    \end{itemize}
    these are exactly the checks in $\mathcal{F}^{\text{Brick}}_{\text{join}}$. Hence $\mathcal{F}^{\text{Brick}}_{\text{join}}$ accepts whenever $\mathcal{P'}^{\text{Brick}}$ completes, and the I/O output $\{\text{Join}, s_{\mathit{init}}\}$ is produced in both games at the same time with the same content.

    \emph{Case~2: Failed join.} The adversary may attempt to prevent or corrupt the opening via two strategies:
    \begin{itemize}
        \item \emph{Message delay:} Both simulators record $\mathcal{A}$'s message-delivery decisions and trigger the join only when $\mathcal{P'}^{\text{Brick}}$ completes the join procedure. Since the $\mathit{Attachment}$ is sent to $\mathcal{F}^{\text{Brick}}_{\text{join}}$ only upon completion, delivery delays affect both execution identically.
        \item \emph{Off-chain deviation:} If corrupted parties send invalid messages or withhold signatures, the simulated $\mathcal{P'}^{\text{Brick}}$ cannot collect $2f{+}1$ valid warden agreements (assuming the corruption threshold $f$ holds and $\mathcal{F}_{\text{sig}}$ provides EUF-CMA security). Thus no successful output is produced in $\mathsf{EXEC}^{\mathcal{F}^{\text{Brick}}_{\text{layer2-submit}}}_{\mathcal{S}_{\text{Brick-submit}},\mathcal{E}}(k)$. Correspondingly, $\mathcal{S}_{\text{Brick-join}}$ will not prepare an Join request, so no output is produced in $\mathcal{F}^{\text{Brick}}_{\text{join}}$. Hence $\mathcal{F}^{\text{Brick}}_{\text{join}}$.
        \item \emph{On-chain deviation:} If corrupted participants do not publish required L1 transactions, the honest clients in $\mathcal{P'}^{\text{Brick}}$ do not complete the opening and generates output in\\ $\mathsf{EXEC}^{\mathcal{F}^{\text{Brick}}_{\text{layer2-submit}}}_{\mathcal{S}_{\text{Brick-submit}},\mathcal{E}}(k)$, and the L1 check in $\mathcal{F}^{\text{Brick}}_{\text{join}}$ fails in $\mathcal{F}^{\text{Brick}}_{\text{join}}$. Hence $\mathcal{F}^{\text{Brick}}_{\text{join}}$. Hence no output is produced in either game.
    \end{itemize}
    In all cases, the I/O outputs are identical.

    \medskip
    \noindent\textbf{(2) On-chain transactions.} The set of transactions published to $\mathcal{F}_{\text{ledger}}$ is determined by the simulated protocol $\mathcal{P'}^{\text{Brick}}$, which runs identically in both simulators with the same forwarded requests. The game hop here only affects when the ideal functionality generates I/O output, not which transactions are published. Hence, on-chain transactions are still computationally indistinguishable.

    \medskip
    \noindent\textbf{(3) Adversarial leakage.} Both simulators maintain synchronized corruption status and execute the same protocol logic inside $\mathcal{P'}^{\text{Brick}}$. Since the game hop only interposes $\mathcal{F}^{\text{Brick}}_{\text{join}}$ between the simulator and the I/O output (and does not alter the simulator's internal execution or its interaction with corrupted parties), the leakage delivered to $\mathcal{A}$ is still computationally indistinguishable.

    \medskip
    Conclusively,
    $\mathsf{EXEC}^{\mathcal{F}^{\text{Brick}}_{\text{layer2-submit}}}_{\mathcal{S}_{\text{Brick-submit}}, \mathcal{E}}(k)\stackrel{c}{\approx}\mathsf{EXEC}^{\mathcal{F}^{\text{Brick}}_{\text{layer2-join}}}_{\mathcal{S}_{\text{Brick-join}}, \mathcal{E}}(k)$.
\end{proof}

Next step, we extend the ideal functionality $\mathcal{F}^{\text{Brick}}_{\text{layer2-join}}$ with the subroutine $\mathcal{F}^{\text{Brick}}_{\text{update}}$ and defined as $\mathcal{F}^{\text{Brick}}_{\text{layer2,update}} = (\mathcal{F}_{\text{client-update}},\\ \mathcal{F}_{\text{ledger}} \mid \mathcal{F}^{\text{Brick}}_{\text{submit}}, \mathcal{F}^{\text{Brick}}_{\text{join}}, \mathcal{F}^{\text{Brick}}_{\text{update}})$, the according simulator $\mathcal{S}_{\text{Brick-update}}$ defined as follow:

\vspace{1em}\begin{functionality}{Description of $\mathcal{M}_{\text{client-update}}$ of $\mathcal{F}^{\text{Brick}}_{\text{layer2-update}}$}{

\textbf{Implemented role(s):} \{client-update\}

\noindent \textbf{Main:}

\vspace{0.5em}

\textbf{recv} \{Submit, $\mathit{request}$\} \textbf{from} I/O: \label{request:submit}

\begin{enumerate}[itemsep=0.5em]
\item \textbf{send} \{Submit, $\mathit{request}$, $\mathsf{internalState}$\} \textbf{to} $(\pcur, \scur, \mathcal{F}^{\text{Brick}}_{\text{submit}}:\text{submit})$,
\textbf{wait for} \{Submit, $\mathit{response}$\} s.t. $\mathit{response} \in \{\text{true}, \text{false}\}$;
\item \textbf{if} $\mathit{response} =$ true: $\mathsf{requestQueue}.\text{add}(\mathit{request})$;
\textbf{send} $\mathit{request}$ \textbf{to} $\mathcal{S}$ via NET;
\end{enumerate}

\hrule
\vspace{0.5em}

\textbf{recv} \{Join, $\mathit{Attachment}$\} \textbf{from} NET: \label{request:join}

\begin{enumerate}[itemsep=0.5em]
\item \textbf{send} \{Join, $\mathit{Attachment}$, $\mathsf{internalState}$\} \textbf{to} $(\pcur, \scur, \mathcal{F}^{\text{Brick}}_{\text{join}}:\text{join})$,
\textbf{wait for} \{Join, $\mathit{response}$\} s.t. $\mathit{response} \in \{\text{true}, \text{false}\}$;
\item \textbf{if} $\mathit{response} =$ true: update $\mathsf{internalState}$ according to $\mathit{Attachment}$;
\textbf{reply} \{Join, $s_{\mathit{init}}$\} via I/O;
\end{enumerate}

\hrule
\vspace{1em}

\textbf{recv} \{Update, $\mathit{Attachment}$\} \textbf{from} NET: \label{request:update}

\begin{enumerate}[itemsep=0.5em]
\item \textbf{send} \{Update, $\mathit{Attachment}$, $\mathsf{internalState}$\} \textbf{to} $(\pcur, \scur, \mathcal{F}^{\text{Brick}}_{\text{update}}:\text{update})$,
\textbf{wait for} \{Update, $\mathit{response}$, $\mathit{newState}$, $\mathit{executedReq}$\};
\item \textbf{if} $\mathit{response} =$ true: update $\mathsf{internalState}$ with $\mathit{newState}$ and $\mathit{executedReq}$;
\end{enumerate}

\hrule
\vspace{0.5em}

\textbf{recv} any other request \textbf{from} I/O:
\begin{enumerate}[itemsep=0.5em]
    \item Forward request to $\mathcal{S}$ through NET;
\end{enumerate}

\hrule
\vspace{1em}

\textbf{recv} any other request \textbf{from} NET:

\begin{enumerate}
    \item Output the message to $\mathcal{E}$ through I/O;
\end{enumerate}

}\end{functionality}\vspace{.5em}

\begin{functionality}{Description of simulator $\mathcal{S}_{\text{Brick-update}}$}{

The simulator $\mathcal{S}_{\text{Brick-update}}$ behaves identically to $\mathcal{S}_{\text{Brick-join}}$, except for the following additional behavior upon detecting a completed state update:

\vspace{0.5em}
\textbf{State update interaction with $\mathcal{F}^{\text{Brick}}_{\text{layer2-update}}$:}

\begin{enumerate}[itemsep=0.5em]

\item In addition to simulating the honest clients' and wardens' actions as defined in Appendix~\ref{Apd:Brickreal}, $\mathcal{S}_{\text{Brick-update}}$ monitors the simulated protocol $\mathcal{P'}^{\text{Brick}}$ for the completion of state updates. Specifically, it detects when a new state $\{s, i\}$ is marked as updated in the simulated client's $\mathsf{stateList}$.

\item Upon detecting a completed state update, $\mathcal{S}_{\text{Brick-update}}$ prepares $\mathit{Attachment}$ by extracting from the simulation state:
\begin{itemize}
    \item $\mathit{executedReq} = \{s, i\}$: the executed request together with its sequence number;
    \item $\mathit{newState}$: the updated state recorded in the simulated client's $\mathsf{stateList}$;
    
\end{itemize}

\item $\mathcal{S}_{\text{Brick-update}}$ sends $\{\text{Update}, \mathit{Attachment}\}$ to $\mathcal{F}^{\text{Brick}}_{\text{layer2-update}}$ via NET;

\end{enumerate}

}\end{functionality}\vspace{1em}

\begin{lemma}
\label{lem:Brick4}
    For all PPT adversaries $\mathcal{A}$, there exists a PPT simulator $\mathcal{S}_{\text{Brick-update}}$ such that for all PPT environments $\mathcal{E}$ and all security parameters $k\in\mathbb{N}$,
    \[
    \mathsf{EXEC}^{\mathcal{F}^{\text{Brick}}_{\text{layer2-join}}}_{\mathcal{S}_{\text{Brick-join}},\mathcal{E}}(k)
    \ \stackrel{c}{\approx}\
    \mathsf{EXEC}^{\mathcal{F}^{\text{Brick}}_{\text{layer2-update}}}_{\mathcal{S}_{\text{Brick-update}},\mathcal{E}}(k),
    \]
    where $\stackrel{c}{\approx}$ denotes computational indistinguishability.
\end{lemma}

\begin{proof}
    Fix an arbitrary PPT environment $\mathcal{E}$. The only difference between the two games is that $\mathcal{F}^{\text{Brick}}_{\text{layer2-update}}$ routes the Update request from simualtor through NET to the subroutine $\mathcal{F}^{\text{Brick}}_{\text{update}}$, whereas in $\mathcal{F}^{\text{Brick}}_{\text{layer2-join}}$ state updates are handled entirely by the simulator. We analyze the three observable components.

    \medskip
    \noindent As defined in both ideal functionalities, the Update request does not directly produce I/O outputs to $\mathcal{E}$, it only modifies $\mathsf{internalState}$. In $\mathsf{EXEC}^{\mathcal{F}^{\text{Brick}}_{\text{layer2-join}}}_{\mathcal{S}_{\text{Brick-join}},\mathcal{E}}(k)$, only the $\mathsf{internalState}$ of the client in $\mathcal{P'}^{\text{Brick}}$ will be changed during the simulation. In $\mathsf{EXEC}^{\mathcal{F}^{\text{Brick}}_{\text{layer2-update}}}_{\mathcal{S}_{\text{Brick-update}},\mathcal{E}}(k)$, $\mathcal{F}^{\text{Brick}}_{\text{update}}$ additionally verifies the update before applying it. The subroutine checks:
    \begin{itemize}
        \item $\forall$ honest client $\mathit{pid}$: a matching Update request for $\{s, i\}$ exists in $\mathsf{requestQueue}$, and all honest clients agree on the same $\{s, i\}$;
        \item No conflict with $\mathsf{executedRequest}$, and balance conservation $s_A + s_B = s_{\mathit{int}_A} + s_{\mathit{int}_B}$;
        \item The sequence number $i$ is the immediate successor of the latest in $\mathsf{executedRequest}$.
    \end{itemize}

    Then we discuss the three types of execution ensemble:
    
    \noindent\textbf{(1) I/O outputs.} Although the Update request itself produces no I/O output, a divergence in $\mathsf{internalState}$ would cause observable differences in subsequent Read or Settlement outputs. Here we prove that $\mathsf{internalState}$ changes identically in both executions.

    In $\mathcal{P'}^{\text{Brick}}$, an honest client considers a state $\{s, i\}$ as successfully updated only after: (i)~both clients have signed $\{s, i\}$, and (ii)~at least $2f{+}1$ wardens have returned valid signatures. The honest client will reject conflicting transactions and incorrect sequence numbers, honest wardens likewise refuse to sign such transactions. These are precisely the checks enforced by $\mathcal{F}^{\text{Brick}}_{\text{update}}$. Thus $\mathcal{F}^{\text{Brick}}_{\text{update}}$ accepts an update if and only if the corresponding update was accepted in $\mathcal{P'}^{\text{Brick}}$, and $\mathsf{internalState}$ changes identically.

    \medskip
    \noindent\textbf{(2) On-chain transactions.} The Update procedure in the Brick channel does not publish transactions to $\mathcal{F}_{\text{ledger}}$ (state updates are purely off-chain). Hence, on-chain transactions are identical in both executions.

    \medskip
    \noindent\textbf{(3) Adversarial leakage.} Both simulators receive the same set of accepted requests (determined by $\mathcal{F}^{\text{Brick}}_{\text{submit}}$), run the same protocol logic inside $\mathcal{P'}^{\text{Brick}}$, and maintain synchronized corruption status. The game hop only interposes $\mathcal{F}^{\text{Brick}}_{\text{update}}$ as a filter on the Update request. Since the subroutine accepts exactly those updates that $\mathcal{P'}^{\text{Brick}}$ would accept and cause the $\mathsf{internalState}$ of the ideal functionalities to change, additionally, the ideal signature functionality guarantees the signatures are indistinguishable, therefore the leakage is computationally indistinguishable.

    Conclusively,
    $\mathsf{EXEC}^{\mathcal{F}^{\text{Brick}}_{\text{layer2-join}}}_{\mathcal{S}_{\text{Brick-join}}, \mathcal{E}}(k)
    \stackrel{c}{\approx}
    \mathsf{EXEC}^{\mathcal{F}^{\text{Brick}}_{\text{layer2-update}}}_{\mathcal{S}_{\text{Brick-update}}, \mathcal{E}}(k)$.\end{proof}

Next step, we extend the ideal functionality $\mathcal{F}^{\text{Brick}}_{\text{layer2-update}}$ with the subroutine $\mathcal{F}^{\text{Brick}}_{\text{read}}$ and defined as $\mathcal{F}^{\text{Brick}}_{\text{layer2,read}} = (\mathcal{F}_{\text{client-read}}, \mathcal{F}_{\text{ledger}} \mid \mathcal{F}^{\text{Brick}}_{\text{submit}}, \mathcal{F}^{\text{Brick}}_{\text{join}}, \mathcal{F}^{\text{Brick}}_{\text{update}}, \mathcal{F}^{\text{Brick}}_{\text{read}})$, the according simulator $\mathcal{S}_{\text{Brick-read}}$ defined as follow:

\vspace{1em}\begin{functionality}{Description of $\mathcal{M}_{\text{client-read}}$ of $\mathcal{F}^{\text{Brick}}_{\text{layer2-read}}$}{

\textbf{Implemented role(s):} \{client-read\}

\noindent \textbf{Main:}

\vspace{0.5em}

\textbf{recv} \{Submit, $\mathit{request}$\} \textbf{from} I/O: \label{request:submit}

\begin{enumerate}[itemsep=0.5em]
\item \textbf{send} \{Submit, $\mathit{request}$, $\mathsf{internalState}$\} \textbf{to} $(\pcur, \scur, \mathcal{F}^{\text{Brick}}_{\text{submit}}:\text{submit})$,
\textbf{wait for} \{Submit, $\mathit{response}$\} s.t. $\mathit{response} \in \{\text{true}, \text{false}\}$;
\item \textbf{if} $\mathit{response} =$ true: $\mathsf{requestQueue}.\text{add}(\mathit{request})$;
\textbf{send} $\mathit{request}$ \textbf{to} $\mathcal{S}$ via NET;
\end{enumerate}

\hrule
\vspace{0.5em}

\textbf{recv} \{Join, $\mathit{Attachment}$\} \textbf{from} NET: \label{request:join}

\begin{enumerate}[itemsep=0.5em]
\item \textbf{send} \{Join, $\mathit{Attachment}$, $\mathsf{internalState}$\} \textbf{to} $(\pcur, \scur, \mathcal{F}^{\text{Brick}}_{\text{join}}:\text{join})$,
\textbf{wait for} \{Join, $\mathit{response}$\} s.t. $\mathit{response} \in \{\text{true}, \text{false}\}$;
\item \textbf{if} $\mathit{response} =$ true: update $\mathsf{internalState}$ according to $\mathit{Attachment}$;
\textbf{reply} \{Join, $s_{\mathit{init}}$\} via I/O;
\end{enumerate}

\hrule
\vspace{1em}

\textbf{recv} \{Update, $\mathit{Attachment}$\} \textbf{from} NET: \label{request:update}

\begin{enumerate}[itemsep=0.5em]
\item \textbf{send} \{Update, $\mathit{Attachment}$, $\mathsf{internalState}$\} \textbf{to} $(\pcur, \scur, \mathcal{F}^{\text{Brick}}_{\text{update}}:\text{update})$,
\textbf{wait for} \{Update, $\mathit{response}$, $\mathit{newState}$, $\mathit{executedReq}$\};
\item \textbf{if} $\mathit{response} =$ true: update $\mathsf{internalState}$ with $\mathit{newState}$ and $\mathit{executedReq}$;
\end{enumerate}

\hrule
\vspace{0.5em}

\textbf{recv} \{Read\} \textbf{from} I/O: \label{request:read}

\begin{enumerate}[itemsep=0.5em]
\item \textbf{send} \{Read, $\mathsf{internalState}$\} \textbf{to} $(\pcur, \scur, \mathcal{F}^{\text{Brick}}_{\text{read}}:\text{read})$,
\textbf{wait for} $\mathit{ReadResult}$;
\item \textbf{if} $\mathit{ReadResult} \neq \bot$: \textbf{reply} \{Read, $\mathit{ReadResult}$\} via I/O;
\end{enumerate}

\hrule
\vspace{0.5em}

\textbf{recv} any other request \textbf{from} I/O:
\begin{enumerate}[itemsep=0.5em]
    \item Forward request to $\mathcal{S}$ through NET;
\end{enumerate}

\hrule
\vspace{1em}

\textbf{recv} any other request \textbf{from} NET:

\begin{enumerate}
    \item Output the message to $\mathcal{E}$ through I/O;
\end{enumerate}

}\end{functionality}\vspace{.5em}

\begin{functionality}{Description of simulator $\mathcal{S}_{\text{Brick-read}}$}{

The simulator $\mathcal{S}_{\text{Brick-read}}$ behaves identically to $\mathcal{S}_{\text{Brick-update}}$, except for the following additional behavior when handling read-delivery queries from the ideal functionality:

\vspace{0.5em}
\textbf{Read delivery interaction with $\mathcal{F}^{\text{Brick}}_{\text{layer2-read}}$:}

\begin{enumerate}[itemsep=0.5em]

\item When $\mathcal{F}^{\text{Brick}}_{\text{read}}$ sends a responsive message $\{\text{ReadDelivery}, i_{\mathit{ptr}}, i_{\max}\}$ to $\mathcal{S}_{\text{Brick-read}}$ via NET (querying whether a pending state update should be considered delivered to the reading client), $\mathcal{S}_{\text{Brick-read}}$ determines the delivery status based on $\mathcal{A}$'s scheduling decisions in the simulated protocol $\mathcal{P'}^{\text{Brick}}$.

\item Concretely, $\mathcal{S}_{\text{Brick-read}}$ checks whether, in $\mathcal{P'}^{\text{Brick}}$, the message carrying the state update at sequence number $i_{\max}$ has been delivered to the requesting client according to $\mathcal{A}$'s message-delivery decisions:
\begin{itemize}
    \item \textbf{if} delivered: $\mathcal{S}_{\text{Brick-read}}$ replies $\{\text{ReadDelivery}, \text{true}\}$; \cmt{Client observes the latest state}
    \item \textbf{else}: $\mathcal{S}_{\text{Brick-read}}$ replies $\{\text{ReadDelivery}, \text{false}\}$; \cmt{Client observes the previously read state}
\end{itemize}

\item The ideal functionality $\mathcal{F}^{\text{Brick}}_{\text{read}}$ then produces the I/O output to $\mathcal{E}$ accordingly.

\end{enumerate}

}\end{functionality}\vspace{1em}
\begin{lemma}
\label{lem:Brick5}
    For all PPT adversaries $\mathcal{A}$, there exists a PPT simulator $\mathcal{S}_{\text{Brick-read}}$ such that for all PPT environments $\mathcal{E}$ and all security parameters $k\in\mathbb{N}$,
    \[
    \mathsf{EXEC}^{\mathcal{F}^{\text{Brick}}_{\text{layer2-update}}}_{\mathcal{S}_{\text{Brick-update}},\mathcal{E}}(k)
    \ \stackrel{c}{\approx}\
    \mathsf{EXEC}^{\mathcal{F}^{\text{Brick}}_{\text{layer2-read}}}_{\mathcal{S}_{\text{Brick-read}},\mathcal{E}}(k),
    \]
    where $\stackrel{c}{\approx}$ denotes computational indistinguishability.
\end{lemma}

\begin{proof}
    Fix an arbitrary PPT environment $\mathcal{E}$. The only difference between the two execution is that in $\mathsf{EXEC}^{\mathcal{F}^{\text{Brick}}_{\text{layer2-read}}}_{\mathcal{S}_{\text{Brick-read}},\mathcal{E}}(k)$ the readresult is dedecided by $\mathcal{F}^{\text{Brick}}_{\text{read}}$ based on the ideal functionality's $\mathsf{internalState}$, whereas in $\mathsf{EXEC}^{\mathcal{F}^{\text{Brick}}_{\text{layer2-update}}}_{\mathcal{S}_{\text{Brick-update}},\mathcal{E}}(k)$ read results are produced entirely by the simulator. We analyze the three observable components of the execution ensemble.

    In both executions, a read request does not modify $\mathsf{internalState}$, it only returns records already stored. In $\mathsf{EXEC}^{\mathcal{F}^{\text{Brick}}_{\text{layer2-update}}}_{\mathcal{S}_{\text{Brick-update}},\mathcal{E}}(k)$, the read result is extracted from the local state of the simulated client inside $\mathcal{S}_{\text{Brick-update}}$ and forwarded to $\mathcal{E}$ by the ideal functionality. In $\mathsf{EXEC}^{\mathcal{F}^{\text{Brick}}_{\text{layer2-read}}}_{\mathcal{S}_{\text{Brick-read}},\mathcal{E}}(k)$, the read result is produced by $\mathcal{F}^{\text{Brick}}_{\text{read}}$ from $\mathsf{internalState}$, where $\mathsf{internalState}$ is updated exclusively by requests passing $\mathcal{F}^{\text{Brick}}_{\text{submit}}$ and updates passing $\mathcal{F}^{\text{Brick}}_{\text{update}}$. Two situations could cause a distinguishable difference: \emph{(i)}~inconsistency between the simulated client's local state and $\mathsf{internalState}$, or \emph{(ii)}~adversarial network control causing different read results.

    \medskip
    \noindent\textbf{(1) I/O outputs.} For situation~(i), as has been proved in Lemma~\ref{lem:Brick4}, $\mathsf{internalState}$ changes identically in both executions. Hence the I/O outputs for Read requests are computationally indistinguishable.

     As for situation~(ii), the adversarial network control affects only the timing of message delivery, not the content of delivered messages. By definition, both simulators use the same delivery schedule determined by $\mathcal{A}$, and $\mathsf{internalState}$ changes only after validated events. In $\mathsf{EXEC}^{\mathcal{F}^{\text{Brick}}_{\text{layer2-read}}}_{\mathcal{S}_{\text{Brick-read}},\mathcal{E}}(k)$, $\mathcal{F}^{\text{Brick}}_{\text{read}}$ queries $\mathcal{S}_{\text{Brick-read}}$ responsively for a boolean delivery decision; $\mathcal{S}_{\text{Brick-read}}$ replies based on whether the relevant message has been delivered in $\mathcal{P'}^{\text{Brick}}$ according to $\mathcal{A}$'s choices. These are the same choices that $\mathcal{S}_{\text{Brick-update}}$ uses in $\mathsf{EXEC}^{\mathcal{F}^{\text{Brick}}_{\text{layer2-update}}}_{\mathcal{S}_{\text{Brick-update}},\mathcal{E}}(k)$ to determine the simulated client's read output. Hence, the read outputs coincide under all adversarial delivery schedules.

    \medskip
    \noindent\textbf{(2) On-chain transactions.} Read requests do not publish any transactions to $\mathcal{F}_{\text{ledger}}$. Hence on-chain transactions are identical in both games.

    \medskip
    \noindent\textbf{(3) Adversarial leakage.} The game hop only interposes $\mathcal{F}^{\text{Brick}}_{\text{read}}$ on how the read result is generated, which is an I/O-facing operation that does not create network messages to corrupted parties. Both simulators maintain synchronized corruption status and execute the same protocol logic inside $\mathcal{P'}^{\text{Brick}}$. Hence, the leakage delivered to $\mathcal{A}$ is computationally indistinguishable.

    \medskip
    Conclusively,
    $\mathsf{EXEC}^{\mathcal{F}^{\text{Brick}}_{\text{layer2-update}}}_{\mathcal{S}_{\text{Brick-update}}, \mathcal{E}}(k)
    \stackrel{c}{\approx}
    \mathsf{EXEC}^{\mathcal{F}^{\text{Brick}}_{\text{layer2-read}}}_{\mathcal{S}_{\text{Brick-read}}, \mathcal{E}}(k)$.
\end{proof}

Next step, we extend the ideal functionality $\mathcal{F}^{\text{Brick}}_{\text{layer2-read}}$ with the subroutine $\mathcal{F}^{\text{Brick}}_{\text{settlement}}$ and defined as $\mathcal{F}^{\text{Brick}}_{\text{layer2,settlement}} = (\mathcal{F}_{\text{client-settlement}},\\ \mathcal{F}_{\text{ledger}} \mid \mathcal{F}^{\text{Brick}}_{\text{submit}}, \mathcal{F}^{\text{Brick}}_{\text{join}}, \mathcal{F}^{\text{Brick}}_{\text{update}}, \mathcal{F}^{\text{Brick}}_{\text{read}}, \mathcal{F}^{\text{Brick}}_{\text{settlement}})$, the according simulator $\mathcal{S}_{\text{Brick-settlement}}$ defined as follow:

\vspace{1em}\begin{functionality}{Description of $\mathcal{M}_{\text{client-settlement}}$ of $\mathcal{F}^{\text{Brick}}_{\text{layer2-settlement}}$}{

\textbf{Implemented role(s):} \{client-settlement\}

\noindent \textbf{Main:}

\vspace{0.5em}

\textbf{recv} \{Submit, $\mathit{request}$\} \textbf{from} I/O: \label{request:submit}

\begin{enumerate}[itemsep=0.5em]
\item \textbf{send} \{Submit, $\mathit{request}$, $\mathsf{internalState}$\} \textbf{to} $(\pcur, \scur, \mathcal{F}^{\text{Brick}}_{\text{submit}}:\text{submit})$,
\textbf{wait for} \{Submit, $\mathit{response}$\} s.t. $\mathit{response} \in \{\text{true}, \text{false}\}$;
\item \textbf{if} $\mathit{response} =$ true: $\mathsf{requestQueue}.\text{add}(\mathit{request})$;
\textbf{send} $\mathit{request}$ \textbf{to} $\mathcal{S}$ via NET;
\end{enumerate}

\hrule
\vspace{0.5em}

\textbf{recv} \{Join, $\mathit{Attachment}$\} \textbf{from} NET: \label{request:join}

\begin{enumerate}[itemsep=0.5em]
\item \textbf{send} \{Join, $\mathit{Attachment}$, $\mathsf{internalState}$\} \textbf{to} $(\pcur, \scur, \mathcal{F}^{\text{Brick}}_{\text{join}}:\text{join})$,
\textbf{wait for} \{Join, $\mathit{response}$\} s.t. $\mathit{response} \in \{\text{true}, \text{false}\}$;
\item \textbf{if} $\mathit{response} =$ true: update $\mathsf{internalState}$ according to $\mathit{Attachment}$;
\textbf{reply} \{Join, $s_{\mathit{init}}$\} via I/O;
\end{enumerate}

\hrule
\vspace{1em}

\textbf{recv} \{Update, $\mathit{Attachment}$\} \textbf{from} NET: \label{request:update}

\begin{enumerate}[itemsep=0.5em]
\item \textbf{send} \{Update, $\mathit{Attachment}$, $\mathsf{internalState}$\} \textbf{to} $(\pcur, \scur, \mathcal{F}^{\text{Brick}}_{\text{update}}:\text{update})$,
\textbf{wait for} \{Update, $\mathit{response}$, $\mathit{newState}$, $\mathit{executedReq}$\};
\item \textbf{if} $\mathit{response} =$ true: update $\mathsf{internalState}$ with $\mathit{newState}$ and $\mathit{executedReq}$;
\end{enumerate}

\hrule
\vspace{0.5em}

\textbf{recv} \{Read\} \textbf{from} I/O: \label{request:read}

\begin{enumerate}[itemsep=0.5em]
\item \textbf{send} \{Read, $\mathsf{internalState}$\} \textbf{to} $(\pcur, \scur, \mathcal{F}^{\text{Brick}}_{\text{read}}:\text{read})$,
\textbf{wait for} $\mathit{ReadResult}$;
\item \textbf{if} $\mathit{ReadResult} \neq \bot$: \textbf{reply} \{Read, $\mathit{ReadResult}$\} via I/O;
\end{enumerate}

\hrule
\vspace{0.5em}

\textbf{recv} \{Settlement, $\mathit{Attachment}$\} \textbf{from} NET: \label{request:settlement}

\begin{enumerate}[itemsep=0.5em]
\item \textbf{send} \{Settlement, $\mathit{Attachment}$, $\mathsf{internalState}$\} \textbf{to} $(\pcur, \scur, \mathcal{F}^{\text{Brick}}_{\text{settlement}}:\\\text{settlement})$,
\textbf{wait for} \{Settlement, $\mathit{response}$, $s_{\mathit{settle}}$\} s.t. $\mathit{response} \in \{\text{true}, \text{false}\}$;
\item \textbf{if} $\mathit{response} =$ true: update $\mathsf{internalState}$;
\textbf{reply} \{Settlement, $s_{\mathit{settle}}$\} via I/O;
\end{enumerate}

\hrule
\vspace{0.5em}

\textbf{recv} any other request \textbf{from} I/O:
\begin{enumerate}[itemsep=0.5em]
    \item Forward request to $\mathcal{S}$ through NET;
\end{enumerate}

\hrule
\vspace{1em}

\textbf{recv} any other request \textbf{from} NET:

\begin{enumerate}
    \item Output the message to $\mathcal{E}$ through I/O;
\end{enumerate}

}\end{functionality}\vspace{.5em}

\begin{functionality}{Description of simulator $\mathcal{S}_{\text{Brick-settlement}}$}{

The simulator $\mathcal{S}_{\text{Brick-settlement}}$ behaves identically to $\mathcal{S}_{\text{Brick-read}}$, except for the following additional behavior upon detecting a completed channel-settlement procedure:

\vspace{0.5em}
\textbf{Settlement interaction with $\mathcal{F}^{\text{Brick}}_{\text{layer2-settlement}}$:}

\begin{enumerate}[itemsep=0.5em]

\item $\mathcal{S}_{\text{Brick-settlement}}$ monitors the simulated protocol $\mathcal{P'}^{\text{Brick}}$. When it detects that a client entity is about to produce a successful settlement output inside the simulation, it intercepts this output and proceeds according to the settlement type.

\item \textbf{if} the simulated output is $\{\text{Settlement}, \texttt{collaborate}, s_L\}$:
\begin{itemize}[itemsep=0.3em]
    \item $\mathcal{S}_{\text{Brick-settlement}}$ prepares $\mathit{Attachment}$ by extracting from the simulation state:
    \begin{itemize}
        \item $\mathit{settlementType} \leftarrow \texttt{collaborate}$;
        \item $\mathit{TX}_{\mathit{close}}$: the collaborative closing transaction formed by the simulated clients, carrying the latest agreed state $s_L$ and signed by both clients;
    \end{itemize}
\end{itemize}

\item \textbf{if} the simulated output is $\{\text{Settlement}, \texttt{unilateral}, s_L\}$:
\begin{itemize}[itemsep=0.3em]
    \item $\mathcal{S}_{\text{Brick-settlement}}$ prepares $\mathit{Attachment}$ by extracting from the simulation state:
    \begin{itemize}
        \item $\mathit{settlementType} \leftarrow \texttt{unilateral}$;
    \end{itemize}
\end{itemize}

\item $\mathcal{S}_{\text{Brick-settlement}}$ sends $\{\text{Settlement}, \mathit{settlementType}, \mathit{Attachment}\}$ to\\ $\mathcal{F}^{\text{Brick}}_{\text{layer2-settlement}}$ via NET.

\end{enumerate}

}\end{functionality}\vspace{1em}

\begin{lemma}
\label{lem:Brick6}
    For all PPT adversaries $\mathcal{A}$, there exists a PPT simulator $\mathcal{S}_{\text{Brick-settlement}}$ such that for all PPT environments $\mathcal{E}$ and all security parameters $k\in\mathbb{N}$,
    \[
    \mathsf{EXEC}^{\mathcal{F}^{\text{Brick}}_{\text{layer2-read}}}_{\mathcal{S}_{\text{Brick-read}},\mathcal{E}}(k)
    \ \stackrel{c}{\approx}\
    \mathsf{EXEC}^{\mathcal{F}^{\text{Brick}}_{\text{layer2-settlement}}}_{\mathcal{S}_{\text{Brick-settlement}},\mathcal{E}}(k),
    \]
    where $\stackrel{c}{\approx}$ denotes computational indistinguishability.
\end{lemma}

\begin{proof}
    Fix an arbitrary PPT environment $\mathcal{E}$. The only difference between the two executions is that in $\mathsf{EXEC}^{\mathcal{F}^{\text{Brick}}_{\text{layer2-settlement}}}_{\mathcal{S}_{\text{Brick-settlement}},\mathcal{E}}(k)$, the Settlement request from simulator through NET is forwarded to the subroutine $\mathcal{F}^{\text{Brick}}_{\text{settlement}}$ to generate output, whereas in $\mathsf{EXEC}^{\mathcal{F}^{\text{Brick}}_{\text{layer2-read}}}_{\mathcal{S}_{\text{Brick-read}},\mathcal{E}}(k)$ the settlement output is produced entirely by the simulator. We analyze the three observable components, considering the two adversarial strategies: \emph{(i)}~influencing message delivery (the network is asynchronous), and \emph{(ii)}~causing parties to deviate from the protocol. We treat the two settlement types separately.

    \medskip
    \noindent\textbf{(1.1) I/O outputs for collaborative settlement.} The I/O output affected is $\{\text{Settlement}, \texttt{collaborate}, s_L\}$. We consider two cases.

    \emph{Case~1: Successful collaborative settlement.} In \\$\mathsf{EXEC}^{\mathcal{F}^{\text{Brick}}_{\text{layer2-read}}}_{\mathcal{S}_{\text{Brick-read}},\mathcal{E}}(k)$, the simulator $\mathcal{S}_{\text{Brick-read}}$ produces a settlement output when the simulated $\mathcal{P'}^{\text{Brick}}$ completes the collaborative closing. In $\mathsf{EXEC}^{\mathcal{F}^{\text{Brick}}_{\text{layer2-settlement}}}_{\mathcal{S}_{\text{Brick-settlement}},\mathcal{E}}(k)$, the simulator $\mathcal{S}_{\text{Brick-settlement}}$ instead prepares an $\mathit{Attachment}$ and sends it to $\mathcal{F}^{\text{Brick}}_{\text{settlement}}$, which outputs to $\mathcal{E}$ only if all checks pass. A successful collaborative closing in $\mathcal{P'}^{\text{Brick}}$ implies:
    \begin{itemize}
        \item $\forall$ honest client $\mathit{pid}$, there is a matching $(\text{Settlement}, \texttt{collaborate})$ request exists in $\mathsf{requestQueue}$;
        \item $\mathit{TX}_{\mathit{close}}$ carries the latest agreed state $ s_{\mathit{settle}}$ according to $\mathsf{stateList}$;
        \item $\mathit{TX}_{\mathit{close}}$ is committed on $\mathcal{F}_{\text{ledger}}$;
    \end{itemize}
    these are exactly the checks in $\mathcal{F}^{\text{Brick}}_{\text{settlement}}$ for the collaborative case. Hence $\mathcal{F}^{\text{Brick}}_{\text{settlement}}$ accepts whenever $\mathcal{P'}^{\text{Brick}}$ completes, and the output is produced at the same time with the same content.

    \emph{Case~2: Failed collaborative settlement (strategy~(ii), off-chain).} The adversary may cause a corrupted client to send an invalid closing message (e.g., a non-latest state or a forged signature). Since $\mathcal{F}_{\text{sig}}$ provides EUF-CMA security, the adversary cannot forge a valid signature on behalf of an honest client. The honest client in $\mathcal{P'}^{\text{Brick}}$ will reject any closing message carrying a state it did not agree to. Therefore, no successful output is produced in $\mathsf{EXEC}^{\mathcal{F}^{\text{Brick}}_{\text{layer2-read}}}_{\mathcal{S}_{\text{Brick-read}},\mathcal{E}}(k)$. Correspondingly, $\mathcal{S}_{\text{Brick-settlement}}$ cannot prepare an $\mathit{Attachment}$ in which the collaborative checks of $\mathcal{F}^{\text{Brick}}_{\text{settlement}}$ pass, so no output is produced in $\mathsf{EXEC}^{\mathcal{F}^{\text{Brick}}_{\text{layer2-settlement}}}_{\mathcal{S}_{\text{Brick-settlement}},\mathcal{E}}(k)$ also.

    \medskip
    \noindent\textbf{(1.2) I/O outputs for unilateral settlement.} The I/O output affected is $\{\text{Settlement}, \texttt{unilateral}, s_L\}$. We consider two cases.

    \emph{Case~1: Successful unilateral settlement.} In $\mathcal{P'}^{\text{Brick}}$, a unilateral settlement succeeds when $\mathit{TX}_{\mathit{unilateral}}$ is committed on L1, at least $2f{+}1$ wardens publish their settlement transactions $\{\mathit{TX}_{\mathit{settle}}\}$, and the committed state corresponds to the latest state. In $\mathsf{EXEC}^{\mathcal{F}^{\text{Brick}}_{\text{layer2-settlement}}}_{\mathcal{S}_{\text{Brick-settlement}},\mathcal{E}}(k)$, $\mathcal{F}^{\text{Brick}}_{\text{settlement}}$ checks:
    \begin{itemize}
    
        \item $\mathit{TX}_{\mathit{unilateral}} \in \mathit{L1ReadResult}.\{\mathit{TX}\}_{\text{L1}}$;
        \item $s_{\mathit{settle}} = \mathit{L1ReadResult}.\mathit{State}_{\text{L1}}$;
    \end{itemize}
    Since $\mathcal{S}_{\text{Brick-settlement}}$ triggers the settlement message to $\mathcal{F}^{\text{Brick}}_{\text{settlement}}$ precisely when $\mathcal{P'}^{\text{Brick}}$ completes the unilateral procedure, and a successful completion implies all the above conditions hold, $\mathcal{F}^{\text{Brick}}_{\text{settlement}}$ accepts and the output is identical.

    \emph{Case~2: Corrupted wardens publish a old state or do not react (strategy~(ii), on-chain).} Suppose a coalition of corrupted wardens publishes settlement transactions containing a non-latest state on L1. Under the corruption threshold ($\leq f$ corrupted wardens out of $3f{+}1$), at least $2f{+}1$ honest wardens remain, who will publish their stored state on L1. At least one honest warden holds the latest executed state and will publish it on L1. In $\mathcal{P'}^{\text{Brick}}$, the pubslihed old state will not be counted since the seqeunce number is lower than the latest state published by honest warden, resulting in the latest state being committed. In $\mathsf{EXEC}^{\mathcal{F}^{\text{Brick}}_{\text{layer2-settlement}}}_{\mathcal{S}_{\text{Brick-settlement}},\mathcal{E}}(k)$, $\mathcal{F}^{\text{Brick}}_{\text{settlement}}$ checks that $s_{\mathit{settle}}$ matches the committed L1 state, which holds after the fraud-proof mechanism resolves the dispute. Since the simulator keeps corruption status synchronized, $\mathcal{A}$ cannot cause corrupted wardens to behave differently across the two games. Hence the outputs coincide.

    \medskip
    \noindent\textbf{(1.3) I/O outputs for message delivery (strategy~(i)).} For both settlement types, both simulators record $\mathcal{A}$'s message-delivery decisions and trigger the settlement message to $\mathcal{F}^{\text{Brick}}_{\text{settlement}}$ only when $\mathcal{P'}^{\text{Brick}}$ completes the procedure. Delivery delays affect both games identically and cannot create a distinguishable difference.

    \medskip
    \noindent\textbf{(2) On-chain transactions.} The settlement procedure publishes $\mathit{TX}_{\mathit{close}}$ (collaborative), or $\mathit{TX}_{\mathit{unilateral}}$, $\{\mathit{TX}_{\mathit{settle}}\}$, and $\{\mathit{TX}_{\mathit{fraud}}\}$ (unilateral) to $\mathcal{F}_{\text{ledger}}$. These are generated by the simulated $\mathcal{P'}^{\text{Brick}}$, which runs identically in both simulators. The game hop only affects when the ideal functionality generates I/O output, not which transactions are published, since the message delivery and corruption are synchronized. Transactions may contain different signature values due to independent randomness in $\mathcal{F}_{\text{sig}}$, but by EUF-CMA security the distributions are computationally indistinguishable.

    \medskip
    \noindent\textbf{(3) Adversarial leakage.} Both simulators maintain synchronized corruption status and execute the same protocol logic inside $\mathcal{P'}^{\text{Brick}}$ with same input request from the ideal functioanlity. The game hop interposes $\mathcal{F}^{\text{Brick}}_{\text{settlement}}$ between the simulator and the I/O output but does not alter the simulator's internal execution or its interaction with corrupted parties. Hence the leakage delivered to $\mathcal{A}$ is identical.

    \medskip
    Conclusively,
    $\mathsf{EXEC}^{\mathcal{F}^{\text{Brick}}_{\text{layer2-read}}}_{\mathcal{S}_{\text{Brick-read}}, \mathcal{E}}(k)\stackrel{c}{\approx}\mathsf{EXEC}^{\mathcal{F}^{\text{Brick}}_{\text{layer2-settlement}}}_{\mathcal{S}_{\text{Brick-settlement}}, \mathcal{E}}(k)$.
\end{proof}

As a final step, we extend the ideal functionality $\mathcal{F}^{\text{Brick}}_{\text{layer2-settlement}}$ with the subroutine $\mathcal{F}^{\text{Brick}}_{\text{updRnd}}$ to reach $\mathcal{F}^{\text{Brick}}_{\text{layer2}} = (\mathcal{F}_{\text{client}}, \mathcal{F}_{\text{ledger}} \mid \mathcal{F}^{\text{Brick}}_{\text{submit}}, \mathcal{F}^{\text{Brick}}_{\text{join}}, \mathcal{F}^{\text{Brick}}_{\text{update}}, \mathcal{F}^{\text{Brick}}_{\text{read}}, \mathcal{F}^{\text{Brick}}_{\text{settlement}}, \mathcal{F}^{\text{Brick}}_{\text{updRnd}})$, the according simulator $\mathcal{S}_{\text{Brick-updRnd}}$ defined as follow:

\vspace{1em}\begin{functionality}{Description of $\mathcal{M}_{\text{client}}$ of $\mathcal{F}^{\text{Brick}}_{\text{layer2}}$}{

\textbf{Implemented role(s):} \{client\}

\noindent \textbf{Main:}

\vspace{0.5em}

\textbf{recv} \{Submit, $\mathit{request}$\} \textbf{from} I/O: \label{request:submit}

\begin{enumerate}[itemsep=0.5em]
\item \textbf{send} \{Submit, $\mathit{request}$, $\mathsf{internalState}$\} \textbf{to} $(\pcur, \scur, \mathcal{F}^{\text{Brick}}_{\text{submit}}:\text{submit})$,
\textbf{wait for} \{Submit, $\mathit{response}$\} s.t. $\mathit{response} \in \{\text{true}, \text{false}\}$;
\item \textbf{if} $\mathit{response} =$ true: $\mathsf{requestQueue}.\text{add}(\mathit{request})$;
\textbf{send} $\mathit{request}$ \textbf{to} $\mathcal{S}$ via NET;
\end{enumerate}

\hrule
\vspace{0.5em}

\textbf{recv} \{Join, $\mathit{Attachment}$\} \textbf{from} NET: \label{request:join}

\begin{enumerate}[itemsep=0.5em]
\item \textbf{send} \{Join, $\mathit{Attachment}$, $\mathsf{internalState}$\} \textbf{to} $(\pcur, \scur, \mathcal{F}^{\text{Brick}}_{\text{join}}:\text{join})$,
\textbf{wait for} \{Join, $\mathit{response}$\} s.t. $\mathit{response} \in \{\text{true}, \text{false}\}$;
\item \textbf{if} $\mathit{response} =$ true: update $\mathsf{internalState}$ according to $\mathit{Attachment}$;
\textbf{reply} \{Join, $s_{\mathit{init}}$\} via I/O;
\end{enumerate}

\hrule
\vspace{0.5em}

\textbf{recv} \{Update, $\mathit{Attachment}$\} \textbf{from} NET: \label{request:update}

\begin{enumerate}[itemsep=0.5em]
\item \textbf{send} \{Update, $\mathit{Attachment}$, $\mathsf{internalState}$\} \textbf{to} $(\pcur, \scur, \mathcal{F}^{\text{Brick}}_{\text{update}}:\text{update})$,
\textbf{wait for} \{Update, $\mathit{response}$, $\mathit{newState}$, $\mathit{executedReq}$\};
\item \textbf{if} $\mathit{response} =$ true: update $\mathsf{internalState}$ with $\mathit{newState}$ and $\mathit{executedReq}$;
\end{enumerate}

\hrule
\vspace{0.5em}

\textbf{recv} \{Read\} \textbf{from} I/O: \label{request:read}

\begin{enumerate}[itemsep=0.5em]
\item \textbf{send} \{Read, $\mathsf{internalState}$\} \textbf{to} $(\pcur, \scur, \mathcal{F}^{\text{Brick}}_{\text{read}}:\text{read})$,
\textbf{wait for} $\mathit{ReadResult}$;
\item \textbf{if} $\mathit{ReadResult} \neq \bot$: \textbf{reply} \{Read, $\mathit{ReadResult}$\} via I/O;
\end{enumerate}

\hrule
\vspace{0.5em}

\textbf{recv} \{Settlement, $\mathit{Attachment}$\} \textbf{from} NET: \label{request:settlement}

\begin{enumerate}[itemsep=0.5em]
\item \textbf{send} \{Settlement, $\mathit{Attachment}$, $\mathsf{internalState}$\} \textbf{to} $(\pcur, \scur, \mathcal{F}^{\text{Brick}}_{\text{settlement}}:\\\text{settlement})$,
\textbf{wait for} \{Settlement, $\mathit{response}$, $s_{\mathit{settle}}$\} s.t. $\mathit{response} \in \{\text{true}, \text{false}\}$;
\item \textbf{if} $\mathit{response} =$ true: update $\mathsf{internalState}$;
\textbf{reply} \{Settlement, $s_{\mathit{settle}}$\} via I/O;
\end{enumerate}

\hrule
\vspace{0.5em}

\textbf{recv} \{UpdateRound\} \textbf{from} NET: \label{request:updRnd}

\begin{enumerate}[itemsep=0.5em]
\item \textbf{send} \{UpdateRound, $\mathsf{internalState}$\} \textbf{to} $(\pcur, \scur, \mathcal{F}^{\text{Brick}}_{\text{updRnd}}:\text{updRnd})$,
\textbf{wait for} \{UpdateRound, $\mathit{response}$\} s.t. $\mathit{response} \in \{\text{true}, \text{false}\}$;
\item \textbf{if} $\mathit{response} =$ true: $\mathsf{round} \leftarrow \mathsf{round} + 1$;
\textbf{reply} \{UpdateRound, $\mathit{response}$\} via NET;
\end{enumerate}

\hrule
\vspace{0.5em}

\textbf{recv} \{GetCurRound\} \textbf{from} I/O or NET:
\begin{enumerate}[itemsep=0.5em]
    \item \textbf{reply} \{GetCurRound, $\mathsf{round}$\};
\end{enumerate}

}\end{functionality}\vspace{.5em}

\begin{functionality}{Description of simulator $\mathcal{S}_{\text{Brick-updRnd}}$}{

The simulator $\mathcal{S}_{\text{Brick-updRnd}}$ behaves identically to $\mathcal{S}_{\text{Brick-settlement}}$, except for the following additional behavior when handling round-update requests:

\vspace{0.5em}
\textbf{Round update interaction with $\mathcal{F}^{\text{Brick}}_{\text{layer2}}$:}

\begin{enumerate}[itemsep=0.5em]

\item Whenever $\mathcal{S}_{\text{Brick-updRnd}}$ receives a round-update instruction from $\mathcal{A}$ (i.e., $\mathcal{A}$ advances the clock in the simulated $\mathcal{P'}^{\text{Brick}}$), $\mathcal{S}_{\text{Brick-updRnd}}$ simultaneously sends $\{\text{UpdateRound}\}$ to $\mathcal{F}^{\text{Brick}}_{\text{layer2}}$ via NET.

\end{enumerate}

}\end{functionality}\vspace{1em}

\begin{lemma}
\label{lem:Brick7}
    For all PPT adversaries $\mathcal{A}$, there exists a PPT simulator $\mathcal{S}_{\text{Brick-updRnd}}$ such that for all PPT environments $\mathcal{E}$ and all security parameters $k\in\mathbb{N}$,
    \[
    \mathsf{EXEC}^{\mathcal{F}^{\text{Brick}}_{\text{layer2-settlement}}}_{\mathcal{S}_{\text{Brick-settlement}},\mathcal{E}}(k)
    \ \stackrel{c}{\approx}\
    \mathsf{EXEC}^{\mathcal{F}^{\text{Brick}}_{\text{layer2}}}_{\mathcal{S}_{\text{Brick-updRnd}},\mathcal{E}}(k),
    \]
    where $\stackrel{c}{\approx}$ denotes computational indistinguishability.
\end{lemma}

\begin{proof}
    Fix an arbitrary PPT environment $\mathcal{E}$. The only difference between the two games is that $\mathsf{EXEC}^{\mathcal{F}^{\text{Brick}}_{\text{layer2}}}_{\mathcal{S}_{\text{Brick-updRnd}},\mathcal{E}}(k)$ forward the UpdateRound request from simualtor through NET to the subroutine $\mathcal{F}^{\text{Brick}}_{\text{updRnd}}$ to decide the round update, whereas in $\mathcal{F}^{\text{Brick}}_{\text{layer2-settlement}}$ round updates are handled entirely by the simulator. We analyze the three observable components.

    The UpdateRound request affects a single observable quantity: the value of $\mathsf{round}$ returned by the GetCurRound request. In \\$\mathsf{EXEC}^{\mathcal{F}^{\text{Brick}}_{\text{layer2-settlement}}}_{\mathcal{S}_{\text{Brick-settlement}},\mathcal{E}}(k)$, $\mathcal{S}_{\text{Brick-settlement}}$ maintains the round counter inside the simulated $\mathcal{P'}^{\text{Brick}}$, whose clock is advanced by $\mathcal{A}$, and forwards the round value to $\mathcal{E}$ upon a GetCurRound request. In $\mathsf{EXEC}^{\mathcal{F}^{\text{Brick}}_{\text{layer2}}}_{\mathcal{S}_{\text{Brick-updRnd}},\mathcal{E}}(k)$, $\mathcal{S}_{\text{Brick-updRnd}}$ additionally forwards each round-update request from $\mathcal{A}$ to $\mathcal{F}^{\text{Brick}}_{\text{updRnd}}$, which updates $\mathsf{round}$ in $\mathsf{internalState}$.

    \medskip
    \noindent\textbf{(1) I/O outputs.} The only I/O output affected is $\{\text{GetCurRound}, \mathsf{round}\}$. Under asynchronous communication, $\mathcal{F}^{\text{Brick}}_{\text{updRnd}}$ imposes no time bound on collaborative operations (join, state update, collaborative settlement), and unconditionally admits round advancement for these request types. The only liveness constraint enforced by $\mathcal{F}^{\text{Brick}}_{\text{updRnd}}$ is on unilateral settlement: a round-update request is rejected whenever there exists a pending unilateral settlement request from an honest client at time~$t$ whose corresponding $s_{\mathit{settle}}$ is not yet reflected in $\mathsf{onchainState}$ within $T_{L_1}$.

We argue that this check is always satisfied in the simulated game. By the definition of $\mathcal{S}_{\text{Brick-updRnd}}$, once the real protocol $\mathcal{P'}^{\text{Brick}}$ completes the unilateral settlement procedure, i.e., the corresponding $\mathit{TX}_{\mathit{unilateral}}$ and at least $2f{+}1$ warden settlement transactions are committed on L1, the simulator notifies $\mathcal{F}^{\text{Brick}}_{\text{settlement}}$, which updates $\mathsf{onchainState}$ to the resolved settlement state. The simulator only permits the clock in $\mathcal{P'}^{\text{Brick}}$ to advance past the unilateral-settlement deadline after this on-chain publication has occurred, so by the time a \{UpdateRound\} request reaches $\mathcal{F}^{\text{Brick}}_{\text{updRnd}}$, the pending-unilateral check passes. Every round-update instruction from $\mathcal{A}$ is therefore accepted, and $\mathsf{round}$ in $\mathsf{internalState}$ evolves in lockstep with the simulated clock in $\mathcal{P'}^{\text{Brick}}$. The GetCurRound output delivered to $\mathcal{E}$ is identical in both games. All other I/O outputs (Join, Read, Settlement) are unaffected by this game hop.

    \medskip
    \noindent\textbf{(2) On-chain transactions.} The UpdateRound request does not publish any transactions to $\mathcal{F}_{\text{ledger}}$. Hence on-chain transactions are identical in both games.

    \medskip
    \noindent\textbf{(3) Adversarial leakage.} The game hop only interposes $\mathcal{F}^{\text{Brick}}_{\text{updRnd}}$ on the UpdateRound request. Since $\mathcal{F}^{\text{Brick}}_{\text{updRnd}}$ unconditionally accepts and does not generate any network messages to corrupted parties, the simulator's internal execution of $\mathcal{P'}^{\text{Brick}}$ and its interaction with corrupted parties are unaffected. The leakage delivered to $\mathcal{A}$ is identical.

    \medskip
    Combining (1)--(3):
    $\mathsf{EXEC}^{\mathcal{F}^{\text{Brick}}_{\text{layer2-settlement}}}_{\mathcal{S}_{\text{Brick-settlement}}, \mathcal{E}}(k)\stackrel{c}{\approx}\mathsf{EXEC}^{\mathcal{F}^{\text{Brick}}_{\text{layer2}}}_{\mathcal{S}_{\text{Brick-updRnd}}, \mathcal{E}}(k)$.
\end{proof}

\ThmrealizeBrick*
\begin{proof}
    Let $\mathcal{A}$ be any PPT adversary and let $\mathcal{E}$ be any PPT environment. By Lemmas~\ref{lem:Brick1}--\ref{lem:Brick7}, the sequence of game hops yields a chain of computationally indistinguishable execution ensembles, each adjacent pair differing only in how one subroutine is implemented:
    \[
    \mathsf{EXEC}^{\mathcal{P}^{\text{Brick}}}_{\mathcal{A},\mathcal{E}}
    \ \stackrel{c}{\approx}\ 
    \mathsf{EXEC}^{\mathcal{F}^{\text{Brick}}_{\text{dummy}}}_{\mathcal{S}_{\text{Brick}},\mathcal{E}}
    \ \stackrel{c}{\approx}\ 
    \cdots
    \ \stackrel{c}{\approx}\ 
    \mathsf{EXEC}^{\mathcal{F}^{\text{Brick}}_{\text{layer2}}}_{\mathcal{S}_{\text{Brick-updRnd}},\mathcal{E}}.
    \]
    By transitivity of computational indistinguishability, we obtain
    $
    \mathsf{EXEC}^{\mathcal{P}^{\text{Brick}}}_{\mathcal{A},\mathcal{E}}
    \stackrel{c}{\approx}
    \mathsf{EXEC}^{\mathcal{F}^{\text{Brick}}_{\text{layer2}}}_{\mathcal{S}_{\text{Brick-updRnd}},\mathcal{E}}$, which proves that $\mathcal{P}^{\text{Brick}}$ iUC-realizes $\mathcal{F}^{\text{Brick}}_{\text{layer2}}$.
\end{proof}
\section{Case Study: The Liquid Sidechain}
\label{apdx:Liquid}

\subsection{$\mathcal{F}_{\text{ledger}}$ instantiation for Liquid sidechain}
\label{apd:LiquidLedger}

\subsubsection*{Submission functionality $\mathcal{F}_{\text{submit}_{\text{L1}}}$}

The submission subroutine accepts a request to submit a transaction to L1 if and only if the transaction conforms to one of the Liquid transaction types. For the Liquid sidechain, the following transactions are accepted by $\mathcal{F}^{\text{Liquid}}_{\text{submit}_{\text{L1}}}$: client deposit transactions ($\mathit{TX}_{\mathit{deposit}}$), which lock funds on L1 and carry the initial sidechain state; and operator-submitted settlment transactions $\mathit{TX}_{\mathit{settlement}}$ for the client, which release funds from the sidechain back to L1 at settlement according to the included L2 state. Sidechain chain transactions, that are submitted by clients to operators through $\mathcal{F}^{\text{Liquid}}_{\text{com}}$ and BFT block-consensus messages (\{UpdatePropose\}, \{UpdatePrecommitment\}, \{UpdateFinal\}) are handled off-chain and are \emph{not} submitted to L1.

\subsubsection*{Update functionality $\mathcal{F}_{\text{update}_{\text{L1}}}$}

The update subroutine refines the base $\mathcal{F}_{\text{update}_{\text{L1}}}$ by enforcing Liquid-specific commitment semantics for sidechain entry and exit:

\begin{itemize}
    \item \textbf{Deposit confirmation.} A deposit transaction $\mathit{TX}_{\mathit{deposit}}$ is treated as committed for the purposes of joining the sidechain only after it has been committed on L1 for at least $100$ blocks. Before this bound is reached, operators do not enqueue the corresponding peg-in request.
    \item \textbf{Sidechain state settlement.} If $\mathit{TX}_{\mathit{settlement}}$ is committed on L1, the committed state $\mathit{State}_{\text{L1}}$ is set to the settlement state carried in $\mathit{TX}_{\mathit{settlement}}$, which must match the latest sidechain state agreed upon by the BFT quorum ($2f{+}1$ operator signatures on the block containing the corresponding request). Any peg-out transaction whose state is inconsistent with the latest BFT-agreed state is rejected at the update stage.
\end{itemize}

\subsubsection*{Preservation of L1 safety and liveness.}

The instantiation preserves the original security guarantees of $\mathcal{F}_{\text{ledger}}$. Liquid-specific logic is purely additive: the new transaction types ($\mathit{TX}_{\mathit{deposit}}$, $\mathit{TX}_{\mathit{settlement}}$) are accepted as valid L1 payloads, and the deposit-confirmation and peg-out-resolution rules are deterministic predicates over transactions already committed on L1. No existing L1 transaction is rejected or reordered by the instantiation, so the underlying ledger's security still holds. 

\subsection{Liquid Sidechain Real Protocol}
\label{apd:LiquidReal}

We formally define the real-world protocol implementation $\mathcal{P}^{\text{Liquid}}$ as
$\mathcal{P}^{\text{Liquid}} := (\mathcal{P}^{\text{Liquid}}_{\text{client}} \mid \mathcal{P}^{\text{Liquid}}_{\text{operator}}, \mathcal{F}^{\text{Liquid}}_{\text{com}}, \mathcal{F}_{\text{sig}})$.
The client and operator are the two participating roles. Clients are the main parties instructed by the environment to send out requests and operators are responsible for maintaining the Liquid sidechain via a three-phase BFT consensus protocol. We assume $n = 3f+1$ operators in total, of which the adversary may corrupt at most~$f$. The structure is shown in Figure~\ref{fig:liquid}.

\begin{figure}
    \centering
    \includegraphics[width=\linewidth]{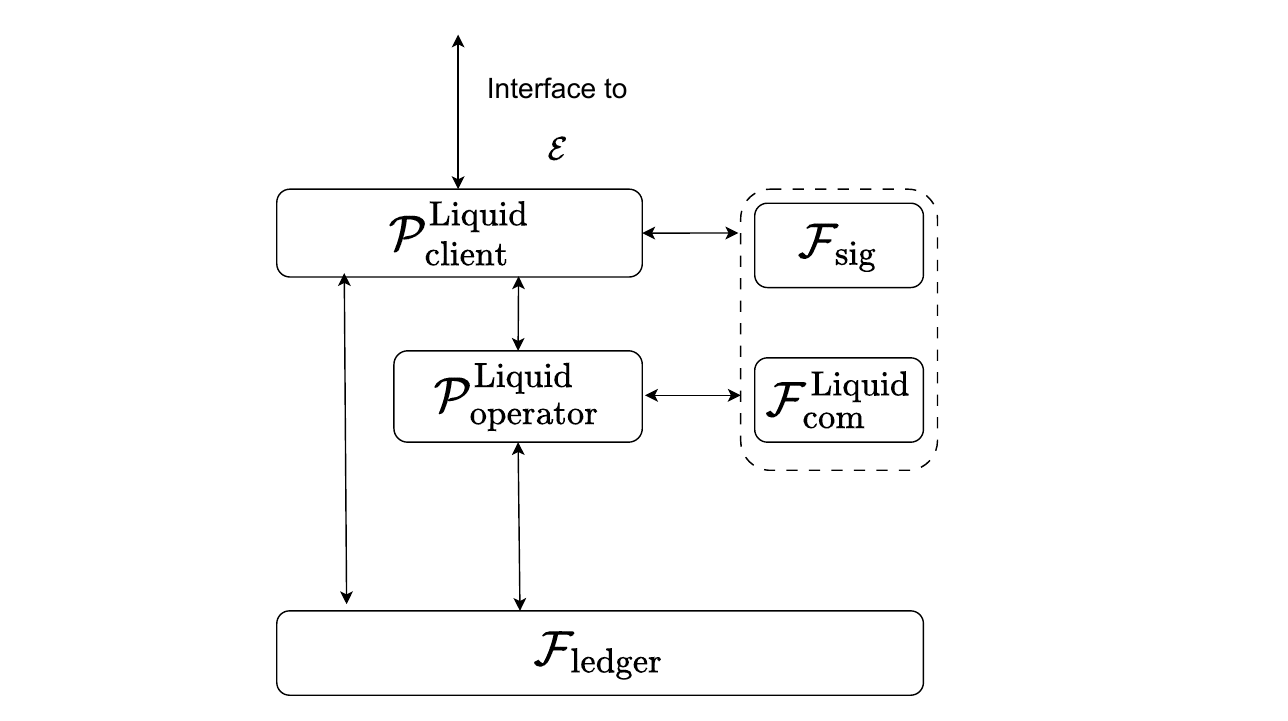}
    \caption{The Liquid sidechain protocol}
    \label{fig:liquid}
\end{figure}

\subsubsection{Client Protocol}

The client machine $\mathcal{M}^{\text{Liquid}}_{\text{client}}$ specifies the behavior of an honest client in the Liquid sidechain.

\vspace{1em}\begin{functionality}{Description of protocol $\mathcal{P}^{\text{Liquid}}_{\text{client}} = (\text{client})$}{

\textbf{Participating roles:} \{client\}

\noindent \textbf{Corruption model:} dynamic corruption, at least one honest

\noindent \textbf{Protocol parameters:}
\begin{itemize}
    \item $n = 3f+1$: operator committee size
    \item $f$: corruption threshold for operators
\end{itemize}

}\end{functionality}\vspace{.5em}
\begin{functionality}{Description of $\mathcal{M}^{\text{Liquid}}_{\text{client}}$}{

\textbf{Implemented role(s):} \{client\}

\noindent \textbf{Subroutines:} $\mathcal{F}_{\text{sig}}$: \{signer, verifier\}, $\mathcal{F}^{\text{Liquid}}_{\text{com}}$: com, $\mathcal{F}_{\text{ledger}}$: client\textsubscript{L1}, $\mathcal{P}^{\text{Liquid}}_{\text{operator}}$: operator

\noindent \textbf{Internal state:}
\begin{itemize}
    \item $\mathsf{round} \in \mathbb{N}_{\geq 0}$, $\mathsf{round} = 0$ \hfill\cmt{Current round}
    \item $\mathsf{requestQueue} \subset \{0,1\}^{*}$, $\mathsf{requestQueue} = \emptyset$ \hfill\cmt{Queued unexecuted requests}
    \item $\mathsf{executedRequest} \subset \{0,1\}^{*}$, $\mathsf{executedRequest} = \emptyset$ \hfill\cmt{Executed requests}
    \item $\mathsf{stateList} \subset \{0,1\}^{*}$, $\mathsf{stateList} = \emptyset$ \hfill\cmt{L2 state list}
    \item $\mathsf{onchainState} \subset \{0,1\}^{*}$, $\mathsf{onchainState} = \emptyset$ \hfill\cmt{L1 committed state}
    \item $\mathsf{identities} \subset \{0,1\}^{*}$, $\mathsf{identities} = \{Address, \{pid_{op}\}\}$ \hfill\cmt{Sidechain operators and sidechain L1 address}
\end{itemize}

\noindent \textbf{CheckID}(\emph{pid}, \emph{sid}, \emph{role}): Accept all messages with the same \emph{sid}.

\noindent \textbf{Corruption behavior:}
\begin{itemize}
    \item \textbf{DetermineCorrStatus}(\emph{pid}, \emph{sid}, \emph{role}): Return $\mathsf{corr}$.
    \item \textbf{LeakedData}(\emph{pid}, \emph{sid}, \emph{role}): Return \texttt{Internal State}.
\end{itemize}

\vspace{0.5em}
\noindent \textbf{Main:}

\vspace{0.5em}

\textbf{recv} \{Join, Identities, $s_{\mathit{init}}$\} \textbf{from} I/O:

\begin{enumerate}[itemsep=0.5em]
    \item Prepare deposit transaction $\mathit{TX}_{\mathit{deposit}}=\{\pcur, Address, s_{\mathit{init}}, \epsilon\}$ from $s_{\mathit{init}}$ \cmt{ parse destination address $\mathit{Address}$ from Identities}

    \item \textbf{send} \{SubmitL1, $\mathit{TX}_{\mathit{deposit}}$, $\mathit{Address}$\} \textbf{to} $(\pcur, \scur, \mathcal{F}_{\text{ledger}}: \text{client}_{\text{L1}})$;

    \item \textbf{send} \{ReadL1\} \textbf{to} $(\pcur, \scur, \mathcal{F}_{\text{ledger}}:\text{client}_{\text{L1}})$;
\textbf{wait for} \{ReadL1, $\mathit{L1ReadResult} = \{\{\mathit{TX}\}_{\text{L1}}, \mathit{State}_{\text{L1}}, \{\mathit{pid}\}\}$\};

\item \textbf{if} $\mathit{TX}_{\mathit{deposit}} \in \{\mathit{TX}\}_{\text{L1}}$, \textbf{send} \{Sign, $\mathit{TX}_{\mathit{peg\text{-}in}}$\} \textbf{to} $(\pcur, \scur, \mathcal{F}_{\text{sig}}: \text{signer})$;
    \textbf{wait for} \{Signature, $\sigma$\};

    \item Prepare $\mathit{TX}_{\mathit{peg\text{-}in}}=\{\epsilon, \pcur, s_{\mathit{init}}, \mathit{TX}_{\mathit{deposit}}\}$ \textbf{send} \{Message, \{Join, $\mathit{TX}_{\mathit{peg\text{-}in}}$, $\sigma$, \emph{receiver}\}\} \textbf{to} $(\pcur, \scur, \mathcal{F}^{\text{Liquid}}_{\text{com}}: \text{com})$, where $\mathit{receiver} = (\_, \scur, \mathcal{P}^{\text{Liquid}}_{\text{operator}}:\text{operator})$; \cmt{Send to all operators}

    \item \textbf{send} \{Read, \emph{receiver}\} \textbf{to} $(\pcur, \scur, \mathcal{F}^{\text{Liquid}}_{\text{com}}: \text{com})$, where $\mathit{receiver} = (\_, \scur,\\ \mathcal{P}^{\text{Liquid}}_{\text{operator}}: \text{operator})$;
\textbf{wait for} \{Read, $\mathit{L2ReadResult} = \{\mathsf{executedRequest}_{\text{op}},\\ \mathsf{stateList}_{\text{op}}\}$\}; \cmt{Read sidechain view from operators}

\item \textbf{if} $\mathit{TX}_{\mathit{peg\text{-}in}} \in \mathit{L2ReadResult}.\mathsf{executedRequest}_{\text{op}}$: \cmt{Peg-in confirmed in the operator-side sidechain history}
\begin{itemize}
    \item $\mathsf{stateList} \leftarrow \mathit{L2ReadResult}.\mathsf{stateList}_{\text{op}}$; \cmt{Adopt operator-confirmed L2 state}
    \item $\mathsf{onchainState} \leftarrow s_{\mathit{init}}$; \cmt{Initial state committed on L1}
    \item $\mathsf{identities}.\text{add}(\pcur)$; \cmt{Register self as a participating client}
    \item \textbf{reply} \{Join, $s_{\mathit{init}}$\} via I/O;
\end{itemize}
\end{enumerate}

\hrule
\vspace{0.5em}

\textbf{recv} \{Submit, $\mathit{TX}$\} \textbf{from} I/O, \textbf{s.t.} request is valid:

\begin{enumerate}[itemsep=0.5em]
\item \textbf{send} \{Sign, $\mathit{TX}$\} \textbf{to} $(\pcur, \scur, \mathcal{F}_{\text{sig}}: \text{signer})$,
\textbf{wait for} \{Signature, $\sigma$\};

\item \textbf{send} \{Message, \{Submit, $\mathit{TX}$, $\sigma$, \emph{receiver}\}\} \textbf{to} $(\pcur, \scur, \mathcal{F}^{\text{Liquid}}_{\text{com}}: \text{com})$, where $\mathit{receiver} = (\_, \scur, \mathcal{P}^{\text{Liquid}}_{\text{operator}}: \text{operator})$; \cmt{Send to all operators}

\item $\mathsf{requestQueue}.\text{add}(\mathit{TX})$;
\end{enumerate}

\hrule
\vspace{0.5em}

\textbf{recv} \{Settlement\} \textbf{from} I/O:

\begin{enumerate}[itemsep=0.5em]

\item Prepare peg-out transaction $\mathit{TX}_{\mathit{peg\text{-}out}}=\{\pcur, \epsilon, s_{\mathit{settle}}, \epsilon\}$;

\item \textbf{send} \{Sign, $\mathit{TX}_{\mathit{peg\text{-}out}}$\} \textbf{to} $(\pcur, \scur, \mathcal{F}_{\text{sig}}: \text{signer})$;
\textbf{wait for} \{Signature, $\sigma$\};

\item \textbf{send} \{Message, \{Settlement, $\mathit{TX}_{\mathit{peg\text{-}out}}$, $\sigma$, \emph{receiver}\}\} \textbf{to} $(\pcur, \scur, \mathcal{F}^{\text{Liquid}}_{\text{com}}: \text{com})$, where $\mathit{receiver} = (\_, \scur, \mathcal{P}^{\text{Liquid}}_{\text{operator}}: \text{operator})$; \cmt{Submit peg-out to all operators}

\item \textbf{send} \{Read, \emph{receiver}\} \textbf{to} $(\pcur, \scur, \mathcal{F}^{\text{Liquid}}_{\text{com}}: \text{com})$, where $\mathit{receiver} = (\_, \scur,\\ \mathcal{P}^{\text{Liquid}}_{\text{operator}}: \text{operator})$,
\textbf{wait for} \{Read, $\mathit{L2ReadResult} = \{\mathsf{executedRequest}_{\text{op}},\\ \mathsf{stateList}_{\text{op}}, \mathsf{onchainState}_{\text{op}}\}$\}; \cmt{Read sidechain view from operators}

\item \textbf{send} \{ReadL1\} \textbf{to} $(\pcur, \scur, \mathcal{F}_{\text{ledger}}: \text{client}_{\text{L1}})$;
\textbf{wait for} \{ReadL1, $\mathit{L1ReadResult} = \{\{\mathit{TX}\}_{\text{L1}}, \mathit{State}_{\text{L1}}, \{\mathit{pid}\}\}$\};

\item Check that the operator has confirmed and finalized the peg-out:
\begin{itemize}
    \item $\mathit{TX}_{\mathit{peg\text{-}out}} \in \mathit{L2ReadResult}.\mathsf{executedRequest}_{\text{op}}$; \cmt{Peg-out is included in the operator-confirmed sidechain history}
    \item $\exists\, \mathit{TX}_{\mathit{settle}} \in \mathit{L1ReadResult}.\{\mathit{TX}\}_{\text{L1}}$ \textbf{s.t.} $\mathit{TX}_{\mathit{settle}}$ is an operator-published settlement transaction carrying $s_{\mathit{settle}}$; \cmt{Operator has posted the settlement transaction on L1}
    \item $s_{\mathit{settle}} \in \mathit{L1ReadResult}.\mathit{State}_{\text{L1}}$; \cmt{Settlement state is recorded on L1}
\end{itemize}

\item \textbf{if} all checks pass:
\begin{itemize}
    \item $\mathsf{onchainState} \leftarrow s_{\mathit{settle}}$; \cmt{Update L1-committed state}
    \item $\forall\, \mathit{pid}' \in \mathit{L1ReadResult}.\{\mathit{pid}\}$: $\mathsf{identities}.\text{add}(\mathit{pid}')$; \cmt{Register newly visible participants}
    \item \textbf{reply} \{Settlement, $s_{\mathit{settle}}$\} via I/O;
\end{itemize}

\end{enumerate}

\hrule
\vspace{0.5em}

\textbf{recv} \{Read\} \textbf{from} I/O:

\begin{enumerate}[itemsep=0.5em]

\item \textbf{send} \{Read, \emph{receiver}\} \textbf{to} $(\pcur, \scur,\\ \mathcal{F}^{\text{Liquid}}_{\text{com}}: \text{com})$, where $\mathit{receiver} = (\_, \scur, \mathcal{P}^{\text{Liquid}}_{\text{operator}}: \text{operator})$;
\textbf{wait for} \{Read, $\mathit{L2ReadResult} = \{\mathsf{executedRequest}_{\text{op}}, \mathsf{stateList}_{\text{op}}\}$\};

\item Reconstruct the sidechain from $\mathit{L2ReadResult}$: $\forall\, \mathit{Block}_k\\ \in \mathit{L2ReadResult}.\mathsf{executedRequest}_{\text{op}}$ \textbf{s.t.} $\mathit{Block}_k$ carries $\geq 2f{+}1$ operator signatures: \cmt{Accept blocks finalized by BFT quorum}
\begin{itemize}
    \item $\forall\, \mathit{TX} \in \mathit{Block}_k$ with $\mathit{TX} \in \mathsf{requestQueue}$: $\mathsf{requestQueue}.\text{remove}(\mathit{TX})$, \\$\mathsf{executedRequest}.\text{add}(\mathit{TX})$;
\end{itemize}

\item $\mathsf{stateList} \leftarrow \mathit{L2ReadResult}.\mathsf{stateList}_{\text{op}}$;

\item \textbf{reply} \{Read, $\mathit{ReadResult} = \{\mathsf{executedRequest}, \mathsf{stateList}, \mathsf{onchainState}\}$\} via I/O;

\end{enumerate}

\hrule
\vspace{0.5em}

\textbf{recv} \{GetCurRound\} \textbf{from} I/O:
\begin{enumerate}[itemsep=0.5em]
\item \textbf{send} \{GetCurRound\} \textbf{to} $(\pcur, \scur, \mathcal{F}^{\text{Liquid}}_{\text{com}}: \text{clock})$,
\textbf{wait for} \{GetCurRound, $\mathit{round}$\};
\item \textbf{reply} \{GetCurRound, $\mathit{round}$\} via I/O;
\end{enumerate}

}\end{functionality}\vspace{1em}
\subsubsection{Operator Protocol}

The operator machine $\mathcal{M}^{\text{Liquid}}_{\text{operator}}$ maintains the Liquid sidechain. Operators run a simplified three-round BFT consensus: every three rounds, $\mathcal{F}^{\text{Liquid}}_{\text{com}}$ notifies the leader to propose a new block (\{UpdateSideChain\}); other operators precommit (\{UpdatePrecommitment\}); after $2f{+}1$ precommitments, a final commit is broadcast (\{UpdateFinal\}). Transaction validity is checked via the function $\mathsf{Val}()$. 

\vspace{1em}\begin{functionality}{Description of protocol $\mathcal{P}^{\text{Liquid}}_{\text{operator}} = (\text{operator})$}{

\textbf{Participating roles:} \{operator\}

\noindent \textbf{Corruption model:} dynamic corruption

\noindent \textbf{Protocol parameters:}
\begin{itemize}
    \item $n = 3f+1$: operator committee size; $f$: corruption threshold
    \item $\mathsf{Val}()$: transaction validity checking algorithm
\end{itemize}

}\end{functionality}\vspace{.5em}
\begin{functionality}{Description of $\mathcal{M}^{\text{Liquid}}_{\text{operator}}$}{

\textbf{Implemented role(s):} \{operator\}

\noindent \textbf{Subroutines:} $\mathcal{F}_{\text{sig}}$: \{signer, verifier\}, $\mathcal{F}^{\text{Liquid}}_{\text{com}}$: com, $\mathcal{F}_{\text{ledger}}$: client\textsubscript{L1}

\noindent \textbf{Internal state:}
\begin{itemize}
    \item $\mathsf{requestQueue} \subset \{0,1\}^{*}$, $\mathsf{requestQueue} = \emptyset$ \hfill\cmt{Queued requests}
    \item $\mathsf{executedRequest} \subset \{0,1\}^{*}$, $\mathsf{executedRequest} = \emptyset$ \hfill\cmt{Executed requests}
    \item $\mathsf{stateList} \subset \{0,1\}^{*}$, $\mathsf{stateList} = \emptyset$ \hfill\cmt{L2 state list}
    \item $\mathsf{onchainState} \subset \{0,1\}^{*}$, $\mathsf{onchainState} = \emptyset$ \hfill\cmt{L1 committed state}
    \item $\mathsf{identities} \subset \{0,1\}^{*}$, $\mathsf{identities} = \emptyset$ \hfill\cmt{Registered identities}
\end{itemize}

\noindent \textbf{CheckID}(\emph{pid}, \emph{sid}, \emph{role}): Accept all messages with the same \emph{sid}. Only accept \emph{role} = operator.

\vspace{0.5em}
\noindent \textbf{Main:}

\vspace{0.5em}

\textbf{recv} \{Join, $\mathit{TX}_{\mathit{peg\text{-}in}}$, $\sigma$, $\mathit{pid}_{\mathit{call}}$\} \textbf{from} $(\pcur, \scur, \mathcal{F}^{\text{Liquid}}_{\text{com}}: \text{com})$:

\begin{enumerate}[itemsep=0.5em]

\item \textbf{send} \{Verify, $\mathit{TX}_{\mathit{peg\text{-}in}}$, $\sigma$\} \textbf{to} $(\pcur, \scur, \mathcal{F}_{\text{sig}}: \text{verifier})$;
\textbf{wait for} \{VerResult, $b_{\sigma}$\}; \cmt{Verify caller's signature on the peg-in transaction}

\item Check that $\mathsf{Val}(\mathit{TX}_{\mathit{peg\text{-}in}}) = \text{true}$; \cmt{Peg-in transaction is well-formed}

\item \textbf{send} \{ReadL1\} \textbf{to} $(\pcur, \scur, \mathcal{F}_{\text{ledger}}:\text{client}_{\text{L1}})$;
\textbf{wait for} \{ReadL1, $\mathit{L1ReadResult} = \{\{\mathit{TX}\}_{\text{L1}}, \mathit{State}_{\text{L1}}, \{\mathit{pid}\}\}$\};

\item Check that $\mathit{TX}_{\mathit{deposit}} \in \{\mathit{TX}\}_{\text{L1}}$; \cmt{Peg-in transaction is committed on L1}

\item \textbf{if} $b_{\sigma} = \text{true}$ and all checks pass: $\mathsf{requestQueue}.\text{add}(\mathit{TX}_{\mathit{peg\text{-}in}})$; \cmt{Enqueue peg-in for BFT consensus}

\end{enumerate}

\hrule
\vspace{0.5em}

\textbf{recv} \{Submit, $\mathit{TX}$, $\sigma$, $\mathit{pid}_{\mathit{call}}$\} \textbf{from} $(\pcur, \scur, \mathcal{F}^{\text{Liquid}}_{\text{com}}: \text{com})$:

\begin{enumerate}[itemsep=0.5em]
\item \textbf{send} \{Verify, $\mathit{TX}$, $\sigma$\} \textbf{to} $(\pcur, \scur, \mathcal{F}_{\text{sig}}: \text{verifier})$;
\textbf{wait for} \{VerResult, $b_{\sigma}$\}; \cmt{Verify caller's signature on the transaction}

\item Check that $\mathsf{Val}(\mathit{TX}) = \text{true}$; \cmt{Transaction is well-formed and satisfies validity predicate}

\item \textbf{if} $b_{\sigma} = \text{true}$ and all checks pass: $\mathsf{requestQueue}.\text{add}(\mathit{TX})$;
\end{enumerate}

\hrule
\vspace{0.5em}

\textbf{recv} \{Settlement, $\mathit{TX}_{\mathit{peg\text{-}out}}$, $\sigma$, $\mathit{pid}_{\mathit{call}}$\} \textbf{from} $(\pcur, \scur, \mathcal{F}^{\text{Liquid}}_{\text{com}}: \text{com})$:

\begin{enumerate}[itemsep=0.5em]
\item \textbf{send} \{Verify, $\mathit{TX}_{\mathit{peg\text{-}out}}$, $\sigma$\} \textbf{to} $(\pcur, \scur, \mathcal{F}_{\text{sig}}: \text{verifier})$;
\textbf{wait for} \{VerResult, $b_{\sigma}$\}; \cmt{Verify caller's signature on the peg-out transaction}

\item Check that $\mathsf{Val}(\mathit{TX}_{\mathit{peg\text{-}out}}) = \text{true}$; \cmt{Peg-out transaction is well-formed}

\item \textbf{if} $b_{\sigma} = \text{true}$ and all checks pass: $\mathsf{requestQueue}.\text{add}(\mathit{TX}_{\mathit{peg\text{-}out}})$;
\end{enumerate}

\hrule
\vspace{0.5em}

\textbf{recv} \{Read, $\mathit{pid}_{\mathit{call}}$\} \textbf{from} $(\pcur, \scur, \mathcal{F}^{\text{Liquid}}_{\text{com}}: \text{com})$:

\begin{enumerate}[itemsep=0.5em]
\item \textbf{reply} \{Read, $\mathit{L2ReadResult} = \{\mathsf{executedRequest}, \mathsf{stateList}\}$\} to the caller via\\ $(\pcur, \scur, \mathcal{F}^{\text{Liquid}}_{\text{com}}: \text{com})$;
\end{enumerate}

\hrule
\vspace{0.5em}

\textbf{recv} \{UpdateSideChain, $\mathit{pid}_{\mathit{call}}$\} \textbf{from} $(\pcur, \scur, \mathcal{F}^{\text{Liquid}}_{\text{com}}: \text{com})$:
\begin{enumerate}[itemsep=0.5em]
\item Generate a valid $\mathit{Block}$ from all valid requests in $\mathsf{requestQueue}$; check no double-spending;
\item \textbf{send} \{Message, \{UpdatePropose, $\mathit{Block}$, \emph{receiver}\}\} \textbf{to} $(\pcur, \scur, \mathcal{F}^{\text{Liquid}}_{\text{com}}: \text{com})$, where $\mathit{receiver} = (\_, \scur, \mathcal{P}^{\text{Liquid}}_{\text{operator}}: \text{operator})$; \cmt{Send to all operators}
\end{enumerate}

\hrule
\vspace{0.5em}

\textbf{recv} \{UpdatePropose, $\mathit{Block}$, $\sigma_{\text{leader}}$, $\mathit{pid}_{\mathit{call}}$\} \textbf{from} $(\pcur, \scur, \mathcal{F}^{\text{Liquid}}_{\text{com}}: \text{com})$:
\begin{enumerate}[itemsep=0.5em]
\item \textbf{send} \{Verify, $\mathit{Block}$, $\sigma_{\text{leader}}$\} \textbf{to} $(\pcur, \scur, \mathcal{F}_{\text{sig}}: \text{verifier})$;
\textbf{wait for} \{VerResult, $b_{\sigma}$\}; \cmt{Verify proposing operator's signature on the block}

\item Check validity of $\mathit{Block}$: correct form, valid signatures on included transactions, no double-spending;

\item \textbf{if} $b_{\sigma} = \text{true}$ and all checks pass: \textbf{send} \{Sign, $\mathit{Block}$\} \textbf{to} $(\pcur, \scur, \mathcal{F}_{\text{sig}}: \text{signer})$,
\textbf{wait for} \{Signature, $\mathit{precommitment}$\};

\item \textbf{send} \{Message, \{UpdatePrecommitment, $\mathit{Block}$, $\mathit{precommitment}$, \emph{receiver}\}\} \textbf{to}\\ $(\pcur, \scur, \mathcal{F}^{\text{Liquid}}_{\text{com}}: \text{com})$, where $\mathit{receiver} = (\_, \scur, \mathcal{P}^{\text{Liquid}}_{\text{operator}}: \text{operator})$;
\end{enumerate}

\hrule
\vspace{0.5em}

\textbf{recv} \{UpdatePrecommitment, $\mathit{Block}$, $\mathit{precommitment}$, $\mathit{pid}_{\mathit{call}}$\} \textbf{from} $(\pcur, \scur, \mathcal{F}^{\text{Liquid}}_{\text{com}}: \text{com})$:
\begin{enumerate}[itemsep=0.5em]
\item \textbf{send} \{Verify, $\mathit{Block}$, $\mathit{precommitment}$\} \textbf{to} $(\pcur, \scur, \mathcal{F}_{\text{sig}}: \text{verifier})$;
\textbf{wait for} \{VerResult, $b_{\sigma}$\}; \cmt{Verify peer operator's precommitment signature}

\item \textbf{if} $b_{\sigma} = \text{true}$ and already received $\geq 2f$ valid precommitments for $\mathit{Block}$: \textbf{send} \{Sign, $\mathit{Block}$\} \textbf{to} $(\pcur, \scur, \mathcal{F}_{\text{sig}}: \text{signer})$,
\textbf{wait for} \{Signature, $\sigma$\};

\item \textbf{send} \{Message, \{UpdateFinal, $\mathit{Block}$, $\sigma$, \emph{receiver}\}\} \textbf{to} $(\pcur, \scur, \mathcal{F}^{\text{Liquid}}_{\text{com}}: \text{com})$, where $\mathit{receiver} = (\_, \scur, \mathcal{P}^{\text{Liquid}}_{\text{operator}}: \text{operator})$;
\end{enumerate}

\hrule
\vspace{0.5em}

\textbf{recv} \{UpdateFinal, $\mathit{Block}$, $\sigma$, $\mathit{pid}_{\mathit{call}}$\} \textbf{from} $(\pcur, \scur, \mathcal{F}^{\text{Liquid}}_{\text{com}}: \text{com})$:
\begin{enumerate}[itemsep=0.5em]
\item \textbf{send} \{Verify, $\mathit{Block}$, $\sigma$\} \textbf{to} $(\pcur, \scur, \mathcal{F}_{\text{sig}}: \text{verifier})$;
\textbf{wait for} \{VerResult, $b_{\sigma}$\}; \cmt{Verify peer operator's final-commit signature}

\item \textbf{if} $b_{\sigma} = \text{true}$ and already received $\geq 2f$ valid final-commit signatures for $\mathit{Block}$:
\begin{itemize}
    \item Update $\mathsf{stateList}$, $\mathsf{executedRequest}$, $\mathsf{requestQueue}$ according to $\mathit{Block}$;
    \item \textbf{send} \{ReadL1\} \textbf{to} $(\pcur, \scur, \mathcal{F}_{\text{ledger}}: \text{client}_{\text{L1}})$;
    \textbf{wait for} \{ReadL1,\\ $\mathit{L1ReadResult}$\};
    \item For peg-in transactions in $\mathit{Block}$ committed on L1: update $\mathsf{onchainState}$ and $\mathsf{identities}$;
    \item For peg-out transactions in $\mathit{Block}$: \textbf{send} \{SubmitL1, $\mathit{TX}_{\mathit{peg\text{-}out}}$\} \textbf{to}\\ $(\pcur, \scur, \mathcal{F}_{\text{ledger}}: \text{client}_{\text{L1}})$; update $\mathsf{onchainState}$ and $\mathsf{identities}$;
\end{itemize}
\end{enumerate}

}\end{functionality}\vspace{1em}

\subsection{Liquid Sidechain Ideal Functionality}
\label{apd:LiquidIdeal}

We define the ideal functionality for the Liquid sidechain as
$\mathcal{F}^{\text{Liquid}}_{\text{layer2}}=(\mathcal{F}_{\text{client}}, \mathcal{F}_{\text{ledger}} \mid \mathcal{F}^{\text{Liquid}}_{\text{join}}, \mathcal{F}^{\text{Liquid}}_{\text{submit}}, \mathcal{F}^{\text{Liquid}}_{\text{update}}, \mathcal{F}^{\text{Liquid}}_{\text{read}},
\mathcal{F}^{\text{Liquid}}_{\text{settlement}}, \mathcal{F}^{\text{Liquid}}_{\text{updRnd}})$. The subroutines' ideal functionalities are defined as follows. To better specify the Liquid sidechain ideal functionality, the $\mathsf{executedRequest}$ is extended to include a block identifier for each request entry.

\subsubsection{Submit Functionality}
\label{apd:Liquidsubmit}

The submit subroutine $\mathcal{F}^{\text{Liquid}}_{\text{submit}}$ verifies that incoming requests are one of three types valid requests: \textit{(i)}~join requests to enter the sidechain, \textit{(ii)}~transaction requests to update the sidechain state, or \textit{(iii)}~settlement requests to leave the sidechain. 

\vspace{1em}\begin{functionality}{Description of subroutine $\mathcal{F}^{\text{Liquid}}_{\text{submit}} = (\text{submit})$}{

\textbf{Participating roles:} \{submit\}

\noindent \textbf{Corruption model:} incorruptible

}\end{functionality}\vspace{.5em}

\begin{functionality}{Description of $\mathcal{M}^{\text{Liquid}}_{\text{submit}}$}{

\textbf{Implemented role(s):} \{submit\}

\noindent \textbf{CheckID}(\emph{pid}, \emph{sid}, \emph{role}): Accept all messages with the same \emph{sid}.

\vspace{0.5em}
\noindent \textbf{Main:}

\vspace{0.5em}

\textbf{recv} \{Submit, $\mathit{request}$, $\mathsf{internalState}$\} \textbf{from} I/O:

\begin{enumerate}[itemsep=0.5em]

\item Check that $\mathit{request}$ belongs to one of the following valid types:
\begin{itemize}
    \item Join request: $\mathit{request} = (\text{Join}, s_{\mathit{init}}, \mathit{Identities})$, \textbf{s.t.}:
    \begin{itemize}
        \item $s_{\mathit{init}} \in \mathbb{N}_{\geq 0}$; \cmt{Well-formed initial balance}
    \end{itemize}

    \item State update request: $\mathit{request} = (\text{Submit}, \mathit{TX})$ with $\mathit{TX} = (\mathit{sender}, \mathit{receiver}, \mathit{value},\\ \mathit{data})$, \textbf{s.t.}:
    \begin{itemize}
        \item $\mathit{sender}, \mathit{receiver} \in \{0,1\}^{*}$ and $\mathit{value} \in \mathbb{N}_{\geq 0}$; \cmt{Well-formed transaction}
        \item $\mathit{sender} \in \mathsf{identities}$; \cmt{Sender is a registered participant}
        \item Let $\mathit{bal}(\mathit{sender})$ denote the balance of $\mathit{sender}$ derived from $\mathsf{stateList}$. Then $\mathit{value} \leq \mathit{bal}(\mathit{sender})$; \cmt{Sender has sufficient balance}
        \item $\nexists\, \mathit{TX}' \in \mathsf{executedRequest} \cup \mathsf{requestQueue}$ that conflicts with $\mathit{TX}$; \cmt{No double-spending}
    \end{itemize}

    \item Settlement request: $\mathit{request} = (\text{Settlement})$;
\end{itemize}

\item Check that the join request has been executed (i.e., $\mathsf{stateList} \neq \emptyset$) or $\mathit{request}$ is a join request; \cmt{No action before joining}

\item Check that $\nexists$ a settlement request from the caller already pending in $\mathsf{requestQueue}$; \cmt{At most one pending settlement per client}

\item \textbf{if} all checks pass: \textbf{reply} \{Submit, true\};
\item \textbf{else}: \textbf{reply} \{Submit, false\};

\end{enumerate}

}\end{functionality}\vspace{1em}

\subsubsection{Join Functionality}
\label{apd:Liquidjoin}

The join subroutine $\mathcal{F}^{\text{Liquid}}_{\text{join}}$ validates the sidechain-joining procedure. It checks that \textit{(i)}~the peg-in transaction is consistent with the initial state, \textit{(ii)}~the sidechain state (maintained by operators) already reflects the peg-in transaction, and \textit{(iii)}~the deposit transaction is committed on L1.

\vspace{1em}\begin{functionality}{Description of subroutine $\mathcal{F}^{\text{Liquid}}_{\text{join}} = (\text{join})$}{

\textbf{Participating roles:} \{join\}

\noindent \textbf{Corruption model:} incorruptible


}\end{functionality}\vspace{.5em}

\begin{functionality}{Description of $\mathcal{M}^{\text{Liquid}}_{\text{join}}$}{

\textbf{Implemented role(s):} \{join\}

\noindent \textbf{CheckID}(\emph{pid}, \emph{sid}, \emph{role}): Accept all messages with the same \emph{sid}.

\vspace{0.5em}
\noindent \textbf{Main:}

\vspace{0.5em}

\textbf{recv} \{Join, $\mathit{Attachment}=\{s_{\mathit{init}}, \mathit{TX}_{\mathit{peg\text{-}in}}\}$, $\mathsf{internalState}$\} \textbf{from} I/O:

\begin{enumerate}[itemsep=0.5em]

\item Parse $\mathit{TX}_{\mathit{peg\text{-}in}} = (\mathit{pid}_{\mathit{init}}, \epsilon, s_{\mathit{init}}', \epsilon)$. Check:
\begin{itemize}
    \item $s_{\mathit{init}}' = s_{\mathit{init}}$; \cmt{Peg-in carries the agreed initial state}
    \item $\mathit{pid}_{\mathit{init}} \in \mathsf{identities}$; \cmt{Initiator is a registered participant}
\end{itemize}

\item Check the following conditions on $\mathsf{internalState}$:
\begin{itemize}
    \item $s_{\mathit{init}} \in \mathsf{stateList}$; \cmt{Initial state is recorded in the sidechain}
    \item $\exists\, \mathit{req} \in \mathsf{executedRequest}$ \textbf{s.t.} $\mathit{TX}_{\mathit{peg\text{-}in}} \in \mathit{req}$; \cmt{A request containing the peg-in has been executed by operators}
\end{itemize}

\item \textbf{send} \{ReadL1\} \textbf{to} $(\pcur, \scur, \mathcal{F}_{\text{ledger}}: \text{client}_{\text{L1}})$,
\textbf{wait for} \{ReadL1, $\mathit{output} = \{\{\mathit{TX}\}_{\text{L1}}, \mathit{State}_{\text{L1}}, \{\mathit{pid}\}\}$\};

\item Check the following conditions on $\mathit{output}$:
\begin{itemize}
    \item $\exists\, \mathit{TX}_{\mathit{deposit}} \in \{\mathit{TX}\}_{\text{L1}}$ \textbf{s.t.} $\mathit{TX}_{\mathit{deposit}}$ is a deposit from $\mathit{pid}_{\mathit{init}}$ consistent with $s_{\mathit{init}}$; \cmt{Deposit transaction is committed on L1}
\end{itemize}

\item \textbf{if} all checks pass: \textbf{reply} \{Join, true\};

\end{enumerate}

}\end{functionality}\vspace{1em}

\begin{lemma}
\label{lem:OpenL}
    The subroutine $\mathcal{F}^{\text{Liquid}}_{\text{join}}$ guarantees \emph{correct L2 initialization}.
\end{lemma}

\begin{proof}
    We prove by contradiction. Suppose correct L2 initialization is violated: either (i)~some honest participant receives a successful initialization output whose committed initial state differs from the proposed state, or (ii)~the proposed state is not eventually committed on L1. However, $\mathcal{F}^{\text{Liquid}}_{\text{join}}$ outputs success only if (i)~the peg-in transaction is consistent with $s_{\mathit{init}}$ (Step~1), (ii)~the sidechain state reflects the peg-in (Step~2), and (iii)~$\mathit{TX}_{\mathit{peg\text{-}in}}$ is committed on L1 (Steps~3--4). As long as fewer than $f_{\mathit{OP}} = \frac{1}{3}n_{\mathit{OP}}$ operators are corrupted, $\mathcal{F}_{\text{sig}}$ prevents signature forgery, and $\mathcal{F}_{\text{ledger}}$ provides a secure L1 ledger, no adversary can cause $\mathcal{F}^{\text{Liquid}}_{\text{join}}$ to accept a mismatched state. This contradicts the assumption.
\end{proof}

\subsubsection{Update Functionality}

The update subroutine $\mathcal{F}^{\text{Liquid}}_{\text{update}}$ verifies state updates in the form of newly committed blocks. It checks that \textit{(i)}~the previous block reference is consistent with the newly proposed block, \textit{(ii)}~executed requests do not conflict with existing state, and \textit{(iii)}~the new state is correctly computed.

\vspace{1em}\begin{functionality}{Description of subroutine $\mathcal{F}^{\text{Liquid}}_{\text{update}} = (\text{update})$}{

\textbf{Participating roles:} \{update\}

\noindent\textbf{Corruption model:} incorruptible

\noindent \textbf{Protocol parameters:}
\begin{itemize}
    \item $n = 3f+1$: operator committee size
    \item $f$: corruption threshold
\end{itemize}

}\end{functionality}\vspace{.5em}

\begin{functionality}{Description of $\mathcal{M}^{\text{Liquid}}_{\text{update}}$}{

\textbf{Implemented role(s):} \{update\}

\noindent \textbf{CheckID}(\emph{pid}, \emph{sid}, \emph{role}): Accept all messages with the same \emph{sid}.

\vspace{0.5em}
\noindent \textbf{Main:}

\vspace{0.5em}

\textbf{recv} \{Update, $\mathit{Attachment}=\{\mathit{Block}=\{\mathit{prevBlock}, \mathit{newStateList}, \mathit{executedReq}\}\}$, $\mathsf{internalState}$\} \textbf{from} I/O:

\begin{enumerate}[itemsep=0.5em]

\item Check that $\mathit{prevBlock}$ matches the latest block in $\mathsf{executedRequest}$; \cmt{Block chain continuity}

\item Check that $\forall\, \mathit{TX}_j = (\mathit{sender}_j, \mathit{receiver}_j, \mathit{value}_j, \mathit{data}_j) \in \mathit{executedReq}$:
\begin{itemize}
    \item $\nexists\, \mathit{TX}' \in \mathsf{executedRequest}$ that conflicts with $\mathit{TX}_j$; \cmt{No conflict with already-executed requests}
    \item Let $\mathit{bal}(\mathit{sender}_j)$ denote the balance of $\mathit{sender}_j$ derived from $\mathsf{stateList}$. Then $\mathit{value}_j \leq \mathit{bal}(\mathit{sender}_j)$; \cmt{Sender has sufficient balance}
\end{itemize}

\item Check that $\forall\, \mathit{TX}_j \in \mathit{executedReq}$ with $(\mathit{sender}_j, \scur, \text{client}) \notin \mathsf{CorruptionSet}$: $\exists\, \mathit{TX}' \in \mathsf{requestQueue}$ \textbf{s.t.} $\mathit{TX}' = \mathit{TX}_j$; \cmt{Each honest sender's transaction is recorded in the request queue}

\item Check that $\mathit{newStateList}$ is the correct output of sequentially executing all $\mathit{TX} \in \mathit{executedReq}$ starting from $\mathsf{stateList}$; \cmt{State transition is correctly computed}

\item \textbf{if} all checks pass: \textbf{reply} \{Update, true, $\mathit{newStateList}$, $\mathit{executedReq}$\};

\end{enumerate}

}\end{functionality}\vspace{1em}

\begin{lemma}
    The subroutines $\mathcal{F}^{\text{Liquid}}_{\text{update}}$ and $\mathcal{F}^{\text{Liquid}}_{\text{read}}$ jointly guarantee $f_{L_2}$-safety.
\end{lemma}

\begin{proof}
    Suppose \emph{safety} is violated. By definition, this means that either self-consistency or view-consistency fails. Since $\mathcal{F}^{\text{Liquid}}_{\text{read}}$ always selects the latest state recorded in $\mathsf{stateList}$ as the read result, and each new conflicting state will not be recorded. Thus, view-consistency holds for different honest clients. Self-consistency follows from the synchronous communication assumption and the guarantee that at least one honest operator responds with the latest sidechain state. Hence all read results are prefix-comparable, contradicting the assumption.
\end{proof}

\subsubsection{Read Functionality}

Under synchronous communication and the assumption that a client is connected to at least one honest operator, the read subroutine $\mathcal{F}^{\text{Liquid}}_{\text{read}}$ returns the sidechain blocks including transaction data and states, according to the influence of the adversary.

\vspace{1em}\begin{functionality}{Description of subroutine $\mathcal{F}^{\text{Liquid}}_{\text{read}} = (\text{read})$}{

\textbf{Participating roles:} \{read\} \\ \textbf{Corruption model:} incorruptible

\noindent \textbf{Protocol parameters:} $n = 3f+1$

}\end{functionality}\vspace{.5em}

\begin{functionality}{Description of $\mathcal{M}^{\text{Liquid}}_{\text{read}}$}{

\textbf{Implemented role(s):} \{read\}

\noindent \textbf{CheckID}(\emph{pid}, \emph{sid}, \emph{role}): Accept all messages with the same \emph{sid}.

\vspace{0.5em}
\noindent \textbf{Main:}

\vspace{0.5em}

\textbf{recv} \{Read, $\mathsf{internalState}$\} \textbf{from} I/O:

\begin{enumerate}[itemsep=0.5em]

\item Let $i_{\mathit{ptr}} \leftarrow \mathsf{lastReadPointer}$; let $i_{\max}$ be the index of the latest executed block in $\mathsf{executedRequest}$;

\item \textbf{send responsively} \{ReadDelivery, $i_{\mathit{ptr}}$, $i_{\max}$\} \textbf{to} $\mathcal{S}$ via NET;
\textbf{wait for} \{ReadDelivery, $b$\} s.t. $b \in \{\text{true}, \text{false}\}$; \cmt{Query $\mathcal{S}$ for delivery decision}

\item \textbf{if} $b = \text{true}$: \cmt{Adversary delivers up to the latest executed request}
\begin{itemize}
    \item Let $\mathit{latestState}$ be the on-chain state corresponding to index $i_{\max}$ from $\mathsf{stateList}$, i.e., the state after applying all blocks up to $i_{\max}$; \cmt{Latest state reflected on L1}
    \item Let $\mathit{transitionData} \leftarrow \{\mathit{Block}_k \in \mathsf{executedRequest} \mid i_{\mathit{ptr}} < k \leq i_{\max}\}$; \cmt{All blocks from last read pointer to the latest executed block}
    \item $\mathsf{lastReadPointer} \leftarrow i_{\max}$; \cmt{Advance read pointer to the latest delivered block}
\end{itemize}

\item \textbf{if} $b = \text{false}$: \cmt{Adversary withholds new blocks; return the state at the current read pointer}
\begin{itemize}
    \item Let $\mathit{latestState}$ be the state at index $i_{\mathit{ptr}}$ from $\mathsf{stateList}$; \cmt{State at the last delivered block}
    \item Let $\mathit{transitionData} \leftarrow \emptyset$; \cmt{No new blocks delivered}
\end{itemize}

\item \textbf{if} $\mathit{latestState} \neq \bot$:
\textbf{reply} \{Read, $\mathit{ReadResult} = \{\mathit{latestState}, \mathit{transitionData},\\ \mathsf{onchainState}\}$\};

\item \textbf{else}: \textbf{reply} \{Read, $\bot$\};

\end{enumerate}

}\end{functionality}\vspace{1em}

\subsubsection{Settlement Functionality}

The settlement subroutine $\mathcal{F}^{\text{Liquid}}_{\text{settlement}}$ verifies that \textit{(i)}~a peg-out transaction request exists in\\ $\mathsf{executedRequest}$, \textit{(ii)}~the settlement transaciont is committed on L1, \textit{(iii)}~the peg-out transaction and settlement transaction is consistent with the latest state, and 

\vspace{1em}\begin{functionality}{Description of subroutine $\mathcal{F}^{\text{Liquid}}_{\text{settlement}} = (\text{settlement})$}{

\textbf{Participating roles:} \{settlement\}

\noindent\textbf{Corruption model:} incorruptible

}\end{functionality}\vspace{.5em}

\begin{functionality}{Description of $\mathcal{M}^{\text{Liquid}}_{\text{settlement}}$}{

\textbf{Implemented role(s):} \{settlement\}

\noindent \textbf{CheckID}(\emph{pid}, \emph{sid}, \emph{role}): Accept all messages with the same \emph{sid}.

\vspace{0.5em}
\noindent \textbf{Main:}

\vspace{0.5em}

\textbf{recv} \{Settlement, $\mathit{Attachment} = \{\mathit{TX}_{\mathit{peg\text{-}out}}, \mathit{TX}_{\mathit{settle}}\}$, $\mathsf{internalState}$\} \textbf{from} I/O:

\begin{enumerate}[itemsep=0.5em]

\item Parse $\mathit{TX}_{\mathit{peg\text{-}out}} = (\mathit{pid}_{\mathit{client}}, \epsilon, s_{\mathit{peg\text{-}out}}, \epsilon)$ and $\mathit{TX}_{\mathit{settle}} = (\mathit{pid}_{\mathit{op}}, Address, s_{\mathit{settle}}, \epsilon)$;

\item Let $s^{*}$ be the latest state in $\mathsf{stateList}$. Check the off-chain peg-out conditions:
\begin{itemize}
    \item $\exists\, \mathit{req} \in \mathsf{executedRequest}$ \textbf{s.t.} $\mathit{TX}_{\mathit{peg\text{-}out}} \in \mathit{req}$; \cmt{Peg-out request has been executed and recorded by operators in the sidechain history}
    \item $\mathit{pid}_{\mathit{client}} \in \mathsf{identities}$; \cmt{Settling client is a registered participant}
    \item $s_{\mathit{peg\text{-}out}} = s^{*}$; \cmt{Peg-out carries the latest valid L2 state}
\end{itemize}

\item \textbf{send} \{ReadL1\} \textbf{to} $(\pcur, \scur, \mathcal{F}_{\text{ledger}}: \text{client}_{\text{L1}})$;
\textbf{wait for} \{ReadL1, $\mathit{output} = \{\{\mathit{TX}\}_{\text{L1}}, \mathit{State}_{\text{L1}}, \{\mathit{pid}\}\}$\};

\item Check the on-chain settlement conditions on $\mathit{output}$:
\begin{itemize}
    \item $\mathit{TX}_{\mathit{settle}} \in \{\mathit{TX}\}_{\text{L1}}$; \cmt{Settlement transaction is committed on L1}
    \item $\mathit{pid}_{\mathit{op}} \in \mathsf{identities}$ and $\mathit{pid}_{\mathit{op}}$ is an operator; \cmt{Settlement transaction is published by an operator}
    \item $s_{\mathit{settle}} = s^{*}$; \cmt{Settlement transaction carries the latest L2 state}
    \item $s^{*} \in \mathit{State}_{\text{L1}}$; \cmt{Latest L2 state is recorded on L1}
\end{itemize}

\item \textbf{if} all checks pass: \textbf{reply} \{Settlement, true, $s^{*}$\};

\end{enumerate}

}\end{functionality}\vspace{1em}

\begin{lemma}
\label{lem:SettleL}
    The subroutine $\mathcal{F}^{\text{Liquid}}_{\text{settlement}}$ guarantees \emph{correct L2 settlement}.
\end{lemma}

\begin{proof}
    According to the definition, \emph{correct L2 settlement} is violated if either (i)~an honest client receives a successful settlement output even though the output state is neither the latest state nor a state agreed upon by honest clients, or (ii)~the output state is not committed on the L1 blockchain. However, $\mathcal{F}^{\text{Liquid}}_{\text{settlement}}$ outputs success only after verifying both off-chain and on-chain conditions: off-chain, the peg-out is recorded in $\mathsf{executedRequest}$ by the operator quorum, the settling client is registered, and the peg-out carries the latest state $s^{*}$ from $\mathsf{stateList}$ (Step~2); on-chain, the operator-published settlement transaction $\mathit{TX}_{\mathit{settle}}$ is committed on L1, is issued by a registered operator, carries the same $s^{*}$, and $s^{*}$ is reflected in $\mathit{State}_{\text{L1}}$ (Step~4). Therefore correct L2 settlement holds here.
\end{proof}

\subsubsection{Update Round Functionality}

Since the Liquid sidechain assumes synchronous communication with a delay bound $\delta$, the round-update subroutine further enforces that no honest client's request remains uncommitted beyond $T_{L_2}=(3f+4)\delta$ rounds, which aims to be realized by the 3-phase BFT protocol under synchronous communication.

\vspace{1em}\begin{functionality}{Description of subroutine $\mathcal{F}^{\text{Liquid}}_{\text{updRnd}} = (\text{updRnd})$}{

\textbf{Participating roles:} \{updRnd\}

\noindent\textbf{Corruption model:} incorruptible

\noindent \textbf{Protocol parameters:}
\begin{itemize}
    \item $T_{L_2} = (3f+4)\delta$: off-chain commitment latency goal
\end{itemize}

}\end{functionality}\vspace{.5em}

\begin{functionality}{Description of $\mathcal{M}^{\text{Liquid}}_{\text{updRnd}}$}{

\textbf{Implemented role(s):} \{updRnd\}

\noindent \textbf{CheckID}(\emph{pid}, \emph{sid}, \emph{role}): Accept all messages with the same \emph{sid}.

\vspace{0.5em}
\noindent \textbf{Main:}

\vspace{0.5em}

\textbf{recv} \{UpdateRound, $\mathsf{internalState}$\} \textbf{from} I/O:

\begin{enumerate}[itemsep=0.5em]

\item Let $\mathit{round}_{\mathit{cur}}$ denote the current round from $\mathsf{internalState}$;

\item Check that $\forall\, \mathit{req} \in \mathsf{requestQueue}$ \textbf{s.t.} the sender of $\mathit{req}$ satisfies $(\mathit{sender}, \scur, \text{client})\\ \notin \mathsf{CorruptionSet}$:
\begin{itemize}
    \item Let $\mathit{round}_{\mathit{sub}}$ denote the round at which $\mathit{req}$ was added to $\mathsf{requestQueue}$;
    \item $\mathit{round}_{\mathit{cur}} - \mathit{round}_{\mathit{sub}} \leq T_{L_2}$; \cmt{No honest client's request has been pending beyond the delay, which is aimed to be realized by 3-round BFT.}
\end{itemize}

\item \textbf{if} check passes: \textbf{reply} \{UpdateRound, true\};
\item \textbf{else}: \textbf{reply} \{UpdateRound, false\};

\end{enumerate}

}\end{functionality}\vspace{1em}

\begin{lemma}
\label{lem:LiveL}
The ideal functionality 
$\mathcal{F}^{\text{Liquid}}_{\text{layer2}}$ guarantees 
the following liveness properties:
\begin{enumerate}
    \item \textbf{(Sidechain joining.)}
    $(f_{L_2} + f_{L_1},\, T_{L_2} + T_{L_1})$-liveness: under $f_{L_2}$ and $f_{L_1}$ corruption, an accepted Join request results in output $\{\text{Join}, s_{\mathit{init}}\}$ within $T_{L_2} + T_{L_1}$. 

    \item \textbf{(State update.)}
    $(f_{L_2},\, T_{L_2})$-liveness: under $f_{L_2}$ corruption, an accepted Update request results in the corresponding entry being added to $\mathsf{executedRequest}$ and changes in $\mathsf{stateList}$ within $T_{L_2}$. 

    \item \textbf{(Settlement.)}
    $(f_{L_2} + f_{L_1},\, T_{L_2} + T_{L_1})$-liveness: under $f_{L_2}$ and $f_{L_1}$ corruption, an accepted Settlement request results in $\mathsf{onchainState}$ reflecting $s_{\mathit{settle}}$ and $\{\text{Settlement}, s_{\mathit{settle}}\}$ output within $T_{L_2} + T_{L_1}$. 
\end{enumerate}
\end{lemma}

\begin{proof}
We prove each property in turn.

\medskip
\noindent\textbf{(1) Sidechain joining.}
Suppose liveness for join requests is violated. A join request involves two phases: (i)~on-chain deposit on L1, requiring $T_{L_1}$ for inclusion, and (ii)~sidechain-side processing by operators, requiring at most $T_{L_2}$ under our trust assumption. $\mathcal{F}^{\text{Liquid}}_{\text{join}}$ outputs success only if the peg-in is consistent with $s_{\mathit{init}}$, the sidechain state reflects the peg-in (recorded in $\mathsf{stateList}$ and $\mathsf{executedRequest}$), and the deposit is committed on L1. Under at most $f_{L_2} + f_{L_1}$ corruption, the deposit could be observed on-chain and the corresponding peg-in is included in a sidechain block within $T_{L_2}$. The L1 liveness of $\mathcal{F}_{\text{ledger}}$ guarantees deposit inclusion within $T_{L_1}$. $\mathcal{F}^{\text{Liquid}}_{\text{updRnd}}$ ensures the off-chain execution phase does not exceed $T_{L_2}$. The total end-to-end bound is $T_{L_2} + T_{L_1}$, contradicting the assumption.

\medskip
\noindent\textbf{(2) State update.}
Suppose liveness for state updates is violated, meaning $\mathcal{F}^{\text{Liquid}}_{\text{update}}$ never produces a positive output under at most $f_{L_2}$ corrupted operators. Since $\mathcal{F}^{\text{Liquid}}_{\text{updRnd}}$ prevents time from advancing past the $T_{L_2}$-round deadline before the request is included in $\mathsf{executedRequest}$ within $T_{L_2}$. For output liveness, $\mathcal{F}^{\text{Liquid}}_{\text{read}}$ returns the latest state from $\mathsf{stateList}$; under our assumption, the read result reflects the updated state within $T_{L_2}$. This contradicts the assumption.

\medskip
\noindent\textbf{(3) Settlement.}
Suppose liveness for settlement requests cannot be satisfied. $\mathcal{F}^{\text{Liquid}}_{\text{settlement}}$ outputs success only after verifying that the peg-out is recorded in $\mathsf{executedRequest}$, the settlement state matches the latest state in $\mathsf{stateList}$, and the published settlement transaction is committed on L1 carrying that latest state. The L2-side condition (peg-out in $\mathsf{executedRequest}$ with consistent state) is captured by $\mathcal{F}^{\text{Liquid}}_{\text{update}}$, which under $f_{L_2}$ corruption admits the request within $T_{L_2}$ as enforced by $\mathcal{F}^{\text{Liquid}}_{\text{updRnd}}$. The L1-side condition is captured by $\mathcal{F}_{\text{ledger}}$, which under $f_{L_1}$ corruption commits the settlement transaction within $T_{L_1}$ via its own liveness guarantee. The simulator therefore sends the corresponding message to $\mathcal{F}^{\text{Liquid}}_{\text{settlement}}$ within $T_{L_2} + T_{L_1}$, at which point all checks pass. The total bound is $T_{L_2} + T_{L_1}$, contradicting the assumption.

\end{proof}


\subsection{Security Proof}
\label{apd:LiquidProof}

After proposing the ideal functionality and real-world implementation, we now show the security of the Liquid sidechain protocol. To start with we first show the ideal functionality captures all the security properties:

\ThmidealLiquid*

\begin{proof}
    According to Lemma~\ref{lem:OpenL}--\ref{lem:LiveL}, the ideal functionality $\mathcal{F}^{\text{Liquid}}_{\text{layer2}}$ guarantees all security properties.
\end{proof}

After defining the ideal functionality $\mathcal{F}^{\text{Liquid}}_{\text{layer2}}$, we prove that the real Liquid protocol iUC-realizes it. The proof is done in 7 steps of successive game replacement. We first define a simulator $\mathcal{S}_{\text{Liquid}}$ that internally simulates a full run of $\mathcal{P}^{\text{Liquid}}$, and a dummy functionality $\mathcal{F}^{\text{Liquid}}_{\text{dummy}}$ that relays messages between $\mathcal{E}$ and $\mathcal{S}_{\text{Liquid}}$. This base ideal execution yields the same distribution of messages to $\mathcal{E}$ as the real execution. We use the execution ensemble $\mathsf{EXEC}$ to denote the messages observed by $\mathcal{E}$, including the output for the input request and the leakage to adversary, when interacting with adversary $\mathcal{A}$, real protocol $\mathcal{P}$, ideal functionality $\mathcal{F}$ and simulator $\mathcal{S}$ in the proofs that follow.

In each subsequent step, we incrementally add interaction between the simulator and the ideal functionality and extend the functionality, thereby forming the corresponding subroutines, while keeping the changes transparent to both $\mathcal{E}$ and $\mathcal{A}$. We continue until we obtain the target functionality $\mathcal{F}^{\text{Liquid}}_{\text{layer2}}$ defined by our framework. At every step, the simulator is adjusted so that the new ideal execution is indistinguishable from the previous one. For each transition, we discuss the differences relative to the prior step and prove that, given the same inputs from $\mathcal{E}$ and $\mathcal{A}$, the resulting outputs remain the same up to computational indistinguishability under any adversarial influence strategy.

We begin by defining the dummy ideal functionality
$\mathcal{F}^{\text{Liquid}}_{\text{dummy}}=(\mathcal{F}_{\text{client-dummy}}, \mathcal{F}_{\text{ledger}} \mid \perp)$
and the simulator $\mathcal{S}_{\text{Liquid}}$ as follows. The dummy functionality forwards every request from $\mathcal{E}$ to the simulator and returns the simulator's response unchanged. The simulator $\mathcal{S}_{\text{Liquid}}$ runs $\mathcal{P'}^{\text{Liquid}}$ internally and produces identical outputs.

\vspace{1em}\begin{functionality}{Description of $\mathcal{M}_{\text{client-dummy}}$ of $\mathcal{F}^{\text{Liquid}}_{\text{dummy}}$}{

\textbf{Implemented role(s):} \{client-dummy\}

\noindent \textbf{Main:}

\textbf{recv} any request \textbf{from} I/O:
\begin{enumerate}[itemsep=0.5em]
    \item Forward request to $\mathcal{S}$ through NET;
\end{enumerate}

\hrule\vspace{0.5em}

\textbf{recv} any message \textbf{from} NET:
\begin{enumerate}
    \item Output the message to $\mathcal{E}$ through I/O;
\end{enumerate}

}\end{functionality}\vspace{1em}

\vspace{1em}\begin{functionality}{Description of simulator $\mathcal{S}_{\text{Liquid}}$}{

$\mathcal{S}_{\text{Liquid}}$ internally simulates $\mathcal{P'}^{\text{Liquid}}$, a copy of the real protocol $\mathcal{P}^{\text{Liquid}}$ as defined in Section~\ref{apd:LiquidReal}.

\vspace{0.5em}
\textbf{Real protocol simulation:}
\begin{itemize}[itemsep=0.3em]
    \item $\mathcal{S}_{\text{Liquid}}$ simulates honest clients' and operators' actions inside $\mathcal{P'}^{\text{Liquid}}$ according to the real protocol, including the three-round BFT consensus.
    \item If participants are corrupted, $\mathcal{S}_{\text{Liquid}}$ leaks the corresponding messages sent to corrupted entities to the adversary $\mathcal{A}$ and continues simulating honest parties based on $\mathcal{A}$'s instructions.
\end{itemize}

\textbf{Network communication from/to the environment:}
\begin{itemize}[itemsep=0.3em]
    \item Messages that $\mathcal{S}_{\text{Liquid}}$ receives on the network interface (from $\mathcal{E}$/$\mathcal{A}$) are forwarded to $\mathcal{P'}^{\text{Liquid}}$.
    \item Messages sent by $\mathcal{P'}^{\text{Liquid}}$ on its network interface (to $\mathcal{E}$/$\mathcal{A}$) are forwarded to the environment.
\end{itemize}

\textbf{Input requests and outputs:}
\begin{itemize}[itemsep=0.3em]
    \item Unlike $\mathcal{P}^{\text{Liquid}}$, which receives inputs directly from $\mathcal{E}$, the simulation $\mathcal{P'}^{\text{Liquid}}$ receives requests forwarded from $\mathcal{F}^{\text{Liquid}}_{\text{layer2}}$. Instead of sending outputs directly to $\mathcal{E}$, $\mathcal{S}_{\text{Liquid}}$ sends them to $\mathcal{F}^{\text{Liquid}}_{\text{layer2}}$.
\end{itemize}

\textbf{Message delivery:}
\begin{itemize}[itemsep=0.3em]
    \item The Liquid sidechain assumes synchronous communication via $\mathcal{F}^{\text{Liquid}}_{\text{com}}$ with delay bound $\delta$. The simulator bookkeeps all messages in $\mathcal{P'}^{\text{Liquid}}$ and triggers delivery according to the adversary's scheduling decisions, mirroring the real-world protocol.
\end{itemize}

\textbf{Corruption handling:}
\begin{itemize}[itemsep=0.3em]
    \item $\mathcal{S}_{\text{Liquid}}$ keeps the corruption status of entities in $\mathcal{P}^{\text{Liquid}}$, $\mathcal{P'}^{\text{Liquid}}$ and $\mathcal{F}^{\text{Liquid}}_{\text{layer2}}$ synchronized. When an entity in $\mathcal{P'}^{\text{Liquid}}$ becomes corrupted, $\mathcal{S}_{\text{Liquid}}$ corrupts the corresponding entity in $\mathcal{F}^{\text{Liquid}}_{\text{layer2}}$ before continuing.
    \item Adversarial commands for corrupted participants (e.g., publishing on $\mathcal{F}_{\text{ledger}}$, sending via $\mathcal{F}^{\text{Liquid}}_{\text{com}}$) are forwarded to $\mathcal{P'}^{\text{Liquid}}$.
    \item When a corrupted participant in $\mathcal{P'}^{\text{Liquid}}$ wants to output to $\mathcal{E}$, $\mathcal{S}_{\text{Liquid}}$ instructs the corresponding entity in $\mathcal{F}^{\text{Liquid}}_{\text{layer2}}$ to output.
\end{itemize}
}\end{functionality}\vspace{1em}

\begin{lemma}
\label{lem:Liquid1}
    For all PPT adversaries $\mathcal{A}$, there exists a PPT simulator $\mathcal{S}_{\text{Liquid}}$ such that for all PPT environments $\mathcal{E}$ and all security parameters $k\in\mathbb{N}$,
    $\mathsf{EXEC}^{\mathcal{P}^{\text{Liquid}}}_{\mathcal{A},\mathcal{E}}(k)\ \stackrel{c}{\approx}\
    \mathsf{EXEC}^{\mathcal{F}^{\text{Liquid}}_{\text{dummy}}}_{\mathcal{S}_{\text{Liquid}},\mathcal{E}}(k)$,
    where $\stackrel{c}{\approx}$ denotes computational indistinguishability.
\end{lemma}

\begin{proof}
    Fix an arbitrary PPT environment $\mathcal{E}$ and adversary $\mathcal{A}$. We argue that the execution ensembles in the real and ideal worlds are computationally indistinguishable by analyzing the three components observable by $\mathcal{E}$.

    \medskip
    \noindent\textbf{Observable components.} In the real world, the execution ensemble $\mathsf{EXEC}^{\mathcal{P}^{\text{Liquid}}}_{\mathcal{A},\mathcal{E}}(k)$ consists of:
    \begin{enumerate}
        \item \emph{I/O outputs} delivered to $\mathcal{E}$ by $\mathcal{P}^{\text{Liquid}}_{\text{client}}$:
        $\{\text{Join}, s_{\mathit{init}}\}$,
        $\{\text{Settlement}, s_{\mathit{settle}}\}$,
        $\{\text{Read}, \mathit{ReadResult}\}$,
        $\{\text{GetCurRound}, \mathit{round}\}$.
        \item \emph{On-chain transactions} committed on $\mathcal{F}_{\text{ledger}}$ during execution:
        $\mathit{TX}_{\mathit{deposit}}$ and $\mathit{TX}_{\mathit{peg\text{-}settlement}}$.
        \item \emph{Adversarial leakage}, comprising messages received by corrupted parties (clients and operators) during the protocol execution, the corrupted parties' internal state, and the transactions they publish on L1.
    \end{enumerate}

    In the ideal world, the simulator $\mathcal{S}_{\text{Liquid}}$ internally runs $\mathcal{P'}^{\text{Liquid}}$ and interacts with the dummy functionality $\mathcal{F}^{\text{Liquid}}_{\text{dummy}}$, which by definition forwards every request from $\mathcal{E}$ to $\mathcal{S}_{\text{Liquid}}$ unchanged and relays the simulator's responses back to $\mathcal{E}$. We show that each component is computationally indistinguishable across the two worlds.

    \medskip
    \noindent\textbf{(1) I/O outputs.} Since $\mathcal{F}^{\text{Liquid}}_{\text{dummy}}$ acts as a transparent relay, $\mathcal{S}_{\text{Liquid}}$ receives exactly the same sequence of requests as $\mathcal{P}^{\text{Liquid}}_{\text{client}}$ would in the real world. By construction, $\mathcal{S}_{\text{Liquid}}$ executes the same client and operator logic inside $\mathcal{P'}^{\text{Liquid}}$ under the same adversarial scheduling, producing the same I/O outputs. The only potential difference arises from the randomness of $\mathcal{F}_{\text{sig}}$: signature strings carried by client transactions and operator-signed peg-out transactions may differ between the two worlds because fresh randomness is sampled independently. However, since $\mathcal{F}_{\text{sig}}$ realizes EUF-CMA security, signatures generated on the same messages are computationally indistinguishable. Hence the I/O outputs are computationally indistinguishable.

    \medskip
    \noindent\textbf{(2) On-chain transactions.} Since $\mathcal{F}^{\text{Liquid}}_{\text{dummy}}$ forwards all requests to $\mathcal{S}_{\text{Liquid}}$, the simulated protocol $\mathcal{P'}^{\text{Liquid}}$ generates and publishes the same set of transactions to $\mathcal{F}_{\text{ledger}}$ as $\mathcal{P}^{\text{Liquid}}$ would in the real world, including client deposit transactions and operator-signed peg-out settlement transactions. Transactions may contain different signature values due to independent randomness in $\mathcal{F}_{\text{sig}}$, but by the EUF-CMA security of the signature scheme, the transaction distributions are computationally indistinguishable.

    \medskip
    \noindent\textbf{(3) Adversarial leakage.} By the definition of $\mathcal{S}_{\text{Liquid}}$, the corruption status of all entities (clients and operators) is kept synchronized between $\mathcal{P'}^{\text{Liquid}}$ and $\mathcal{F}^{\text{Liquid}}_{\text{dummy}}$. Since $\mathcal{F}^{\text{Liquid}}_{\text{dummy}}$ forwards all $\mathcal{E}$-requests to the simulator, $\mathcal{S}_{\text{Liquid}}$ can reconstruct the same internal state and message history for corrupted parties as in the real execution. Consequently, the leakage delivered to $\mathcal{A}$ is identical up to signature randomness, which is again computationally indistinguishable by EUF-CMA security.

    \medskip
    Conclusively, we have:
    $\mathsf{EXEC}^{\mathcal{P}^{\text{Liquid}}}_{\mathcal{A}, \mathcal{E}}(k)\stackrel{c}{\approx}\mathsf{EXEC}^{\mathcal{F}^{\text{Liquid}}_{\text{dummy}}}_{\mathcal{S}_{\text{Liquid}}, \mathcal{E}}(k)$.
\end{proof}

Next, we extend the dummy functionality with the submission subroutine, yielding
$\mathcal{F}^{\text{Liquid}}_{\text{layer2-submit}} = (\mathcal{F}_{\text{client-submit}}, \mathcal{F}_{\text{ledger}} \mid \mathcal{F}^{\text{Liquid}}_{\text{submit}})$.
The subroutine $\mathcal{F}^{\text{Liquid}}_{\text{submit}}$ is defined in Appendix~\ref{apd:Liquidsubmit}. The intermediate client functionality adds the Submit request routing through $\mathcal{F}^{\text{Liquid}}_{\text{submit}}$; all other requests are forwarded to $\mathcal{S}$ unchanged.

\vspace{1em}\begin{functionality}{Description of $\mathcal{M}_{\text{client-submit}}$ of $\mathcal{F}^{\text{Liquid}}_{\text{layer2-submit}}$}{

\textbf{Implemented role(s):} \{client-submit\}

\noindent\textbf{Subroutines:}
$\mathcal{F}^{\text{Liquid}}_{\text{submit}}$: submit,
$\mathcal{F}_{\text{ledger}}$: client\textsubscript{L1}

\noindent \textbf{Internal state:} (same as $\mathcal{F}_{\text{client}}$)

\vspace{0.5em}
\noindent \textbf{Main:}

\vspace{0.5em}

\textbf{recv} \{Submit, $\mathit{request}$\} \textbf{from} I/O:
\begin{enumerate}[itemsep=0.5em]
\item \textbf{send} \{Submit, $\mathit{request}$, $\mathsf{internalState}$\} \textbf{to} $(\pcur, \scur, \mathcal{F}^{\text{Liquid}}_{\text{submit}}:\text{submit})$,
\textbf{wait for} \{Submit, $\mathit{response}$\} s.t. $\mathit{response} \in \{\text{true}, \text{false}\}$;
\item \textbf{if} $\mathit{response} =$ true: $\mathsf{requestQueue}.\text{add}(\mathit{request})$;
\textbf{send} $\mathit{request}$ \textbf{to} $\mathcal{S}$ via NET;
\end{enumerate}

\hrule\vspace{0.5em}

\textbf{recv} \{ReadL1\} \textbf{from} I/O or NET:
\begin{enumerate}[itemsep=0.5em]
    \item \textbf{send} \{Read\} \textbf{to} $(\pcur, \scur, \mathcal{F}_{\text{ledger}}: \text{client}_{\text{L1}})$;
    \textbf{wait for} $\mathit{L1ReadResult}$;
    \item \textbf{reply} \{ReadL1, $\mathit{L1ReadResult}$\} via I/O;
\end{enumerate}

\hrule\vspace{0.5em}

\textbf{recv} any other message \textbf{from} I/O:
\begin{enumerate}
    \item Forward to $\mathcal{S}$ through NET;
\end{enumerate}

\hrule\vspace{0.5em}

\textbf{recv} any message \textbf{from} NET:
\begin{enumerate}
    \item Output the message to $\mathcal{E}$ through I/O;
\end{enumerate}

}\end{functionality}\vspace{.5em}

\begin{functionality}{Description of simulator $\mathcal{S}_{\text{Liquid-submit}}$}{
The simulator $\mathcal{S}_{\text{Liquid-submit}}$ behaves the same as $\mathcal{S}_{\text{Liquid}}$.
}\end{functionality}\vspace{1em}

\begin{lemma}
\label{lem:Liquid2}
    For all PPT adversaries $\mathcal{A}$, there exists a PPT simulator $\mathcal{S}_{\text{Liquid-submit}}$ such that for all PPT environments $\mathcal{E}$ and all security parameters $k\in\mathbb{N}$,
    \[
    \mathsf{EXEC}^{\mathcal{F}^{\text{Liquid}}_{\text{dummy}}}_{\mathcal{S}_{\text{Liquid}},\mathcal{E}}(k)
    \ \stackrel{c}{\approx}\
    \mathsf{EXEC}^{\mathcal{F}^{\text{Liquid}}_{\text{layer2-submit}}}_{\mathcal{S}_{\text{Liquid-submit}},\mathcal{E}}(k),
    \]
    where $\stackrel{c}{\approx}$ denotes computational indistinguishability.
\end{lemma}

\begin{proof}
    Fix an arbitrary PPT environment $\mathcal{E}$. The simulator logic is identical in both executions; the only difference is the ideal functionality through which requests are forwarded to the simulator. We analyze the three observable components.

    \medskip
     In $\mathcal{F}^{\text{Liquid}}_{\text{dummy}}$, every request from $\mathcal{E}$ is forwarded directly to $\mathcal{S}_{\text{Liquid}}$. In $\mathcal{F}^{\text{Liquid}}_{\text{layer2-submit}}$, the subroutine $\mathcal{F}^{\text{Liquid}}_{\text{submit}}$ intercepts each request and checks semantic validity before forwarding. Requests that fail these checks are rejected and never reach the simulator.

    \medskip
    \noindent\textbf{(1) I/O outputs.} In the real protocol $\mathcal{P}^{\text{Liquid}}$ (and hence in $\mathcal{P'}^{\text{Liquid}}$), the client machine already enforces the same validity checks for all request types (Join, Submit, Settlement): only predefined request types are processed, and malformed or out-of-phase requests produce no output. Therefore, any request rejected by $\mathcal{F}^{\text{Liquid}}_{\text{submit}}$ would also produce no output inside the simulation. For accepted requests, both simulators execute the same client logic and produce the same I/O outputs. Hence, the I/O outputs are identical.

    \medskip
    \noindent\textbf{(2) On-chain transactions.} Since only accepted requests trigger protocol simulation of $\mathcal{P'}^{\text{Liquid}}$ inside the simulator, and the set of accepted requests is the same in both executions, the transactions published to $\mathcal{F}_{\text{ledger}}$ are identical except for signature values. By EUF-CMA security, the on-chain transactions are computationally indistinguishable.

    \medskip
    \noindent\textbf{(3) Adversarial leakage.} Since invalid requests are rejected by the ideal functionality and ignored by honest parties during simulation, no leakage is generated toward the adversary in either world. In that case, the only leakage is generated from the simulated protocol execution inside the simulator that takes the same input requests. The message during execution could include different signatures, which, according to EUF-CMA security, are computationally indistinguishable. Moreover, the corruption status remains synchronized across both worlds. Consequently, the leakage delivered to $\mathcal{A}$ is computationally indistinguishable in both executions.

    \medskip
    Conclusively,
    $\mathsf{EXEC}^{\mathcal{F}^{\text{Liquid}}_{\text{dummy}}}_{\mathcal{S}_{\text{Liquid}}, \mathcal{E}}(k) \stackrel{c}{\approx} \mathsf{EXEC}^{\mathcal{F}^{\text{Liquid}}_{\text{layer2-submit}}}_{\mathcal{S}_{\text{Liquid-submit}}, \mathcal{E}}(k)$.
\end{proof}

Next, we extend $\mathcal{F}^{\text{Liquid}}_{\text{layer2-submit}}$ with the join subroutine, yielding $\mathcal{F}^{\text{Liquid}}_{\text{layer2-join}} = (\mathcal{F}_{\text{client-join}}, \mathcal{F}_{\text{ledger}} \mid \mathcal{F}^{\text{Liquid}}_{\text{submit}}, \mathcal{F}^{\text{Liquid}}_{\text{join}})$. The intermediate client functionality additionally routes the Join request (from NET) through $\mathcal{F}^{\text{Liquid}}_{\text{join}}$.

\vspace{1em}\begin{functionality}{Description of $\mathcal{M}_{\text{client-join}}$ of $\mathcal{F}^{\text{Liquid}}_{\text{layer2-join}}$}{

\textbf{Implemented role(s):} \{client-join\}

\noindent\textbf{Subroutines:}
$\mathcal{F}^{\text{Liquid}}_{\text{submit}}$: submit,
$\mathcal{F}^{\text{Liquid}}_{\text{join}}$: join,
$\mathcal{F}_{\text{ledger}}$: client\textsubscript{L1}

\noindent \textbf{Internal state:} (same as $\mathcal{F}_{\text{client}}$)

\vspace{0.5em}
\noindent \textbf{Main:}

\vspace{0.5em}

\textbf{recv} \{Submit, $\mathit{request}$\} \textbf{from} I/O:
\begin{enumerate}[itemsep=0.5em]
\item \textbf{send} \{Submit, $\mathit{request}$, $\mathsf{internalState}$\} \textbf{to} $(\pcur, \scur, \mathcal{F}^{\text{Liquid}}_{\text{submit}}:\text{submit})$,
\textbf{wait for} \{Submit, $\mathit{response}$\} s.t. $\mathit{response} \in \{\text{true}, \text{false}\}$;
\item \textbf{if} $\mathit{response} =$ true: $\mathsf{requestQueue}.\text{add}(\mathit{request})$;
\textbf{send} $\mathit{request}$ \textbf{to} $\mathcal{S}$ via NET;
\end{enumerate}

\hrule\vspace{0.5em}

\textbf{recv} \{Join, $\mathit{Attachment}$\} \textbf{from} NET:
\begin{enumerate}[itemsep=0.5em]
\item \textbf{send} \{Join, $\mathit{Attachment}$, $\mathsf{internalState}$\} \textbf{to} $(\pcur, \scur, \mathcal{F}^{\text{Liquid}}_{\text{join}}:\text{join})$,
\textbf{wait for} \{Join, $\mathit{response}$\} s.t. $\mathit{response} \in \{\text{true}, \text{false}\}$;
\item \textbf{if} $\mathit{response} =$ true: update $\mathsf{internalState}$ according to $\mathit{Attachment}$;
\textbf{reply} \{Join, $s_{\mathit{init}}$\} via I/O;
\end{enumerate}

\hrule\vspace{0.5em}

\textbf{recv} \{ReadL1\} \textbf{from} I/O or NET:
\begin{enumerate}[itemsep=0.5em]
    \item \textbf{send} \{Read\} \textbf{to} $(\pcur, \scur, \mathcal{F}_{\text{ledger}}: \text{client}_{\text{L1}})$;
    \textbf{wait for} $\mathit{L1ReadResult}$;
    \item \textbf{reply} \{ReadL1, $\mathit{L1ReadResult}$\} via I/O;
\end{enumerate}

\hrule\vspace{0.5em}

\textbf{recv} any other message \textbf{from} NET:
\begin{enumerate}
    \item Output the message to $\mathcal{E}$ through I/O;
\end{enumerate}

}\end{functionality}\vspace{.5em}

\begin{functionality}{Description of simulator $\mathcal{S}_{\text{Liquid-join}}$}{

The simulator $\mathcal{S}_{\text{Liquid-join}}$ behaves identically to $\mathcal{S}_{\text{Liquid-submit}}$, except for the following additional behavior upon detecting a completed sidechain joining:

\vspace{0.5em}
\textbf{Join interaction with $\mathcal{F}^{\text{Liquid}}_{\text{layer2-join}}$:}

\begin{enumerate}[itemsep=0.5em]

\item $\mathcal{S}_{\text{Liquid-join}}$ monitors the simulated protocol $\mathcal{P'}^{\text{Liquid}}$. When it detects that a client entity is about to produce the I/O output $\{\text{Join}, s_{\mathit{init}}\}$ (the peg-in transaction is committed on L1 and included in the sidechain), $\mathcal{S}_{\text{Liquid-join}}$ intercepts this output.

\item $\mathcal{S}_{\text{Liquid-join}}$ prepares $\mathit{Attachment}$ by extracting from the simulation state:
\begin{itemize}
    \item $s_{\mathit{init}}$: the initial state included in the simulated client's Join output;
    \item $\mathit{TX}_{\mathit{peg\text{-}in}}$: the peg-in transaction committed on $\mathcal{F}_{\text{ledger}}$ in $\mathcal{P'}^{\text{Liquid}}$;
\end{itemize}

\item $\mathcal{S}_{\text{Liquid-join}}$ sends $\{\text{Join}, \mathit{Attachment}\}$ to $\mathcal{F}^{\text{Liquid}}_{\text{layer2-join}}$ via NET.

\end{enumerate}

}\end{functionality}\vspace{1em}

\begin{lemma}
\label{lem:Liquid3}
    For all PPT adversaries $\mathcal{A}$, there exists a PPT simulator $\mathcal{S}_{\text{Liquid-join}}$ such that for all PPT environments $\mathcal{E}$ and all security parameters $k\in\mathbb{N}$,
    \[
    \mathsf{EXEC}^{\mathcal{F}^{\text{Liquid}}_{\text{layer2-submit}}}_{\mathcal{S}_{\text{Liquid-submit}},\mathcal{E}}(k)
    \ \stackrel{c}{\approx}\
    \mathsf{EXEC}^{\mathcal{F}^{\text{Liquid}}_{\text{layer2-join}}}_{\mathcal{S}_{\text{Liquid-join}},\mathcal{E}}(k),
    \]
    where $\stackrel{c}{\approx}$ denotes computational indistinguishability.
\end{lemma}

\begin{proof}
    Fix an arbitrary PPT environment $\mathcal{E}$. The only difference between the two executions is that $\mathcal{F}^{\text{Liquid}}_{\text{layer2-join}}$ routes the Join request (from NET) through $\mathcal{F}^{\text{Liquid}}_{\text{join}}$, whereas in $\mathcal{F}^{\text{Liquid}}_{\text{layer2-submit}}$ the join output is produced entirely by the simulator. We analyze the three observable components.

    \medskip
    \noindent\textbf{(1) I/O outputs.} The only I/O output affected is $\{\text{Join}, s_{\mathit{init}}\}$. We consider two cases.

    \emph{Case~1: Successful join.} In $\mathsf{EXEC}^{\mathcal{F}^{\text{Liquid}}_{\text{layer2-submit}}}_{\mathcal{S}_{\text{Liquid-submit}},\mathcal{E}}(k)$, the simulator produces a join output when the simulated $\mathcal{P'}^{\text{Liquid}}$ completes the joining procedure: the deposit transaction is committed on L1 and the BFT-finalized sidechain state, jointly maintained by the operator committee, reflects the peg-in. In $\mathsf{EXEC}^{\mathcal{F}^{\text{Liquid}}_{\text{layer2-join}}}_{\mathcal{S}_{\text{Liquid-join}},\mathcal{E}}(k)$, the simulator instead extracts $\mathit{Attachment} = \{s_{\mathit{init}}, \mathit{TX}_{\mathit{peg\text{-}in}}\}$ from the simulation state and sends it to $\mathcal{F}^{\text{Liquid}}_{\text{join}}$, which outputs to $\mathcal{E}$ only if all checks pass. A successful join in $\mathcal{P'}^{\text{Liquid}}$ implies that $\mathit{TX}_{\mathit{peg\text{-}in}}$ is consistent with $s_{\mathit{init}}$, the deposit is committed on L1, and $s_{\mathit{init}}$ is recorded in the BFT-finalized sidechain state. These are exactly the checks in $\mathcal{F}^{\text{Liquid}}_{\text{join}}$ (Steps~1--4). Under the BFT honest-majority assumption (at most $f_{L_2} = \tfrac{1}{3} n_{\mathit{OP}}$ corrupted operators) and the EUF-CMA security of $\mathcal{F}_{\text{sig}}$ (which prevents the adversary from forging operator signatures on a malformed peg-in), $\mathcal{F}^{\text{Liquid}}_{\text{join}}$ accepts whenever $\mathcal{P'}^{\text{Liquid}}$ completes.

    \emph{Case~2: Failed join.} If the adversary tries to make the joining procedure fail by delaying the communication network or creating incorrect peg-in transactions into the sidechain block, $\mathcal{P'}^{\text{Liquid}}$ does not complete the join due to the honest participants' checking with $\mathcal{F}_{\text{sig}}$. Correspondingly, $\mathcal{S}_{\text{Liquid-join}}$ does not send the Join request to $\mathcal{F}^{\text{Liquid}}_{\text{join}}$, so no output is produced.

    In all cases, the I/O outputs are identical.

    \medskip
    \noindent\textbf{(2) On-chain transactions.} The set of transactions published to $\mathcal{F}_{\text{ledger}}$ is fully determined by $\mathcal{P'}^{\text{Liquid}}$, which runs the real-world protocol on the input requests forwarded by the same $\mathcal{F}^{\text{Liquid}}_{\text{submit}}$ in both executions. Transactions published by honest participants, including client deposit transactions and operator-signed peg-out settlement transactions, therefore differ across the two worlds only in their signature values, which are computationally indistinguishable by the EUF-CMA security of $\mathcal{F}_{\text{sig}}$. When participants are corrupted, the transactions they publish on L1 are likewise computationally indistinguishable, since the simulator forwards all corrupted-party messages to $\mathcal{P'}^{\text{Liquid}}$ and signature randomness is the only source of variation.

    \medskip
    \noindent\textbf{(3) Adversarial leakage.} In both executions, the simulator extracts the plaintext leakage (under our plaintext-leakage assumption) from the ideal functionality and uses it to drive $\mathcal{P'}^{\text{Liquid}}$, generating leakage to the adversary according to the corruption status. Since the two executions feed $\mathcal{P'}^{\text{Liquid}}$ the same inputs (those accepted by $\mathcal{F}^{\text{Liquid}}_{\text{submit}}$) and run the same real-world protocol, $\mathsf{internalState}$ at every honest participant, including client and operator roles, evolves identically across the two executions. The interposition of $\mathcal{F}^{\text{Liquid}}_{\text{join}}$ affects only how the join I/O output is produced (via the ideal functionality's checks rather than directly by the simulator); it does not alter any network-side message exchange or any honest party's internal computation, so the leakage delivered to corrupted parties is unaffected. Consequently, the leaked $\mathsf{internalState}$ and received messages delivered to corrupted parties are computationally indistinguishable across the two executions.

    \medskip
    Conclusively,
    $\mathsf{EXEC}^{\mathcal{F}^{\text{Liquid}}_{\text{layer2-submit}}}_{\mathcal{S}_{\text{Liquid-submit}}, \mathcal{E}}(k)\stackrel{c}{\approx}\mathsf{EXEC}^{\mathcal{F}^{\text{Liquid}}_{\text{layer2-join}}}_{\mathcal{S}_{\text{Liquid-join}}, \mathcal{E}}(k)$.
\end{proof}

Next, we extend $\mathcal{F}^{\text{Liquid}}_{\text{layer2-join}}$ with the update subroutine, yielding $\mathcal{F}^{\text{Liquid}}_{\text{layer2-update}} = (\mathcal{F}_{\text{client-update}}, \mathcal{F}_{\text{ledger}} \mid \mathcal{F}^{\text{Liquid}}_{\text{submit}}, \mathcal{F}^{\text{Liquid}}_{\text{join}}, \mathcal{F}^{\text{Liquid}}_{\text{update}})$.

\vspace{1em}\begin{functionality}{Description of $\mathcal{M}_{\text{client-update}}$ of $\mathcal{F}^{\text{Liquid}}_{\text{layer2-update}}$}{

\textbf{Implemented role(s):} \{client-update\}

\noindent\textbf{Subroutines:}
$\mathcal{F}^{\text{Liquid}}_{\text{submit}}$: submit,
$\mathcal{F}^{\text{Liquid}}_{\text{join}}$: join,
$\mathcal{F}^{\text{Liquid}}_{\text{update}}$: update,
$\mathcal{F}_{\text{ledger}}$: client\textsubscript{L1}

\noindent \textbf{Internal state:} (same as $\mathcal{F}_{\text{client}}$)

\vspace{0.5em}
\noindent \textbf{Main:}

\vspace{0.5em}

\textbf{recv} \{Submit, $\mathit{request}$\} \textbf{from} I/O:
\begin{enumerate}[itemsep=0.5em]
\item \textbf{send} \{Submit, $\mathit{request}$, $\mathsf{internalState}$\} \textbf{to} $(\pcur, \scur, \mathcal{F}^{\text{Liquid}}_{\text{submit}}:\text{submit})$,
\textbf{wait for} \{Submit, $\mathit{response}$\} s.t. $\mathit{response} \in \{\text{true}, \text{false}\}$;
\item \textbf{if} $\mathit{response} =$ true: $\mathsf{requestQueue}.\text{add}(\mathit{request})$;
\textbf{send} $\mathit{request}$ \textbf{to} $\mathcal{S}$ via NET;
\end{enumerate}

\hrule\vspace{0.5em}

\textbf{recv} \{Join, $\mathit{Attachment}$\} \textbf{from} NET:
\begin{enumerate}[itemsep=0.5em]
\item \textbf{send} \{Join, $\mathit{Attachment}$, $\mathsf{internalState}$\} \textbf{to} $(\pcur, \scur, \mathcal{F}^{\text{Liquid}}_{\text{join}}:\text{join})$,
\textbf{wait for} \{Join, $\mathit{response}$\} s.t. $\mathit{response} \in \{\text{true}, \text{false}\}$;
\item \textbf{if} $\mathit{response} =$ true: update $\mathsf{internalState}$ according to $\mathit{Attachment}$;
\textbf{reply} \{Join, $s_{\mathit{init}}$\} via I/O;
\end{enumerate}

\hrule\vspace{0.5em}

\textbf{recv} \{Update, $\mathit{Attachment}$\} \textbf{from} NET:
\begin{enumerate}[itemsep=0.5em]
\item \textbf{send} \{Update, $\mathit{Attachment}$, $\mathsf{internalState}$\} \textbf{to} $(\pcur, \scur, \mathcal{F}^{\text{Liquid}}_{\text{update}}:\text{update})$,
\textbf{wait for} \{Update, $\mathit{response}$, $\mathit{newState}$, $\mathit{executedReq}$\};
\item \textbf{if} $\mathit{response} =$ true: update $\mathsf{internalState}$ with $\mathit{newState}$ and $\mathit{executedReq}$;
\end{enumerate}

\hrule\vspace{0.5em}

\textbf{recv} \{ReadL1\} \textbf{from} I/O or NET:
\begin{enumerate}[itemsep=0.5em]
    \item \textbf{send} \{Read\} \textbf{to} $(\pcur, \scur, \mathcal{F}_{\text{ledger}}: \text{client}_{\text{L1}})$;
    \textbf{wait for} $\mathit{L1ReadResult}$;
    \item \textbf{reply} \{ReadL1, $\mathit{L1ReadResult}$\} via I/O;
\end{enumerate}

\hrule\vspace{0.5em}

\textbf{recv} any other message \textbf{from} NET:
\begin{enumerate}
    \item Output the message to $\mathcal{E}$ through I/O;
\end{enumerate}

}\end{functionality}\vspace{.5em}

\begin{functionality}{Description of simulator $\mathcal{S}_{\text{Liquid-update}}$}{

The simulator $\mathcal{S}_{\text{Liquid-update}}$ behaves identically to $\mathcal{S}_{\text{Liquid-join}}$, except for the following additional behavior upon detecting a completed state update (BFT block finalization):

\vspace{0.5em}
\textbf{State update interaction with $\mathcal{F}^{\text{Liquid}}_{\text{layer2-update}}$:}

\begin{enumerate}[itemsep=0.5em]

\item $\mathcal{S}_{\text{Liquid-update}}$ monitors $\mathcal{P'}^{\text{Liquid}}$ for the completion of BFT consensus on a new block. Specifically, it detects when $2f{+}1$ valid operator signatures have been collected and the block satisfies all local validity checks.

\item Upon detecting a completed block, $\mathcal{S}_{\text{Liquid-update}}$ prepares $\mathit{Attachment}$ by extracting:
\begin{itemize}
    \item $\mathit{Block} = \{\mathit{prevBlock}, \mathit{newStateList}, \mathit{executedReq}\}$: the newly committed block;
    \item $\{\mathit{Agreement}\}$: the set of $\geq 2f{+}1$ operator approvals (abstract agreements rather than cryptographic signatures);
\end{itemize}

\item $\mathcal{S}_{\text{Liquid-update}}$ sends $\{\text{Update}, \mathit{Attachment}\}$ to $\mathcal{F}^{\text{Liquid}}_{\text{layer2-update}}$ via NET.

\end{enumerate}

}\end{functionality}\vspace{1em}

\begin{lemma}
\label{lem:Liquid4}
    For all PPT adversaries $\mathcal{A}$, there exists a PPT simulator $\mathcal{S}_{\text{Liquid-update}}$ such that for all PPT environments $\mathcal{E}$ and all security parameters $k\in\mathbb{N}$,
    \[
    \mathsf{EXEC}^{\mathcal{F}^{\text{Liquid}}_{\text{layer2-join}}}_{\mathcal{S}_{\text{Liquid-join}},\mathcal{E}}(k)
    \ \stackrel{c}{\approx}\
    \mathsf{EXEC}^{\mathcal{F}^{\text{Liquid}}_{\text{layer2-update}}}_{\mathcal{S}_{\text{Liquid-update}},\mathcal{E}}(k),
    \]
    where $\stackrel{c}{\approx}$ denotes computational indistinguishability.
\end{lemma}

\begin{proof}
    Fix an arbitrary PPT environment $\mathcal{E}$. The only difference between the two games is that $\mathcal{F}^{\text{Liquid}}_{\text{layer2-update}}$ routes the Update request (from NET) through the subroutine $\mathcal{F}^{\text{Liquid}}_{\text{update}}$. We analyze the three observable components.

    \medskip
    \noindent As defined in both ideal functionalities, the Update request does not directly produce I/O outputs to $\mathcal{E}$; it only modifies $\mathsf{internalState}$. However, since reads and settlements derive from $\mathsf{internalState}$, a divergence between the two executions would propagate to observable Read or Settlement outputs. We must therefore argue that $\mathsf{internalState}$ at every honest participant evolves identically across the two executions.

    \medskip
    \noindent\textbf{(1) I/O outputs.} Although the Update request itself produces no I/O output, a divergence in $\mathsf{internalState}$ would cause observable differences in subsequent Read or Settlement outputs. In $\mathcal{P'}^{\text{Liquid}}$, a block is accepted as a valid state update only after the three-round BFT consensus among the operator committee completes, yielding (i)~at least $2f{+}1$ operator signatures on the block, (ii)~a valid predecessor reference into the BFT-finalized chain, (iii)~no conflicts with previously executed requests, and (iv)~a correctly computed post-state. These are precisely the checks enforced by $\mathcal{F}^{\text{Liquid}}_{\text{update}}$ (Steps~1--4). Since $\mathcal{S}_{\text{Liquid-update}}$ sends the Update message precisely when a block reaches BFT finality in $\mathcal{P'}^{\text{Liquid}}$, $\mathcal{F}^{\text{Liquid}}_{\text{update}}$ accepts if and only if the corresponding block was committed in $\mathcal{P'}^{\text{Liquid}}$. The BFT honest-majority assumption (at most $f_{L_2} = \tfrac{1}{3} n_{\mathit{OP}}$ corrupted operators) prevents the corrupted operator coalition from finalizing an inconsistent block, and the EUF-CMA security of $\mathcal{F}_{\text{sig}}$ prevents forgery of the $2f{+}1$ operator signatures required for finalization. And also prevent forgery of the client's transactions, which passes the check with $\mathsf{requestQueue}$ in the ideal functioanlity. Therefore $\mathsf{internalState}$ changes identically across the two executions.

    \medskip
    \noindent\textbf{(2) On-chain transactions.} The BFT consensus operates off-chain among the operator committee, and the Update request itself publishes no transactions to $\mathcal{F}_{\text{ledger}}$. Hence on-chain transactions are identical in both games.

    \medskip
    \noindent\textbf{(3) Adversarial leakage.} In both executions, the simulator extracts the plaintext leakage (under our plaintext-leakage assumption) from the ideal functionality and uses it to drive $\mathcal{P'}^{\text{Liquid}}$, generating leakage to the adversary according to the corruption status. Since the two executions feed $\mathcal{P'}^{\text{Liquid}}$ the same accepted inputs and run the same real-world protocol, $\mathsf{internalState}$ at every honest participant, including client and operator roles, evolves identically across the two executions. The interposition of $\mathcal{F}^{\text{Liquid}}_{\text{update}}$ affects only how block acceptance is gated within the ideal functionality; it does not alter any network-side message exchange, any honest party's internal computation, or the BFT-protocol message flow among operators. Consequently, the leaked $\mathsf{internalState}$ and received messages delivered to corrupted parties are computationally indistinguishable across the two executions.

    \medskip
    Conclusively,
    $\mathsf{EXEC}^{\mathcal{F}^{\text{Liquid}}_{\text{layer2-join}}}_{\mathcal{S}_{\text{Liquid-join}}, \mathcal{E}}(k)
    \stackrel{c}{\approx}
    \mathsf{EXEC}^{\mathcal{F}^{\text{Liquid}}_{\text{layer2-update}}}_{\mathcal{S}_{\text{Liquid-update}}, \mathcal{E}}(k)$.
\end{proof}

Next, we extend $\mathcal{F}^{\text{Liquid}}_{\text{layer2-update}}$ with the read subroutine, yielding $\mathcal{F}^{\text{Liquid}}_{\text{layer2-read}} = (\mathcal{F}_{\text{client-read}}, \mathcal{F}_{\text{ledger}} \mid \mathcal{F}^{\text{Liquid}}_{\text{submit}}, \mathcal{F}^{\text{Liquid}}_{\text{join}}, \mathcal{F}^{\text{Liquid}}_{\text{update}}, \mathcal{F}^{\text{Liquid}}_{\text{read}})$.

\vspace{1em}\begin{functionality}{Description of $\mathcal{M}_{\text{client-read}}$ of $\mathcal{F}^{\text{Liquid}}_{\text{layer2-read}}$}{

\textbf{Implemented role(s):} \{client-read\}

\noindent\textbf{Subroutines:}
$\mathcal{F}^{\text{Liquid}}_{\text{submit}}$: submit,
$\mathcal{F}^{\text{Liquid}}_{\text{join}}$: join,
$\mathcal{F}^{\text{Liquid}}_{\text{update}}$: update,
$\mathcal{F}^{\text{Liquid}}_{\text{read}}$: read,
$\mathcal{F}_{\text{ledger}}$: client\textsubscript{L1}

\noindent \textbf{Internal state:} (same as $\mathcal{F}_{\text{client}}$)

\vspace{0.5em}
\noindent \textbf{Main:}

\vspace{0.5em}

\textbf{recv} \{Submit, $\mathit{request}$\} \textbf{from} I/O:
\begin{enumerate}[itemsep=0.5em]
\item \textbf{send} \{Submit, $\mathit{request}$, $\mathsf{internalState}$\} \textbf{to} $(\pcur, \scur, \mathcal{F}^{\text{Liquid}}_{\text{submit}}:\text{submit})$,
\textbf{wait for} \{Submit, $\mathit{response}$\} s.t. $\mathit{response} \in \{\text{true}, \text{false}\}$;
\item \textbf{if} $\mathit{response} =$ true: $\mathsf{requestQueue}.\text{add}(\mathit{request})$;
\textbf{send} $\mathit{request}$ \textbf{to} $\mathcal{S}$ via NET;
\end{enumerate}

\hrule\vspace{0.5em}

\textbf{recv} \{Join, $\mathit{Attachment}$\} \textbf{from} NET:
\begin{enumerate}[itemsep=0.5em]
\item \textbf{send} \{Join, $\mathit{Attachment}$, $\mathsf{internalState}$\} \textbf{to} $(\pcur, \scur, \mathcal{F}^{\text{Liquid}}_{\text{join}}:\text{join})$,
\textbf{wait for} \{Join, $\mathit{response}$\} s.t. $\mathit{response} \in \{\text{true}, \text{false}\}$;
\item \textbf{if} $\mathit{response} =$ true: update $\mathsf{internalState}$ according to $\mathit{Attachment}$;
\textbf{reply} \{Join, $s_{\mathit{init}}$\} via I/O;
\end{enumerate}

\hrule\vspace{0.5em}

\textbf{recv} \{Update, $\mathit{Attachment}$\} \textbf{from} NET:
\begin{enumerate}[itemsep=0.5em]
\item \textbf{send} \{Update, $\mathit{Attachment}$, $\mathsf{internalState}$\} \textbf{to} $(\pcur, \scur, \mathcal{F}^{\text{Liquid}}_{\text{update}}:\text{update})$,
\textbf{wait for} \{Update, $\mathit{response}$, $\mathit{newState}$, $\mathit{executedReq}$\};
\item \textbf{if} $\mathit{response} =$ true: update $\mathsf{internalState}$ with $\mathit{newState}$ and $\mathit{executedReq}$;
\end{enumerate}

\hrule\vspace{0.5em}

\textbf{recv} \{Read\} \textbf{from} I/O:
\begin{enumerate}[itemsep=0.5em]
\item \textbf{send} \{Read, $\mathsf{internalState}$\} \textbf{to} $(\pcur, \scur, \mathcal{F}^{\text{Liquid}}_{\text{read}}:\text{read})$,
\textbf{wait for} $\mathit{ReadResult}$;
\item \textbf{if} $\mathit{ReadResult} \neq \bot$: \textbf{reply} \{Read, $\mathit{ReadResult}$\} via I/O;
\end{enumerate}

\hrule\vspace{0.5em}

\textbf{recv} \{ReadL1\} \textbf{from} I/O or NET:
\begin{enumerate}[itemsep=0.5em]
    \item \textbf{send} \{Read\} \textbf{to} $(\pcur, \scur, \mathcal{F}_{\text{ledger}}: \text{client}_{\text{L1}})$;
    \textbf{wait for} $\mathit{L1ReadResult}$;
    \item \textbf{reply} \{ReadL1, $\mathit{L1ReadResult}$\} via I/O;
\end{enumerate}

\hrule\vspace{0.5em}

\textbf{recv} any other message \textbf{from} NET:
\begin{enumerate}
    \item Output the message to $\mathcal{E}$ through I/O;
\end{enumerate}

}\end{functionality}\vspace{.5em}

\begin{functionality}{Description of simulator $\mathcal{S}_{\text{Liquid-read}}$}{

The simulator $\mathcal{S}_{\text{Liquid-read}}$ behaves identically to $\mathcal{S}_{\text{Liquid-update}}$, except for the following additional behavior when handling read-delivery queries from the ideal functionality:

\vspace{0.5em}
\textbf{Read delivery interaction with $\mathcal{F}^{\text{Liquid}}_{\text{layer2-read}}$:}

\begin{enumerate}[itemsep=0.5em]

\item When $\mathcal{F}^{\text{Liquid}}_{\text{read}}$ sends a responsive query to $\mathcal{S}_{\text{Liquid-read}}$ via NET to determine message delivery, $\mathcal{S}_{\text{Liquid-read}}$ replies based on $\mathcal{A}$'s scheduling decisions in $\mathcal{P'}^{\text{Liquid}}$.

\item $\mathcal{F}^{\text{Liquid}}_{\text{read}}$ then produces the I/O output to $\mathcal{E}$ accordingly.

\end{enumerate}

}\end{functionality}\vspace{1em}

\begin{lemma}
\label{lem:Liquid5}
    For all PPT adversaries $\mathcal{A}$, there exists a PPT simulator $\mathcal{S}_{\text{Liquid-read}}$ such that for all PPT environments $\mathcal{E}$ and all security parameters $k\in\mathbb{N}$,
    \[
    \mathsf{EXEC}^{\mathcal{F}^{\text{Liquid}}_{\text{layer2-update}}}_{\mathcal{S}_{\text{Liquid-update}},\mathcal{E}}(k)
    \ \stackrel{c}{\approx}\
    \mathsf{EXEC}^{\mathcal{F}^{\text{Liquid}}_{\text{layer2-read}}}_{\mathcal{S}_{\text{Liquid-read}},\mathcal{E}}(k),
    \]
    where $\stackrel{c}{\approx}$ denotes computational indistinguishability.
\end{lemma}

\begin{proof}
    Fix an arbitrary PPT environment $\mathcal{E}$. The only difference is that in $\mathsf{EXEC}^{\mathcal{F}^{\text{Liquid}}_{\text{layer2-read}}}_{\mathcal{S}_{\text{Liquid-read}},\mathcal{E}}(k)$ the read result is decided by $\mathcal{F}^{\text{Liquid}}_{\text{read}}$ based on $\mathsf{internalState}$, whereas in $\mathsf{EXEC}^{\mathcal{F}^{\text{Liquid}}_{\text{layer2-update}}}_{\mathcal{S}_{\text{Liquid-update}},\mathcal{E}}(k)$ read results are produced entirely by the simulator. We analyze the three observable components.

    \medskip
    \noindent\textbf{(1) I/O outputs.} By Lemma~\ref{lem:Liquid4}, $\mathsf{internalState}$ at every honest participant is identical in both executions. In both executions, read results are reconstructed from the same source: the BFT-finalized sidechain state recorded in $\mathsf{internalState}$, intersected with the L1-confirmed deposits and peg-out transactions read from $\mathcal{F}_{\text{ledger}}$. Two situations could in principle cause a distinguishable difference: \emph{(i)}~inconsistency between the simulated state and $\mathsf{internalState}$, or \emph{(ii)}~adversarial network control over the delivery schedule producing divergent results.

    For situation~(i), Lemma~\ref{lem:Liquid4} establishes that $\mathsf{internalState}$ evolves identically in both executions, since $\mathcal{F}^{\text{Liquid}}_{\text{update}}$ accepts a block if and only if BFT consensus finalizes it in $\mathcal{P'}^{\text{Liquid}}$, subject only to the unforgeability of $\mathcal{F}_{\text{sig}}$. For situation~(ii), $\mathcal{F}^{\text{Liquid}}_{\text{read}}$ uses the same delivery schedule supplied by $\mathcal{S}_{\text{Liquid-read}}$, mirroring the schedule that $\mathcal{S}_{\text{Liquid-update}}$ uses to determine the simulated client's read output in the previous game. Under the synchronous communication assumption among operators and the BFT honest-majority assumption (at most $f_{L_2} = \tfrac{1}{3} n_{\mathit{OP}}$ corrupted operators), at least one honest operator responds with the latest BFT-finalized state under any adversarial schedule. Hence the read outputs coincide.

    \medskip
    \noindent\textbf{(2) On-chain transactions.} Read requests do not publish any transactions to $\mathcal{F}_{\text{ledger}}$. Hence on-chain transactions are identical in both games.

    \medskip
    \noindent\textbf{(3) Adversarial leakage.} The game hop only interposes $\mathcal{F}^{\text{Liquid}}_{\text{read}}$ on the Read request, which is an I/O-facing operation that produces no network messages to corrupted parties. Both simulators maintain synchronized corruption status across all role types (client and operator) and run the same protocol logic inside $\mathcal{P'}^{\text{Liquid}}$. The leakage is computationally indistinguishable.

    \medskip
    Conclusively,
    $\mathsf{EXEC}^{\mathcal{F}^{\text{Liquid}}_{\text{layer2-update}}}_{\mathcal{S}_{\text{Liquid-update}}, \mathcal{E}}(k)
    \stackrel{c}{\approx}
    \mathsf{EXEC}^{\mathcal{F}^{\text{Liquid}}_{\text{layer2-read}}}_{\mathcal{S}_{\text{Liquid-read}}, \mathcal{E}}(k)$.
\end{proof}

Next, we extend $\mathcal{F}^{\text{Liquid}}_{\text{layer2-read}}$ with the settlement subroutine, yielding $\mathcal{F}^{\text{Liquid}}_{\text{layer2-settlement}} = (\mathcal{F}_{\text{client-settlement}}, \mathcal{F}_{\text{ledger}} \mid \mathcal{F}^{\text{Liquid}}_{\text{submit}}, \mathcal{F}^{\text{Liquid}}_{\text{join}},\\ \mathcal{F}^{\text{Liquid}}_{\text{update}}, \mathcal{F}^{\text{Liquid}}_{\text{read}}, \mathcal{F}^{\text{Liquid}}_{\text{settlement}})$.

\vspace{1em}\begin{functionality}{Description of $\mathcal{M}_{\text{client-settlement}}$ of $\mathcal{F}^{\text{Liquid}}_{\text{layer2-settlement}}$}{

\textbf{Implemented role(s):} \{client-settlement\}

\noindent\textbf{Subroutines:}
$\mathcal{F}^{\text{Liquid}}_{\text{submit}}$: submit,
$\mathcal{F}^{\text{Liquid}}_{\text{join}}$: join,
$\mathcal{F}^{\text{Liquid}}_{\text{update}}$: update,
$\mathcal{F}^{\text{Liquid}}_{\text{read}}$: read,
$\mathcal{F}^{\text{Liquid}}_{\text{settlement}}$: settlement,
$\mathcal{F}_{\text{ledger}}$: client\textsubscript{L1}

\noindent \textbf{Internal state:} (same as $\mathcal{F}_{\text{client}}$)

\vspace{0.5em}
\noindent \textbf{Main:}

\vspace{0.5em}

\textbf{recv} \{Submit, $\mathit{request}$\} \textbf{from} I/O:
\begin{enumerate}[itemsep=0.5em]
\item \textbf{send} \{Submit, $\mathit{request}$, $\mathsf{internalState}$\} \textbf{to} $(\pcur, \scur, \mathcal{F}^{\text{Liquid}}_{\text{submit}}:\text{submit})$,
\textbf{wait for} \{Submit, $\mathit{response}$\} s.t. $\mathit{response} \in \{\text{true}, \text{false}\}$;
\item \textbf{if} $\mathit{response} =$ true: $\mathsf{requestQueue}.\text{add}(\mathit{request})$;
\textbf{send} $\mathit{request}$ \textbf{to} $\mathcal{S}$ via NET;
\end{enumerate}

\hrule\vspace{0.5em}

\textbf{recv} \{Join, $\mathit{Attachment}$\} \textbf{from} NET:
\begin{enumerate}[itemsep=0.5em]
\item \textbf{send} \{Join, $\mathit{Attachment}$, $\mathsf{internalState}$\} \textbf{to} $(\pcur, \scur, \mathcal{F}^{\text{Liquid}}_{\text{join}}:\text{join})$,
\textbf{wait for} \{Join, $\mathit{response}$\} s.t. $\mathit{response} \in \{\text{true}, \text{false}\}$;
\item \textbf{if} $\mathit{response} =$ true: update $\mathsf{internalState}$ according to $\mathit{Attachment}$;
\textbf{reply} \{Join, $s_{\mathit{init}}$\} via I/O;
\end{enumerate}

\hrule\vspace{0.5em}

\textbf{recv} \{Update, $\mathit{Attachment}$\} \textbf{from} NET:
\begin{enumerate}[itemsep=0.5em]
\item \textbf{send} \{Update, $\mathit{Attachment}$, $\mathsf{internalState}$\} \textbf{to} $(\pcur, \scur, \mathcal{F}^{\text{Liquid}}_{\text{update}}:\text{update})$,
\textbf{wait for} \{Update, $\mathit{response}$, $\mathit{newState}$, $\mathit{executedReq}$\};
\item \textbf{if} $\mathit{response} =$ true: update $\mathsf{internalState}$ with $\mathit{newState}$ and $\mathit{executedReq}$;
\end{enumerate}

\hrule\vspace{0.5em}

\textbf{recv} \{Read\} \textbf{from} I/O:
\begin{enumerate}[itemsep=0.5em]
\item \textbf{send} \{Read, $\mathsf{internalState}$\} \textbf{to} $(\pcur, \scur, \mathcal{F}^{\text{Liquid}}_{\text{read}}:\text{read})$,
\textbf{wait for} $\mathit{ReadResult}$;
\item \textbf{if} $\mathit{ReadResult} \neq \bot$: \textbf{reply} \{Read, $\mathit{ReadResult}$\} via I/O;
\end{enumerate}

\hrule\vspace{0.5em}

\textbf{recv} \{Settlement, $\mathit{Attachment}$\} \textbf{from} NET:
\begin{enumerate}[itemsep=0.5em]
\item \textbf{send} \{Settlement, $\mathit{Attachment}$, $\mathsf{internalState}$\} \textbf{to} $(\pcur, \scur, \mathcal{F}^{\text{Liquid}}_{\text{settlement}}:\\\text{settlement})$,
\textbf{wait for} \{Settlement, $\mathit{response}$, $s_{\mathit{settle}}$\} s.t. $\mathit{response} \in \{\text{true}, \text{false}\}$;
\item \textbf{if} $\mathit{response} =$ true: update $\mathsf{internalState}$;
\textbf{reply} \{Settlement, $s_{\mathit{settle}}$\} via I/O;
\end{enumerate}

\hrule\vspace{0.5em}

\textbf{recv} \{ReadL1\} \textbf{from} I/O or NET:
\begin{enumerate}[itemsep=0.5em]
    \item \textbf{send} \{Read\} \textbf{to} $(\pcur, \scur, \mathcal{F}_{\text{ledger}}: \text{client}_{\text{L1}})$;
    \textbf{wait for} $\mathit{L1ReadResult}$;
    \item \textbf{reply} \{ReadL1, $\mathit{L1ReadResult}$\} via I/O;
\end{enumerate}

\hrule\vspace{0.5em}

\textbf{recv} any other message \textbf{from} NET:
\begin{enumerate}
    \item Output the message to $\mathcal{E}$ through I/O;
\end{enumerate}

}\end{functionality}\vspace{.5em}

\begin{functionality}{Description of simulator $\mathcal{S}_{\text{Liquid-settlement}}$}{

The simulator $\mathcal{S}_{\text{Liquid-settlement}}$ behaves identically to $\mathcal{S}_{\text{Liquid-read}}$, except for the following additional behavior upon detecting a completed settlement:

\vspace{0.5em}
\textbf{Settlement interaction with $\mathcal{F}^{\text{Liquid}}_{\text{layer2-settlement}}$:}

\begin{enumerate}[itemsep=0.5em]

\item $\mathcal{S}_{\text{Liquid-settlement}}$ monitors $\mathcal{P'}^{\text{Liquid}}$. When it detects that a client entity is about to produce a successful settlement output $\{\text{Settlement}, s_{\mathit{settle}}\}$ (the peg-out has been executed by operators and committed on L1), it intercepts this output.

\item $\mathcal{S}_{\text{Liquid-settlement}}$ prepares $\mathit{Attachment}$ by extracting:
\begin{itemize}
    \item $\mathit{TX}_{\mathit{peg\text{-}out}}$: the peg-out transaction committed on $\mathcal{F}_{\text{ledger}}$ in $\mathcal{P'}^{\text{Liquid}}$;
\end{itemize}

\item $\mathcal{S}_{\text{Liquid-settlement}}$ sends $\{\text{Settlement}, \mathit{Attachment}\}$ to $\mathcal{F}^{\text{Liquid}}_{\text{layer2-settlement}}$ via NET.

\end{enumerate}

}\end{functionality}\vspace{1em}

\begin{lemma}
\label{lem:Liquid6}
    For all PPT adversaries $\mathcal{A}$, there exists a PPT simulator $\mathcal{S}_{\text{Liquid-settlement}}$ such that for all PPT environments $\mathcal{E}$ and all security parameters $k\in\mathbb{N}$,
    \[
    \mathsf{EXEC}^{\mathcal{F}^{\text{Liquid}}_{\text{layer2-read}}}_{\mathcal{S}_{\text{Liquid-read}},\mathcal{E}}(k)
    \ \stackrel{c}{\approx}\
    \mathsf{EXEC}^{\mathcal{F}^{\text{Liquid}}_{\text{layer2-settlement}}}_{\mathcal{S}_{\text{Liquid-settlement}},\mathcal{E}}(k),
    \]
    where $\stackrel{c}{\approx}$ denotes computational indistinguishability.
\end{lemma}

\begin{proof}
    Fix an arbitrary PPT environment $\mathcal{E}$. The only difference between the two executions is that $\mathcal{F}^{\text{Liquid}}_{\text{layer2-settlement}}$ routes the Settlement request (from NET) through $\mathcal{F}^{\text{Liquid}}_{\text{settlement}}$. We analyze the three observable components.

    \medskip
    \noindent\textbf{(1) I/O outputs.} The I/O output affected is $\{\text{Settlement}, s_{\mathit{settle}}\}$. We consider two cases.

    \emph{Case~1: Successful settlement.} In $\mathcal{P'}^{\text{Liquid}}$, settlement succeeds when the peg-out transaction has been included in a BFT-finalized sidechain block and the corresponding $\mathit{TX}_{\mathit{peg\text{-}settlement}}$ is committed on L1. The simulator $\mathcal{S}_{\text{Liquid-settlement}}$ triggers $\mathcal{F}^{\text{Liquid}}_{\text{settlement}}$ precisely when this occurs. The checks in $\mathcal{F}^{\text{Liquid}}_{\text{settlement}}$: (1)~the peg-out is recorded in $\mathsf{executedRequest}$, (2)~$s_{\mathit{settle}}$ is consistent with the latest BFT-finalized sidechain state, and (3)~$\mathit{TX}_{\mathit{peg\text{-}settlement}}$ is committed on L1, are exactly the conditions that hold when the simulated settlement succeeds. Under the BFT honest-majority assumption (at most $f_{L_2} = \tfrac{1}{3} n_{\mathit{OP}}$ corrupted operators), L1 liveness and safety, and the EUF-CMA security of $\mathcal{F}_{\text{sig}}$ (which prevents the adversary from forging the operator signatures required to finalize an inconsistent peg-out), $\mathcal{F}^{\text{Liquid}}_{\text{settlement}}$ accepts whenever $\mathcal{P'}^{\text{Liquid}}$ completes. Hence the output is identical.

    \emph{Case~2: Failed settlement.} If the adversary tries to make the settlement procedure fail by delaying the communication network or creating incorrect peg-in transactions into the sidechain block, $\mathcal{P'}^{\text{Liquid}}$ does not complete the settlement due to the honest participants' checking with $\mathcal{F}_{\text{sig}}$. Correspondingly, $\mathcal{S}_{\text{Liquid-settlement}}$ does not trigger $\mathcal{F}^{\text{Liquid}}_{\text{settlement}}$, so no output is produced in either execution.

    \medskip
    \noindent\textbf{(2) On-chain transactions.} The peg-out transaction $\mathit{TX}_{\mathit{peg\text{-}settlement}}$ published to $\mathcal{F}_{\text{ledger}}$ is generated by $\mathcal{P'}^{\text{Liquid}}$, which runs identically in both simulators. The game hop only affects when the ideal functionality generates the I/O output, not which transactions are published. By the EUF-CMA security of $\mathcal{F}_{\text{sig}}$, the distributions of on-chain transactions are computationally indistinguishable.

    \medskip
    \noindent\textbf{(3) Adversarial leakage.} Both simulators maintain synchronized corruption status across all role types (client and operator) and execute the same protocol logic inside $\mathcal{P'}^{\text{Liquid}}$. The game hop interposes $\mathcal{F}^{\text{Liquid}}_{\text{settlement}}$ between the simulator and the I/O output but does not alter the simulator's internal execution or its interaction with corrupted parties. The leakage is identical.

    \medskip
    Conclusively,
    $\mathsf{EXEC}^{\mathcal{F}^{\text{Liquid}}_{\text{layer2-read}}}_{\mathcal{S}_{\text{Liquid-read}}, \mathcal{E}}(k)\stackrel{c}{\approx}\mathsf{EXEC}^{\mathcal{F}^{\text{Liquid}}_{\text{layer2-settlement}}}_{\mathcal{S}_{\text{Liquid-settlement}}, \mathcal{E}}(k)$.
\end{proof}

As a final step, we extend $\mathcal{F}^{\text{Liquid}}_{\text{layer2-settlement}}$ with the subroutine $\mathcal{F}^{\text{Liquid}}_{\text{updRnd}}$ to reach $\mathcal{F}^{\text{Liquid}}_{\text{layer2}} = (\mathcal{F}_{\text{client}}, \mathcal{F}_{\text{ledger}} \mid \mathcal{F}^{\text{Liquid}}_{\text{submit}}, \mathcal{F}^{\text{Liquid}}_{\text{join}}, \mathcal{F}^{\text{Liquid}}_{\text{update}},\\ \mathcal{F}^{\text{Liquid}}_{\text{read}}, \mathcal{F}^{\text{Liquid}}_{\text{settlement}}, \mathcal{F}^{\text{Liquid}}_{\text{updRnd}})$.

\vspace{1em}\begin{functionality}{Description of $\mathcal{M}_{\text{client}}$ of $\mathcal{F}^{\text{Liquid}}_{\text{layer2}}$}{

\textbf{Implemented role(s):} \{client\}

\noindent \textbf{Main:}

\vspace{0.5em}

\textbf{recv} \{Submit, $\mathit{request}$\} \textbf{from} I/O:
\begin{enumerate}[itemsep=0.5em]
\item \textbf{send} \{Submit, $\mathit{request}$, $\mathsf{internalState}$\} \textbf{to} $(\pcur, \scur, \mathcal{F}^{\text{Liquid}}_{\text{submit}}:\text{submit})$,
\textbf{wait for} \{Submit, $\mathit{response}$\} s.t. $\mathit{response} \in \{\text{true}, \text{false}\}$;
\item \textbf{if} $\mathit{response} =$ true: $\mathsf{requestQueue}.\text{add}(\mathit{request})$;
\textbf{send} $\mathit{request}$ \textbf{to} $\mathcal{S}$ via NET;
\end{enumerate}

\hrule\vspace{0.5em}

\textbf{recv} \{Join, $\mathit{Attachment}$\} \textbf{from} NET:
\begin{enumerate}[itemsep=0.5em]
\item \textbf{send} \{Join, $\mathit{Attachment}$, $\mathsf{internalState}$\} \textbf{to} $(\pcur, \scur, \mathcal{F}^{\text{Liquid}}_{\text{join}}:\text{join})$,
\textbf{wait for} \{Join, $\mathit{response}$\} s.t. $\mathit{response} \in \{\text{true}, \text{false}\}$;
\item \textbf{if} $\mathit{response} =$ true: update $\mathsf{internalState}$ according to $\mathit{Attachment}$;
\textbf{reply} \{Join, $s_{\mathit{init}}$\} via I/O;
\end{enumerate}

\hrule\vspace{0.5em}

\textbf{recv} \{Update, $\mathit{Attachment}$\} \textbf{from} NET:
\begin{enumerate}[itemsep=0.5em]
\item \textbf{send} \{Update, $\mathit{Attachment}$, $\mathsf{internalState}$\} \textbf{to} $(\pcur, \scur, \mathcal{F}^{\text{Liquid}}_{\text{update}}:\text{update})$,
\textbf{wait for} \{Update, $\mathit{response}$, $\mathit{newState}$, $\mathit{executedReq}$\};
\item \textbf{if} $\mathit{response} =$ true: update $\mathsf{internalState}$ with $\mathit{newState}$ and $\mathit{executedReq}$;
\end{enumerate}

\hrule\vspace{0.5em}

\textbf{recv} \{Read\} \textbf{from} I/O:
\begin{enumerate}[itemsep=0.5em]
\item \textbf{send} \{Read, $\mathsf{internalState}$\} \textbf{to} $(\pcur, \scur, \mathcal{F}^{\text{Liquid}}_{\text{read}}:\text{read})$,
\textbf{wait for} $\mathit{ReadResult}$;
\item \textbf{if} $\mathit{ReadResult} \neq \bot$: \textbf{reply} \{Read, $\mathit{ReadResult}$\} via I/O;
\end{enumerate}

\hrule\vspace{0.5em}

\textbf{recv} \{Settlement, $\mathit{Attachment}$\} \textbf{from} NET:
\begin{enumerate}[itemsep=0.5em]
\item \textbf{send} \{Settlement, $\mathit{Attachment}$, $\mathsf{internalState}$\} \textbf{to} $(\pcur, \scur, \mathcal{F}^{\text{Liquid}}_{\text{settlement}}:\\\text{settlement})$,
\textbf{wait for} \{Settlement, $\mathit{response}$, $s_{\mathit{settle}}$\} s.t. $\mathit{response} \in \{\text{true}, \text{false}\}$;
\item \textbf{if} $\mathit{response} =$ true: update $\mathsf{internalState}$;
\textbf{reply} \{Settlement, $s_{\mathit{settle}}$\} via I/O;
\end{enumerate}

\hrule\vspace{0.5em}

\textbf{recv} \{UpdateRound\} \textbf{from} NET:
\begin{enumerate}[itemsep=0.5em]
\item \textbf{send} \{UpdateRound, $\mathsf{internalState}$\} \textbf{to} $(\pcur, \scur, \mathcal{F}^{\text{Liquid}}_{\text{updRnd}}:\text{updRnd})$,
\textbf{wait for} \{UpdateRound, $\mathit{response}$\} s.t. $\mathit{response} \in \{\text{true}, \text{false}\}$;
\item \textbf{if} $\mathit{response} =$ true: $\mathsf{round} \leftarrow \mathsf{round} + 1$;
\textbf{reply} \{UpdateRound, $\mathit{response}$\} via NET;
\end{enumerate}

\hrule\vspace{0.5em}

\textbf{recv} \{GetCurRound\} \textbf{from} I/O or NET:
\begin{enumerate}[itemsep=0.5em]
    \item \textbf{reply} \{GetCurRound, $\mathsf{round}$\};
\end{enumerate}

\hrule\vspace{0.5em}

\textbf{recv} \{ReadL1\} \textbf{from} I/O or NET:
\begin{enumerate}[itemsep=0.5em]
    \item \textbf{send} \{Read\} \textbf{to} $(\pcur, \scur, \mathcal{F}_{\text{ledger}}: \text{client}_{\text{L1}})$;
    \textbf{wait for} $\mathit{L1ReadResult}$;
    \item \textbf{reply} \{ReadL1, $\mathit{L1ReadResult}$\} via I/O;
\end{enumerate}

}\end{functionality}\vspace{.5em}

\begin{functionality}{Description of simulator $\mathcal{S}_{\text{Liquid-updRnd}}$}{

The simulator $\mathcal{S}_{\text{Liquid-updRnd}}$ behaves identically to $\mathcal{S}_{\text{Liquid-settlement}}$, except for the following additional behavior when handling round-update requests:

\vspace{0.5em}
\textbf{Round update interaction with $\mathcal{F}^{\text{Liquid}}_{\text{layer2}}$:}

\begin{enumerate}[itemsep=0.5em]

\item Whenever $\mathcal{S}_{\text{Liquid-updRnd}}$ receives a round-update instruction from $\mathcal{A}$ (i.e., $\mathcal{A}$ advances the clock in the simulated $\mathcal{P'}^{\text{Liquid}}$), $\mathcal{S}_{\text{Liquid-updRnd}}$ simultaneously sends $\{\text{UpdateRound}\}$ to $\mathcal{F}^{\text{Liquid}}_{\text{layer2}}$ via NET.

\item $\mathcal{F}^{\text{Liquid}}_{\text{updRnd}}$ checks that no honest client's request has been pending for more than $T_{L_2}$ rounds. If the check passes, $\mathsf{round}$ in $\mathsf{internalState}$ is incremented.

\end{enumerate}

}\end{functionality}\vspace{1em}

\begin{lemma}
\label{lem:Liquid7}
    For all PPT adversaries $\mathcal{A}$, there exists a PPT simulator $\mathcal{S}_{\text{Liquid-updRnd}}$ such that for all PPT environments $\mathcal{E}$ and all security parameters $k\in\mathbb{N}$,
    \[
    \mathsf{EXEC}^{\mathcal{F}^{\text{Liquid}}_{\text{layer2-settlement}}}_{\mathcal{S}_{\text{Liquid-settlement}},\mathcal{E}}(k)
    \ \stackrel{c}{\approx}\
    \mathsf{EXEC}^{\mathcal{F}^{\text{Liquid}}_{\text{layer2}}}_{\mathcal{S}_{\text{Liquid-updRnd}},\mathcal{E}}(k),
    \]
    where $\stackrel{c}{\approx}$ denotes computational indistinguishability.
\end{lemma}

\begin{proof}
    Fix an arbitrary PPT environment $\mathcal{E}$. The only difference between the two games is that $\mathcal{F}^{\text{Liquid}}_{\text{layer2}}$ routes the UpdateRound request (from NET) through $\mathcal{F}^{\text{Liquid}}_{\text{updRnd}}$. We analyze the three observable components.

    \medskip
    \noindent\noindent\textbf{(1) I/O outputs.} The only I/O output affected is $\{\text{GetCurRound}, \mathsf{round}\}$. In $\mathsf{EXEC}^{\mathcal{F}^{\text{Liquid}}_{\text{layer2-settlement}}}_{\mathcal{S}_{\text{Liquid-settlement}},\mathcal{E}}(k)$, $\mathcal{S}_{\text{Liquid-settlement}}$ maintains the round counter inside the simulated $\mathcal{P'}^{\text{Liquid}}$, whose clock is advanced by $\mathcal{A}$. In $\mathsf{EXEC}^{\mathcal{F}^{\text{Liquid}}_{\text{layer2}}}_{\mathcal{S}_{\text{Liquid-updRnd}},\mathcal{E}}(k)$, $\mathcal{S}_{\text{Liquid-updRnd}}$ additionally forwards each round-update request from $\mathcal{A}$ to $\mathcal{F}^{\text{Liquid}}_{\text{updRnd}}$, which enforces a liveness predicate: no honest client's request may have been pending in the simulated protocol for more than $T_{L_2}$ rounds before the clock is allowed to advance.

In $\mathcal{P'}^{\text{Liquid}}$, $\mathcal{F}^{\text{Liquid}}_{\text{com}}$ guarantees synchronous message delivery among the operator committee within $\delta$ rounds per message, and the three-phase BFT consensus among operators (with at most $f_{L_2} = \tfrac{1}{3} n_{\mathit{OP}}$ corrupted) finalizes every honest client's request within $T_{L_2} = (3f{+}1)\delta$ rounds: each of the three BFT phases requires one round of all-to-all message exchange among the operator committee, bounded by $\delta$ per round, and the protocol's leader-rotation worst case absorbs at most $f$ failed leaders before reaching a finalizing honest operator proposer. Once finalization occurs, the corresponding state transition enters $\mathsf{internalState}$ via $\mathcal{F}^{\text{Liquid}}_{\text{update}}$, which by Lemma~\ref{lem:Liquid4} accepts the same set of blocks in both executions.

Consequently, whenever $\mathcal{A}$ attempts to advance the clock past $T_{L_2}$ rounds since some honest client's request was issued, $\mathcal{P'}^{\text{Liquid}}$ has already finalized that request through BFT, $\mathcal{F}^{\text{Liquid}}_{\text{update}}$ has already accepted it into $\mathsf{internalState}$, and the $T_{L_2}$-check in $\mathcal{F}^{\text{Liquid}}_{\text{updRnd}}$ therefore passes. The check passes at exactly the same rounds in which $\mathcal{P'}^{\text{Liquid}}$ would otherwise allow its clock to advance, the round counters evolve identically across the two executions, and $\mathsf{internalState}$ at every honest participant remains synchronised.

    \medskip
    \noindent\textbf{(2) On-chain transactions.} The UpdateRound request itself does not publish any transactions to $\mathcal{F}_{\text{ledger}}$; it only checks the synchronous-delivery predicate against the simulated round counter. The peg-in and peg-out transactions are generated by the Join and Settlement requests respectively in $\mathcal{P'}^{\text{Liquid}}$, both of which run identically in both simulators. Hence on-chain transactions are identical in both games.

    \medskip
    \noindent\textbf{(3) Adversarial leakage.} The game hop only interposes $\mathcal{F}^{\text{Liquid}}_{\text{updRnd}}$ on the UpdateRound request. Since $\mathcal{F}^{\text{Liquid}}_{\text{updRnd}}$ generates no network messages to corrupted parties and the $T_{L_2}$-check passes at the same times in both executions, the simulator's internal execution of $\mathcal{P'}^{\text{Liquid}}$ and its interaction with corrupted parties (across both role types: client and operator) are unaffected. The leakage delivered to $\mathcal{A}$ is identical.

    \medskip
    Conclusively,
    $\mathsf{EXEC}^{\mathcal{F}^{\text{Liquid}}_{\text{layer2-settlement}}}_{\mathcal{S}_{\text{Liquid-settlement}}, \mathcal{E}}(k)\stackrel{c}{\approx}\mathsf{EXEC}^{\mathcal{F}^{\text{Liquid}}_{\text{layer2}}}_{\mathcal{S}_{\text{Liquid-updRnd}}, \mathcal{E}}(k)$.
\end{proof}

\ThmrealizeLiquid*
\begin{proof}
    Let $\mathcal{A}$ be any PPT adversary and let $\mathcal{E}$ be any PPT environment. By Lemmas~\ref{lem:Liquid1}--\ref{lem:Liquid7}, the sequence of game hops yields a chain of computationally indistinguishable execution ensembles, each adjacent pair differing only in how one subroutine is implemented:
    \[
    \mathsf{EXEC}^{\mathcal{P}^{\text{Liquid}}}_{\mathcal{A},\mathcal{E}}
    \ \stackrel{c}{\approx}\ 
    \mathsf{EXEC}^{\mathcal{F}^{\text{Liquid}}_{\text{dummy}}}_{\mathcal{S}_{\text{Liquid}},\mathcal{E}}
    \ \stackrel{c}{\approx}\ 
    \cdots
    \ \stackrel{c}{\approx}\ 
    \mathsf{EXEC}^{\mathcal{F}^{\text{Liquid}}_{\text{layer2}}}_{\mathcal{S}_{\text{Liquid-updRnd}},\mathcal{E}}.
    \]
    By transitivity of computational indistinguishability, we obtain
    $
    \mathsf{EXEC}^{\mathcal{P}^{\text{Liquid}}}_{\mathcal{A},\mathcal{E}}
    \stackrel{c}{\approx}
    \mathsf{EXEC}^{\mathcal{F}^{\text{Liquid}}_{\text{layer2}}}_{\mathcal{S}_{\text{Liquid-updRnd}},\mathcal{E}}$, which proves that $\mathcal{P}^{\text{Liquid}}$ iUC-realizes $\mathcal{F}^{\text{Liquid}}_{\text{layer2}}$.
\end{proof}
\section{Case Study: The Arbitrum Nitro Rollup}
\label{apdx:arbitrum}

\subsection{$\mathcal{F}_{\text{ledger}}$ instantiation for Arbitrum Nitro}
\label{apd:ArbitrumLedger}

\subsubsection*{Submission functionality $\mathcal{F}^{\text{Arbitrum}}_{\text{submit}_{\text{L1}}}$}

The submission subroutine accepts a request to submit a transaction to L1 if and only if the transaction conforms to one of the Arbitrum transaction types. For Arbitrum Nitro, the following transactions are accepted by $\mathcal{F}^{\text{Arbitrum}}_{\text{submit}_{\text{L1}}}$: client deposit transactions ($\mathit{TX}_{\mathit{deposit}}$); client-signed self-submitted transactions ($\mathit{TX}_{\mathit{self}}$), which bypass the operator and go directly to L1; operator-published transaction batches with claimed post-states ($\mathit{batch}, \mathit{resultState}$); peg-out transactions ($\mathit{TX}_{\mathit{peg\text{-}out}}$) for both collaborative and escape-hatch settlement; and fraud-proof transactions ($\mathit{TX}_{\mathit{fraud}}$) published by clients (verifiers).

\subsubsection*{Update functionality $\mathcal{F}^{\text{Arbitrum}}_{\text{update}_{\text{L1}}}$}

The update subroutine refines the base $\mathcal{F}_{\text{update}_{\text{L1}}}$ by enforcing Arbitrum-specific commitment semantics, centred around the optimistic challenge period:

\begin{itemize}
    \item \textbf{Operator-published state transitions.} When an operator publishes the batch transaction $\mathit{TX}_{\mathit{batch}} = (\mathit{batch}, \mathit{resultState})$ to L1, $\mathit{TX}_{\mathit{batch}}$ carrying the operator's claimed L2 execution state $\mathit{resultState}$ is committed in $\{\mathit{TX}\}_{\text{L1}}$ ( representing stored in a blockchain BLOB), only when an L1-verifiable predicate over $\{\mathit{TX}\}_{\text{L1}}$ holds: $\mathit{TX}_{\mathit{batch}}$ has survived the challenge period $T_{\text{challenge}}$ without a valid fraud-proof transaction $\mathit{TX}_{\mathit{fraud}}$ being committed against it. L1 maintainers do not execute the L2 transition; the L1-finalized state $\mathit{State}_{\text{L1}}$ is advanced to $\mathit{resultState}$ 

    \item \textbf{Settlement.} A peg-out transaction $\mathit{TX}_{\mathit{peg\text{-}out}}$ (from either the collaborative or escape-hatch path) enters $\{\mathit{TX}\}_{\text{L1}}$ at publication time. $\mathit{State}_{\text{L1}}$ is advanced to reflect the post-settlement state carried in $\mathit{TX}_{\mathit{peg\text{-}out}}$ only after the same L1-verifiable predicate holds: $\mathit{TX}_{\mathit{peg\text{-}out}}$ has survived $T_{\text{challenge}}$ in $\{\mathit{TX}\}_{\text{L1}}$ without a valid fraud-proof transaction being committed against it.
\end{itemize}

\subsubsection*{Preservation of L1 safety and liveness.}

The instantiation preserves the original security guarantees of $\mathcal{F}_{\text{ledger}}$. Arbitrum-specific logic is purely additive: the new transaction types are accepted as valid L1 payloads, and the challenge-period and fraud-proof rules are deterministic predicates over transactions already committed on L1. And under the assumption of at least one honest client (verifier), incorrect transaction batches will not be committed on-chain. No existing L1 transaction is rejected or reordered by the instantiation, so the underlying ledger's security is inherited without change. 


\subsection{Arbitrum Nitro Real Protocol}
\label{apd:ArbitrumReal}

We formally define the real-world protocol implementation $\mathcal{P}^{\text{Arbitrum}}$ as
$\mathcal{P}^{\text{Arbitrum}} := (\mathcal{P}^{\text{Arbitrum}}_{\text{client}}: \text{client}, \mathcal{F}_{\text{ledger}} \mid \mathcal{P}^{\text{Arbitrum}}_{\text{operator}}, \mathcal{P}^{\text{Arbitrum}}_{\text{client}}: \text{verifier}, \mathcal{F}_{\text{sig}}, \mathcal{F}^{\text{Arbitrum}}_{\text{com}} )$.
The protocol consists of two participating parties: client and operator. Each client machine implements two roles, $\{\text{client}, \text{verifier}\}$: the client role is the public interface used by the environment for joining, submitting, settling, and reading; the verifier role is a private, internal role that monitors L1 publications from operators and posts fraud proofs when invalid state transitions are detected. The verifier role is not addressable from the environment and produces no I/O outputs, only its L1 publications (fraud-proof transactions) are observable, on the same footing as any other on-chain transaction. Operators execute received transaction requests and publish results to L1 periodically. We assume at least one honest operator and at least one honest client (so that at least one honest verifier exists). The structure is shown in Figure~\ref{fig:arbitrum}.
 
\subsubsection{Client Protocol}
 
The client machine $\mathcal{M}^{\text{Arbitrum}}_{\text{client}}$ defines the code run by users of the rollup protocol, implementing both the public client role and the private verifier role. As a client, it submits transactions collaboratively through the operator or directly to L1 via the self-submit path, and performs settlement via the collaborative or escape-hatch mechanism. As a verifier, it monitors operator batch publications on L1 and posts fraud proofs when an invalid state transition is detected; the verifier role is internal therefore it has no I/O interface to the environment.
 
\vspace{1em}\begin{functionality}{Description of protocol $\mathcal{P}^{\text{Arbitrum}}_{\text{client}} = (\text{client} \mid \text{verifier})$}{
 
\textbf{Participating roles:} \{client, verifier\}
 
\noindent \textbf{Corruption model:} dynamic corruption of the client role, at least one honest; the verifier role is private and follows the corruption status of the client machine that hosts it.
 
\noindent \textbf{Protocol parameters:}
\begin{itemize}
    \item $T_{\text{challenge}}$: challenge period duration
    \item $\mathsf{Val}(\mathit{TX})$: transaction validity checking algorithm. Given $\mathit{TX} = (\mathit{sender}, \mathit{receiver},\\ \mathit{value}, \mathit{data})$, $\mathsf{Val}(\mathit{TX})$ returns True if and only if:
    \begin{itemize}
        \item $\mathit{TX}$ is well-formed: $\mathit{sender}, \mathit{receiver} \in \{0,1\}^{*}$, $\mathit{value} \in \mathbb{N}_{\geq 0}$; \cmt{Correct transaction format}
        \item $\mathit{sender} \in \mathsf{identities}$; \cmt{Sender is a registered participant}
        \item Let $\mathit{bal}(\mathit{sender})$ denote the balance of $\mathit{sender}$ derived from $\mathsf{stateList}$. Then $\mathit{value} \leq \mathit{bal}(\mathit{sender})$; \cmt{Sender has sufficient balance}
        \item $\mathit{TX}$ does not conflict with any $\mathit{TX}' \in \mathsf{executedRequest}$; \cmt{No double-spending}
    \end{itemize}
    \item $\mathsf{Check}(\mathit{batch}, \mathit{resultState}, \mathit{prevState})$: batch validity checking algorithm used by the verifier role. Returns true if and only if:
    \begin{itemize}
        \item $\forall\, \mathit{TX} \in \mathit{batch}$: $\mathsf{Val}(\mathit{TX}) = \text{true}$; \cmt{Each transaction is well-formed}
        \item $\mathit{resultState}$ is the correct output of sequentially executing all $\mathit{TX} \in \mathit{batch}$ starting from $\mathit{prevState}$; \cmt{State transition is correct}
        \end{itemize}
        \item $\mathsf{GenF}(\mathit{batch}, \mathit{resultState}, \mathit{prevState})$: fraud-proof generation algorithm. Returns a fraud proof transaction $TX_{\textit{fraud}}$ if $\mathsf{Check}$ fails, otherwise $\bot$.
    
\end{itemize}
 
}\end{functionality}\vspace{.5em}
 
\begin{functionality}{Description of $\mathcal{M}^{\text{Arbitrum}}_{\text{client}}$}{

\textbf{Implemented role(s):} \{client, verifier\} \cmt{verifier is a private internal role}

\noindent \textbf{Subroutines:} $\mathcal{F}_{\text{sig}}$: \{signer, verifier\}, $\mathcal{F}^{\text{Arbitrum}}_{\text{com}}$: com, $\mathcal{F}_{\text{ledger}}$: client\textsubscript{L1}, $\mathcal{P}^{\text{Arbitrum}}_{\text{operator}}$: operator

\noindent \textbf{Internal state:}
\begin{itemize}
    \item $\mathsf{round} \in \mathbb{N}_{\geq 0}$, $\mathsf{round} = 0$ \hfill\cmt{Current round}
    \item $\mathsf{requestQueue} \subset \{0,1\}^{*}$, $\mathsf{requestQueue} = \emptyset$ \hfill\cmt{Queued unexecuted requests}
    \item $\mathsf{executedRequest} \subset \{0,1\}^{*}$, $\mathsf{executedRequest} = \emptyset$ \hfill\cmt{Executed requests}
    \item $\mathsf{stateList} \subset \{0,1\}^{*}$, $\mathsf{stateList} = \emptyset$ \hfill\cmt{L2 state list}
    \item $\mathsf{onchainState} \subset \{0,1\}^{*}$, $\mathsf{onchainState} = \emptyset$ \hfill\cmt{L1 committed state}
    \item $\mathsf{identities} \subset \{0,1\}^{*}$, $\mathsf{identities} = \{\mathsf{pid}_{op}, ContractAddr\}$ \hfill\cmt{Operator identity and contract address}
    \item $\mathsf{lastConfirmedState} \in \{0,1\}^{*}$, $\mathsf{lastConfirmedState} = \emptyset$ \hfill\cmt{Verifier's last confirmed L2 state}
\end{itemize}

\noindent \textbf{CheckID}(\emph{pid}, \emph{sid}, \emph{role}): Accept all messages with the same \emph{sid}. The \emph{verifier} role is private: it only accepts messages from $\mathcal{F}^{\text{Arbitrum}}_{\text{com}}$ (\{UpdateCheck\}) and from internal subroutines; it does not accept I/O requests from the environment.

\noindent \textbf{Corruption behavior:}
\begin{itemize}
    \item \textbf{DetermineCorrStatus}(\emph{pid}, \emph{sid}, \emph{role}): Return $\mathsf{corr}$ for the \emph{client} role; the \emph{verifier} role inherits the corruption status of the hosting machine.
    \item \textbf{LeakedData}(\emph{pid}, \emph{sid}, \emph{role}): Return \texttt{Internal State}.
\end{itemize}

\vspace{0.5em}
\noindent \textbf{Main (client role):}

\vspace{0.5em}

\textbf{recv} \{Join, $\mathit{ContractAddr}$, $s_{\mathit{init}}$\} \textbf{from} I/O:

\begin{enumerate}[itemsep=0.5em]
    \item Create $\mathit{TX}_{\mathit{deposit}} = \{\pcur, \mathit{ContractAddr}, s_{\mathit{init}}, \epsilon\}$;

    \item \textbf{send} \{SubmitL1, $\mathit{TX}_{\mathit{deposit}}$\} \textbf{to} $(\pcur, \scur, \mathcal{F}_{\text{ledger}}: \text{client}_{\text{L1}})$;

    \item \textbf{send} \{ReadL1\} \textbf{to} $(\pcur, \scur, \mathcal{F}_{\text{ledger}}:\text{client}_{\text{L1}})$;
    \textbf{wait for} \{ReadL1, $\mathit{L1ReadResult} = \{\{\mathit{TX}\}_{\text{L1}}, \mathit{State}_{\text{L1}}, \{\mathit{pid}\}\}$\};

    \item \textbf{if} $\mathit{TX}_{\mathit{deposit}} \in \{\mathit{TX}\}_{\text{L1}}$: create $\mathit{TX}_{\mathit{peg\text{-}in}} = \{\pcur, \epsilon, s_{\mathit{init}}, \epsilon\}$;

    \item \textbf{send} \{Sign, $\mathit{TX}_{\mathit{peg\text{-}in}}$\} \textbf{to} $(\pcur, \scur, \mathcal{F}_{\text{sig}}: \text{signer})$;
    \textbf{wait for} \{Signature, $\sigma$\};

    \item \textbf{send} \{Message, \{Join, $\mathit{TX}_{\mathit{peg\text{-}in}}$, $\sigma$, \emph{receiver}\}\} \textbf{to} $(\pcur, \scur, \mathcal{F}^{\text{Arbitrum}}_{\text{com}}:\text{com})$, where $\mathit{receiver} = (\mathsf{pid}_{op}, \scur, \mathcal{P}^{\text{Arbitrum}}_{\text{operator}}:\text{operator})$;

    \item $\mathsf{requestQueue}.\text{add}(\mathit{TX}_{\mathit{peg\text{-}in}})$; \cmt{Pending until L1 commitment + challenge window survival}

    \item \textbf{send} \{ReadL1\} \textbf{to} $(\pcur, \scur, \mathcal{F}_{\text{ledger}}:\text{client}_{\text{L1}})$;
    \textbf{wait for} \{ReadL1, $\mathit{L1ReadResult}$\}; $\mathsf{identities}.\text{add}(\mathit{ContractAddr})$;

    \item \textbf{if} $\mathit{TX}_{\mathit{peg\text{-}in}} \in \mathit{L1ReadResult}.\{\mathit{TX}\}_{\text{L1}}$ and $\nexists\, \mathit{TX}_{\mathit{fraud}}$ within $T_{\text{challenge}}$:
    \begin{itemize}
        \item $\mathsf{requestQueue}.\text{remove}(\mathit{TX}_{\mathit{peg\text{-}in}})$;
        $\mathsf{executedRequest}.\text{add}(\mathit{TX}_{\mathit{peg\text{-}in}})$;
        \item $\mathsf{stateList} \leftarrow s_{\mathit{init}}$; $\mathsf{onchainState} \leftarrow s_{\mathit{init}}$;
        \item $\forall\, \mathit{ID} \in \mathit{L1ReadResult}.\{\mathit{pid}\}$: $\mathsf{identities}.\text{add}(\mathit{ID})$;
        \item \textbf{reply} \{Join, $s_{\mathit{init}}$\} via I/O;
    \end{itemize}
\end{enumerate}

\hrule
\vspace{0.5em}

\textbf{recv} \{Submit, \texttt{collaborate}, $\mathit{TX}$\} \textbf{from} I/O, \textbf{s.t.} $\mathsf{Val}(\mathit{TX})$ = True:

\begin{enumerate}[itemsep=0.5em]
\item \textbf{send} \{Sign, $\mathit{TX}$\} \textbf{to} $(\pcur, \scur, \mathcal{F}_{\text{sig}}: \text{signer})$;
\textbf{wait for} \{Signature, $\sigma$\};

\item \textbf{send} \{Message, \{Submit, $\mathit{TX}$, $\sigma$, \emph{receiver}\}\} \textbf{to} $(\pcur, \scur, \mathcal{F}^{\text{Arbitrum}}_{\text{com}}:\text{com})$, where $\mathit{receiver} = (\mathsf{pid}_{op}, \scur, \mathcal{P}^{\text{Arbitrum}}_{\text{operator}}:\text{operator})$;

\item $\mathsf{requestQueue}.\text{add}(\mathit{TX})$; \cmt{Pending; moved to $\mathsf{executedRequest}$ on next Read after L1 batch survives $T_{\text{challenge}}$}
\end{enumerate}

\hrule
\vspace{0.5em}

\textbf{recv} \{Submit, \texttt{selfsubmit}, $\mathit{TX}_{\textit{self}}$\} \textbf{from} I/O, \textbf{s.t.} $\mathsf{Val}(\mathit{TX}_{\textit{self}})$ = True:\cmt{transaction including contract address $ContractAddr$}

\begin{enumerate}[itemsep=0.5em]
\item \textbf{send} \{Sign, $\mathit{TX}_{\textit{self}}$\} \textbf{to} $(\pcur, \scur, \mathcal{F}_{\text{sig}}: \text{signer})$;
\textbf{wait for} \{Signature, $\sigma$\};

\item \textbf{send} \{SubmitL1, \{$\mathit{TX}_{\textit{self}}$, $\sigma$\}\} \textbf{to} $(\pcur, \scur, \mathcal{F}_{\text{ledger}}:\text{client}_{\text{L1}})$;

\item \textbf{send} \{ReadL1\} \textbf{to} $(\pcur, \scur, \mathcal{F}_{\text{ledger}}:\text{client}_{\text{L1}})$;
\textbf{wait for} $\mathit{L1ReadResult}$;

\item \textbf{if} $\mathit{TX}_{\textit{self}} \in \mathit{L1ReadResult}.\{\mathit{TX}\}_{\text{L1}}$: $\mathsf{requestQueue}.\text{add}(\mathit{TX}_{\textit{self}})$; \cmt{Pending; cleared on next Read after $T_{\text{challenge}}$}
\end{enumerate}

\hrule
\vspace{0.5em}

\textbf{recv} \{Settlement, \texttt{collaborate}\} \textbf{from} I/O:

\begin{enumerate}[itemsep=0.5em]
\item Create $\mathit{TX}_{\mathit{peg\text{-}out}} = \{\mathit{ContractAddr}, \pcur, \mathsf{stateList}, \epsilon\}$;

\item \textbf{send} \{Sign, $\mathit{TX}_{\mathit{peg\text{-}out}}$\} \textbf{to} $(\pcur, \scur, \mathcal{F}_{\text{sig}}: \text{signer})$;
\textbf{wait for} \{Signature, $\sigma$\};

\item \textbf{send} \{Message, \{Settlement, $\mathit{TX}_{\mathit{peg\text{-}out}}$, $\sigma$, \emph{receiver}\}\} \textbf{to} $(\pcur, \scur, \mathcal{F}^{\text{Arbitrum}}_{\text{com}}:\text{com})$, where $\mathit{receiver} = (\mathsf{pid}_{op}, \scur, \mathcal{P}^{\text{Arbitrum}}_{\text{operator}}:\text{operator})$;

\item $\mathsf{requestQueue}.\text{add}(\mathit{TX}_{\mathit{peg\text{-}out}})$;

\item \textbf{send} \{ReadL1\} \textbf{to} $(\pcur, \scur, \mathcal{F}_{\text{ledger}}:\text{client}_{\text{L1}})$;
\textbf{wait for} $\mathit{L1ReadResult}$;

\item $\forall\, \mathit{TX}' \in \mathit{L1ReadResult}.\{\mathit{TX}\}_{\text{L1}}$ \textbf{s.t.} $\mathit{TX}' \in \mathsf{requestQueue}$ and $\nexists\, \mathit{TX}_{\mathit{fraud}}$ for $\mathit{TX}'$ within $T_{\text{challenge}}$: \cmt{Sweep cleared peg-out and any other now-final pending requests}
\begin{itemize}
    \item $\mathsf{requestQueue}.\text{remove}(\mathit{TX}')$;
    $\mathsf{executedRequest}.\text{add}(\mathit{TX}')$;
\end{itemize}

\item \textbf{if} $\mathit{TX}_{\mathit{peg\text{-}out}} \in \mathsf{executedRequest}$:
$\mathsf{stateList} \leftarrow \mathit{L1ReadResult}.\mathit{State}_{\text{L1}}$;\\
$\mathsf{onchainState} \leftarrow \mathit{L1ReadResult}.\mathit{State}_{\text{L1}}$;
\textbf{reply} \{Settlement, $\mathsf{onchainState}$\} via I/O;
\end{enumerate}

\hrule
\vspace{0.5em}

\textbf{recv} \{Settlement, \texttt{escape-hatch}\} \textbf{from} I/O:

\begin{enumerate}[itemsep=0.5em]
\item Create $\mathit{TX}_{\mathit{peg\text{-}out}} = \{\mathit{ContractAddr}, \pcur, \mathsf{stateList}, \texttt{self}\}$;

\item \textbf{send} \{SubmitL1, $\mathit{TX}_{\mathit{peg\text{-}out}}$\} \textbf{to} $(\pcur, \scur, \mathcal{F}_{\text{ledger}}:\text{client}_{\text{L1}})$;

\item $\mathsf{requestQueue}.\text{add}(\mathit{TX}_{\mathit{peg\text{-}out}})$;

\item \textbf{send} \{ReadL1\} \textbf{to} $(\pcur, \scur, \mathcal{F}_{\text{ledger}}:\text{client}_{\text{L1}})$;
\textbf{wait for} $\mathit{L1ReadResult}$;

\item $\forall\, \mathit{TX}' \in \mathit{L1ReadResult}.\{\mathit{TX}\}_{\text{L1}}$ \textbf{s.t.} $\mathit{TX}' \in \mathsf{requestQueue}$ and $\nexists\, \mathit{TX}_{\mathit{fraud}}$ for $\mathit{TX}'$ within $T_{\text{challenge}}$:
\begin{itemize}
    \item $\mathsf{requestQueue}.\text{remove}(\mathit{TX}')$;
    $\mathsf{executedRequest}.\text{add}(\mathit{TX}')$;
\end{itemize}

\item \textbf{if} $\mathit{TX}_{\mathit{peg\text{-}out}} \in \mathsf{executedRequest}$:
$\mathsf{stateList} \leftarrow \mathit{L1ReadResult}.\mathit{State}_{\text{L1}}$;\\
$\mathsf{onchainState} \leftarrow \mathit{L1ReadResult}.\mathit{State}_{\text{L1}}$;
\textbf{reply} \{Settlement, $\mathsf{onchainState}$\} via I/O;
\end{enumerate}

\hrule
\vspace{0.5em}

\textbf{recv} \{Read\} \textbf{from} I/O:
\begin{enumerate}[itemsep=0.5em]
\item \textbf{send} \{ReadL1\} \textbf{to} $(\pcur, \scur, \mathcal{F}_{\text{ledger}}:\text{client}_{\text{L1}})$;
\textbf{wait for} \{ReadL1, $\mathit{L1ReadResult} = \{\{\mathit{TX}\}_{\text{L1}}, \mathit{State}_{\text{L1}}, \{\mathit{pid}\}\}$\};

\item Let $\mathit{TX}_{\mathit{batch}}^{*} \in \{\mathit{TX}\}_{\text{L1}}$ be the most recent checkpoint transaction that has survived $T_{\text{challenge}}$ without a valid fraud proof. Parse $\{\mathit{batch}, \mathit{resultState}\} \leftarrow \mathit{TX}_{\mathit{batch}}^{*}.\mathit{data}$. 

\item $\forall\, \mathit{TX} \in \mathit{batch}$ \textbf{s.t.} $\mathit{TX} \in \mathsf{requestQueue}$:
\begin{itemize}
    \item $\mathsf{requestQueue}.\text{remove}(\mathit{TX})$;
    $\mathsf{executedRequest}.\text{add}(\mathit{TX})$; 
\end{itemize}

\item $\mathsf{stateList} \leftarrow \mathit{resultState}$, 
$\mathsf{onchainState} \leftarrow \mathit{State}_{\text{L1}}$; 

\item $\forall\, \mathit{pid}' \in \{\mathit{pid}\}$: $\mathsf{identities}.\text{add}(\mathit{pid}')$;

\item \textbf{reply} \{Read, $\mathit{ReadResult} = \{\mathsf{executedRequest}, \mathsf{stateList}, \mathsf{onchainState}\}$\} via I/O;
\end{enumerate}

\hrule
\vspace{0.5em}

\textbf{recv} \{GetCurRound\} \textbf{from} I/O:
\begin{enumerate}[itemsep=0.5em]
\item \textbf{send} \{GetCurRound\} \textbf{to} $(\pcur, \scur, \mathcal{F}^{\text{Arbitrum}}_{\text{com}}: \text{clock})$;
\textbf{wait for} \{GetCurRound, $\mathit{round}$\};
\item \textbf{reply} \{GetCurRound, $\mathit{round}$\} via I/O;
\end{enumerate}

\hrule
\vspace{1em}
\noindent \textbf{Main (verifier role, private, no I/O interface):}

\vspace{0.5em}

\textbf{recv} \{UpdateCheck\} \textbf{from} $(\pcur, \scur, \mathcal{F}^{\text{Arbitrum}}_{\text{com}}: \text{com})$:
\begin{enumerate}[itemsep=0.5em]
    \item \textbf{send} \{ReadL1\} \textbf{to} $(\pcur, \scur, \mathcal{F}_{\text{ledger}}:\text{client}_{\text{L1}})$;
    \textbf{wait for} \{ReadL1, $\mathit{L1ReadResult} = \{\{\mathit{TX}\}_{\text{L1}}, \mathit{State}_{\text{L1}}, \{\mathit{pid}\}\}$\};

    \item Let $\mathit{TX}_{\mathit{batch}} \in \{\mathit{TX}\}_{\text{L1}}$ be the latest unconfirmed checkpoint transaction within $T_{\text{challenge}}$ from the operator addressed to $\mathit{ContractAddr}$, and parse $\{\mathit{batch},\\ \mathit{resultState}\} \leftarrow \mathit{TX}_{\mathit{batch}}.\mathit{data}$; \cmt{Both batch and claimed post-state read from checkpoint data field}

    \item Let $b \leftarrow \mathsf{Check}(\mathit{batch}, \mathit{resultState}, \mathsf{lastConfirmedState})$; \cmt{Verifier-role check on published batch}

    \item \textbf{if} $b = \text{false}$:
    \begin{itemize}
        \item Let $\mathit{TX}_{\mathit{fraud}} \leftarrow \mathsf{GenF}(\mathit{batch}, \mathit{resultState}, \mathsf{lastConfirmedState})$; \cmt{Generate fraud-proof transaction}
        \item \textbf{if} $\mathit{TX}_{\mathit{fraud}} \neq \bot$: \textbf{send} \{SubmitL1, $\mathit{TX}_{\mathit{fraud}}$\} \textbf{to} $(\pcur, \scur, \mathcal{F}_{\text{ledger}}:\text{client}_{\text{L1}})$; \cmt{Fraud proof published on L1}
    \end{itemize}

    \item \textbf{if} $b = \text{true}$: $\mathsf{lastConfirmedState} \leftarrow \mathit{resultState}$;
\end{enumerate}

\noindent\emph{Verifier-role note.} The verifier role does not produce any I/O output to $\mathcal{E}$. Its only externally observable action is the publication of $\mathit{TX}_{\mathit{fraud}}$ on $\mathcal{F}_{\text{ledger}}$, which is an L1 transaction observable on the same footing as any other on-chain transaction.

}\end{functionality}\vspace{1em}
\begin{figure}
    \centering
    \includegraphics[width=\linewidth]{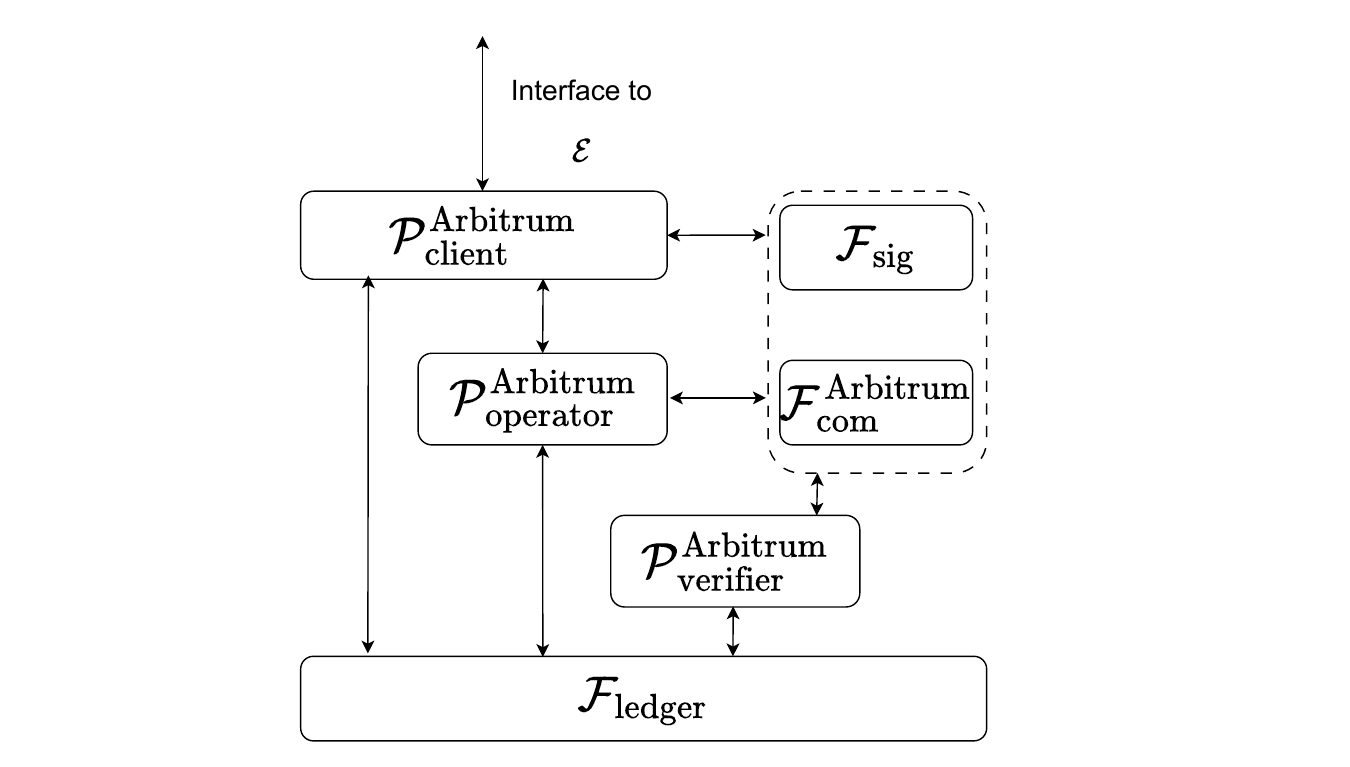}
    \caption{The Arbitrum Nitro protocol}
    \label{fig:arbitrum}
\end{figure}

\subsubsection{Operator Protocol}

The operator machine $\mathcal{M}^{\text{Arbitrum}}_{\text{operator}}$ executes received transaction requests and publishes results to L1 periodically (triggered by $\mathcal{F}^{\text{Arbitrum}}_{\text{com}}$ via \{UpdateRequest\}). We assume at least one honest operator.

\vspace{1em}\begin{functionality}{Description of protocol $\mathcal{P}^{\text{Arbitrum}}_{\text{operator}} = (\text{operator})$}{

\textbf{Participating roles:} \{operator\}

\noindent \textbf{Corruption model:} dynamic corruption

\noindent \textbf{Protocol parameters:}
\begin{itemize}
    \item $\mathsf{Val}(\mathit{TX}, \mathit{batch})$: transaction validity checking algorithm. Given $\mathit{TX} = (\mathit{sender},\\ \mathit{receiver}, \mathit{value}, \mathit{data})$ and the enclosing $\mathit{batch}$ of $(\mathit{TX}', \sigma')$ pairs, $\mathsf{Val}(\mathit{TX}, \mathit{batch})$ returns True if and only if:
    \begin{itemize}
        \item $\mathit{TX}$ is well-formed: $\mathit{sender}, \mathit{receiver} \in \{0,1\}^{*}$, $\mathit{value} \in \mathbb{N}_{\geq 0}$; \cmt{Correct transaction format}
        \item $\mathit{sender} \in \mathsf{identities}$; \cmt{Sender is a registered participant}
        \item Let $\mathit{bal}(\mathit{sender})$ denote the balance of $\mathit{sender}$ derived from $\mathsf{stateList}$. Then $\mathit{value} \leq \mathit{bal}(\mathit{sender})$; \cmt{Sender has sufficient balance}
        \item $\mathit{TX}$ does not conflict with any $\mathit{TX}' \in \mathsf{executedRequest}$; \cmt{No double-spending}
        \item $\exists\, (\mathit{TX}, \sigma) \in \mathit{batch}$ \textbf{s.t.} $\mathcal{F}_{\text{sig}}.\text{Verify}(\mathit{TX}, \sigma, \mathit{sender}) = \text{true}$; \cmt{$\mathit{batch}$ contains a valid signature of $\mathit{TX}$ from the sender}
    \end{itemize}
\end{itemize}

}\end{functionality}\vspace{.5em}

\begin{functionality}{Description of $\mathcal{M}^{\text{Arbitrum}}_{\text{operator}}$}{

\textbf{Implemented role(s):} \{operator\}

\noindent \textbf{Subroutines:} $\mathcal{F}_{\text{sig}}$: \{signer, verifier\}, $\mathcal{F}^{\text{Arbitrum}}_{\text{com}}$: com, $\mathcal{F}_{\text{ledger}}$: client\textsubscript{L1}

\noindent \textbf{Internal state:}
\begin{itemize}
    \item $\mathsf{requestQueue} \subset \{0,1\}^{*}$, $\mathsf{requestQueue} = \emptyset$ \hfill\cmt{Queued unexecuted requests}
    \item $\mathsf{stateList} \subset \{0,1\}^{*} \times \mathbb{N}_{\geq 0}$, $\mathsf{stateList} = \emptyset$ \hfill\cmt{Account states for all known clients, of form $(\mathit{pid}, \mathit{bal})$}
    \item $\mathsf{identities} \subset \{0,1\}^{*}$, $\mathsf{identities} = \emptyset$ \hfill\cmt{Registered identities}
\end{itemize}

\noindent \textbf{CheckID}(\emph{pid}, \emph{sid}, \emph{role}): Accept all messages with the same \emph{sid}. Only accept \emph{role} = operator.

\vspace{0.5em}
\noindent \textbf{Main:}

\vspace{0.5em}

\textbf{recv} \{Join, $\mathit{TX}_{\mathit{peg\text{-}in}}$, $\sigma$, $\mathit{pid}_{\mathit{call}}$\} \textbf{from} $(\pcur, \scur, \mathcal{F}^{\text{Arbitrum}}_{\text{com}}: \text{com})$:
\begin{enumerate}[itemsep=0.5em]
    \item Check that $\mathsf{Val}(\mathit{TX}_{\mathit{peg\text{-}in}}) = \text{true}$; \cmt{Peg-in transaction is well-formed}
    \item \textbf{send} \{Verify, $\mathit{TX}_{\mathit{peg\text{-}in}}$, $\sigma$\} \textbf{to} $(\pcur, \scur, \mathcal{F}_{\text{sig}}: \text{verifier})$;
    \textbf{wait for} \{VerResult, $b$\};
    \item \textbf{send} \{ReadL1\} \textbf{to} $(\pcur, \scur, \mathcal{F}_{\text{ledger}}:\text{client}_{\text{L1}})$;
    \textbf{wait for} \{ReadL1, $\mathit{L1ReadResult} = \{\{\mathit{TX}\}_{\text{L1}}, \mathit{State}_{\text{L1}}, \{\mathit{pid}\}\}$\};
    \item Check that $\mathit{TX}_{\mathit{deposit}} \in \{\mathit{TX}\}_{\text{L1}}$; \cmt{Deposit transaction is committed on L1}
    \item \textbf{if} $b = \text{true}$ and all checks pass:
    \begin{itemize}
        \item $\mathsf{identities}.\text{add}(\mathsf{pid}_{\mathit{sender}})$; \cmt{Register joining client}
        \item Parse $s_{\mathit{init}}$ from $\mathit{TX}_{\mathit{peg\text{-}in}}$; $\mathsf{stateList}.\text{add}(\mathsf{pid}_{\mathit{sender}}, s_{\mathit{init}})$; \cmt{Initialize client's account state}
        \item $\mathsf{requestQueue}.\text{add}(\mathit{TX}_{\mathit{peg\text{-}in}}, \sigma)$; \cmt{Enqueue peg-in with client signature}
    \end{itemize}
\end{enumerate}

\hrule
\vspace{0.5em}

\textbf{recv} \{Submit, $\mathit{TX} = (\mathit{sender}, \mathit{receiver}, \mathit{value}, \mathit{data})$, $\sigma$, $\mathit{pid}_{\mathit{call}}$\} \textbf{from} $(\pcur, \scur, \mathcal{F}^{\text{Arbitrum}}_{\text{com}}: \text{com})$:
\begin{enumerate}[itemsep=0.5em]
\item Check that $\mathsf{Val}(\mathit{TX}) = \text{true}$; \cmt{Transaction is well-formed and satisfies validity predicate}
\item \textbf{send} \{Verify, $\mathit{TX}$, $\sigma$\} \textbf{to} $(\pcur, \scur, \mathcal{F}_{\text{sig}}: \text{verifier})$,
\textbf{wait for} \{VerResult, $b$\};
\item \textbf{if} $b = \text{true}$:
\begin{itemize}
    \item $\mathsf{requestQueue}.\text{add}(\mathit{TX}, \sigma)$; \cmt{Signature valid; enqueue request with sender signature}
    \item Let $\mathit{bal}(\mathit{sender})$ be the balance of $\mathit{sender}$ in $\mathsf{stateList}$;
    \item $\mathsf{stateList}[\mathit{sender}].\mathit{bal} \leftarrow \mathit{bal}(\mathit{sender}) - \mathit{value}$; \cmt{Deduct from sender}
    \item $\mathsf{stateList}[\mathit{receiver}].\mathit{bal} \leftarrow \mathit{bal}(\mathit{receiver}) + \mathit{value}$; \cmt{Credit to receiver}
\end{itemize}
\end{enumerate}

\hrule
\vspace{0.5em}

\textbf{recv} \{Settlement, \texttt{collaborate}, $\mathit{TX}_{\mathit{peg\text{-}out}}$, $\sigma$, $\mathit{pid}_{\mathit{call}}$\} \textbf{from} $(\pcur, \scur, \mathcal{F}^{\text{Arbitrum}}_{\text{com}}: \text{com})$:
\begin{enumerate}[itemsep=0.5em]
\item Check that $\mathsf{Val}(\mathit{TX}_{\mathit{peg\text{-}out}}) = \text{true}$; \cmt{Peg-out transaction is well-formed}
\item \textbf{send} \{Verify, $\mathit{TX}_{\mathit{peg\text{-}out}}$, $\sigma$\} \textbf{to} $(\pcur, \scur, \mathcal{F}_{\text{sig}}: \text{verifier})$,
\textbf{wait for} \{VerResult, $b$\};
\item \textbf{if} $b = \text{true}$:
$\mathsf{requestQueue}.\text{add}(\mathit{TX}_{\mathit{peg\text{-}out}}, \sigma)$; \cmt{Signature valid; enqueue settlement request with client signature}
\end{enumerate}

\hrule
\vspace{0.5em}

\textbf{recv} \{UpdateRequest\} \textbf{from} $(\pcur, \scur, \mathcal{F}^{\text{Arbitrum}}_{\text{com}}: \text{com})$:
\begin{enumerate}[itemsep=0.5em]
\item Let $\mathit{batch} \leftarrow \mathsf{requestQueue}$; 

\item Let $\mathit{resultState} \leftarrow \mathsf{stateList}$; \cmt{Current state reflects all executed requests}
\item Construct the checkpoint transaction $\mathit{TX}_{\mathit{batch}} =\\ (\pcur,\, \mathit{ContractAddr},\, \epsilon,\, \{\mathit{batch}, \mathit{resultState}\})$; 

\item \textbf{send} \{SubmitL1, $\mathit{TX}_{\mathit{batch}}$\} \textbf{to} $(\pcur, \scur, \mathcal{F}_{\text{ledger}}:\text{client}_{\text{L1}})$; \cmt{Publish checkpoint transaction to L1}

\item $\mathsf{requestQueue} \leftarrow \emptyset$; 
\end{enumerate}

}\end{functionality}\vspace{1em}

\subsection{Arbitrum Nitro Ideal Functionality}
\label{apd:ArbitrumIdeal}

We define the ideal functionality for Arbitrum Nitro as
$\mathcal{F}^{\text{Arbitrum}}_{\text{layer2}}=(\mathcal{F}_{\text{client}}, \mathcal{F}_{\text{ledger}} \mid \mathcal{F}^{\text{Arbitrum}}_{\text{join}}, \mathcal{F}^{\text{Arbitrum}}_{\text{submit}}, \mathcal{F}^{\text{Arbitrum}}_{\text{update}}, \mathcal{F}^{\text{Arbitrum}}_{\text{read}}, \mathcal{F}^{\text{Arbitrum}}_{\text{settlement}},\\ \mathcal{F}^{\text{Arbitrum}}_{\text{updRnd}})$. The subroutines' ideal functionalities are defined as follows.

\subsubsection{Submit Functionality}
\label{apd:Arbitrumsubmit}

The submit subroutine $\mathcal{F}^{\text{Arbitrum}}_{\text{submit}}$ verifies that incoming requests are one of three valid types: \textit{(i)}~join requests, \textit{(ii)}~transaction submission requests (via either \texttt{collaborate} or \texttt{selfsubmit}), or \textit{(iii)}~settlement requests (via either \texttt{collaborate} or \texttt{escape-hatch}). The subroutine ensures semantic well-formedness.

\vspace{1em}\begin{functionality}{Description of subroutine $\mathcal{F}^{\text{Arbitrum}}_{\text{submit}} = (\text{submit})$}{

\textbf{Participating roles:} \{submit\}

\noindent \textbf{Corruption model:} incorruptible

}\end{functionality}\vspace{.5em}
\begin{functionality}{Description of $\mathcal{M}^{\text{Arbitrum}}_{\text{submit}}$}{

\textbf{Implemented role(s):} \{submit\}

\noindent \textbf{CheckID}(\emph{pid}, \emph{sid}, \emph{role}): Accept all messages with the same \emph{sid}.

\vspace{0.5em}
\noindent \textbf{Main:}

\vspace{0.5em}

\textbf{recv} \{Submit, $\mathit{request}$, $\mathsf{internalState}$\} \textbf{from} I/O:

\begin{enumerate}[itemsep=0.5em]

\item Check that $\mathit{request}$ belongs to one of the following valid types:
\begin{itemize}
    \item Join request: $\mathit{request} = (\text{Join}, s_{\mathit{init}}, \mathit{ContractAddr})$, \textbf{s.t.}:
    \begin{itemize}
        \item $s_{\mathit{init}} \in \mathbb{N}_{\geq 0}$; \cmt{Well-formed initial balance}
        \item $\mathit{ContractAddr} \in \{0,1\}^{*}$; \cmt{Valid contract address}
    \end{itemize}

    \item State update request: $\mathit{request} = (\text{Submit}, \mathit{mode}, \mathit{TX})$ with $\mathit{mode} \in \{\texttt{collaborate},\\ \texttt{selfsubmit}\}$ and $\mathit{TX} = (\mathit{sender}, \mathit{receiver}, \mathit{value}, \mathit{data})$, \textbf{s.t.}:
    \begin{itemize}
        \item $\mathit{sender}, \mathit{receiver} \in \{0,1\}^{*}$ and $\mathit{value} \in \mathbb{N}_{\geq 0}$; \cmt{Well-formed transaction}
        \item $\mathit{sender} \in \mathsf{identities}$; \cmt{Sender is a registered participant}
        \item Let $\mathit{bal}(\mathit{sender})$ denote the balance of $\mathit{sender}$ derived from $\mathsf{stateList}$. Then $\mathit{value} \leq \mathit{bal}(\mathit{sender})$; \cmt{Sender has sufficient balance}
        \item $\nexists\, \mathit{TX}' \in \mathsf{executedRequest} \cup \mathsf{requestQueue}$ that conflicts with $\mathit{TX}$; \cmt{No double-spending}
    \end{itemize}

    \item Settlement request: $\mathit{request} \in \{(\text{Settlement}, \texttt{collaborate}),\\ (\text{Settlement}, \texttt{escape-hatch})\}$;
\end{itemize}

\item Check that the join request has been executed (i.e., $\mathsf{stateList} \neq \emptyset$) or $\mathit{request}$ is a join request; \cmt{No action before joining}

\item Check that $\nexists$ a settlement request from the caller already pending in $\mathsf{requestQueue}$; \cmt{At most one pending settlement per client}

\item \textbf{if} all checks pass: \textbf{reply} \{Submit, true\};
\item \textbf{else}: \textbf{reply} \{Submit, false\};
\end{enumerate}

}\end{functionality}\vspace{1em}

\subsubsection{Join Functionality}
\label{apd:Arbitrumjoin}

The join subroutine $\mathcal{F}^{\text{Arbitrum}}_{\text{join}}$ validates the rollup joining procedure. It checks that \textit{(i)}~the peg-in transaction is consistent with the initial state, and \textit{(ii)}~both the deposit and peg-in transactions are committed on L1 without fraud proof during $T_{\text{challenge}}$.

\vspace{1em}
\begin{functionality}{Description of subroutine $\mathcal{F}^{\text{Arbitrum}}_{\text{join}} = (\text{join})$}{

\textbf{Participating roles:} \{join\}

\noindent \textbf{Corruption model:} incorruptible

\noindent \textbf{Protocol parameters:}
\begin{itemize}
    \item $T_{\text{challenge}}$: challenge period duration
\end{itemize}

}\end{functionality}\vspace{.5em}
\begin{functionality}{Description of $\mathcal{M}^{\text{Arbitrum}}_{\text{join}}$}{

\textbf{Implemented role(s):} \{join\}

\noindent \textbf{CheckID}(\emph{pid}, \emph{sid}, \emph{role}): Accept all messages with the same \emph{sid}.

\vspace{0.5em}
\noindent \textbf{Main:}

\vspace{0.5em}

\textbf{recv} \{Join, $\mathit{Attachment}=\{s_{\mathit{init}}, \mathit{TX}_{\mathit{peg\text{-}in}}\}$, $\mathsf{internalState}$\} \textbf{from} I/O:

\begin{enumerate}[itemsep=0.5em]

\item Parse $\mathit{TX}_{\mathit{peg\text{-}in}} = (\mathit{pid}_{\mathit{init}}, \epsilon, s_{\mathit{init}}', \epsilon)$. Check:
\begin{itemize}
    \item $s_{\mathit{init}}' = s_{\mathit{init}}$; \cmt{Peg-in carries the agreed initial state}
    \item $\mathit{pid}_{\mathit{init}} \in \mathsf{identities}$; \cmt{Initiator is a registered participant}
\end{itemize}

\item \textbf{send} \{ReadL1\} \textbf{to} $(\pcur, \scur, \mathcal{F}_{\text{ledger}}: \text{client}_{\text{L1}})$,
\textbf{wait for} \{ReadL1, $\mathit{L1ReadResult} = \{\{\mathit{TX}\}_{\text{L1}}, \mathit{State}_{\text{L1}}, \{\mathit{pid}\}\}$\};

\item Check the following conditions on $\mathit{output}$:
\begin{itemize}
    \item $\exists\, \mathit{TX}_{\mathit{deposit}} \in \{\mathit{TX}\}_{\text{L1}}$ \textbf{s.t.} $\mathit{TX}_{\mathit{deposit}}$ is a deposit from $\mathit{pid}_{\mathit{init}}$ consistent with $s_{\mathit{init}}$; \cmt{Deposit transaction is committed on L1}
    \item $\mathit{TX}_{\mathit{peg\text{-}in}} \in \{\mathit{TX}\}_{\text{L1}}$; \cmt{Peg-in transaction is committed on L1}
    \item $s_{\mathit{init}} \in \mathit{State}_{\text{L1}}$; \cmt{Initial state is committed on L1}
\end{itemize}

\item \textbf{if} all checks pass: \textbf{reply} \{Join, true, $s_{\mathit{init}}$\};

\end{enumerate}

}\end{functionality}\vspace{1em}

\begin{lemma}
\label{lem:OpenA}
    The subroutine $\mathcal{F}^{\text{Arbitrum}}_{\text{join}}$ guarantees \emph{correct L2 initialization}.
\end{lemma}

\begin{proof}
    We prove by contradiction. Suppose correct L2 initialization is violated, then according to the definition in Section~\ref{sec:def}: either (i)~some honest participant receives a successful initialization output whose committed initial state differs from the proposed state, or (ii)~the proposed state is not eventually committed on L1. However, $\mathcal{F}^{\text{Arbitrum}}_{\text{join}}$ outputs success only if (i)~the peg-in transaction is consistent with $s_{\mathit{init}}$ (Step~1)), and (ii)~$\mathit{TX}_{\mathit{peg\text{-}in}}$ and $s_{\mathit{init}}$ is committed in Steps~2). As long as $\mathcal{F}_{\text{sig}}$ prevents signature forgery and at least one honest verifier publishes fraud proofs for any incorrect publication, no adversary can cause $\mathcal{F}^{\text{Arbitrum}}_{\text{join}}$ to accept a mismatched state. This contradicts the assumption.
\end{proof}

\subsubsection{Update Functionality}

The update subroutine $\mathcal{F}^{\text{Arbitrum}}_{\text{update}}$ verifies state update request from the simulator. It checks that \textit{(i)}~the new state is correctly computed from the request batch, \textit{(ii)}~no conflict with previously executed requests, \textit{(iii)}  and \textit{(iv)}~the new state survives $T_{\text{challenge}}$ without fraud proof. If incorrect execution occurs and no fraud proof is published, the functionality halts.

\vspace{1em}\begin{functionality}{Description of subroutine $\mathcal{F}^{\text{Arbitrum}}_{\text{update}} = (\text{update})$}{

\textbf{Participating roles:} \{update\}

\noindent\textbf{Corruption model:} incorruptible

}\end{functionality}\vspace{.5em}

\begin{functionality}{Description of $\mathcal{M}^{\text{Arbitrum}}_{\text{update}}$}{

\textbf{Implemented role(s):} \{update\}

\noindent \textbf{CheckID}(\emph{pid}, \emph{sid}, \emph{role}): Accept all messages with the same \emph{sid}.

\vspace{0.5em}
\noindent \textbf{Main:}

\vspace{0.5em}

\textbf{recv} \{Update, $\mathit{Attachment}=\{\mathit{newState}, \mathit{requestBatch}\}\}$, $\mathsf{internalState}$\} \textbf{from} I/O:

\begin{enumerate}[itemsep=0.5em]

\item Parse $\mathit{requestBatch} = \{\mathit{TX}_1, \ldots, \mathit{TX}_m\}$. Check that $\forall\, j \in \{1, \ldots, m\}$:
\begin{itemize}
    \item $\nexists\, \mathit{TX}' \in \mathsf{executedRequest}$ that conflicts with $\mathit{TX}_j$; \cmt{No conflict with already-executed requests}
    \item $\mathit{TX}_j$ does not conflict with $\mathsf{stateList}$; \cmt{No conflict with current state}
\end{itemize}

\item Check that $\forall\, \mathit{TX}_j = (\mathit{sender}_j,\\ \mathit{receiver}_j, \mathit{value}_j, \mathit{text}_j) \in \mathit{requestBatch}$ with $(\mathit{sender}_j, \scur, \text{client}) \notin \\\mathsf{CorruptionSet}$:
\begin{itemize}
    \item $\exists\, \mathit{TX}' \in \mathsf{requestQueue}$ \textbf{s.t.} $\mathit{TX}' = \mathit{TX}_j$; \cmt{Each honest sender's transaction is recorded in the request queue}
\end{itemize}

\item Check that $\mathit{newState}$ is the correct output of sequentially executing $\mathit{TX}_1, \ldots, \mathit{TX}_m$ starting from $\mathsf{stateList}$; \cmt{State transition is correctly computed}

\item \textbf{send} \{ReadL1\} \textbf{to} $(\pcur, \scur, \mathcal{F}_{\text{ledger}}: \text{client}_{\text{L1}})$,
\textbf{wait for} \{ReadL1, $\mathit{L1ReadResult} = \{\{\mathit{TX}\}_{\text{L1}}, \mathit{State}_{\text{L1}}, \{\mathit{pid}\}\}$\};

\item Check in $\mathit{L1ReadResult}$, that $\mathit{requestBatch}, \mathit{newState} \subseteq \{\mathit{TX}\}_{\text{L1}}$ and they are the latest commitment; \cmt{Batch is published on L1}

\item \textbf{if} all checks pass: \textbf{reply} \{Update, true, $\mathit{newState}$, $\mathit{requestBatch}$\};

\end{enumerate}

}\end{functionality}\vspace{1em}

\subsubsection{Read Functionality}

The read subroutine $\mathcal{F}^{\text{Arbitrum}}_{\text{read}}$ queries the L1 blockchain and reconciles the result with $\mathsf{internalState}$ to produce the latest state and executed transactions to the clients according to L1 status.

\vspace{1em}\begin{functionality}{Description of subroutine $\mathcal{F}^{\text{Arbitrum}}_{\text{read}} = (\text{read})$}{

\textbf{Participating roles:} \{read\} \\ \textbf{Corruption model:} incorruptible

}\end{functionality}\vspace{.5em}

\begin{functionality}{Description of $\mathcal{M}^{\text{Arbitrum}}_{\text{read}}$}{

\textbf{Implemented role(s):} \{read\}

\noindent \textbf{CheckID}(\emph{pid}, \emph{sid}, \emph{role}): Accept all messages with the same \emph{sid}.

\vspace{0.5em}
\noindent \textbf{Main:}

\vspace{0.5em}

\textbf{recv} \{Read, $\mathsf{internalState}$\} \textbf{from} I/O:
\begin{enumerate}[itemsep=0.5em]

\item \textbf{send} \{ReadL1\} \textbf{to} $(\pcur, \scur, \mathcal{F}_{\text{ledger}}:\text{client}_{\text{L1}})$,
\textbf{wait for} \{ReadL1, $\mathit{L1ReadResult} = \{\{\mathit{TX}\}_{\text{L1}}, \mathit{State}_{\text{L1}}, \{\mathit{pid}\}\}$\};

\item Let $\mathsf{TX}^{\mathit{final}}_{\text{L1}} \leftarrow \{e \in \{\mathit{TX}\}_{\text{L1}} \mid e \text{ is L1-committed L2 published transaction and result}\}$ and let $\mathit{State}^{\mathit{final}}_{\text{L1}}$ be the latest L1-final L2 state derived from $\mathsf{TX}^{\mathit{final}}_{\text{L1}}$;

\item Let $\mathit{ReadResult}.\mathsf{executedRequest} \leftarrow \mathsf{executedRequest} \cap \mathsf{TX}^{\mathit{final}}_{\text{L1}}$;

\item Let $\mathit{ReadResult}.\mathsf{stateList} \leftarrow \mathsf{stateList} \cap \mathit{State}^{\mathit{final}}_{\text{L1}}$;

\item Let $\mathit{ReadResult}.\mathsf{onchainState} \leftarrow \mathsf{onchainState} \cap \mathit{State}^{\mathit{final}}_{\text{L1}}$;

\item Let $\mathit{ReadResult}.\mathsf{identities} \leftarrow \mathsf{identities} \cap \{\mathit{pid}\}$;

\item \textbf{reply} \{Read, $\mathit{ReadResult}$\};\cmt{$\mathsf{internalState}$ read result consistenct with L1 status}
\end{enumerate}

}\end{functionality}\vspace{1em}

\begin{lemma}
    The subroutines $\mathcal{F}^{\text{Arbitrum}}_{\text{update}}$ and $\mathcal{F}^{\text{Arbitrum}}_{\text{read}}$ jointly guarantee $(f_{L_2}+f_{L_1})$-safety.
\end{lemma}

\begin{proof}
    Suppose $\{f_{L_1}, f_{L_2}\}$-\emph{safety} is violated. By Definition~\ref{def:safety}, this means either self-consistency or view-consistency is violated for some honest client. The local execution-correctness check in $\mathcal{F}^{\text{Arbitrum}}_{\text{update}}$, gurantees by $f_{L_2}$ corruption, rules out incorrect state transitions entering $\mathsf{internalState}$, since updates are accepted only when committed on L1 either after surviving the challenge period $T_{\text{challenge}}$ or with $2f{+}1$ valid watchtower attestations. Additionally, $\mathcal{F}^{\text{Arbitrum}}_{\text{read}}$ derives every read result from $\mathsf{internalState}$ consistent with the underlying L1 transcript via $\mathcal{F}_{\text{ledger}}$. Under the assumption that $\mathcal{F}_{\text{ledger}}$ is safe under at most $f_{L_1}$ corrupted L1 participants, the L1 transcript is itself self-consistent and view-consistent, contradicting the assumption that safety was violated.
\end{proof}

\subsubsection{Settlement Functionality}

The settlement subroutine $\mathcal{F}^{\text{Arbitrum}}_{\text{settlement}}$ verifies that the peg-out transaction is consistent with the latest state and committed on L1 without fraud proof during $T_{\text{challenge}}$.

\vspace{1em}\begin{functionality}{Description of subroutine $\mathcal{F}^{\text{Arbitrum}}_{\text{settlement}} = (\text{settlement})$}{

\textbf{Participating roles:} \{settlement\}

\noindent\textbf{Corruption model:} incorruptible

}\end{functionality}\vspace{.5em}

\begin{functionality}{Description of $\mathcal{M}^{\text{Arbitrum}}_{\text{settlement}}$}{

\textbf{Implemented role(s):} \{settlement\}

\noindent \textbf{CheckID}(\emph{pid}, \emph{sid}, \emph{role}): Accept all messages with the same \emph{sid}.

\vspace{0.5em}
\noindent \textbf{Main:}

\vspace{0.5em}

\textbf{recv} \{Settlement, $\mathit{Attachment}=\{\mathit{TX}_{\mathit{peg\text{-}out}}\}$, $\mathsf{internalState}$\} \textbf{from} I/O:

\begin{enumerate}[itemsep=0.5em]

\item Parse $\mathit{TX}_{\mathit{peg\text{-}out}} = (\mathit{ContractAddr}, \mathit{pid}_{\mathit{client}}, s_{\mathit{settle}}, \epsilon)$. Let $s_{\mathit{settle}}$ be the latest state from $\mathsf{stateList}$. Check:
\begin{itemize}
    \item $s_{\mathit{settle}} = s_{\mathit{settle}}$; \cmt{Settlement carries the latest valid state}
    \item $\mathit{pid}_{\mathit{client}} \in \mathsf{identities}$; \cmt{Settling client is a registered participant}
\end{itemize}

\item \textbf{send} \{ReadL1\} \textbf{to} $(\pcur, \scur, \mathcal{F}_{\text{ledger}}: \text{client}_{\text{L1}})$,
\textbf{wait for} \{ReadL1, $\mathit{L1ReadResult} = \{\{\mathit{TX}\}_{\text{L1}}, \mathit{State}_{\text{L1}}, \{\mathit{pid}\}\}$\};

\item Check in $\mathit{L1ReadResult}$:
\begin{itemize}
    \item $\mathit{TX}_{\mathit{peg\text{-}out}} \in \{\mathit{TX}\}_{\text{L1}}$; \cmt{Peg-out transaction is committed on L1}
    \item $s_{\mathit{settle}} \in \mathit{State}_{\text{L1}}$; \cmt{Latest L2 state is recorded on L1}
\end{itemize}

\item \textbf{if} all checks pass: \textbf{reply} \{Settlement, true, $s_{\mathit{settle}}$\};

\end{enumerate}

}\end{functionality}\vspace{1em}

\begin{lemma}
\label{lem:SettleA}
    The subroutine $\mathcal{F}^{\text{Arbitrum}}_{\text{settlement}}$ guarantees \emph{correct L2 settlement}.
\end{lemma}

\begin{proof}
    According to the definition, \emph{correct L2 settlement} is violated if either (i)~an honest client receives a successful settlement output even though the output state is neither the latest state nor a state agreed upon by honest clients, or (ii)~the output state is not committed on the L1 blockchain. However, $\mathcal{F}^{\text{Arbitrum}}_{\text{settlement}}$ outputs success only after verifying that the committed state matches $s_{\mathit{settle}}$ (Step~1) and that the peg-out transaction is committed on L1 without fraud proof (Step~2). As long as $\mathcal{F}_{\text{sig}}$ prevents signature forgery, $\mathcal{F}_{\text{ledger}}$ provides a secure L1 ledger, and at least one honest verifier publishes fraud proofs for incorrect publications, no adversary can cause $\mathcal{F}^{\text{Arbitrum}}_{\text{settlement}}$ to accept a mismatched state. This contradicts the assumption.
\end{proof}

\subsubsection{Update Round Functionality}

The clock subroutine $\mathcal{F}^{\text{Arbitrum}}_{\text{updRnd}}$ 
enforces three liveness conditions. First, it ensures that 
self-submitted transactions from honest clients, which bypass the 
operator and are posted directly to L1, are reflected in the global 
state within $T_{L_1}$ rounds. Second, it ensures that escape-hatch 
settlement requests from honest clients are committed within 
$T_{L_1} + T_{\text{challenge}}$ rounds, accounting for both L1 
inclusion and the fraud-proof challenge period. Third, it ensures 
that at least one honest operator publishes a transaction batch to 
L1 within every $T_{\text{period}}$-round window, which guarantees 
that collaboratively submitted transactions are eventually included. 
Collaborative submit and collaborative settlement requests, which 
depend on operator cooperation, receive liveness guarantees only 
indirectly through this periodic publication requirement.

\vspace{1em}\begin{functionality}{Description of subroutine 
  $\mathcal{F}^{\text{Arbitrum}}_{\text{updRnd}} = 
  (\text{updRnd})$}{

\textbf{Participating roles:} \{updRnd\}

\noindent\textbf{Corruption model:} incorruptible

\noindent \textbf{Protocol parameters:}
\begin{itemize}
    \item $T_{L_1}$: L1 inclusion delay bound for regular transaction;

    \item $T_{\text{challnge}}$: the challenge period;

\end{itemize}

}\end{functionality}\vspace{.5em}

\begin{functionality}{Description of 
  $\mathcal{M}^{\text{Arbitrum}}_{\text{updRnd}}$}{

\textbf{Implemented role(s):} \{updRnd\}

\noindent \textbf{CheckID}(\emph{pid}, \emph{sid}, 
\emph{role}): Accept all messages with the same \emph{sid}.

\vspace{0.5em}
\noindent \textbf{Main:}

\vspace{0.5em}

\textbf{recv} \{UpdateRound, $\mathsf{internalState}$\} 
\textbf{from} I/O:

\begin{enumerate}[itemsep=0.5em]

\item Parse $\mathsf{round}$, $\mathsf{requestQueue}$, 
and $\mathsf{CorruptionSet}$ from $\mathsf{internalState}$.

\item \textbf{send} \{ReadL1\} \textbf{to} 
$(\pcur, \scur, \mathcal{F}_{\text{ledger}}: 
\text{client}_{\text{L1}})$,
\textbf{wait for} \{ReadL1, $\mathit{L1ReadResult} = 
\{\{\mathit{TX}\}_{\text{L1}}, \mathit{State}_{\text{L1}}, 
\{\mathit{pid}\}\}$\};

\item \textbf{// Check 1: Self-submitted transaction 
liveness}

For each $(\_,\, t_{\mathrm{sub}},\, \mathit{pid},\, q) 
\in \mathsf{requestQueue}$ such that 
$\mathit{pid} \notin \mathsf{CorruptionSet}$ and 
$\texttt{getType}(q) = \texttt{selfsubmit}$:
\begin{itemize}
    \item \textbf{if} $\mathsf{round} + 1 > 
    t_{\mathrm{sub}} + T_{L_1} + T_{\text{challenge}}$:
    \textbf{reply} \{UpdateRound, false\};
    \hfill\cmt{Self-submitted tx liveness violated}
\end{itemize}

\item \textbf{// Check 2: Escape-hatch settlement 
liveness}

For each $(\_,\, t_{\mathrm{sub}},\, \mathit{pid},\, q) 
\in \mathsf{requestQueue}$ such that 
$\mathit{pid} \notin \mathsf{CorruptionSet}$ and 
$\texttt{getType}(q) = \texttt{escape-hatch}$:
\begin{itemize}
    \item \textbf{if} $\mathsf{round} + 1 > 
    t_{\mathrm{sub}} + T_{L_1} + T_{\text{challenge}}$:
    \textbf{reply} \{UpdateRound, false\};
    \hfill\cmt{Escape-hatch settlement liveness violated}
\end{itemize}




\item \textbf{reply} \{UpdateRound, true\};
\hfill\cmt{All liveness checks passed}

\end{enumerate}

}\end{functionality}\vspace{1em}

\begin{lemma}
\label{lem:LiveA}
The ideal functionality 
$\mathcal{F}^{\text{Arbitrum}}_{\text{layer2}}$ guarantees 
the following liveness properties:
\begin{enumerate}
    \item \textbf{(Protocol joining.)}
    $(f_{L_2} + f_{L_1},\, T_{L_2} + T_{L_1} + T_{\text{challenge}})$-liveness: under $f_{L_2}$ and $f_{L_1}$ corruption, an accepted Join request results in output $\{\text{Join}, s_{\mathit{init}}\}$ within $T_{L_2} + T_{L_1} + T_{\text{challenge}}$. The off-chain latency $T_{L_2}$ is not enforced by $\mathcal{F}^{\text{Arbitrum}}_{\text{updRnd}}$ due to asynchronous communication; the L1 latency $T_{L_1}$ and $T_{\text{challenge}}$ is inherited from $\mathcal{F}_{\text{ledger}}$.

    \item \textbf{(State update.)} The functionality offers two state update liveness.
    \begin{itemize}
        \item \emph{Through operator.} $(f_{L_2} + f_{L_1},\, T_{L_2} + T_{L_1} + T_{\text{challenge}})$-liveness: under $f_{L_2}$ and $f_{L_1}$ corruption, an accepted update transaction request results in being added to\\ $\mathsf{executedRequest}$ and changes in $\mathsf{stateList}$ within $T_{L_2} + T_{L_1} + T_{\text{challenge}}$.

        \item \emph{Self-submit.} $(f_{L_1},\, T_{L_1} + T_{\text{challenge}})$-liveness: under $f_{L_1}$ corruption, an accepted self-submit update transaction request results in being added to $\mathsf{executedRequest}$ and changes in $\mathsf{stateList}$ within $T_{L_1} + T_{\text{challenge}}$.
    \end{itemize}
    The off-chain latency $T_{L_2}$ is not enforced by $\mathcal{F}^{\text{Arbitrum}}_{\text{updRnd}}$ due to asynchronous communication; the L1 latency $T_{L_1}$ and $T_{\text{challenge}}$ is inherited from $\mathcal{F}_{\text{ledger}}$.

    \item \textbf{(Settlement.)} The functionality offers two settlement liveness.
    \begin{itemize}
        \item \emph{Collaborative.} $(f_{L_2} + f_{L_1},\, T_{L_2} + T_{L_1} + T_{\text{challenge}})$-liveness: under $f_{L_2}$ and $f_{L_1}$ corruption, an accepted settlement request results in $\mathsf{onchainState}$ reflecting $s_{\mathit{settle}}$ and\\ $\{\text{Settlement}, s_{\mathit{settle}}\}$ output within $T_{L_2} + T_{L_1} + T_{\text{challenge}}$.
        
    \end{itemize}
\end{enumerate}
\end{lemma}

\begin{proof}
We argue each property separately.

\medskip
\noindent\textbf{(1) Protocol joining.}
Suppose liveness for join requests is violated. $\mathcal{F}^{\text{Arbitrum}}_{\text{join}}$ outputs success only if the peg-in transaction is consistent with $s_{\mathit{init}}$ and both the deposit and peg-in are committed on L1 without fraud proof during $T_{\text{challenge}}$. The join procedure is fully captured by $\mathcal{F}^{\text{Arbitrum}}_{\text{join}}$ waiting for the simulator's trigger, since $\mathcal{F}^{\text{Arbitrum}}_{\text{updRnd}}$ imposes no constraint on the off-chain latency $T_{L_2}$ controlled by the adversary. The L1-commitment check resolves within latency $T_{L_1}$ plus the challenge window $T_{\text{challenge}}$. The join is therefore accessible via a read request within $T_{L_2} + T_{L_1} + T_{\text{challenge}}$, yielding $(f_{L_2} + f_{L_1},\, T_{L_2} + T_{L_1} + T_{\text{challenge}})$-liveness and contradicting the assumption.

\medskip
\noindent\textbf{(2) State update.}
Suppose liveness for state update requests is violated. $\mathcal{F}^{\text{Arbitrum}}_{\text{update}}$ outputs success only if the new state is correctly computed from the request batch, no conflict exists with previously executed requests, each honest sender's transaction appears in $\mathsf{requestQueue}$, and the batch is committed on L1 without fraud proof. We distinguish two cases based on the path taken by the honest client submitting the request at time~$t$.

\emph{Collaborative path.} The collaborative path is fully captured by $\mathcal{F}^{\text{Arbitrum}}_{\text{update}}$ waiting for the simulator's trigger, since $\mathcal{F}^{\text{Arbitrum}}_{\text{updRnd}}$ imposes no constraint on the off-chain latency $T_{L_2}$ controlled by the adversary. The L1-commitment check resolves within latency $T_{L_1}$ plus the challenge window $T_{\text{challenge}}$. The request is therefore accessible via a read request within $T_{L_2} + T_{L_1} + T_{\text{challenge}}$, yielding $(f_{L_2} + f_{L_1},\, T_{L_2} + T_{L_1} + T_{\text{challenge}})$-liveness.

\emph{Self-submit path.} $\mathcal{F}^{\text{Arbitrum}}_{\text{updRnd}}$ captures the requirement for round update when the client bypasses the operator and submits the transaction directly to L1 via the self-submit mechanism. By L1 liveness, the transaction is committed on L1 within $T_{L_1}$. The honest client acting as verifier ensures only correct state transitions survive $T_{\text{challenge}}$. The request is therefore accessible via a read request within $T_{L_1} + T_{\text{challenge}}$, and the adversary cannot influence the bound through delaying off-chain communication. This yields $(f_{L_1},\, T_{L_1} + T_{\text{challenge}})$-liveness, depending only on L1 liveness and the challenge period.

In both cases, the corresponding time bound is achieved, contradicting the assumption.

\medskip
\noindent\textbf{(3) Settlement.}
Suppose liveness for settlement requests is violated. $\mathcal{F}^{\text{Arbitrum}}_{\text{settlement}}$ outputs success only if the peg-out transaction is consistent with the latest state and is committed on L1 without a fraud proof during $T_{\text{challenge}}$.

\emph{Collaborative path.} The collaborative path is fully captured by $\mathcal{F}^{\text{Arbitrum}}_{\text{settlement}}$ waiting for the simulator's trigger, since $\mathcal{F}^{\text{Arbitrum}}_{\text{updRnd}}$ imposes no constraint on the off-chain latency $T_{L_2}$ controlled by the adversary. The L1-commitment check resolves within latency $T_{L_1}$ plus the challenge window $T_{\text{challenge}}$. The settlement is therefore accessible via a read request within $T_{L_2} + T_{L_1} + T_{\text{challenge}}$, yielding $(f_{L_2} + f_{L_1},\, T_{L_2} + T_{L_1} + T_{\text{challenge}})$-liveness.

\emph{Escape-hatch path.} $\mathcal{F}^{\text{Arbitrum}}_{\text{updRnd}}$ captures the requirement for round update when the client bypasses the operator and publishes the peg-out directly to L1, so settlement liveness holds even if every operator is corrupted. By L1 liveness, the peg-out is committed on L1 within $T_{L_1}$ plus $T_{\text{challenge}}$. The settlement is therefore accessible via a read request within $T_{L_1} + T_{\text{challenge}}$, and the adversary cannot influence the bound through delaying off-chain communication. This yields $(f_{L_1},\, T_{L_1} + T_{\text{challenge}})$-liveness, depending only on L1 liveness and the challenge period.

In both cases, the corresponding time bound is achieved, contradicting the assumption.

\medskip
\noindent In summary, Arbitrum's liveness guarantees derive from three mechanisms rather than from $\mathcal{F}^{\text{Arbitrum}}_{\text{updRnd}}$: (i)~L1 liveness ensures on-chain inclusion within $T_{L_1}$; (ii)~the existence of at least one honest operator enables the collaborative path with bound $T_{L_2} + T_{L_1} + T_{\text{challenge}}$, while the self-submit and escape-hatch mechanisms provide a fallback that depends only on L1 liveness, with bound $T_{L_1} + T_{\text{challenge}}$; (iii)~the honest client acting as verifier ensures that only correctly computed state transitions survive the challenge period $T_{\text{challenge}}$. Settlement enjoys the strongest corruption tolerance since the escape-hatch mechanism depends only on L1 liveness.
\end{proof}


\subsection{Security Proof}
\label{apd:ArbitrumProof}

After proposing the ideal functionality and real-world implementation, we now show the security of the Arbitrum Nitro rollup protocol. To start with we first show the ideal functionality captures all the security properties:

\ThmidealArbitrum*

\begin{proof}
    According to Lemma~\ref{lem:OpenA}--\ref{lem:LiveA}, the ideal functionality\\ $\mathcal{F}^{\text{Arbitrum}}_{\text{layer2}}$ guarantees all security properties.
\end{proof}

After defining the ideal functionality $\mathcal{F}^{\text{Arbitrum}}_{\text{layer2}}$, we prove that the real Arbitrum protocol iUC-realizes it. The proof is done in 7 steps of successive game replacement. We first define a simulator $\mathcal{S}_{\text{Arbitrum}}$ that internally simulates a full run of $\mathcal{P}^{\text{Arbitrum}}$, and a dummy functionality $\mathcal{F}^{\text{Arbitrum}}_{\text{dummy}}$ that relays messages between $\mathcal{E}$ and $\mathcal{S}_{\text{Arbitrum}}$. This base ideal execution yields the same distribution of messages to $\mathcal{E}$ as the real execution. We use the execution ensemble $\mathsf{EXEC}$ to denote the messages observed by $\mathcal{E}$, including the output for the input request and the leakage to adversary, when interacting with adversary $\mathcal{A}$, real protocol $\mathcal{P}$, ideal functionality $\mathcal{F}$ and simulator $\mathcal{S}$ in the proofs that follow.

In each subsequent step, we incrementally add interaction between the simulator and the ideal functionality and extend the functionality, thereby forming the corresponding subroutines, while keeping the changes transparent to both $\mathcal{E}$ and $\mathcal{A}$. We continue until we obtain the target functionality $\mathcal{F}^{\text{Arbitrum}}_{\text{layer2}}$ defined by our framework. At every step, the simulator is adjusted so that the new ideal execution is indistinguishable from the previous one. For each transition, we discuss the differences relative to the prior step and prove that, given the same inputs from $\mathcal{E}$ and $\mathcal{A}$, the resulting outputs remain the same up to computational indistinguishability under any adversarial influence strategy.

We begin by defining the dummy ideal functionality
$\mathcal{F}^{\text{Arbitrum}}_{\text{dummy}}=(\mathcal{F}_{\text{client-dummy}}, \mathcal{F}_{\text{ledger}} \mid \perp)$
and the simulator $\mathcal{S}_{\text{Arbitrum}}$ as follows. The dummy functionality forwards every request from $\mathcal{E}$ to the simulator and returns the simulator's response unchanged. The simulator $\mathcal{S}_{\text{Arbitrum}}$ runs $\mathcal{P'}^{\text{Arbitrum}}$ internally and produces identical outputs.

\vspace{1em}\begin{functionality}{Description of $\mathcal{M}_{\text{client-dummy}}$ of $\mathcal{F}^{\text{Arbitrum}}_{\text{dummy}}$}{

\textbf{Implemented role(s):} \{client-dummy\}

\noindent \textbf{Main:}

\textbf{recv} any request \textbf{from} I/O:
\begin{enumerate}[itemsep=0.5em]
    \item Forward request to $\mathcal{S}$ through NET;
\end{enumerate}

\hrule
\vspace{0.5em}

\textbf{recv} any message \textbf{from} NET:

\begin{enumerate}
    \item Output the message to $\mathcal{E}$ through I/O;
\end{enumerate}

}\end{functionality}\vspace{.5em}

\begin{functionality}{Description of simulator $\mathcal{S}_{\text{Arbitrum}}$}{

$\mathcal{S}_{\text{Arbitrum}}$ internally simulates $\mathcal{P'}^{\text{Arbitrum}}$, a copy of the real protocol $\mathcal{P}^{\text{Arbitrum}}$ as defined in Section~\ref{apd:ArbitrumReal}.

\vspace{0.5em}
\textbf{Real protocol simulation:}
\begin{itemize}[itemsep=0.3em]
    \item $\mathcal{S}_{\text{Arbitrum}}$ simulates honest clients', operators', and verifier's actions inside $\mathcal{P'}^{\text{Arbitrum}}$ according to the real protocol.
    \item If participants are corrupted, $\mathcal{S}_{\text{Arbitrum}}$ leaks the corresponding messages sent to corrupted entities to the adversary $\mathcal{A}$ and continues simulating honest parties based on $\mathcal{A}$'s instructions.
\end{itemize}

\textbf{Network communication from/to the environment:}
\begin{itemize}[itemsep=0.3em]
    \item Messages that $\mathcal{S}_{\text{Arbitrum}}$ receives on the network interface (from $\mathcal{E}$/$\mathcal{A}$) are forwarded to $\mathcal{P'}^{\text{Arbitrum}}$.
    \item Messages sent by $\mathcal{P'}^{\text{Arbitrum}}$ on its network interface (to $\mathcal{E}$/$\mathcal{A}$) are forwarded to the environment.
\end{itemize}

\textbf{Input requests and outputs:}
\begin{itemize}[itemsep=0.3em]
    \item Unlike $\mathcal{P}^{\text{Arbitrum}}$, which receives inputs directly from $\mathcal{E}$, the simulation $\mathcal{P'}^{\text{Arbitrum}}$ receives requests forwarded from $\mathcal{F}^{\text{Arbitrum}}_{\text{layer2}}$. Instead of sending outputs directly to $\mathcal{E}$, $\mathcal{S}_{\text{Arbitrum}}$ sends them to $\mathcal{F}^{\text{Arbitrum}}_{\text{layer2}}$.
\end{itemize}

\textbf{Message delivery:}
\begin{itemize}[itemsep=0.3em]
    \item The Arbitrum protocol assumes asynchronous communication via $\mathcal{F}^{\text{Arbitrum}}_{\text{com}}$. The simulator bookkeeps all messages in $\mathcal{P'}^{\text{Arbitrum}}$ and triggers delivery according to the adversary's scheduling decisions, mirroring the real-world protocol.
\end{itemize}

\textbf{Corruption handling:}
\begin{itemize}[itemsep=0.3em]
    \item $\mathcal{S}_{\text{Arbitrum}}$ keeps the corruption status of entities in $\mathcal{P}^{\text{Arbitrum}}$, $\mathcal{P'}^{\text{Arbitrum}}$ and $\mathcal{F}^{\text{Arbitrum}}_{\text{layer2}}$ synchronized. When an entity in $\mathcal{P'}^{\text{Arbitrum}}$ becomes corrupted, $\mathcal{S}_{\text{Arbitrum}}$ corrupts the corresponding entity in $\mathcal{F}^{\text{Arbitrum}}_{\text{layer2}}$ before continuing.
    \item Adversarial commands for corrupted participants (e.g., publishing on $\mathcal{F}_{\text{ledger}}$, sending via $\mathcal{F}^{\text{Arbitrum}}_{\text{com}}$) are forwarded to $\mathcal{P'}^{\text{Arbitrum}}$.
    \item When a corrupted participant in $\mathcal{P'}^{\text{Arbitrum}}$ wants to output to $\mathcal{E}$, $\mathcal{S}_{\text{Arbitrum}}$ instructs the corresponding entity in $\mathcal{F}^{\text{Arbitrum}}_{\text{layer2}}$ to output.
\end{itemize}
}\end{functionality}\vspace{1em}

\begin{lemma}
\label{lem:Arb1}
    For all PPT adversaries $\mathcal{A}$, there exists a PPT simulator $\mathcal{S}_{\text{Arbitrum}}$ such that for all PPT environments $\mathcal{E}$ and all security parameters $k\in\mathbb{N}$,
    $\mathsf{EXEC}^{\mathcal{P}^{\text{Arbitrum}}}_{\mathcal{A},\mathcal{E}}(k)\ \stackrel{c}{\approx}\
    \mathsf{EXEC}^{\mathcal{F}^{\text{Arbitrum}}_{\text{dummy}}}_{\mathcal{S}_{\text{Arbitrum}},\mathcal{E}}(k)$,
    where $\stackrel{c}{\approx}$ denotes computational indistinguishability.
\end{lemma}

\begin{proof}
    Fix an arbitrary PPT environment $\mathcal{E}$ and adversary $\mathcal{A}$. We argue that the execution ensembles in the real and ideal worlds are computationally indistinguishable by analyzing the three components observable by $\mathcal{E}$.

    \medskip
    \noindent\textbf{Observable components.} In the real world, the execution ensemble $\mathsf{EXEC}^{\mathcal{P}^{\text{Arbitrum}}}_{\mathcal{A},\mathcal{E}}(k)$ consists of:
    \begin{enumerate}
        \item \emph{I/O outputs} delivered to $\mathcal{E}$ by $\mathcal{P}^{\text{Arbitrum}}_{\text{client}}$:
        $\{\text{Join}, s_{\mathit{init}}\}$,
        $\{\text{Settlement},\\ \mathsf{onchainState}\}$,
        $\{\text{Read}, \mathit{ReadResult}\}$,
        $\{\text{GetCurRound}, \mathit{round}\}$.
        \item \emph{On-chain transactions} committed on $\mathcal{F}_{\text{ledger}}$ during execution:
        $\mathit{TX}_{\mathit{deposit}}$, $\mathit{TX}_{\mathit{peg\text{-}in}}$, $\mathit{TX}_{\mathit{peg\text{-}out}}$, operator batch publications, $\mathit{TX}_{fraud}$.
        \item \emph{Adversarial leakage}, comprising messages received by corrupted parties during the protocol execution, and the corrupted parties' internal state.
    \end{enumerate}

    In the ideal world, the simulator $\mathcal{S}_{\text{Arbitrum}}$ internally runs $\mathcal{P'}^{\text{Arbitrum}}$ and interacts with the dummy functionality $\mathcal{F}^{\text{Arbitrum}}_{\text{dummy}}$, which by definition forwards every request from $\mathcal{E}$ to $\mathcal{S}_{\text{Arbitrum}}$ unchanged and relays the simulator's responses back to $\mathcal{E}$. We show that each component is computationally indistinguishable across the two worlds.

    \medskip
    \noindent\textbf{(1) I/O outputs.} Since $\mathcal{F}^{\text{Arbitrum}}_{\text{dummy}}$ acts as a transparent relay, $\mathcal{S}_{\text{Arbitrum}}$ receives exactly the same sequence of requests as $\mathcal{P}^{\text{Arbitrum}}_{\text{client}}$ would in the real world. By construction, $\mathcal{S}_{\text{Arbitrum}}$ executes the same client, operator, and verifier logic inside $\mathcal{P'}^{\text{Arbitrum}}$ under the same adversarial scheduling, producing the same I/O outputs. The only potential difference arises from the randomness of $\mathcal{F}_{\text{sig}}$: signature strings may differ between the two worlds because fresh randomness is sampled independently. However, since $\mathcal{F}_{\text{sig}}$ realizes EUF-CMA security, signatures generated on the same messages are computationally indistinguishable. Hence the I/O outputs are computationally indistinguishable.

    \medskip
    \noindent\textbf{(2) On-chain transactions.} Since $\mathcal{F}^{\text{Arbitrum}}_{\text{dummy}}$ forwards all requests to $\mathcal{S}_{\text{Arbitrum}}$, the simulated protocol $\mathcal{P'}^{\text{Arbitrum}}$ generates and publishes the same set of transactions to $\mathcal{F}_{\text{ledger}}$ as $\mathcal{P}^{\text{Arbitrum}}$ would in the real world. Transactions may contain different signature values due to independent randomness in $\mathcal{F}_{\text{sig}}$, but by the EUF-CMA security of the signature scheme, the transaction distributions are computationally indistinguishable.

    \medskip
    \noindent\textbf{(3) Adversarial leakage.} By the definition of $\mathcal{S}_{\text{Arbitrum}}$, the corruption status of all entities is kept synchronized between $\mathcal{P'}^{\text{Arbitrum}}$ and $\mathcal{F}^{\text{Arbitrum}}_{\text{dummy}}$. Since $\mathcal{F}^{\text{Arbitrum}}_{\text{dummy}}$ forwards all $\mathcal{E}$-requests to the simulator, $\mathcal{S}_{\text{Arbitrum}}$ can reconstruct the same internal state and message history for corrupted parties as in the real execution. Consequently, the leakage delivered to $\mathcal{A}$ is identical up to signature randomness, which is again computationally indistinguishable by EUF-CMA security.

    \medskip
    Conclusively, we have:
    $\mathsf{EXEC}^{\mathcal{P}^{\text{Arbitrum}}}_{\mathcal{A}, \mathcal{E}}(k)\stackrel{c}{\approx}\mathsf{EXEC}^{\mathcal{F}^{\text{Arbitrum}}_{\text{dummy}}}_{\mathcal{S}_{\text{Arbitrum}}, \mathcal{E}}(k)$.
\end{proof}

Next, we extend the dummy functionality $\mathcal{F}^{\text{Arbitrum}}_{\text{dummy}}$ with the submission subroutine, yielding
$\mathcal{F}^{\text{Arbitrum}}_{\text{layer2-submit}} = (\mathcal{F}_{\text{client-submit}}, \mathcal{F}_{\text{ledger}} \mid \mathcal{F}^{\text{Arbitrum}}_{\text{submit}})$.
The subroutine $\mathcal{F}^{\text{Arbitrum}}_{\text{submit}}$ is defined in Appendix~\ref{apd:Arbitrumsubmit}. The intermediate client functionality adds the Submit request that routes through $\mathcal{F}^{\text{Arbitrum}}_{\text{submit}}$, all other requests are forwarded to $\mathcal{S}$ unchanged.

\vspace{1em}\begin{functionality}{Description of $\mathcal{M}_{\text{client-submit}}$ of $\mathcal{F}^{\text{Arbitrum}}_{\text{layer2-submit}}$}{

\textbf{Implemented role(s):} \{client-submit\}

\noindent \textbf{Main:}

\textbf{recv} \{Submit, $\mathit{request}$\} \textbf{from} I/O:

\begin{enumerate}[itemsep=0.5em]
\item \textbf{send} \{Submit, $\mathit{request}$, $\mathsf{internalState}$\} \textbf{to} $(\pcur, \scur, \mathcal{F}^{\text{Arbitrum}}_{\text{submit}}:\text{submit})$,
\textbf{wait for} \{Submit, $\mathit{response}$\} s.t. $\mathit{response} \in \{\text{true}, \text{false}\}$;
\item \textbf{if} $\mathit{response} =$ true: $\mathsf{requestQueue}.\text{add}(\mathit{request})$;
\textbf{send} $\mathit{request}$ \textbf{to} $\mathcal{S}$ via NET;
\end{enumerate}

\hrule
\vspace{0.5em}

\textbf{recv} any request \textbf{from} I/O:
\begin{enumerate}[itemsep=0.5em]
    \item Forward request to $\mathcal{S}$ through NET;
\end{enumerate}

\hrule
\vspace{0.5em}

\textbf{recv} any message \textbf{from} NET:

\begin{enumerate}
    \item Output the message to $\mathcal{E}$ through I/O;
\end{enumerate}

}\end{functionality}\vspace{.5em}

\begin{functionality}{Description of simulator $\mathcal{S}_{\text{Arbitrum-submit}}$}{

The simulator $\mathcal{S}_{\text{Arbitrum-submit}}$ behaves the same as $\mathcal{S}_{\text{Arbitrum}}$.

}\end{functionality}\vspace{1em}

\begin{lemma}
\label{lem:Arb2}
    For all PPT adversaries $\mathcal{A}$, there exists a PPT simulator $\mathcal{S}_{\text{Arbitrum-submit}}$ such that for all PPT environments $\mathcal{E}$ and all security parameters $k\in\mathbb{N}$,
    \[
    \mathsf{EXEC}^{\mathcal{F}^{\text{Arbitrum}}_{\text{dummy}}}_{\mathcal{S}_{\text{Arbitrum}},\mathcal{E}}(k)
    \ \stackrel{c}{\approx}\
    \mathsf{EXEC}^{\mathcal{F}^{\text{Arbitrum}}_{\text{layer2-submit}}}_{\mathcal{S}_{\text{Arbitrum-submit}},\mathcal{E}}(k),
    \]
    where $\stackrel{c}{\approx}$ denotes computational indistinguishability.
\end{lemma}

\begin{proof}
    Fix an arbitrary PPT environment $\mathcal{E}$. The simulator logic is identical in both executions; the only difference is the ideal functionality through which requests are forwarded to the simulator. We analyze the three observable components.

    \medskip
    \noindent\textbf{Difference between the two executions.} In $\mathcal{F}^{\text{Arbitrum}}_{\text{dummy}}$, every request from $\mathcal{E}$ is forwarded directly to $\mathcal{S}_{\text{Arbitrum}}$. In $\mathcal{F}^{\text{Arbitrum}}_{\text{layer2-submit}}$, the subroutine $\mathcal{F}^{\text{Arbitrum}}_{\text{submit}}$ intercepts each request and checks semantic validity before forwarding. Requests that fail these checks are rejected and never reach the simulator.

    \medskip
    \noindent\textbf{(1) I/O outputs.} In the real protocol $\mathcal{P}^{\text{Arbitrum}}$ (and hence in $\mathcal{P'}^{\text{Arbitrum}}$), the client machine already enforces the same validity checks for all four request types (Join, Submit, SubmitFast, Settlement): only predefined request types are processed, and malformed or out-of-phase requests produce no output. Therefore, any request rejected by $\mathcal{F}^{\text{Arbitrum}}_{\text{submit}}$ would also produce no output inside the simulation. For accepted requests, both simulators execute the same client logic and produce the same I/O outputs. Hence, the I/O outputs are identical.

    \medskip
    \noindent\textbf{(2) On-chain transactions.} Since only accepted requests trigger protocol actions in $\mathcal{P'}^{\text{Arbitrum}}$, and the set of accepted requests is the same in both executions, the transactions published to $\mathcal{F}_{\text{ledger}}$ are identical except for signature values. By EUF-CMA security, the on-chain transactions are computationally indistinguishable.

    \medskip
    \noindent\textbf{(3) Adversarial leakage.} Since invalid requests are rejected by the ideal functionality and ignored by honest parties during simulation, no leakage is generated toward the adversary in either world. Moreover, the corruption status remains synchronized across both worlds. Consequently, the leakage delivered to $\mathcal{A}$ is computationally indistinguishable in both executions.

    \medskip
    Conclusively,
    $\mathsf{EXEC}^{\mathcal{F}^{\text{Arbitrum}}_{\text{dummy}}}_{\mathcal{S}_{\text{Arbitrum}}, \mathcal{E}}(k) \stackrel{c}{\approx} \mathsf{EXEC}^{\mathcal{F}^{\text{Arbitrum}}_{\text{layer2-submit}}}_{\mathcal{S}_{\text{Arbitrum-submit}}, \mathcal{E}}(k)$.
\end{proof}

Next, we extend $\mathcal{F}^{\text{Arbitrum}}_{\text{layer2-submit}}$ with the join subroutine, yielding $\mathcal{F}^{\text{Arbitrum}}_{\text{layer2-join}} = (\mathcal{F}_{\text{client-join}}, \mathcal{F}_{\text{ledger}} \mid \mathcal{F}^{\text{Arbitrum}}_{\text{submit}}, \mathcal{F}^{\text{Arbitrum}}_{\text{join}})$. The intermediate client functionality additionally routes the Join request (from NET) through $\mathcal{F}^{\text{Arbitrum}}_{\text{join}}$.

\vspace{1em}\begin{functionality}{Description of $\mathcal{M}_{\text{client-join}}$ of $\mathcal{F}^{\text{Arbitrum}}_{\text{layer2-join}}$}{

\textbf{Implemented role(s):} \{client-join\}

\noindent \textbf{Main:}

\textbf{recv} \{Submit, $\mathit{request}$\} \textbf{from} I/O:

\begin{enumerate}[itemsep=0.5em]
\item \textbf{send} \{Submit, $\mathit{request}$, $\mathsf{internalState}$\} \textbf{to} $(\pcur, \scur, \mathcal{F}^{\text{Arbitrum}}_{\text{submit}}:\text{submit})$,
\textbf{wait for} \{Submit, $\mathit{response}$\} s.t. $\mathit{response} \in \{\text{true}, \text{false}\}$;
\item \textbf{if} $\mathit{response} =$ true: $\mathsf{requestQueue}.\text{add}(\mathit{request})$;
\textbf{send} $\mathit{request}$ \textbf{to} $\mathcal{S}$ via NET;
\end{enumerate}

\hrule
\vspace{0.5em}

\textbf{recv} \{Join, $\mathit{Attachment}$\} \textbf{from} NET:

\begin{enumerate}[itemsep=0.5em]
\item \textbf{send} \{Join, $\mathit{Attachment}$, $\mathsf{internalState}$\} \textbf{to} $(\pcur, \scur, \mathcal{F}^{\text{Arbitrum}}_{\text{join}}:\text{join})$,
\textbf{wait for} \{Join, $\mathit{response}$\} s.t. $\mathit{response} \in \{\text{true}, \text{false}\}$;
\item \textbf{if} $\mathit{response} =$ true: update $\mathsf{internalState}$ according to $\mathit{Attachment}$;
\textbf{reply} \{Join, $s_{\mathit{init}}$\} via I/O;
\end{enumerate}

\hrule
\vspace{0.5em}

\textbf{recv} any request \textbf{from} I/O:
\begin{enumerate}[itemsep=0.5em]
    \item Forward request to $\mathcal{S}$ through NET;
\end{enumerate}

\hrule
\vspace{0.5em}

\textbf{recv} any message \textbf{from} NET:

\begin{enumerate}
    \item Output the message to $\mathcal{E}$ through I/O;
\end{enumerate}

}\end{functionality}\vspace{.5em}

\begin{functionality}{Description of simulator $\mathcal{S}_{\text{Arbitrum-join}}$}{

The simulator $\mathcal{S}_{\text{Arbitrum-join}}$ behaves identically to $\mathcal{S}_{\text{Arbitrum-submit}}$, except for the following additional behavior upon detecting a completed join:

\vspace{0.5em}
\textbf{Join interaction with $\mathcal{F}^{\text{Arbitrum}}_{\text{layer2-join}}$:}

\begin{enumerate}[itemsep=0.5em]

\item $\mathcal{S}_{\text{Arbitrum-join}}$ monitors the simulated protocol $\mathcal{P'}^{\text{Arbitrum}}$. When it detects that a client entity is about to produce the I/O output $\{\text{Join}, s_{\mathit{init}}\}$, $\mathcal{S}_{\text{Arbitrum-join}}$ intercepts this output and proceeds as follows.

\item $\mathcal{S}_{\text{Arbitrum-join}}$ prepares $\mathit{Attachment}$ by extracting from the simulation state:
\begin{itemize}
    \item $s_{\mathit{init}}$: the initial state included in the simulated client's Join output;
    \item $\mathit{TX}_{\mathit{peg\text{-}in}}$: the peg-in transaction committed on $\mathcal{F}_{\text{ledger}}$ in $\mathcal{P'}^{\text{Arbitrum}}$;
\end{itemize}

\item $\mathcal{S}_{\text{Arbitrum-join}}$ sends $\{\text{Join}, \mathit{Attachment}\}$ to $\mathcal{F}^{\text{Arbitrum}}_{\text{layer2-join}}$ via NET.

\end{enumerate}

}\end{functionality}\vspace{1em}

\begin{lemma}
\label{lem:Arb3}
    For all PPT adversaries $\mathcal{A}$, there exists a PPT simulator $\mathcal{S}_{\text{Arbitrum-join}}$ such that for all PPT environments $\mathcal{E}$ and all security parameters $k\in\mathbb{N}$,
    \[
    \mathsf{EXEC}^{\mathcal{F}^{\text{Arbitrum}}_{\text{layer2-submit}}}_{\mathcal{S}_{\text{Arbitrum-submit}},\mathcal{E}}(k)
    \ \stackrel{c}{\approx}\
    \mathsf{EXEC}^{\mathcal{F}^{\text{Arbitrum}}_{\text{layer2-join}}}_{\mathcal{S}_{\text{Arbitrum-join}},\mathcal{E}}(k),
    \]
    where $\stackrel{c}{\approx}$ denotes computational indistinguishability.
\end{lemma}

\begin{proof}
    Fix an arbitrary PPT environment $\mathcal{E}$. The only difference between the two executions is that $\mathcal{F}^{\text{Arbitrum}}_{\text{layer2-join}}$ routes the Join request (from NET) through $\mathcal{F}^{\text{Arbitrum}}_{\text{join}}$, whereas in $\mathcal{F}^{\text{Arbitrum}}_{\text{layer2-submit}}$ the join output is produced entirely by the simulator. We analyze the three observable components.

    \medskip
    \noindent\textbf{(1) I/O outputs.} The only I/O output affected is $\{\text{Join}, s_{\mathit{init}}\}$. We consider two cases.

    \emph{Case~1: Successful join.} In $\mathsf{EXEC}^{\mathcal{F}^{\text{Arbitrum}}_{\text{layer2-submit}}}_{\mathcal{S}_{\text{Arbitrum-submit}},\mathcal{E}}(k)$, the simulator produces a join output when the simulated $\mathcal{P'}^{\text{Arbitrum}}$ completes the joining procedure (deposit committed, peg-in published by operator and survived $T_{\text{challenge}}$). In $\mathsf{EXEC}^{\mathcal{F}^{\text{Arbitrum}}_{\text{layer2-join}}}_{\mathcal{S}_{\text{Arbitrum-join}},\mathcal{E}}(k)$, the simulator instead sends $\mathit{Attachment}$ to $\mathcal{F}^{\text{Arbitrum}}_{\text{join}}$, which outputs to $\mathcal{E}$ only if all checks pass. A successful join in $\mathcal{P'}^{\text{Arbitrum}}$ implies: $\mathit{TX}_{\mathit{peg\text{-}in}}$ is consistent with $s_{\mathit{init}}$, and $\mathit{TX}_{\mathit{peg\text{-}in}}$ is committed on L1 without fraud proof during $T_{\text{challenge}}$. These are exactly the checks in $\mathcal{F}^{\text{Arbitrum}}_{\text{join}}$ (Steps~1)--3)). Under the assumption of at least one honest verifier (who publishes fraud proofs for invalid publications) and EUF-CMA security, $\mathcal{F}^{\text{Arbitrum}}_{\text{join}}$ accepts whenever $\mathcal{P'}^{\text{Arbitrum}}$ completes.

    \emph{Case~2: Failed join.} If the operator does not publish the peg-in or a fraud proof invalidates it, $\mathcal{P'}^{\text{Arbitrum}}$ does not complete the join. Correspondingly, $\mathcal{S}_{\text{Arbitrum-join}}$ does not send the Join request, so no output is produced.

    In all cases, the I/O outputs are identical.

    \medskip
    \noindent\textbf{(2) On-chain transactions.} The set of transactions published to $\mathcal{F}_{\text{ledger}}$ is fully determined by $\mathcal{P'}^{\text{Arbitrum}}$, which runs the real-world protocol on the input requests forwarded by the same $\mathcal{F}^{\text{Arbitrum}}_{\text{submit}}$ in both executions. Transactions published by honest participants therefore differ across the two worlds only in the signature values, which are computationally indistinguishable by the EUF-CMA security of $\mathcal{F}_{\text{sig}}$. When participants are corrupted, the transactions they publish on L1 are likewise computationally indistinguishable, since the simulator forwards all corrupted-party messages to $\mathcal{P'}^{\text{Arbitrum}}$ and signature randomness is the only source of variation.

    \medskip
    \noindent\textbf{(3) Adversarial leakage.} In both executions, the simulator extracts the plaintext leakage (under our plaintext-leakage assumption) from the ideal functionality and uses it to drive $\mathcal{P'}^{\text{Arbitrum}}$, generating leakage to the adversary according to the corruption status. Since the two executions feed $\mathcal{P'}^{\text{Arbitrum}}$ the same inputs (those accepted by $\mathcal{F}^{\text{Arbitrum}}_{\text{submit}}$) and run the same real-world protocol, $\mathsf{internalState}$ at every honest participant evolves identically across the two executions. Consequently, the leaked $\mathsf{internalState}$ and received messages delivered to corrupted parties are computationally indistinguishable across the two executions.

    \medskip
    Conclusively,
    $\mathsf{EXEC}^{\mathcal{F}^{\text{Arbitrum}}_{\text{layer2-submit}}}_{\mathcal{S}_{\text{Arbitrum-submit}}, \mathcal{E}}(k)\stackrel{c}{\approx}\mathsf{EXEC}^{\mathcal{F}^{\text{Arbitrum}}_{\text{layer2-join}}}_{\mathcal{S}_{\text{Arbitrum-join}}, \mathcal{E}}(k)$.
\end{proof}

Next, we extend $\mathcal{F}^{\text{Arbitrum}}_{\text{layer2-join}}$ with the update subroutine, yielding $\mathcal{F}^{\text{Arbitrum}}_{\text{layer2-update}} = (\mathcal{F}_{\text{client-update}}, \mathcal{F}_{\text{ledger}} \mid \mathcal{F}^{\text{Arbitrum}}_{\text{submit}}, \mathcal{F}^{\text{Arbitrum}}_{\text{join}},\\ \mathcal{F}^{\text{Arbitrum}}_{\text{update}})$.

\vspace{1em}\begin{functionality}{Description of $\mathcal{M}_{\text{client-update}}$ of $\mathcal{F}^{\text{Arbitrum}}_{\text{layer2-update}}$}{

\textbf{Implemented role(s):} \{client-update\}

\noindent \textbf{Main:}

\textbf{recv} \{Submit, $\mathit{request}$\} \textbf{from} I/O:

\begin{enumerate}[itemsep=0.5em]
\item \textbf{send} \{Submit, $\mathit{request}$, $\mathsf{internalState}$\} \textbf{to} $(\pcur, \scur, \mathcal{F}^{\text{Arbitrum}}_{\text{submit}}:\text{submit})$,
\textbf{wait for} \{Submit, $\mathit{response}$\} s.t. $\mathit{response} \in \{\text{true}, \text{false}\}$;
\item \textbf{if} $\mathit{response} =$ true: $\mathsf{requestQueue}.\text{add}(\mathit{request})$;
\textbf{send} $\mathit{request}$ \textbf{to} $\mathcal{S}$ via NET;
\end{enumerate}

\hrule
\vspace{0.5em}

\textbf{recv} \{Join, $\mathit{Attachment}$\} \textbf{from} NET:

\begin{enumerate}[itemsep=0.5em]
\item \textbf{send} \{Join, $\mathit{Attachment}$, $\mathsf{internalState}$\} \textbf{to} $(\pcur, \scur, \mathcal{F}^{\text{Arbitrum}}_{\text{join}}:\text{join})$,
\textbf{wait for} \{Join, $\mathit{response}$\} s.t. $\mathit{response} \in \{\text{true}, \text{false}\}$;
\item \textbf{if} $\mathit{response} =$ true: update $\mathsf{internalState}$ according to $\mathit{Attachment}$;
\textbf{reply} \{Join, $s_{\mathit{init}}$\} via I/O;
\end{enumerate}

\hrule
\vspace{0.5em}

\textbf{recv} \{Update, $\mathit{Attachment}$\} \textbf{from} NET:

\begin{enumerate}[itemsep=0.5em]
\item \textbf{send} \{Update, $\mathit{Attachment}$, $\mathsf{internalState}$\} \textbf{to} $(\pcur, \scur, \mathcal{F}^{\text{Arbitrum}}_{\text{update}}:\text{update})$,
\textbf{wait for} \{Update, $\mathit{response}$, $\mathit{newState}$, $\mathit{executedReq}$\};
\item \textbf{if} $\mathit{response} =$ true: update $\mathsf{internalState}$ with $\mathit{newState}$ and $\mathit{executedReq}$;
\end{enumerate}

\hrule
\vspace{0.5em}

\textbf{recv} any request \textbf{from} I/O:
\begin{enumerate}[itemsep=0.5em]
    \item Forward request to $\mathcal{S}$ through NET;
\end{enumerate}

\hrule
\vspace{0.5em}

\textbf{recv} any message \textbf{from} NET:

\begin{enumerate}
    \item Output the message to $\mathcal{E}$ through I/O;
\end{enumerate}

}\end{functionality}\vspace{.5em}

\begin{functionality}{Description of simulator $\mathcal{S}_{\text{Arbitrum-update}}$}{

The simulator $\mathcal{S}_{\text{Arbitrum-update}}$ behaves identically to $\mathcal{S}_{\text{Arbitrum-join}}$, except for the following additional behavior upon detecting a confirmed state update:

\vspace{0.5em}
\textbf{State update interaction with $\mathcal{F}^{\text{Arbitrum}}_{\text{layer2-update}}$:}

\begin{enumerate}[itemsep=0.5em]

\item $\mathcal{S}_{\text{Arbitrum-update}}$ monitors $\mathcal{P'}^{\text{Arbitrum}}$ for operator batch publications that survive $T_{\text{challenge}}$ without fraud proof. Specifically, it detects when a new batch of transactions and resulting state are confirmed on L1.

\item Upon detecting a confirmed batch, $\mathcal{S}_{\text{Arbitrum-update}}$ prepares $\mathit{Attachment}$ by extracting:
\begin{itemize}
    \item $\mathit{requestBatch}$: the published request batch;
    \item $\mathit{newState}$: the published new state;
    \item $\{\mathit{Agreement}\}$: the set of initiator approvals (abstract agreements rather than cryptographic signatures);
\end{itemize}

\item $\mathcal{S}_{\text{Arbitrum-update}}$ sends $\{\text{Update}, \mathit{Attachment}\}$ to $\mathcal{F}^{\text{Arbitrum}}_{\text{layer2-update}}$ via NET.

\end{enumerate}

}\end{functionality}\vspace{1em}

\begin{lemma}
\label{lem:Arb4}
    For all PPT adversaries $\mathcal{A}$, there exists a PPT simulator $\mathcal{S}_{\text{Arbitrum-update}}$ such that for all PPT environments $\mathcal{E}$ and all security parameters $k\in\mathbb{N}$,
    \[
    \mathsf{EXEC}^{\mathcal{F}^{\text{Arbitrum}}_{\text{layer2-join}}}_{\mathcal{S}_{\text{Arbitrum-join}},\mathcal{E}}(k)
    \ \stackrel{c}{\approx}\
    \mathsf{EXEC}^{\mathcal{F}^{\text{Arbitrum}}_{\text{layer2-update}}}_{\mathcal{S}_{\text{Arbitrum-update}},\mathcal{E}}(k),
    \]
    where $\stackrel{c}{\approx}$ denotes computational indistinguishability.
\end{lemma}

\begin{proof}
    Fix an arbitrary PPT environment $\mathcal{E}$. The only difference between the two games is that $\mathcal{F}^{\text{Arbitrum}}_{\text{layer2-update}}$ routes the Update request (from NET) through the subroutine $\mathcal{F}^{\text{Arbitrum}}_{\text{update}}$. We analyze the three observable components.

    \medskip
    \noindent As defined in both ideal functionalities, the Update request does not directly produce I/O outputs to $\mathcal{E}$, it only modifies $\mathsf{internalState}$. However, unlike Arbitrum where updates are purely off-chain, Arbitrum requires posting transaction batches and fraud proofs to L1, which are observable by $\mathcal{E}$. We must argue that the same on-chain view arises in both executions.

    \medskip
    \noindent\textbf{(1) I/O outputs.} Although the Update request itself produces no I/O output, a divergence in $\mathsf{internalState}$ would cause observable differences in subsequent Read or Settlement outputs. In $\mathcal{P'}^{\text{Arbitrum}}$, a batch is accepted as a valid state update only after: (i)~the operator publishes the batch and resulting state on L1, (ii)~the new state is correctly computed from the request batch, (iii)~no fraud proof is published within $T_{\text{challenge}}$, and (iv)~each request has initiator agreement. Conversely, if the batch contains an incorrect transition, the honest client who acts as verifier publishes a fraud proof, and the batch is invalidated. These are precisely the checks enforced by $\mathcal{F}^{\text{Arbitrum}}_{\text{update}}$ (Steps~1--7). Since $\mathcal{S}_{\text{Arbitrum-update}}$ sends the Update message precisely when a batch survives $T_{\text{challenge}}$ in $\mathcal{P'}^{\text{Arbitrum}}$, $\mathcal{F}^{\text{Arbitrum}}_{\text{update}}$ accepts if and only if the corresponding update was accepted in $\mathcal{P'}^{\text{Arbitrum}}$. The existence of a challenge period and the unforgeability of the EUF-CMA secure signature schemes guarantees that $\mathcal{P'}^{\text{Arbitrum}}$ will not change its own client's $\mathsf{internalState}$ based on incorrect transaction for forging transaction, therefore $\mathsf{internalState}$ changes identically.

    \medskip
    \noindent\textbf{(2) On-chain transactions.} The operator batch publications and verifier fraud proofs are generated by $\mathcal{P'}^{\text{Arbitrum}}$, which runs identically in both simulators. The game hop only interposes $\mathcal{F}^{\text{Arbitrum}}_{\text{update}}$ as a filter on which updates are reflected in $\mathsf{internalState}$. Since the filter accepts exactly those updates that $\mathcal{P'}^{\text{Arbitrum}}$ would accept, and the ideal signature functionality guarantees signatures are indistinguishable, the on-chain transactions are computationally indistinguishable.

    \medskip
    \noindent\textbf{(3) Adversarial leakage.} Both simulators receive the same set of accepted requests, run the same protocol logic, and maintain synchronized corruption status. The leakage is computationally indistinguishable.

    \medskip
    Conclusively,
    $\mathsf{EXEC}^{\mathcal{F}^{\text{Arbitrum}}_{\text{layer2-join}}}_{\mathcal{S}_{\text{Arbitrum-join}}, \mathcal{E}}(k)
    \stackrel{c}{\approx}
    \mathsf{EXEC}^{\mathcal{F}^{\text{Arbitrum}}_{\text{layer2-update}}}_{\mathcal{S}_{\text{Arbitrum-update}}, \mathcal{E}}(k)$.
\end{proof}

Next, we extend $\mathcal{F}^{\text{Arbitrum}}_{\text{layer2-update}}$ with the read subroutine, yielding $\mathcal{F}^{\text{Arbitrum}}_{\text{layer2-read}} = (\mathcal{F}_{\text{client-read}}, \mathcal{F}_{\text{ledger}} \mid \mathcal{F}^{\text{Arbitrum}}_{\text{submit}}, \mathcal{F}^{\text{Arbitrum}}_{\text{join}}, \mathcal{F}^{\text{Arbitrum}}_{\text{update}},\\ \mathcal{F}^{\text{Arbitrum}}_{\text{read}})$.

\vspace{1em}\begin{functionality}{Description of $\mathcal{M}_{\text{client-read}}$ of $\mathcal{F}^{\text{Arbitrum}}_{\text{layer2-read}}$}{

\textbf{Implemented role(s):} \{client-read\}

\noindent \textbf{Main:}

\textbf{recv} \{Submit, $\mathit{request}$\} \textbf{from} I/O:

\begin{enumerate}[itemsep=0.5em]
\item \textbf{send} \{Submit, $\mathit{request}$, $\mathsf{internalState}$\} \textbf{to} $(\pcur, \scur, \mathcal{F}^{\text{Arbitrum}}_{\text{submit}}:\text{submit})$,
\textbf{wait for} \{Submit, $\mathit{response}$\} s.t. $\mathit{response} \in \{\text{true}, \text{false}\}$;
\item \textbf{if} $\mathit{response} =$ true: $\mathsf{requestQueue}.\text{add}(\mathit{request})$;
\textbf{send} $\mathit{request}$ \textbf{to} $\mathcal{S}$ via NET;
\end{enumerate}

\hrule
\vspace{0.5em}

\textbf{recv} \{Join, $\mathit{Attachment}$\} \textbf{from} NET:

\begin{enumerate}[itemsep=0.5em]
\item \textbf{send} \{Join, $\mathit{Attachment}$, $\mathsf{internalState}$\} \textbf{to} $(\pcur, \scur, \mathcal{F}^{\text{Arbitrum}}_{\text{join}}:\text{join})$,
\textbf{wait for} \{Join, $\mathit{response}$\} s.t. $\mathit{response} \in \{\text{true}, \text{false}\}$;
\item \textbf{if} $\mathit{response} =$ true: update $\mathsf{internalState}$ according to $\mathit{Attachment}$;
\textbf{reply} \{Join, $s_{\mathit{init}}$\} via I/O;
\end{enumerate}

\hrule
\vspace{0.5em}

\textbf{recv} \{Update, $\mathit{Attachment}$\} \textbf{from} NET:

\begin{enumerate}[itemsep=0.5em]
\item \textbf{send} \{Update, $\mathit{Attachment}$, $\mathsf{internalState}$\} \textbf{to} $(\pcur, \scur, \mathcal{F}^{\text{Arbitrum}}_{\text{update}}:\text{update})$,
\textbf{wait for} \{Update, $\mathit{response}$, $\mathit{newState}$, $\mathit{executedReq}$\};
\item \textbf{if} $\mathit{response} =$ true: update $\mathsf{internalState}$ with $\mathit{newState}$ and $\mathit{executedReq}$;
\end{enumerate}

\hrule
\vspace{0.5em}

\textbf{recv} \{Read\} \textbf{from} I/O:

\begin{enumerate}[itemsep=0.5em]
\item \textbf{send} \{Read, $\mathsf{internalState}$\} \textbf{to} $(\pcur, \scur, \mathcal{F}^{\text{Arbitrum}}_{\text{read}}:\text{read})$,
\textbf{wait for} $\mathit{ReadResult}$;
\item \textbf{if} $\mathit{ReadResult} \neq \bot$: \textbf{reply} \{Read, $\mathit{ReadResult}$\} via I/O;
\end{enumerate}

\hrule
\vspace{0.5em}

\textbf{recv} any request \textbf{from} I/O:
\begin{enumerate}[itemsep=0.5em]
    \item Forward request to $\mathcal{S}$ through NET;
\end{enumerate}

\hrule
\vspace{0.5em}

\textbf{recv} any message \textbf{from} NET:

\begin{enumerate}
    \item Output the message to $\mathcal{E}$ through I/O;
\end{enumerate}

}\end{functionality}\vspace{.5em}

\begin{functionality}{Description of simulator $\mathcal{S}_{\text{Arbitrum-read}}$}{

The simulator $\mathcal{S}_{\text{Arbitrum-read}}$ behaves the same as $\mathcal{S}_{\text{Arbitrum-update}}$.

}\end{functionality}\vspace{1em}

\begin{lemma}
\label{lem:Arb5}
    For all PPT adversaries $\mathcal{A}$, there exists a PPT simulator $\mathcal{S}_{\text{Arbitrum-read}}$ such that for all PPT environments $\mathcal{E}$ and all security parameters $k\in\mathbb{N}$,
    \[
    \mathsf{EXEC}^{\mathcal{F}^{\text{Arbitrum}}_{\text{layer2-update}}}_{\mathcal{S}_{\text{Arbitrum-update}},\mathcal{E}}(k)
    \ \stackrel{c}{\approx}\
    \mathsf{EXEC}^{\mathcal{F}^{\text{Arbitrum}}_{\text{layer2-read}}}_{\mathcal{S}_{\text{Arbitrum-read}},\mathcal{E}}(k),
    \]
    where $\stackrel{c}{\approx}$ denotes computational indistinguishability.
\end{lemma}

\begin{proof}
    Fix an arbitrary PPT environment $\mathcal{E}$. The only difference is that in $\mathsf{EXEC}^{\mathcal{F}^{\text{Arbitrum}}_{\text{layer2-read}}}_{\mathcal{S}_{\text{Arbitrum-read}},\mathcal{E}}(k)$ the read result is decided by $\mathcal{F}^{\text{Arbitrum}}_{\text{read}}$ based on $\mathsf{internalState}$, whereas in $\mathsf{EXEC}^{\mathcal{F}^{\text{Arbitrum}}_{\text{layer2-update}}}_{\mathcal{S}_{\text{Arbitrum-update}},\mathcal{E}}(k)$ read results are produced entirely by the simulator. We analyze the three observable components.

    \medskip
    \noindent\textbf{(1) I/O outputs.} By Lemma~\ref{lem:Arb4}, the L1 ledger state and $\mathsf{internalState}$ are identical in both executions. In both executions, read results are reconstructed from $\mathcal{F}_{\text{ledger}}$ using the same reconstruction rule: the intersection of $\mathit{L1ReadResult}$ and $\mathsf{internalState}$. Since both rely on the same ledger and the same $\mathsf{internalState}$, the read outputs to $\mathcal{E}$ are distributionally identical.

    \medskip
    \noindent\textbf{(2) On-chain transactions.} Read requests do not publish any transactions to $\mathcal{F}_{\text{ledger}}$. Hence on-chain transactions are identical in both games.

    \medskip
    \noindent\textbf{(3) Adversarial leakage.} The game hop only interposes $\mathcal{F}^{\text{Arbitrum}}_{\text{read}}$ on the Read request, which is an I/O-facing operation. Both simulators maintain synchronized corruption status. The leakage is computationally indistinguishable.

    \medskip
    Conclusively,
    $\mathsf{EXEC}^{\mathcal{F}^{\text{Arbitrum}}_{\text{layer2-update}}}_{\mathcal{S}_{\text{Arbitrum-update}}, \mathcal{E}}(k)
    \stackrel{c}{\approx}
    \mathsf{EXEC}^{\mathcal{F}^{\text{Arbitrum}}_{\text{layer2-read}}}_{\mathcal{S}_{\text{Arbitrum-read}}, \mathcal{E}}(k)$.
\end{proof}

Next, we extend $\mathcal{F}^{\text{Arbitrum}}_{\text{layer2-read}}$ with the settlement subroutine, yielding $\mathcal{F}^{\text{Arbitrum}}_{\text{layer2-settlement}} = (\mathcal{F}_{\text{client-settlement}}, \mathcal{F}_{\text{ledger}} \mid \mathcal{F}^{\text{Arbitrum}}_{\text{submit}}, \mathcal{F}^{\text{Arbitrum}}_{\text{join}},\\ \mathcal{F}^{\text{Arbitrum}}_{\text{update}}, \mathcal{F}^{\text{Arbitrum}}_{\text{read}}, \mathcal{F}^{\text{Arbitrum}}_{\text{settlement}})$.

\vspace{1em}\begin{functionality}{Description of $\mathcal{M}_{\text{client-settlement}}$ of $\mathcal{F}^{\text{Arbitrum}}_{\text{layer2-settlement}}$}{

\textbf{Implemented role(s):} \{client-settlement\}

\noindent \textbf{Main:}

\textbf{recv} \{Submit, $\mathit{request}$\} \textbf{from} I/O:

\begin{enumerate}[itemsep=0.5em]
\item \textbf{send} \{Submit, $\mathit{request}$, $\mathsf{internalState}$\} \textbf{to} $(\pcur, \scur, \mathcal{F}^{\text{Arbitrum}}_{\text{submit}}:\text{submit})$,
\textbf{wait for} \{Submit, $\mathit{response}$\} s.t. $\mathit{response} \in \{\text{true}, \text{false}\}$;
\item \textbf{if} $\mathit{response} =$ true: $\mathsf{requestQueue}.\text{add}(\mathit{request})$;
\textbf{send} $\mathit{request}$ \textbf{to} $\mathcal{S}$ via NET;
\end{enumerate}

\hrule
\vspace{0.5em}

\textbf{recv} \{Join, $\mathit{Attachment}$\} \textbf{from} NET:

\begin{enumerate}[itemsep=0.5em]
\item \textbf{send} \{Join, $\mathit{Attachment}$, $\mathsf{internalState}$\} \textbf{to} $(\pcur, \scur, \mathcal{F}^{\text{Arbitrum}}_{\text{join}}:\text{join})$,
\textbf{wait for} \{Join, $\mathit{response}$\} s.t. $\mathit{response} \in \{\text{true}, \text{false}\}$;
\item \textbf{if} $\mathit{response} =$ true: update $\mathsf{internalState}$ according to $\mathit{Attachment}$;
\textbf{reply} \{Join, $s_{\mathit{init}}$\} via I/O;
\end{enumerate}

\hrule
\vspace{0.5em}

\textbf{recv} \{Update, $\mathit{Attachment}$\} \textbf{from} NET:

\begin{enumerate}[itemsep=0.5em]
\item \textbf{send} \{Update, $\mathit{Attachment}$, $\mathsf{internalState}$\} \textbf{to} $(\pcur, \scur, \mathcal{F}^{\text{Arbitrum}}_{\text{update}}:\text{update})$,
\textbf{wait for} \{Update, $\mathit{response}$, $\mathit{newState}$, $\mathit{executedReq}$\};
\item \textbf{if} $\mathit{response} =$ true: update $\mathsf{internalState}$ with $\mathit{newState}$ and $\mathit{executedReq}$;
\end{enumerate}

\hrule
\vspace{0.5em}

\textbf{recv} \{Read\} \textbf{from} I/O:

\begin{enumerate}[itemsep=0.5em]
\item \textbf{send} \{Read, $\mathsf{internalState}$\} \textbf{to} $(\pcur, \scur, \mathcal{F}^{\text{Arbitrum}}_{\text{read}}:\text{read})$,
\textbf{wait for} $\mathit{ReadResult}$;
\item \textbf{if} $\mathit{ReadResult} \neq \bot$: \textbf{reply} \{Read, $\mathit{ReadResult}$\} via I/O;
\end{enumerate}

\hrule
\vspace{0.5em}

\textbf{recv} \{Settlement, $\mathit{Attachment}$\} \textbf{from} NET:

\begin{enumerate}[itemsep=0.5em]
\item \textbf{send} \{Settlement, $\mathit{Attachment}$, $\mathsf{internalState}$\} \textbf{to} $(\pcur, \scur, \mathcal{F}^{\text{Arbitrum}}_{\text{settlement}}:\\\text{settlement})$,
\textbf{wait for} \{Settlement, $\mathit{response}$, $s_{\mathit{settle}}$\} s.t. $\mathit{response} \in \{\text{true}, \text{false}\}$;
\item \textbf{if} $\mathit{response} =$ true: update $\mathsf{internalState}$;
\textbf{reply} \{Settlement, $s_{\mathit{settle}}$\} via I/O;
\end{enumerate}

\hrule
\vspace{0.5em}

\textbf{recv} any request \textbf{from} I/O:
\begin{enumerate}[itemsep=0.5em]
    \item Forward request to $\mathcal{S}$ through NET;
\end{enumerate}

\hrule
\vspace{0.5em}

\textbf{recv} any message \textbf{from} NET:

\begin{enumerate}
    \item Output the message to $\mathcal{E}$ through I/O;
\end{enumerate}

}\end{functionality}\vspace{.5em}

\begin{functionality}{Description of simulator $\mathcal{S}_{\text{Arbitrum-settlement}}$}{

The simulator $\mathcal{S}_{\text{Arbitrum-settlement}}$ behaves identically to $\mathcal{S}_{\text{Arbitrum-read}}$, except for the following additional behavior upon detecting a completed settlement:

\vspace{0.5em}
\textbf{Settlement interaction with $\mathcal{F}^{\text{Arbitrum}}_{\text{layer2-settlement}}$:}

\begin{enumerate}[itemsep=0.5em]

\item $\mathcal{S}_{\text{Arbitrum-settlement}}$ monitors $\mathcal{P'}^{\text{Arbitrum}}$. When it detects that a client entity is about to produce a successful settlement output inside the simulation (either via \texttt{collaborate} or \texttt{escape-hatch}), it intercepts this output.

\item $\mathcal{S}_{\text{Arbitrum-settlement}}$ prepares $\mathit{Attachment}$ by extracting:
\begin{itemize}
    \item $\mathit{TX}_{\mathit{peg\text{-}out}}$: the peg-out transaction committed on $\mathcal{F}_{\text{ledger}}$ in $\mathcal{P'}^{\text{Arbitrum}}$;
\end{itemize}

\item $\mathcal{S}_{\text{Arbitrum-settlement}}$ sends $\{\text{Settlement}, \mathit{Attachment}\}$ to $\mathcal{F}^{\text{Arbitrum}}_{\text{layer2-settlement}}$ via NET.

\end{enumerate}

}\end{functionality}\vspace{1em}

\begin{lemma}
\label{lem:Arb6}
    For all PPT adversaries $\mathcal{A}$, there exists a PPT simulator $\mathcal{S}_{\text{Arbitrum-settlement}}$ such that for all PPT environments $\mathcal{E}$ and all security parameters $k\in\mathbb{N}$,
    \[
    \mathsf{EXEC}^{\mathcal{F}^{\text{Arbitrum}}_{\text{layer2-read}}}_{\mathcal{S}_{\text{Arbitrum-read}},\mathcal{E}}(k)
    \ \stackrel{c}{\approx}\
    \mathsf{EXEC}^{\mathcal{F}^{\text{Arbitrum}}_{\text{layer2-settlement}}}_{\mathcal{S}_{\text{Arbitrum-settlement}},\mathcal{E}}(k),
    \]
    where $\stackrel{c}{\approx}$ denotes computational indistinguishability.
\end{lemma}

\begin{proof}
    Fix an arbitrary PPT environment $\mathcal{E}$. The only difference between the two executions is that $\mathcal{F}^{\text{Arbitrum}}_{\text{layer2-settlement}}$ routes the Settlement request (from NET) through $\mathcal{F}^{\text{Arbitrum}}_{\text{settlement}}$. We analyze the three observable components, noting that settlement in Arbitrum can proceed via two paths: \texttt{collaborate} (through the operator) or \texttt{escape-hatch} (directly to L1).

    \medskip
    \noindent\textbf{(1) I/O outputs.} The I/O output affected is $\{\text{Settlement}, \mathsf{onchainState}\}$.

    \emph{Case~1: Successful settlement (either path).} In $\mathcal{P'}^{\text{Arbitrum}}$, a settlement succeeds when $\mathit{TX}_{\mathit{peg\text{-}out}}$ is committed on L1 without fraud proof during $T_{\text{challenge}}$. The simulator $\mathcal{S}_{\text{Arbitrum-settlement}}$ triggers $\mathcal{F}^{\text{Arbitrum}}_{\text{settlement}}$ precisely when this occurs. The checks in $\mathcal{F}^{\text{Arbitrum}}_{\text{settlement}}$: (1)~$s_{\mathit{settle}}$ is consistent with $\mathit{TX}_{\mathit{peg\text{-}out}}$, and (2)~$\mathit{TX}_{\mathit{peg\text{-}out}}$ is committed on L1 without fraud proof, are exactly the conditions that hold when the simulated settlement succeeds. Under the assumption of at least one honest verifier, L1 liveness and safety, and EUF-CMA signatures, $\mathcal{F}^{\text{Arbitrum}}_{\text{settlement}}$ accepts whenever $\mathcal{P'}^{\text{Arbitrum}}$ completes. Hence the output is identical.

    \emph{Case~2: Failed settlement.} If a fraud proof invalidates the peg-out or L1 conditions are not met, no output is produced in either execution.

    \medskip
    \noindent\textbf{(2) On-chain transactions.} The peg-out transaction published to $\mathcal{F}_{\text{ledger}}$ (whether via operator or escape-hatch) is generated by $\mathcal{P'}^{\text{Arbitrum}}$, which runs identically in both simulators. The game hop only affects when the ideal functionality generates I/O output. By EUF-CMA security, the distributions are computationally indistinguishable.

    \medskip
    \noindent\textbf{(3) Adversarial leakage.} Both simulators maintain synchronized corruption status and execute the same protocol logic. The game hop interposes $\mathcal{F}^{\text{Arbitrum}}_{\text{settlement}}$ between the simulator and the I/O output but does not alter the simulator's internal execution. The leakage is identical.

    \medskip
    Conclusively,
    $\mathsf{EXEC}^{\mathcal{F}^{\text{Arbitrum}}_{\text{layer2-read}}}_{\mathcal{S}_{\text{Arbitrum-read}}, \mathcal{E}}(k)\stackrel{c}{\approx}\mathsf{EXEC}^{\mathcal{F}^{\text{Arbitrum}}_{\text{layer2-settlement}}}_{\mathcal{S}_{\text{Arbitrum-settlement}}, \mathcal{E}}(k)$.
\end{proof}

As a final step, we extend $\mathcal{F}^{\text{Arbitrum}}_{\text{layer2-settlement}}$ with the subroutine $\mathcal{F}^{\text{Arbitrum}}_{\text{updRnd}}$ to reach $\mathcal{F}^{\text{Arbitrum}}_{\text{layer2}} = (\mathcal{F}_{\text{client}}, \mathcal{F}_{\text{ledger}} \mid \mathcal{F}^{\text{Arbitrum}}_{\text{submit}}, \mathcal{F}^{\text{Arbitrum}}_{\text{join}},\\ \mathcal{F}^{\text{Arbitrum}}_{\text{update}}, \mathcal{F}^{\text{Arbitrum}}_{\text{read}}, \mathcal{F}^{\text{Arbitrum}}_{\text{settlement}}, \mathcal{F}^{\text{Arbitrum}}_{\text{updRnd}})$.

\vspace{1em}\begin{functionality}{Description of $\mathcal{M}_{\text{client}}$ of $\mathcal{F}^{\text{Arbitrum}}_{\text{layer2}}$}{

\textbf{Implemented role(s):} \{client\}

\noindent \textbf{Main:}

\vspace{0.5em}

\textbf{recv} \{Submit, $\mathit{request}$\} \textbf{from} I/O:

\begin{enumerate}[itemsep=0.5em]
\item \textbf{send} \{Submit, $\mathit{request}$, $\mathsf{internalState}$\} \textbf{to} $(\pcur, \scur, \mathcal{F}^{\text{Arbitrum}}_{\text{submit}}:\text{submit})$,
\textbf{wait for} \{Submit, $\mathit{response}$\} s.t. $\mathit{response} \in \{\text{true}, \text{false}\}$;
\item \textbf{if} $\mathit{response} =$ true: $\mathsf{requestQueue}.\text{add}(\mathit{request})$;
\textbf{send} $\mathit{request}$ \textbf{to} $\mathcal{S}$ via NET;
\end{enumerate}

\hrule
\vspace{0.5em}

\textbf{recv} \{Join, $\mathit{Attachment}$\} \textbf{from} NET:

\begin{enumerate}[itemsep=0.5em]
\item \textbf{send} \{Join, $\mathit{Attachment}$, $\mathsf{internalState}$\} \textbf{to} $(\pcur, \scur, \mathcal{F}^{\text{Arbitrum}}_{\text{join}}:\text{join})$,
\textbf{wait for} \{Join, $\mathit{response}$\} s.t. $\mathit{response} \in \{\text{true}, \text{false}\}$;
\item \textbf{if} $\mathit{response} =$ true: update $\mathsf{internalState}$ according to $\mathit{Attachment}$;
\textbf{reply} \{Join, $s_{\mathit{init}}$\} via I/O;
\end{enumerate}

\hrule
\vspace{0.5em}

\textbf{recv} \{Update, $\mathit{Attachment}$\} \textbf{from} NET:

\begin{enumerate}[itemsep=0.5em]
\item \textbf{send} \{Update, $\mathit{Attachment}$, $\mathsf{internalState}$\} \textbf{to} $(\pcur, \scur, \mathcal{F}^{\text{Arbitrum}}_{\text{update}}:\text{update})$,
\textbf{wait for} \{Update, $\mathit{response}$, $\mathit{newState}$, $\mathit{executedReq}$\};
\item \textbf{if} $\mathit{response} =$ true: update $\mathsf{internalState}$ with $\mathit{newState}$ and $\mathit{executedReq}$;
\end{enumerate}

\hrule
\vspace{0.5em}

\textbf{recv} \{Read\} \textbf{from} I/O:

\begin{enumerate}[itemsep=0.5em]
\item \textbf{send} \{Read, $\mathsf{internalState}$\} \textbf{to} $(\pcur, \scur, \mathcal{F}^{\text{Arbitrum}}_{\text{read}}:\text{read})$,
\textbf{wait for} $\mathit{ReadResult}$;
\item \textbf{if} $\mathit{ReadResult} \neq \bot$: \textbf{reply} \{Read, $\mathit{ReadResult}$\} via I/O;
\end{enumerate}

\hrule
\vspace{0.5em}

\textbf{recv} \{Settlement, $\mathit{Attachment}$\} \textbf{from} NET:

\begin{enumerate}[itemsep=0.5em]
\item \textbf{send} \{Settlement, $\mathit{Attachment}$, $\mathsf{internalState}$\} \textbf{to} $(\pcur, \scur, \mathcal{F}^{\text{Arbitrum}}_{\text{settlement}}:\\\text{settlement})$,
\textbf{wait for} \{Settlement, $\mathit{response}$, $s_{\mathit{settle}}$\} s.t. $\mathit{response} \in \{\text{true}, \text{false}\}$;
\item \textbf{if} $\mathit{response} =$ true: update $\mathsf{internalState}$;
\textbf{reply} \{Settlement, $s_{\mathit{settle}}$\} via I/O;
\end{enumerate}

\hrule
\vspace{0.5em}

\textbf{recv} \{UpdateRound\} \textbf{from} NET:

\begin{enumerate}[itemsep=0.5em]
\item \textbf{send} \{UpdateRound, $\mathsf{internalState}$\} \textbf{to} $(\pcur, \scur, \mathcal{F}^{\text{Arbitrum}}_{\text{updRnd}}:\text{updRnd})$,
\textbf{wait for} \{UpdateRound, $\mathit{response}$\} s.t. $\mathit{response} \in \{\text{true}, \text{false}\}$;
\item \textbf{if} $\mathit{response} =$ true: $\mathsf{round} \leftarrow \mathsf{round} + 1$;
\textbf{reply} \{UpdateRound, $\mathit{response}$\} via NET;
\end{enumerate}

\hrule
\vspace{0.5em}

\textbf{recv} \{GetCurRound\} \textbf{from} I/O or NET:
\begin{enumerate}[itemsep=0.5em]
    \item \textbf{reply} \{GetCurRound, $\mathsf{round}$\};
\end{enumerate}

\hrule
\vspace{0.5em}

\textbf{recv} \{ReadL1\} \textbf{from} I/O or NET:
\begin{enumerate}[itemsep=0.5em]
    \item \textbf{send} \{Read\} \textbf{to} $(\pcur, \scur, \mathcal{F}_{\text{ledger}}: \text{client}_{\text{L1}})$;
    \textbf{wait for} $\mathit{L1ReadResult}$;
    \item \textbf{reply} \{ReadL1, $\mathit{L1ReadResult}$\} via I/O;
\end{enumerate}

}\end{functionality}\vspace{.5em}

\begin{functionality}{Description of simulator $\mathcal{S}_{\text{Arbitrum-updRnd}}$}{

The simulator $\mathcal{S}_{\text{Arbitrum-updRnd}}$ behaves identically to $\mathcal{S}_{\text{Arbitrum-settlement}}$, except for the following additional behavior when handling round-update requests:

\vspace{0.5em}
\textbf{Round update interaction with $\mathcal{F}^{\text{Arbitrum}}_{\text{layer2}}$:}

\begin{enumerate}[itemsep=0.5em]

\item Whenever $\mathcal{S}_{\text{Arbitrum-updRnd}}$ receives a round-update instruction from $\mathcal{A}$ (i.e., $\mathcal{A}$ advances the clock in the simulated $\mathcal{P'}^{\text{Arbitrum}}$), $\mathcal{S}_{\text{Arbitrum-updRnd}}$ simultaneously sends $\{\text{UpdateRound}\}$ to $\mathcal{F}^{\text{Arbitrum}}_{\text{layer2}}$ via NET.

\end{enumerate}

}\end{functionality}\vspace{1em}

\begin{lemma}
\label{lem:Arb7}
    For all PPT adversaries $\mathcal{A}$, there exists a PPT simulator $\mathcal{S}_{\text{Arbitrum-updRnd}}$ such that for all PPT environments $\mathcal{E}$ and all security parameters $k\in\mathbb{N}$,
    \[
    \mathsf{EXEC}^{\mathcal{F}^{\text{Arbitrum}}_{\text{layer2-settlement}}}_{\mathcal{S}_{\text{Arbitrum-settlement}},\mathcal{E}}(k)
    \ \stackrel{c}{\approx}\
    \mathsf{EXEC}^{\mathcal{F}^{\text{Arbitrum}}_{\text{layer2}}}_{\mathcal{S}_{\text{Arbitrum-updRnd}},\mathcal{E}}(k),
    \]
    where $\stackrel{c}{\approx}$ denotes computational indistinguishability.
\end{lemma}

\begin{proof}
    Fix an arbitrary PPT environment $\mathcal{E}$. The only difference between the two games is that $\mathcal{F}^{\text{Arbitrum}}_{\text{layer2}}$ routes the UpdateRound request (from NET) through $\mathcal{F}^{\text{Arbitrum}}_{\text{updRnd}}$. We analyze the three observable components.

    \medskip
    \noindent\textbf{(1) I/O outputs.} The only I/O output affected is $\{\text{GetCurRound}, \mathsf{round}\}$. In $\mathsf{EXEC}^{\mathcal{F}^{\text{Arbitrum}}_{\text{layer2-settlement}}}_{\mathcal{S}_{\text{Arbitrum-settlement}},\mathcal{E}}(k)$, $\mathcal{S}_{\text{Arbitrum-settlement}}$ maintains the round counter inside the simulated $\mathcal{P'}^{\text{Arbitrum}}$, whose clock is advanced by $\mathcal{A}$. In $\mathsf{EXEC}^{\mathcal{F}^{\text{Arbitrum}}_{\text{layer2}}}_{\mathcal{S}_{\text{Arbitrum-updRnd}},\mathcal{E}}(k)$, $\mathcal{S}_{\text{Arbitrum-updRnd}}$ additionally forwards each round-update request from $\mathcal{A}$ to $\mathcal{F}^{\text{Arbitrum}}_{\text{updRnd}}$, which checks that a new transaction batch has been published on L1 within $T_{\text{period}}$ rounds before accepting. In $\mathcal{P'}^{\text{Arbitrum}}$, the operator is triggered by $\mathcal{F}^{\text{Arbitrum}}_{\text{com}}$ to publish a batch every $T_{\text{period}}$ rounds via the \{UpdateRequest\} request. As long as at least one honest operator publishes periodically and $\mathcal{F}_{\text{ledger}}$ guarantees L1 liveness, the L1 ledger reflects a new batch within every $T_{\text{period}}$-round window. Therefore, the check in $\mathcal{F}^{\text{Arbitrum}}_{\text{updRnd}}$ passes at exactly the same times as the round advances in $\mathcal{P'}^{\text{Arbitrum}}$, and the round counters evolve identically. All other I/O outputs (Join, Read, Settlement) are unaffected by this game hop.

    \medskip
    \noindent\textbf{(2) On-chain transactions.} The UpdateRound request itself does not publish any new transactions to $\mathcal{F}_{\text{ledger}}$; it only reads L1 to verify that a recent batch exists. The batch publications are generated by the operator's \{UpdateRequest\} request in $\mathcal{P'}^{\text{Arbitrum}}$, which runs identically in both simulators. Hence on-chain transactions are identical in both games.

    \medskip
    \noindent\textbf{(3) Adversarial leakage.} The game hop only interposes $\mathcal{F}^{\text{Arbitrum}}_{\text{updRnd}}$ on the UpdateRound request. Since $\mathcal{F}^{\text{Arbitrum}}_{\text{updRnd}}$ only reads L1 and does not generate any network messages to corrupted parties, the simulator's internal execution of $\mathcal{P'}^{\text{Arbitrum}}$ and its interaction with corrupted parties are unaffected. The leakage delivered to $\mathcal{A}$ is identical.

    \medskip
    Conclusively,
    $\mathsf{EXEC}^{\mathcal{F}^{\text{Arbitrum}}_{\text{layer2-settlement}}}_{\mathcal{S}_{\text{Arbitrum-settlement}}, \mathcal{E}}(k)\stackrel{c}{\approx}\mathsf{EXEC}^{\mathcal{F}^{\text{Arbitrum}}_{\text{layer2}}}_{\mathcal{S}_{\text{Arbitrum-updRnd}}, \mathcal{E}}(k)$.
\end{proof}

\ThmrealizeArbitrum*
\begin{proof}
    Let $\mathcal{A}$ be any PPT adversary and let $\mathcal{E}$ be any PPT environment. By Lemmas~\ref{lem:Arb1}--\ref{lem:Arb7}, the sequence of game hops yields a chain of computationally indistinguishable execution ensembles, each adjacent pair differing only in how one subroutine is implemented:
    \[
    \mathsf{EXEC}^{\mathcal{P}^{\text{Arbitrum}}}_{\mathcal{A},\mathcal{E}}
    \ \stackrel{c}{\approx}\ 
    \mathsf{EXEC}^{\mathcal{F}^{\text{Arbitrum}}_{\text{dummy}}}_{\mathcal{S}_{\text{Arbitrum}},\mathcal{E}}
    \ \stackrel{c}{\approx}\ 
    \cdots
    \ \stackrel{c}{\approx}\ 
    \mathsf{EXEC}^{\mathcal{F}^{\text{Arbitrum}}_{\text{layer2}}}_{\mathcal{S}_{\text{Arbitrum-updRnd}},\mathcal{E}}.
    \]
    By transitivity of computational indistinguishability, we obtain
    $
    \mathsf{EXEC}^{\mathcal{P}^{\text{Arbitrum}}}_{\mathcal{A},\mathcal{E}}
    \stackrel{c}{\approx}
    \mathsf{EXEC}^{\mathcal{F}^{\text{Arbitrum}}_{\text{layer2}}}_{\mathcal{S}_{\text{Arbitrum-updRnd}},\mathcal{E}}$, which proves that $\mathcal{P}^{\text{Arbitrum}}$ iUC-realizes $\mathcal{F}^{\text{Arbitrum}}_{\text{layer2}}$.
\end{proof}

\section{Missing Proofs of Comparative Analysis}
\label{apdx:proof}

\ThmlivenessPS*

\begin{proof}
    To begin, observe from the ideal functionality subroutine $\mathcal{F}_{\text{update}}$ that in PCF and sidechain-based protocols, all execution occurs off-chain without the L1 interaction requirement. Consequently, the liveness latency is determined solely by the off-chain environment and is always \( T_{L_2} \).

Suppose the protocol simultaneously realizes \( f_{\text{OP}} \)-\emph{safety} and 
\[
\left\{ \left\lfloor \frac{n_{\text{OP}} - (f_{\text{OP}} + 1)}{2} \right\rfloor + 1, \; T_{L_2} \right\}\text{-liveness}.
\]
This implies that to execute a request, it suffices to obtain agreement from only
\[
n_{\text{OP}} - \left\lfloor \frac{n_{\text{OP}} - (f_{\text{OP}} + 1)}{2} \right\rfloor - 1
\]
operators.

Assume the adversary corrupts \( f_{\text{OP}} \) operators. Then, by coordinating with up to half of the remaining honest operators, the adversary can potentially form a coalition of size
\[
f_{\text{OP}} + \left\lfloor \frac{n_{\text{OP}} - f_{\text{OP}}}{2} \right\rfloor.
\]

We now show that:
\[
f_{\text{OP}} + \left\lfloor \frac{n_{\text{OP}} - f_{\text{OP}}}{2} \right\rfloor \geq n_{\text{OP}} - \left\lfloor \frac{n_{\text{OP}} - (f_{\text{OP}} + 1)}{2} \right\rfloor - 1.
\]

Let \( x := n_{\text{OP}} - f_{\text{OP}} > 0 \). Then the inequality becomes:
\[
\left\lfloor \frac{x}{2} \right\rfloor + \left\lfloor \frac{x - 1}{2} \right\rfloor \geq x - 1.
\]

This identity holds for all integers \( x > 0 \), with equality:
\[
\left\lfloor \frac{x}{2} \right\rfloor + \left\lfloor \frac{x - 1}{2} \right\rfloor = x - 1.
\]

Therefore, the adversary is capable of gathering sufficient votes to satisfy the liveness condition. This implies that the adversary could cause inconsistent state transitions, contradicting the assumption that the protocol realizes \( f_{\text{OP}} \)-safety. 

Thus, there exists a fundamental trade-off between liveness and safety under adversarial corruption thresholds.
\end{proof}

\ThmlivenessA*
\begin{proof}
We prove this by contradiction. Assume that a secure rollup protocol realizes \emph{$(f_{L_2}+f_{L_1})$-safety} while also achieving a stronger liveness tolerance than stated that $\{n_{OP} + f_{L_1},\, T_{L_1}\}$-liveness even \emph{without} the self-submit path. This would imply that even if all operators are corrupted, the protocol still guarantees that a request from an honest client is eventually executed within $T_{L_1}$ via the collaborative path alone.

However, the collaborative path requires at least one operator to include the request in a batch and publish it on the L1 blockchain. If all operators are corrupted, they can simply ignore the honest client's request indefinitely, preventing it from being published or executed. This violates liveness on the collaborative path and contradicts the assumption. Hence, $\{n_{OP} + f_{L_1},\, T_{L_1}\}$-liveness is unachievable when liveness depends solely on operator cooperation.

As for the \emph{self-submit} situation: when an honest client detects censorship by the operators, the client can bypass them and submit the transaction directly to L1 through $\mathcal{F}_{\text{ledger}}$. By L1 liveness, the transaction is committed on L1 within $T_{L_1}$, and the honest verifier ensures only correct state transitions survive the dispute window $T_{\text{challenge}}$. Consequently, the self-submit path yields $\{f_{L_1},\, T_{L_1}\}$-liveness and depends only on L1 liveness and the honest verifier, tolerating full operator corruption.

In summary, the best achievable liveness for state-update requests is $\{(n_{OP}-1) + f_{L_1},\, T_{L_1} + T_{L_2}\}$ on the collaborative path (requiring at least one honest operator) and $\{f_{L_1},\, T_{L_1}\}$ on the self-submit path requiring only L1 liveness. 

\end{proof}

\ThmdataPS*
\begin{proof}
To start with, since all the $n$ participants are required to publish their joining transactions on the L1 for the protocol to begin, the L1 storage efficiency must be at least $\Omega(n_p)$.

For the L2 storage efficiency, we prove it by contradiction. Assume that \emph{liveness} still holds even if fewer than $m$ executed requests are stored off-chain. In that case, there must exist at least one request that cannot be retrieved through the read interface at a certain round based solely on the L2 storage. Since PCF and sidechain protocols do not retrieve execution data from the L1, the missing request would violate the definition of \emph{liveness}. Thus, the assumption leads to a contradiction, and storing fewer than $m$ requests off-chain breaks \emph{liveness}.
\end{proof}

\ThmdataA*
\begin{proof}
    The L1 storage efficiency consists of two parts. The first part is the $\Omega(n_p)$ storage requirement for the participants’ joining requests, all of which are published on the L1.
    
    For the second part, assume that \emph{liveness} still holds for state update requests in a rollup protocol even if fewer than $m$ requests are stored on the L1. Since the rollup protocol’s read result is derived directly from the L1, this would imply that at least one request cannot be retrieved by a read query at a certain round. This contradicts the definition of \emph{liveness}, and therefore the assumption must be false. Thus, storing fewer than $m$ requests on the L1 would break \emph{liveness}, and the L1 storage efficiency must also be at least $\Omega(m) + \Omega(n_p)$ in total.
\end{proof}

\section{\xr Protocol Design and Security Analysis}
\label{apdx:cross}

\subsection{$\mathcal{F}_{\text{ledger}}$ instantiation for \xr}
\label{apd:FRollLedger}

\subsubsection*{Submission functionality $\mathcal{F}_{\text{submit}_{\text{L1}}}$}

The submission subroutine accepts a request to submit a transaction to L1 if and only if the transaction conforms to one of the FRoll transaction types. For \xr, the following transactions are accepted by $\mathcal{F}^{\text{FRoll}}_{\text{submit}_{\text{L1}}}$: client deposit transactions ($\mathit{TX}_{\mathit{deposit}}$); client-signed self-submitted transactions ($\mathit{TX}_{\mathit{self}}$), covering both regular and fast-finality variants; operator-published transaction batches with claimed post-states ($\mathit{batch}, \mathit{resultState}$); peg-out transactions ($\mathit{TX}_{\mathit{peg\text{-}out}}$) for  escape-hatch settlement; watchtower certificate transactions ($\mathit{TX}_{\mathit{WT}} = \{\mathit{TX}_{\mathit{fast}}, \sigma, \sigma_W\}$) carrying a watchtower signature on a fast-finality transaction; and fraud-proof transactions $\mathit{TX}_{\mathit{fraud}}$ published by clients (verifiers).

\subsubsection*{Update functionality $\mathcal{F}_{\text{update}_{\text{L1}}}$}
The update subroutine extends the base $\mathcal{F}_{\text{update}_{\text{L1}}}$ by enforcing FRoll-specific commitment rule across two finality paths: optimistic (challenge-period-based) and fast (watchtower-quorum-based). The two paths share the same separation between L1 inclusion and L1 finalization: a transaction's appearance in $\{\mathit{TX}\}_{\text{L1}}$ records that L1 has accepted the transaction in BLOB without execution, thus does not by itself change$\mathit{State}_{\text{L1}}$

\begin{itemize}
    \item \textbf{Operator-published state transitions.} When an operator publishes the batch transaction $\mathit{TX}_{\mathit{batch}} = (\mathit{batch}, \mathit{resultState})$ to L1, $\mathit{TX}_{\mathit{batch}}$ carrying the operator's claimed L2 execution state $\mathit{resultState}$ is committed in $\{\mathit{TX}\}_{\text{L1}}$ ( representing stored in a blockchain BLOB), only when an L1-verifiable predicate over $\{\mathit{TX}\}_{\text{L1}}$ holds: $\mathit{TX}_{\mathit{batch}}$ has survived the challenge period $T_{\text{challenge}}$ without a valid fraud-proof transaction $\mathit{TX}_{\mathit{fraud}}$ being committed against it. L1 maintainers do not execute the L2 transition; the L1-finalized state $\mathit{State}_{\text{L1}}$ is advanced to $\mathit{resultState}$ 

    \item \textbf{Fast-finality state transitions.} A fast-finality transaction $\mathit{TX}_{\mathit{fast}}$ along with its execution result  is committed in $\{\mathit{TX}\}_{\text{L1}}$ without waiting for the full challenge period, once at least $f{+}1$ distinct watchtower certificate transactions $\mathit{TX}_{\mathit{WT}}^{(j)} = (\mathit{TX}_{\mathit{fast}}, \sigma, \sigma_{W_j})$ for the same $\mathit{TX}_{\mathit{fast}}$ are committed on L1. The $\mathit{State}_{\text{L1}}$ will not change since it requires no L1 execution.

    \item \textbf{Settlement.} A peg-out transaction $\mathit{TX}_{\mathit{peg\text{-}out}}$ (collaborative or escape-hatch), once submitted, enters $\{\mathit{TX}\}_{\text{L1}}$ as a record without immediately advancing $\mathit{State}_{\text{L1}}$. After $\mathit{TX}_{\mathit{peg\text{-}out}}$ has survived $T_{\text{challenge}}$ without a fraud proof, $\mathit{State}_{\text{L1}}$ is advanced to reflect the post-settlement state carried in $\mathit{TX}_{\mathit{peg\text{-}out}}$.
\end{itemize}

\subsubsection*{Preservation of L1 safety and liveness.}

The instantiation preserves the original security guarantees of $\mathcal{F}_{\text{ledger}}$. FRoll-specific logic is purely additive: the new transaction types are accepted as valid L1 payloads, and the challenge-period, watchtower-quorum, and fraud-proof rules are deterministic predicates over transactions already committed on L1. No existing L1 transaction is rejected or reordered by the instantiation, so the underlying ledger's self-consistency and view-consistency (L1 safety) are inherited without change. Each FRoll transaction is submitted through the unchanged $\text{SubmitL1}$ interface and included in $\{\mathit{TX}\}_{\text{L1}}$ under the same conditions as any other L1 transaction, so the liveness bound carries over: any honest submission is reflected on L1 within $T_{L_1}$. The optimistic challenge-period rule and the $f{+}1$-watchtower quorum rule are both enforced on already-final L1 data, so the resulting state commitments cannot diverge across honest views so long as L1 itself does not. Consequently, $\mathcal{F}^{\text{FRoll}}_{\text{ledger}}$ still realizes $(f_{L_1})$-safety and liveness based on the specification above.

\subsection{\xr Real Protocol}
\label{apd:XRollReal}
 
We formally define the real-world protocol as
$\mathcal{P}^{\text{FRoll}} := (\mathcal{P}^{\text{FRoll}}_{\text{client}}: \text{client}, \mathcal{F}_{\text{ledger}} \mid \mathcal{P}^{\text{FRoll}}_{\text{operator}}, \mathcal{P}^{\text{FRoll}}_{\text{client}}: \text{verifier},\mathcal{P}^{\text{FRoll}}_{\text{watchtower}}, \mathcal{F}_{\text{sig}}, \mathcal{F}^{\text{FRoll}}_{\text{com}})$.
The protocol consists of three machine types: the client machine, the operator machine, and the watchtower machine. Each client machine implements two roles, $\{\text{client}, \text{verifier}\}$: the client role is the public interface used by the environment for joining, submitting (regular and fast-finality), settling, and reading; the verifier role is a private, internal role that monitors L1 publications from operators and posts fraud proofs when invalid state transitions are detected. The verifier role is not addressable from the environment and produces no I/O outputs, only its L1 publications (fraud-proof transactions) are observable, on the same footing as any other on-chain transaction. We assume at least one honest operator, at least one honest client (so that at least one honest verifier exists), and $n = 2f{+}1$ watchtowers of which at most~$f$ may be corrupted. The protocol structure is shown in Figure~\ref{fig:XRoll}.
 
\begin{figure}
    \centering
    \includegraphics[width=\linewidth]{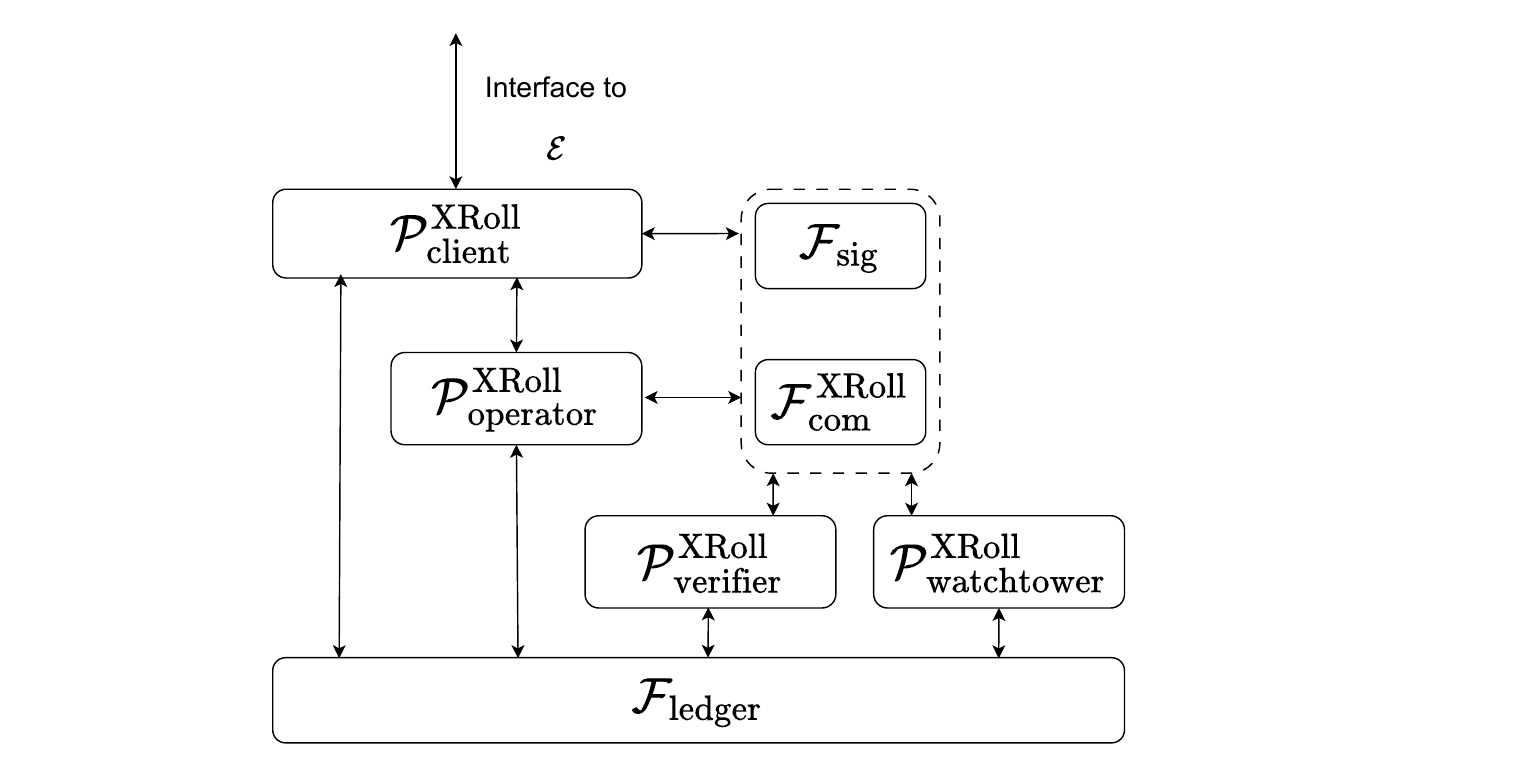}
    \caption{The \xr protocol}
    \label{fig:XRoll}
\end{figure}

\subsubsection{Client Protocol}
 
The client machine $\mathcal{M}^{\text{FRoll}}_{\text{client}}$ implements both the public client role and the private verifier role. As a client, it submits regular and fast-finality transactions, joins, settles, and reads executed requests and states. In addition to submitting regular transactions to the operator, a client may submit selected transactions for fast finality via \texttt{SubmitFast}; such transactions are deemed valid upon receiving on-chain a quorum certificate with $f{+}1$ watchtower signatures. As a verifier, it monitors operator batch publications on L1 and posts fraud proofs when an invalid state transition is detected; the verifier role is internal, therefore it has no I/O interface to the environment.
 
\vspace{1em}\begin{functionality}{Description of protocol $\mathcal{P}^{\text{FRoll}}_{\text{client}} = (\text{client}, \text{verifier})$}{

\textbf{Participating roles:} \{client, verifier\}

\noindent \textbf{Corruption model:} dynamic corruption of the client role, at least one honest; the verifier role is private and follows the corruption status of the client machine that hosts it.

\noindent \textbf{Protocol parameters:}
\begin{itemize}
    \item $T_{\text{challenge}}$: challenge period duration
    \item $\mathsf{Val}(\mathit{TX})$: transaction validity checking algorithm. Given $\mathit{TX} = (\mathit{sender}, \mathit{receiver},\\ \mathit{value}, \mathit{data})$, $\mathsf{Val}(\mathit{TX})$ returns True if and only if:
    \begin{itemize}
        \item $\mathit{TX}$ is well-formed: $\mathit{sender}, \mathit{receiver} \in \{0,1\}^{*}$, $\mathit{value} \in \mathbb{N}_{\geq 0}$; \cmt{Correct transaction format}
        \item $\mathit{sender} \in \mathsf{identities}$; \cmt{Sender is a registered participant}
        \item Let $\mathit{bal}(\mathit{sender})$ denote the balance of $\mathit{sender}$ derived from $\mathsf{stateList}$. Then $\mathit{value} \leq \mathit{bal}(\mathit{sender})$; \cmt{Sender has sufficient balance}
        \item $\mathit{TX}$ does not conflict with any $\mathit{TX}' \in \mathsf{executedRequest}$; \cmt{No double-spending}
    \end{itemize}
    \item $\mathsf{Check}(\mathit{batch}, \mathit{resultState}, \mathit{prevState})$: batch validity checking algorithm used by the verifier role. Returns true if and only if:
    \begin{itemize}
        \item $\forall\, \mathit{TX} \in \mathit{batch}$: $\mathsf{Val}(\mathit{TX}) = \text{true}$; \cmt{Each transaction is well-formed}
        \item $\mathit{resultState}$ is the correct output of sequentially executing all $\mathit{TX} \in \mathit{batch}$ starting from $\mathit{prevState}$; \cmt{State transition is correct}
        \end{itemize}
        \item $\mathsf{GenF}(\mathit{batch}, \mathit{resultState}, \mathit{prevState})$: fraud-proof generation algorithm. Returns a fraud proof transaction $TX_{\textit{fraud}}$ if $\mathsf{Check}$ fails, otherwise $\bot$.
    
\end{itemize}

}\end{functionality}\vspace{.5em}

\begin{functionality}{Description of $\mathcal{M}^{\text{FRoll}}_{\text{client}}$}{

\textbf{Implemented role(s):} \{client, verifier\} \cmt{verifier is a private internal role}

\noindent \textbf{Subroutines:} $\mathcal{F}_{\text{sig}}$: \{signer, verifier\}, $\mathcal{F}^{\text{FRoll}}_{\text{com}}$: com, $\mathcal{F}_{\text{ledger}}$: client\textsubscript{L1}, $\mathcal{P}^{\text{FRoll}}_{\text{operator}}$: operator,\\ $\mathcal{P}^{\text{FRoll}}_{\text{watchtower}}$: watchtower

\noindent \textbf{Internal state:} 
\begin{itemize}
    \item $\mathsf{round} \in \mathbb{N}_{\geq 0}$, $\mathsf{round} = 0$ \hfill\cmt{Current round}
    \item $\mathsf{requestQueue} \subset \{0,1\}^{*}$, $\mathsf{requestQueue} = \emptyset$ \hfill\cmt{Queued unexecuted requests}
    \item $\mathsf{executedRequest} \subset \{0,1\}^{*}$, $\mathsf{executedRequest} = \emptyset$ \hfill\cmt{Executed requests}
    \item $\mathsf{stateList} \subset \{0,1\}^{*}$, $\mathsf{stateList} = \emptyset$ \hfill\cmt{L2 state list}
    \item $\mathsf{onchainState} \subset \{0,1\}^{*}$, $\mathsf{onchainState} = \emptyset$ \hfill\cmt{L1 committed state}
    \item $\mathsf{identities} \subset \{0,1\}^{*}$, $\mathsf{identities} = \{\mathsf{pid}_{\textit{op}}, ContractAddr, \{\mathsf{pid}\}_{\textit{watchtower}}\}$ \hfill\cmt{Operator, watchtower identities and contract address}
    \item $\mathsf{lastConfirmedState} \in \{0,1\}^{*}$, $\mathsf{lastConfirmedState} = \emptyset$ \hfill\cmt{Verifier's last confirmed L2 state}
\end{itemize}

\noindent \textbf{CheckID}(\emph{pid}, \emph{sid}, \emph{role}): Accept all messages with the same \emph{sid}. The \emph{verifier} role is private: it only accepts messages from $\mathcal{F}^{\text{FRoll}}_{\text{com}}$ (\{UpdateCheck\}); it does not accept I/O requests from the environment.

\noindent \textbf{Corruption behavior:}
\begin{itemize}
    \item \textbf{DetermineCorrStatus}(\emph{pid}, \emph{sid}, \emph{role}): Return $\mathsf{corr}$ for the \emph{client} role; the \emph{verifier} role inherits the corruption status of the hosting machine.
    \item \textbf{LeakedData}(\emph{pid}, \emph{sid}, \emph{role}): Return $\mathsf{internalState}$. \cmt{Plaintext assumption; covers both roles' state}
\end{itemize}

\vspace{0.5em}
\noindent \textbf{Main (client role):}

\vspace{0.5em}

\textbf{recv} \{Join, $\mathit{ContractAddr}$, $s_{\mathit{init}}$\} \textbf{from} I/O:

\begin{enumerate}[itemsep=0.5em]
    \item Create $\mathit{TX}_{\mathit{deposit}} = (\pcur, \mathit{ContractAddr}, s_{\mathit{init}}, \epsilon)$;

    \item \textbf{send} \{SubmitL1, $\mathit{TX}_{\mathit{deposit}}$\} \textbf{to} $(\pcur, \scur, \mathcal{F}_{\text{ledger}}: \text{client}_{\text{L1}})$;

    \item \textbf{send} \{ReadL1\} \textbf{to} $(\pcur, \scur, \mathcal{F}_{\text{ledger}}:\text{client}_{\text{L1}})$;
    \textbf{wait for} \{ReadL1, $\mathit{L1ReadResult} = \{\{\mathit{TX}\}_{\text{L1}}, \mathit{State}_{\text{L1}}, \{\mathit{pid}\}\}$\};

    \item \textbf{if} $\mathit{TX}_{\mathit{deposit}} \in \{\mathit{TX}\}_{\text{L1}}$: create $\mathit{TX}_{\mathit{peg\text{-}in}} = (\pcur, \epsilon, s_{\mathit{init}}, \epsilon)$;

    \item \textbf{send} \{Sign, $\mathit{TX}_{\mathit{peg\text{-}in}}$\} \textbf{to} $(\pcur, \scur, \mathcal{F}_{\text{sig}}: \text{signer})$;
    \textbf{wait for} \{Signature, $\sigma$\};

    \item \textbf{send} \{Message, \{Join, $\mathit{TX}_{\mathit{peg\text{-}in}}$, $\sigma$, \emph{receiver}\}\} \textbf{to} $(\pcur, \scur, \mathcal{F}^{\text{FRoll}}_{\text{com}}:\text{com})$, where $\mathit{receiver} = (\pcur, \scur, \mathcal{P}^{\text{FRoll}}_{\text{operator}}:\text{operator})$;

    \item $\mathsf{requestQueue}.\text{add}(\mathit{TX}_{\mathit{peg\text{-}in}})$; \cmt{Pending until L1 commitment + challenge survival}

    \item \textbf{send} \{ReadL1\} \textbf{to} $(\pcur, \scur, \mathcal{F}_{\text{ledger}}:\text{client}_{\text{L1}})$;
    \textbf{wait for} \{ReadL1, $\mathit{L1ReadResult}$\}; $\mathsf{identities}.\text{add}(\mathit{ContractAddr})$;

    \item \textbf{if} $\mathit{TX}_{\mathit{peg\text{-}in}} \in \mathit{L1ReadResult}.\{\mathit{TX}\}_{\text{L1}}$ and ($T_{\text{challenge}}$ survived without fraud proof, or $\geq f{+}1$ watchtower signatures on L1):
    \begin{itemize}
        \item $\mathsf{requestQueue}.\text{remove}(\mathit{TX}_{\mathit{peg\text{-}in}})$;
        $\mathsf{executedRequest}.\text{add}(\mathit{TX}_{\mathit{peg\text{-}in}})$;
        \item $\mathsf{stateList} \leftarrow s_{\mathit{init}}$; $\mathsf{onchainState} \leftarrow s_{\mathit{init}}$;
        \item $\forall\, \mathit{ID} \in \mathit{L1ReadResult}.\{\mathit{pid}\}$: $\mathsf{identities}.\text{add}(\mathit{ID})$;
        \item \textbf{reply} \{Join, $s_{\mathit{init}}$\} via I/O;
    \end{itemize}
\end{enumerate}

\hrule\vspace{0.5em}

\textbf{recv} \{Submit, \texttt{collaborate}, $\mathit{TX}$\} \textbf{from} I/O, \textbf{s.t.} $\mathsf{Val}(\mathit{TX})$ passes:
\begin{enumerate}[itemsep=0.5em]
\item \textbf{send} \{Sign, $\mathit{TX}$\} \textbf{to} $(\pcur, \scur, \mathcal{F}_{\text{sig}}: \text{signer})$;
\textbf{wait for} \{Signature, $\sigma$\};
\item \textbf{send} \{Message, \{Submit, $\mathit{TX}$, $\sigma$, \emph{receiver}\}\} \textbf{to} $(\pcur, \scur, \mathcal{F}^{\text{FRoll}}_{\text{com}}:\text{com})$, where $\mathit{receiver} = (\pcur, \scur, \mathcal{P}^{\text{FRoll}}_{\text{operator}}:\text{operator})$;
\item $\mathsf{requestQueue}.\text{add}(\mathit{TX})$; \cmt{Pending; cleared on next Read after $T_{\text{challenge}}$ or $f{+}1$ watchtower attestations}
\end{enumerate}

\hrule\vspace{0.5em}

\textbf{recv} \{Submit, \texttt{selfsubmit}, $\mathit{TX}$\} \textbf{from} I/O, \textbf{s.t.} $\mathsf{Val}(\mathit{TX})$ passes:
\begin{enumerate}[itemsep=0.5em]
\item \textbf{send} \{Sign, $\mathit{TX}$\} \textbf{to} $(\pcur, \scur, \mathcal{F}_{\text{sig}}: \text{signer})$;
\textbf{wait for} \{Signature, $\sigma$\};
\item \textbf{send} \{SubmitL1, $\mathit{TX}$, $\sigma$\} \textbf{to} $(\pcur, \scur, \mathcal{F}_{\text{ledger}}:\text{client}_{\text{L1}})$;
\item \textbf{send} \{ReadL1\} \textbf{to} $(\pcur, \scur, \mathcal{F}_{\text{ledger}}:\text{client}_{\text{L1}})$; \textbf{wait for} $\mathit{L1ReadResult}$;
\item \textbf{if} $\mathit{TX} \in \mathit{L1ReadResult}.\{\mathit{TX}\}_{\text{L1}}$: $\mathsf{requestQueue}.\text{add}(\mathit{TX})$;
\end{enumerate}

\hrule\vspace{0.5em}

\textbf{recv} \{SubmitFast, \texttt{collaborate}, $\mathit{TX}_{\mathit{fast}}$\} \textbf{from} I/O, \textbf{s.t.} $\mathsf{Val}(\mathit{TX}_{\mathit{fast}})$ passes:
\begin{enumerate}[itemsep=0.5em]
\item \textbf{send} \{Sign, $\mathit{TX}_{\mathit{fast}}$\} \textbf{to} $(\pcur, \scur, \mathcal{F}_{\text{sig}}: \text{signer})$;
\textbf{wait for} \{Signature, $\sigma$\};
\item \textbf{send} \{Message, \{SubmitFast, $\mathit{TX}_{\mathit{fast}}$, $\sigma$, \emph{receiver}\}\} \textbf{to} $(\pcur, \scur, \mathcal{F}^{\text{FRoll}}_{\text{com}}:\text{com})$, where $\mathit{receiver} = (\pcur, \scur, \mathcal{P}^{\text{FRoll}}_{\text{operator}}:\text{operator})$;
\item $\mathsf{requestQueue}.\text{add}(\mathit{TX}_{\mathit{fast}})$; \cmt{Pending; cleared as soon as $\geq f{+}1$ watchtower signatures land on L1}
\end{enumerate}

\hrule\vspace{0.5em}

\textbf{recv} \{SubmitFast, \texttt{selfsubmit}, $\mathit{TX}_{\mathit{fast}}$\} \textbf{from} I/O, \textbf{s.t.} $\mathsf{Val}(\mathit{TX}_{\mathit{fast}})$ passes:
\begin{enumerate}[itemsep=0.5em]
\item \textbf{send} \{Sign, $\mathit{TX}_{\mathit{fast}}$\} \textbf{to} $(\pcur, \scur, \mathcal{F}_{\text{sig}}: \text{signer})$;
\textbf{wait for} \{Signature, $\sigma$\};
\item \textbf{send} \{SubmitL1, $\mathit{TX}_{\mathit{fast}}$, $\sigma$\} \textbf{to} $(\pcur, \scur, \mathcal{F}_{\text{ledger}}:\text{client}_{\text{L1}})$;
\item \textbf{send} \{ReadL1\} \textbf{to} $(\pcur, \scur, \mathcal{F}_{\text{ledger}}:\text{client}_{\text{L1}})$; \textbf{wait for} $\mathit{L1ReadResult}$;
\item \textbf{if} $\mathit{TX}_{\mathit{fast}} \in \mathit{L1ReadResult}.\{\mathit{TX}\}_{\text{L1}}$: $\mathsf{requestQueue}.\text{add}(\mathit{TX}_{\mathit{fast}})$;
\end{enumerate}

\hrule\vspace{0.5em}

\textbf{recv} \{Settlement, \texttt{collaborate}\} \textbf{from} I/O:
\begin{enumerate}[itemsep=0.5em]
\item Create $\mathit{TX}_{\mathit{peg\text{-}out}} = (\mathit{ContractAddr}, \pcur, \mathsf{stateList}, \epsilon)$;
\item \textbf{send} \{Sign, $\mathit{TX}_{\mathit{peg\text{-}out}}$\} \textbf{to} $(\pcur, \scur, \mathcal{F}_{\text{sig}}: \text{signer})$;
\textbf{wait for} \{Signature, $\sigma$\};
\item \textbf{send} \{Message, \{Settlement, $\mathit{TX}_{\mathit{peg\text{-}out}}$, $\sigma$, \emph{receiver}\}\} \textbf{to} $(\pcur, \scur, \mathcal{F}^{\text{FRoll}}_{\text{com}}:\text{com})$, where $\mathit{receiver} = (\pcur, \scur, \mathcal{P}^{\text{FRoll}}_{\text{operator}}:\text{operator})$;
\item $\mathsf{requestQueue}.\text{add}(\mathit{TX}_{\mathit{peg\text{-}out}})$;
\item \textbf{send} \{ReadL1\} \textbf{to} $(\pcur, \scur, \mathcal{F}_{\text{ledger}}:\text{client}_{\text{L1}})$; \textbf{wait for} $\mathit{L1ReadResult}$;
\item $\forall\, \mathit{TX}' \in \mathit{L1ReadResult}.\{\mathit{TX}\}_{\text{L1}}$ \textbf{s.t.} $\mathit{TX}' \in \mathsf{requestQueue}$ and ($T_{\text{challenge}}$ survived without fraud proof, or $\geq f{+}1$ watchtower signatures on L1 for $\mathit{TX}'$): \cmt{Sweep cleared peg-out and any earlier now-final pending requests}
\begin{itemize}
    \item $\mathsf{requestQueue}.\text{remove}(\mathit{TX}')$;
    $\mathsf{executedRequest}.\text{add}(\mathit{TX}')$;
\end{itemize}
\item \textbf{if} $\mathit{TX}_{\mathit{peg\text{-}out}} \in \mathsf{executedRequest}$:
$\mathsf{stateList} \leftarrow \mathit{L1ReadResult}.\mathit{State}_{\text{L1}}$;\\
$\mathsf{onchainState} \leftarrow \mathit{L1ReadResult}.\mathit{State}_{\text{L1}}$;
\textbf{reply} \{Settlement, $\mathsf{onchainState}$\} via I/O;
\end{enumerate}

\hrule\vspace{0.5em}

\textbf{recv} \{Settlement, \texttt{escape-hatch}\} \textbf{from} I/O:
\begin{enumerate}[itemsep=0.5em]
\item Create $\mathit{TX}_{\mathit{peg\text{-}out}} = (\mathit{ContractAddr}, \pcur, \mathsf{stateList}, \epsilon)$;
\item \textbf{send} \{SubmitL1, $\mathit{TX}_{\mathit{peg\text{-}out}}$\} \textbf{to} $(\pcur, \scur, \mathcal{F}_{\text{ledger}}:\text{client}_{\text{L1}})$;
\item $\mathsf{requestQueue}.\text{add}(\mathit{TX}_{\mathit{peg\text{-}out}})$;
\item \textbf{send} \{ReadL1\} \textbf{to} $(\pcur, \scur, \mathcal{F}_{\text{ledger}}:\text{client}_{\text{L1}})$; \textbf{wait for} $\mathit{L1ReadResult}$;
\item $\forall\, \mathit{TX}' \in \mathit{L1ReadResult}.\{\mathit{TX}\}_{\text{L1}}$ \textbf{s.t.} $\mathit{TX}' \in \mathsf{requestQueue}$ and ($T_{\text{challenge}}$ survived without fraud proof, or $\geq f{+}1$ watchtower signatures on L1 for $\mathit{TX}'$):
\begin{itemize}
    \item $\mathsf{requestQueue}.\text{remove}(\mathit{TX}')$;
    $\mathsf{executedRequest}.\text{add}(\mathit{TX}')$;
\end{itemize}
\item \textbf{if} $\mathit{TX}_{\mathit{peg\text{-}out}} \in \mathsf{executedRequest}$:
$\mathsf{stateList} \leftarrow \mathit{L1ReadResult}.\mathit{State}_{\text{L1}}$;\\
$\mathsf{onchainState} \leftarrow \mathit{L1ReadResult}.\mathit{State}_{\text{L1}}$;
\textbf{reply} \{Settlement, $\mathsf{onchainState}$\} via I/O;
\end{enumerate}

\hrule\vspace{0.5em}

\textbf{recv} \{Read\} \textbf{from} I/O:
\begin{enumerate}[itemsep=0.5em]
\item \textbf{send} \{ReadL1\} \textbf{to} $(\pcur, \scur, \mathcal{F}_{\text{ledger}}:\text{client}_{\text{L1}})$;
\textbf{wait for} \{ReadL1, $\mathit{L1ReadResult} = \{\{\mathit{TX}\}_{\text{L1}}, \mathit{State}_{\text{L1}}, \{\mathit{pid}\}\}$\};

\item Let $\mathit{TX}_{\mathit{batch}}^{*} \in \{\mathit{TX}\}_{\text{L1}}$ be the most recent checkpoint transaction that has survived $T_{\text{challenge}}$ without a valid fraud proof. Parse $\{\mathit{batch}, \mathit{resultState}\} \leftarrow \mathit{TX}_{\mathit{batch}}^{*}.\mathit{data}$. \cmt{Latest L1-final checkpoint defines current finalized L2 view}

\item $\forall\, \mathit{TX} \in \mathit{batch}$ \textbf{s.t.} $\mathit{TX} \in \mathsf{requestQueue}$:
\begin{itemize}
    \item $\mathsf{requestQueue}.\text{remove}(\mathit{TX})$;
    $\mathsf{executedRequest}.\text{add}(\mathit{TX})$; \cmt{Cleared once enclosing checkpoint is L1-final}
\end{itemize}

\item $\forall\, \mathit{TX}_{\mathit{fast}} \in \{\mathit{TX}\}_{\text{L1}}$ \textbf{s.t.} $\mathit{TX}_{\mathit{fast}} \in \mathsf{requestQueue}$ and $\geq f{+}1$ watchtower certificate transactions $\mathit{TX}_{\mathit{WT}}^{(j)} = (\mathit{TX}_{\mathit{fast}}, \sigma, \sigma_{W_j})$ are committed on L1 for $\mathit{TX}_{\mathit{fast}}$:
\begin{itemize}
    \item $\mathsf{requestQueue}.\text{remove}(\mathit{TX}_{\mathit{fast}})$;
    $\mathsf{executedRequest}.\text{add}(\mathit{TX}_{\mathit{fast}})$; \cmt{Fast-finality path: cleared once watchtower quorum lands on L1}
\end{itemize}

\item $\mathsf{stateList} \leftarrow \mathit{resultState}$; \cmt{L2 view derived from claimed post-state in latest checkpoint data}\\
$\mathsf{onchainState} \leftarrow \mathit{State}_{\text{L1}}$; \cmt{L1-finalized state}

\item $\forall\, \mathit{pid}' \in \{\mathit{pid}\}$: $\mathsf{identities}.\text{add}(\mathit{pid}')$;

\item \textbf{reply} \{Read, $\mathit{ReadResult} = \{\mathsf{executedRequest}, \mathsf{stateList}, \mathsf{onchainState}\}$\} via I/O;
\end{enumerate}

\hrule\vspace{0.5em}

\textbf{recv} \{GetCurRound\} \textbf{from} I/O:
\begin{enumerate}[itemsep=0.5em]
\item \textbf{send} \{GetCurRound\} \textbf{to} $(\pcur, \scur, \mathcal{F}^{\text{FRoll}}_{\text{com}}: \text{com})$;
\textbf{wait for} \{GetCurRound, $\mathit{round}$\};
\item \textbf{reply} \{GetCurRound, $\mathit{round}$\} via I/O;
\end{enumerate}

\hrule
\vspace{1em}
\noindent \textbf{Main (verifier role, private, no I/O interface):}

\vspace{0.5em}

\textbf{recv} \{UpdateCheck\} \textbf{from} $(\pcur, \scur, \mathcal{F}^{\text{FRoll}}_{\text{com}}: \text{com})$:
\begin{enumerate}[itemsep=0.5em]
    \item \textbf{send} \{ReadL1\} \textbf{to} $(\pcur, \scur, \mathcal{F}_{\text{ledger}}:\text{client}_{\text{L1}})$;
    \textbf{wait for} \{ReadL1, $\mathit{L1ReadResult} = \{\{\mathit{TX}\}_{\text{L1}}, \mathit{State}_{\text{L1}}, \{\mathit{pid}\}\}$\};

    \item Let $\mathit{TX}_{\mathit{batch}} \in \{\mathit{TX}\}_{\text{L1}}$ be the latest unconfirmed checkpoint transaction within $T_{\text{challenge}}$ from the operator addressed to $\mathit{ContractAddr}$, and parse $\{\mathit{batch},\\ \mathit{resultState}\} \leftarrow \mathit{TX}_{\mathit{batch}}.\mathit{data}$; \cmt{Both batch and claimed post-state read from checkpoint data field}

    \item Let $b \leftarrow \mathsf{Check}(\mathit{batch}, \mathit{resultState}, \mathsf{lastConfirmedState})$; \cmt{Verifier-role check on published batch}

    \item \textbf{if} $b = \text{false}$:
    \begin{itemize}
        \item Let $\mathit{TX}_{\mathit{fraud}} \leftarrow \mathsf{GenF}(\mathit{batch}, \mathit{resultState}, \mathsf{lastConfirmedState})$; \cmt{Generate fraud-proof transaction}
        \item \textbf{if} $\mathit{TX}_{\mathit{fraud}} \neq \bot$: \textbf{send} \{SubmitL1, $\mathit{TX}_{\mathit{fraud}}$\} \textbf{to} $(\pcur, \scur, \mathcal{F}_{\text{ledger}}:\text{client}_{\text{L1}})$; \cmt{Fraud proof published on L1}
    \end{itemize}

    \item \textbf{if} $b = \text{true}$: $\mathsf{lastConfirmedState} \leftarrow \mathit{resultState}$;
\end{enumerate}

\noindent\emph{Verifier-role note.} The verifier role does not produce any I/O output to $\mathcal{E}$. Its only externally observable action is the publication of $\mathit{TX}_{\mathit{fraud}}$ on $\mathcal{F}_{\text{ledger}}$, which is an L1 transaction observable on the same footing as any other on-chain transaction.

}\end{functionality}\vspace{1em}

\subsubsection{Operator Protocol}

The operator machine $\mathcal{M}^{\text{FRoll}}_{\text{operator}}$ executes received regular and fast-finality transactions and publishes results to L1 periodically (triggered by $\mathcal{F}^{\text{FRoll}}_{\text{com}}$ via \{UpdateRequest\}). We assume at least one honest operator.

\vspace{1em}\begin{functionality}{Description of protocol $\mathcal{P}^{\text{FRoll}}_{\text{operator}} = (\text{operator})$}{
\textbf{Participating roles:} \{operator\} \\ \textbf{Corruption model:} dynamic corruption

\noindent \textbf{Protocol parameters:} $\mathsf{Val}()$: transaction validity checking algorithm
}\end{functionality}\vspace{.5em}

\begin{functionality}{Description of $\mathcal{M}^{\text{FRoll}}_{\text{operator}}$}{

\textbf{Implemented role(s):} \{operator\}

\noindent \textbf{Subroutines:} $\mathcal{F}_{\text{sig}}$: \{signer, verifier\}, $\mathcal{F}^{\text{FRoll}}_{\text{com}}$: com, $\mathcal{F}_{\text{ledger}}$: client\textsubscript{L1}

\noindent \textbf{Internal state:} 
\begin{itemize}
    \item $\mathsf{requestQueue} \subset \{0,1\}^{*}$, $\mathsf{requestQueue} = \emptyset$ \hfill\cmt{Queued unexecuted requests}
    \item $\mathsf{stateList} \subset \{0,1\}^{*} \times \mathbb{N}_{\geq 0}$, $\mathsf{stateList} = \emptyset$ \hfill\cmt{Account states for all known clients, of form $(\mathit{pid}, \mathit{bal})$}
    \item $\mathsf{identities} \subset \{0,1\}^{*}$, $\mathsf{identities} = \emptyset$ \hfill\cmt{Registered identities}
\end{itemize}

\noindent \textbf{CheckID}(\emph{pid}, \emph{sid}, \emph{role}): Accept all messages with the same \emph{sid}. Only accept \emph{role} = operator.

\vspace{0.5em}
\noindent \textbf{Main:}

\vspace{0.5em}

\textbf{recv} \{Join, $\mathit{TX}_{\mathit{peg\text{-}in}}$, $\sigma$, $\mathit{pid}_{\mathit{call}}$\} \textbf{from} $(\pcur, \scur, \mathcal{F}^{\text{FRoll}}_{\text{com}}: \text{com})$:
\begin{enumerate}[itemsep=0.5em]
    \item Check that $\mathsf{Val}(\mathit{TX}_{\mathit{peg\text{-}in}}) = \text{true}$; \cmt{Peg-in transaction is well-formed}
    \item \textbf{send} \{Verify, $\mathit{TX}_{\mathit{peg\text{-}in}}$, $\sigma$\} \textbf{to} $(\pcur, \scur, \mathcal{F}_{\text{sig}}: \text{verifier})$;
    \textbf{wait for} \{VerResult, $b$\};
    \item \textbf{send} \{ReadL1\} \textbf{to} $(\pcur, \scur, \mathcal{F}_{\text{ledger}}:\text{client}_{\text{L1}})$;
    \textbf{wait for} \{ReadL1, $\mathit{L1ReadResult} = \{\{\mathit{TX}\}_{\text{L1}}, \mathit{State}_{\text{L1}}, \{\mathit{pid}\}\}$\};
    \item Check that $\mathit{TX}_{\mathit{deposit}} \in \{\mathit{TX}\}_{\text{L1}}$; \cmt{Deposit committed on L1}
    \item \textbf{if} $b = \text{true}$ and all checks pass:
    \begin{itemize}
        \item $\mathsf{identities}.\text{add}(\mathsf{pid}_{\mathit{sender}})$; \cmt{Register joining client}
        \item Parse $s_{\mathit{init}}$ from $\mathit{TX}_{\mathit{peg\text{-}in}}$; $\mathsf{stateList}.\text{add}(\mathsf{pid}_{\mathit{sender}}, s_{\mathit{init}})$; \cmt{Initialize client's account state}
        \item $\mathsf{requestQueue}.\text{add}(\mathit{TX}_{\mathit{peg\text{-}in}}, \sigma)$; \cmt{Enqueue peg-in with client signature}
    \end{itemize}
\end{enumerate}

\hrule\vspace{0.5em}

\textbf{recv} \{Submit, $\mathit{TX} = (\mathit{sender}, \mathit{receiver}, \mathit{value}, \mathit{data})$, $\sigma$, $\mathit{pid}_{\mathit{call}}$\} or \{SubmitFast, $\mathit{TX}$, $\sigma$, $\mathit{pid}_{\mathit{call}}$\} \textbf{from} $(\pcur, \scur, \mathcal{F}^{\text{FRoll}}_{\text{com}}: \text{com})$:
\begin{enumerate}[itemsep=0.5em]
\item Check that $\mathsf{Val}(\mathit{TX}) = \text{true}$; \cmt{Transaction is well-formed and satisfies validity predicate}
\item \textbf{send} \{Verify, $\mathit{TX}$, $\sigma$\} \textbf{to} $(\pcur, \scur, \mathcal{F}_{\text{sig}}: \text{verifier})$,
\textbf{wait for} \{VerResult, $b$\};
\item \textbf{if} $b = \text{true}$:
\begin{itemize}
    \item $\mathsf{requestQueue}.\text{add}(\mathit{TX}, \sigma)$; \cmt{Signature valid; enqueue request with sender signature}
    \item $\mathsf{stateList}[\mathit{sender}].\mathit{bal} \leftarrow \mathit{bal}(\mathit{sender}) - \mathit{value}$; \cmt{Deduct from sender}
    \item $\mathsf{stateList}[\mathit{receiver}].\mathit{bal} \leftarrow \mathit{bal}(\mathit{receiver}) + \mathit{value}$; \cmt{Credit to receiver}
\end{itemize}
\end{enumerate}

\hrule\vspace{0.5em}

\textbf{recv} \{Settlement, $\mathit{TX}_{\mathit{peg\text{-}out}}$, $\sigma$, $\mathit{pid}_{\mathit{call}}$\} \textbf{from} $(\pcur, \scur, \mathcal{F}^{\text{FRoll}}_{\text{com}}: \text{com})$:
\begin{enumerate}[itemsep=0.5em]
\item Check that $\mathsf{Val}(\mathit{TX}_{\mathit{peg\text{-}out}}) = \text{true}$; \cmt{Peg-out transaction is well-formed}
\item \textbf{send} \{Verify, $\mathit{TX}_{\mathit{peg\text{-}out}}$, $\sigma$\} \textbf{to} $(\pcur, \scur, \mathcal{F}_{\text{sig}}: \text{verifier})$,
\textbf{wait for} \{VerResult, $b$\};
\item \textbf{if} $b = \text{true}$:
$\mathsf{requestQueue}.\text{add}(\mathit{TX}_{\mathit{peg\text{-}out}}, \sigma)$; \cmt{Signature valid; enqueue settlement request with client signature}
\end{enumerate}

\hrule\vspace{0.5em}

\textbf{recv} \{UpdateRequest\} \textbf{from} $(\pcur, \scur, \mathcal{F}^{\text{FRoll}}_{\text{com}}: \text{com})$:
\begin{enumerate}[itemsep=0.5em]
\item Let $\mathit{batch} \leftarrow \mathsf{requestQueue}$; 

\item Let $\mathit{resultState} \leftarrow \mathsf{stateList}$; \cmt{Current state reflects all executed requests}
\item Construct the checkpoint transaction $\mathit{TX}_{\mathit{batch}} = (\pcur,\, \mathit{ContractAddr},\, \epsilon,\\ \{\mathit{batch}, \mathit{resultState}\})$; 

\item \textbf{send} \{SubmitL1, $\mathit{TX}_{\mathit{batch}}$\} \textbf{to} $(\pcur, \scur, \mathcal{F}_{\text{ledger}}:\text{client}_{\text{L1}})$; \cmt{Publish checkpoint transaction to L1}

\item $\mathsf{requestQueue} \leftarrow \emptyset$; 
\end{enumerate}

}\end{functionality}\vspace{1em}

\subsubsection{Watchtower Protocol}

The watchtower machine $\mathcal{M}^{\text{FRoll}}_{\text{watchtower}}$ monitors L1 for fast-finality transactions and contributes to quorum certificates by issuing signatures.

\vspace{1em}\begin{functionality}{Description of protocol $\mathcal{P}^{\text{FRoll}}_{\text{watchtower}} = (\text{watchtower})$}{
\textbf{Participating roles:} \{watchtower\} \\ \textbf{Corruption model:} dynamic corruption

\noindent \textbf{Protocol parameters:} $n = 2f+1$; $f$: corruption threshold; $\mathsf{Val}()$
}\end{functionality}\vspace{.5em}

\begin{functionality}{Description of $\mathcal{M}^{\text{FRoll}}_{\text{watchtower}}$}{

\textbf{Implemented role(s):} \{watchtower\}

\noindent \textbf{Subroutines:} $\mathcal{F}_{\text{sig}}$: \{signer, verifier\}, $\mathcal{F}_{\text{ledger}}$: client\textsubscript{L1}, $\mathcal{F}^{\text{FRoll}}_{\text{com}}$: com

\noindent \textbf{CheckID}(\emph{pid}, \emph{sid}, \emph{role}): Accept all messages with the same \emph{sid}. Only accept \emph{role} = watchtower.

\vspace{0.5em}
\noindent \textbf{Main:}

\vspace{0.5em}

\textbf{recv} \{CheckFastL1\} \textbf{from} $(\pcur, \scur, \mathcal{F}^{\text{FRoll}}_{\text{com}}: \text{com})$:
\begin{enumerate}[itemsep=0.5em]
  \item \textbf{send} \{ReadL1\} \textbf{to} $(\pcur, \scur, \mathcal{F}_{\text{ledger}}: \text{client}_{\text{L1}})$,
  \textbf{wait for} \{ReadL1, $\mathit{L1ReadResult} = \{\{\mathit{TX}\}_{\text{L1}}, \mathit{State}_{\text{L1}}, \{\mathit{pid}\}\}$\};

  \item $\forall\, \{\mathit{TX}_{\mathit{fast}}, \sigma\} \in \{\mathit{TX}\}_{\text{L1}}$ marked as fast-finality:
  \begin{itemize}
      \item Check that $\mathsf{Val}(\mathit{TX}_{\mathit{fast}}) = \text{true}$;
      \item \textbf{send} \{Verify, $\mathit{TX}_{\mathit{fast}}$, $\sigma$\} \textbf{to} $(\pcur, \scur, \mathcal{F}_{\text{sig}}: \text{verifier})$;
      \textbf{wait for} \{VerResult, $b$\};
  \end{itemize}

  \item \textbf{if} $b = \text{true}$: \textbf{send} \{Sign, $\{\mathit{TX}_{\mathit{fast}}, \sigma\}$\} \textbf{to} $(\pcur, \scur, \mathcal{F}_{\text{sig}}: \text{signer})$,
  \textbf{wait for} \{Signature, $\sigma_W$\};

  \item \textbf{send} \{SubmitL1, $\{\mathit{TX}_{\mathit{fast}}, \sigma, \sigma_W\}$\} \textbf{to} $(\pcur, \scur, \mathcal{F}_{\text{ledger}}: \text{client}_{\text{L1}})$; \cmt{Publish watchtower attestation to L1}
\end{enumerate}

}\end{functionality}\vspace{1em}
\subsection{\xr Ideal Functionality}
\label{apd:XRollideal}

We define the ideal functionality for \xr as
$\mathcal{F}^{\text{FRoll}}_{\text{layer2}}=(\mathcal{F}_{\text{client}}, \mathcal{F}_{\text{ledger}} \mid \mathcal{F}^{\text{FRoll}}_{\text{join}}, \mathcal{F}^{\text{FRoll}}_{\text{submit}}, \mathcal{F}^{\text{FRoll}}_{\text{update}}, \mathcal{F}^{\text{FRoll}}_{\text{read}},
\mathcal{F}^{\text{FRoll}}_{\text{settlement}}, \mathcal{F}^{\text{FRoll}}_{\text{updRnd}})$. The subroutines' ideal functionalities are defined as follows.

\subsubsection{Submit Functionality}
\label{apd:XRollsubmit}

The submit subroutine $\mathcal{F}^{\text{FRoll}}_{\text{submit}}$ handles four request types: \textit{(i)}~join requests, \textit{(ii)}~regular transaction requests, \textit{(iii)}~fast-finality transaction requests, and \textit{(iv)}~settlement requests.

\vspace{1em}\begin{functionality}{Description of subroutine $\mathcal{F}^{\text{FRoll}}_{\text{submit}} = (\text{submit})$}{
\textbf{Participating roles:} \{submit\} \\ \textbf{Corruption model:} incorruptible
}\end{functionality}\vspace{.5em}

\begin{functionality}{Description of $\mathcal{M}^{\text{FRoll}}_{\text{submit}}$}{
\textbf{Implemented role(s):} \{submit\}

\noindent \textbf{CheckID}(\emph{pid}, \emph{sid}, \emph{role}): Accept all messages with the same \emph{sid}.

\vspace{0.5em}\noindent \textbf{Main:}\vspace{0.5em}

\textbf{recv} \{Submit, $\mathit{request}$, $\mathsf{internalState}$\} \textbf{from} I/O:

\begin{enumerate}[itemsep=0.5em]
\item Check that $\mathit{request}$ belongs to one of the following valid types:
\begin{itemize}
    \item Join request: $\mathit{request} = (\text{Join}, s_{\mathit{init}},\\ \mathit{ContractAddr})$, \textbf{s.t.} $s_{\mathit{init}} \in \mathbb{N}_{\geq 0}$ and $\mathit{ContractAddr} \in \{0,1\}^{*}$;

    \item Regular state update: $\mathit{request} = (\text{Submit}, \mathit{mode}, \mathit{TX})$ with $\mathit{mode} \in\\ \{\texttt{collaborate}, \texttt{selfsubmit}\}$ and $\mathit{TX} = (\mathit{sender}, \mathit{receiver}, \mathit{value}, \mathit{data})$, \textbf{s.t.}:
    \begin{itemize}
        \item $\mathit{sender} \in \mathsf{identities}$; $\mathit{value} \leq \mathit{bal}(\mathit{sender})$ from $\mathsf{stateList}$; \cmt{Sufficient balance}
        \item $\nexists\, \mathit{TX}' \in \mathsf{executedRequest} \cup \mathsf{requestQueue}$ conflicting with $\mathit{TX}$; \cmt{No double-spending}
    \end{itemize}

    \item Fast-finality state update: $\mathit{request} = (\text{SubmitFast}, \mathit{mode}, \mathit{TX}_{\mathit{fast}})$ with $\mathit{mode} \in \{\texttt{collaborate}, \texttt{selfsubmit}\}$, subject to the same validity checks as regular updates;

    \item Settlement request: $\mathit{request} \in \{(\text{Settlement}, \texttt{collaborate}),\\ (\text{Settlement}, \texttt{escape-hatch})\}$;
\end{itemize}

\item Check that the join request has been executed (i.e., $\mathsf{stateList} \neq \emptyset$) or $\mathit{request}$ is a join request;

\item Check that $\nexists$ a settlement request from the caller already pending in $\mathsf{requestQueue}$;

\item \textbf{if} all checks pass: \textbf{reply} \{Submit, true\};
\item \textbf{else}: \textbf{reply} \{Submit, false\};
\end{enumerate}
}\end{functionality}\vspace{1em}

\subsubsection{Join Functionality}
\label{apd:XRolljoin}

The join subroutine $\mathcal{F}^{\text{FRoll}}_{\text{join}}$ validates the rollup joining procedure. It checks peg-in consistency, L1 commitment of both deposit transaction and peg-in transaction, and state recording on L1.

\vspace{1em}\begin{functionality}{Description of subroutine $\mathcal{F}^{\text{FRoll}}_{\text{join}} = (\text{join})$}{
\textbf{Participating roles:} \{join\} \\ \textbf{Corruption model:} incorruptible
}\end{functionality}\vspace{.5em}

\begin{functionality}{Description of $\mathcal{M}^{\text{FRoll}}_{\text{join}}$}{
\textbf{Implemented role(s):} \{join\}

\noindent \textbf{CheckID}(\emph{pid}, \emph{sid}, \emph{role}): Accept all messages with the same \emph{sid}.

\vspace{0.5em}\noindent \textbf{Main:}\vspace{0.5em}

\textbf{recv} \{Join, $\mathit{Attachment}=\{s_{\mathit{init}}, \mathit{TX}_{\mathit{peg\text{-}in}}\}$, $\mathsf{internalState}$\} \textbf{from} I/O:

\begin{enumerate}[itemsep=0.5em]

\item Parse $\mathit{TX}_{\mathit{peg\text{-}in}} = (\mathit{pid}_{\mathit{init}}, \epsilon, s_{\mathit{init}}', \epsilon)$. Check:
\begin{itemize}
    \item $s_{\mathit{init}}' = s_{\mathit{init}}$; \cmt{Peg-in carries the agreed initial state}
    \item $\mathit{pid}_{\mathit{init}} \in \mathsf{identities}$; \cmt{Initiator is a registered participant}
\end{itemize}

\item \textbf{send} \{ReadL1\} \textbf{to} $(\pcur, \scur, \mathcal{F}_{\text{ledger}}: \text{client}_{\text{L1}})$,
\textbf{wait for} \{ReadL1, $\mathit{output} = \{\{\mathit{TX}\}_{\text{L1}}, \mathit{State}_{\text{L1}}, \{\mathit{pid}\}\}$\};

\item Check the following conditions on $\mathit{output}$:
\begin{itemize}
    \item $\exists\, \mathit{TX}_{\mathit{deposit}} \in \{\mathit{TX}\}_{\text{L1}}$ \textbf{s.t.} $\mathit{TX}_{\mathit{deposit}}$ is a deposit from $\mathit{pid}_{\mathit{init}}$ consistent with $s_{\mathit{init}}$; \cmt{Deposit committed on L1}
    \item $\mathit{TX}_{\mathit{peg\text{-}in}} \in \{\mathit{TX}\}_{\text{L1}}$; \cmt{Peg-in committed on L1}
    \item $s_{\mathit{init}} \in \mathit{State}_{\text{L1}}$; \cmt{Initial state recorded on L1}
\end{itemize}

\item \textbf{if} all checks pass: \textbf{reply} \{Join, true, $s_{\mathit{init}}$\};

\end{enumerate}
}\end{functionality}\vspace{1em}

\begin{lemma}\label{lem:OpenC}
    The subroutine $\mathcal{F}^{\text{FRoll}}_{\text{join}}$ guarantees \emph{correct L2 initialization}.
\end{lemma}

\begin{proof}
    We prove by contradiction. Suppose correct L2 initialization is violated. However, $\mathcal{F}^{\text{FRoll}}_{\text{join}}$ outputs success only if the peg-in is consistent with $s_{\mathit{init}}$ (Step~1), both deposit and peg-in are committed on L1 (Step~3), and the initial state is recorded on L1 (Step~3). Since $\mathcal{F}_{\text{sig}}$ prevents forgery and at least one honest verifier publishes fraud proofs for incorrect operator publications, no adversary can cause acceptance of an invalid state. This contradicts the assumption.
\end{proof}

\subsubsection{Update Functionality}

The update subroutine $\mathcal{F}^{\text{FRoll}}_{\text{update}}$ handles both regular and fast-finality requests. It checks correct execution, no conflicts, honest-sender request existence in the queue, and L1 commitment. For regular requests, batch and state must be recorded on L1. For fast-finality requests, $f{+}1$ watchtower signatures must be committed on L1.

\vspace{1em}\begin{functionality}{Description of subroutine $\mathcal{F}^{\text{FRoll}}_{\text{update}} = (\text{update})$}{
\textbf{Participating roles:} \{update\} \\ \textbf{Corruption model:} incorruptible

\noindent \textbf{Protocol parameters:} $n = 2f+1$; $f$: corruption threshold
}\end{functionality}\vspace{.5em}

\begin{functionality}{Description of $\mathcal{M}^{\text{FRoll}}_{\text{update}}$}{
\textbf{Implemented role(s):} \{update\}

\noindent \textbf{CheckID}(\emph{pid}, \emph{sid}, \emph{role}): Accept all messages with the same \emph{sid}.

\vspace{0.5em}\noindent \textbf{Main:}\vspace{0.5em}

\textbf{recv} \{Update, $\mathit{Attachment}=\{\mathit{newState}, \mathit{requestBatch}\}$, $\mathsf{internalState}$\} \textbf{from} I/O:

\begin{enumerate}[itemsep=0.5em]

\item Parse $\mathit{requestBatch} = \{\mathit{TX}_1, \ldots, \mathit{TX}_m\}$. Check that $\forall\, j \in \{1, \ldots, m\}$:
\begin{itemize}
    \item $\nexists\, \mathit{TX}' \in \mathsf{executedRequest}$ that conflicts with $\mathit{TX}_j$; \cmt{No conflict with executed requests}
    \item Let $\mathit{bal}(\mathit{sender}_j)$ denote the balance of $\mathit{sender}_j$ from $\mathsf{stateList}$. Then $\mathit{value}_j \leq \mathit{bal}(\mathit{sender}_j)$; \cmt{Sufficient balance}
\end{itemize}

\item Check that $\forall\, \mathit{TX}_j \in \mathit{requestBatch}$ with $(\mathit{sender}_j, \scur, \text{client}) \notin \mathsf{CorruptionSet}$:
\begin{itemize}
    \item $\exists\, \mathit{TX}' \in \mathsf{requestQueue}$ \textbf{s.t.} $\mathit{TX}' = \mathit{TX}_j$; \cmt{Each honest sender's transaction is in the queue}
\end{itemize}

\item Check that $\mathit{newState}$ is the correct output of sequentially executing $\mathit{TX}_1, \ldots, \mathit{TX}_m$ starting from $\mathsf{stateList}$; \cmt{Correct state transition}

\item \textbf{send} \{ReadL1\} \textbf{to} $(\pcur, \scur, \mathcal{F}_{\text{ledger}}: \text{client}_{\text{L1}})$,
\textbf{wait for} \{ReadL1, $\mathit{output} = \{\{\mathit{TX}\}_{\text{L1}}, \mathit{State}_{\text{L1}}, \{\mathit{pid}\}\}$\};

\item Check in $\mathit{L1ReadResult}$, that $\mathit{requestBatch}, \mathit{newState} \subseteq \{\mathit{TX}\}_{\text{L1}}$ and they are the latest commitment; \cmt{Batch is published on L1}

\item \textbf{if} all checks pass: \textbf{reply} \{Update, true, $\mathit{newState}$, $\mathit{requestBatch}$\};

\end{enumerate}
}\end{functionality}\vspace{1em}

\begin{lemma}
    The subroutines $\mathcal{F}^{\text{FRoll}}_{\text{update}}$ and $\mathcal{F}^{\text{FRoll}}_{\text{read}}$ jointly guarantee $(f_{L_2}+f_{L_1})$-safety.
\end{lemma}

\begin{proof}
    Suppose $\{f_{L_2}+ f_{L_1}\}$-\emph{safety} is violated. By Definition~\ref{def:safety}, this means either self-consistency or view-consistency is violated for some honest client. The execution-correctness check in $\mathcal{F}^{\text{FRoll}}_{\text{update}}$, guaranteed by $f_{L_2}$ corruption, rules out incorrect state transitions entering $\mathsf{internalState}$, since updates are accepted only when committed on L1 blockchain. Additionally, $\mathcal{F}^{\text{FRoll}}_{\text{read}}$ derives every read result from $\mathsf{internalState}$ consistent with the underlying L1 transcript via $\mathcal{F}_{\text{ledger}}$. As long as the $\mathcal{F}_{\text{ledger}}$ is safe under at most $f_{L_1}$ corrupted L1 participants, the L1 transcript is itself self-consistent and view-consistent, contradicting the assumption that safety was violated.
\end{proof}

\subsubsection{Read Functionality}

The read subroutine $\mathcal{F}^{\text{FRoll}}_{\text{read}}$ queries L1 and generate consistent read result from $\mathsf{internalState}$ .

\vspace{1em}\begin{functionality}{Description of subroutine $\mathcal{F}^{\text{FRoll}}_{\text{read}} = (\text{read})$}{
\textbf{Participating roles:} \{read\} \\ \textbf{Corruption model:} incorruptible
}\end{functionality}\vspace{.5em}

\begin{functionality}{Description of $\mathcal{M}^{\text{FRoll}}_{\text{read}}$}{
\textbf{Implemented role(s):} \{read\}

\noindent \textbf{CheckID}(\emph{pid}, \emph{sid}, \emph{role}): Accept all messages with the same \emph{sid}.

\vspace{0.5em}\noindent \textbf{Main:}\vspace{0.5em}

\textbf{recv} \{Read, $\mathsf{internalState}$\} \textbf{from} I/O:
\begin{enumerate}[itemsep=0.5em]

\item \textbf{send} \{ReadL1\} \textbf{to} $(\pcur, \scur, \mathcal{F}_{\text{ledger}}:\text{client}_{\text{L1}})$,
\textbf{wait for} \{ReadL1, $\mathit{L1ReadResult} = \{\{\mathit{TX}\}_{\text{L1}}, \mathit{State}_{\text{L1}}, \{\mathit{pid}\}\}$\};

\item Let $\mathsf{TX}^{\mathit{final}}_{\text{L1}} \leftarrow \{e \in \{\mathit{TX}\}_{\text{L1}} \mid e \text{ is L1-committed L2 published transaction and result}\}$ and let $\mathit{State}^{\mathit{final}}_{\text{L1}}$ be the latest L1-final L2 state derived from $\mathsf{TX}^{\mathit{final}}_{\text{L1}}$;

\item Let $\mathit{ReadResult}.\mathsf{executedRequest} \leftarrow \mathsf{executedRequest} \cap \mathsf{TX}^{\mathit{final}}_{\text{L1}}$;

\item Let $\mathit{ReadResult}.\mathsf{stateList} \leftarrow \mathsf{stateList} \cap \mathit{State}^{\mathit{final}}_{\text{L1}}$;

\item Let $\mathit{ReadResult}.\mathsf{onchainState} \leftarrow \mathsf{onchainState} \cap \mathit{State}^{\mathit{final}}_{\text{L1}}$;

\item Let $\mathit{ReadResult}.\mathsf{identities} \leftarrow \mathsf{identities} \cap \{\mathit{pid}\}$;

\item \textbf{reply} \{Read, $\mathit{ReadResult}$\};\cmt{$\mathsf{internalState}$ read result consistenct with L1 status}
\end{enumerate}
}\end{functionality}\vspace{1em}

\subsubsection{Settlement Functionality}

The settlement subroutine $\mathcal{F}^{\text{FRoll}}_{\text{settlement}}$ verifies peg-out consistency with the latest state and L1 commitment.

\vspace{1em}\begin{functionality}{Description of subroutine $\mathcal{F}^{\text{FRoll}}_{\text{settlement}} = (\text{settlement})$}{
\textbf{Participating roles:} \{settlement\} \\ \textbf{Corruption model:} incorruptible
}\end{functionality}\vspace{.5em}

\begin{functionality}{Description of $\mathcal{M}^{\text{FRoll}}_{\text{settlement}}$}{
\textbf{Implemented role(s):} \{settlement\}

\noindent \textbf{CheckID}(\emph{pid}, \emph{sid}, \emph{role}): Accept all messages with the same \emph{sid}.

\vspace{0.5em}\noindent \textbf{Main:}\vspace{0.5em}

\textbf{recv} \{Settlement, $\mathit{Attachment}=\{\mathit{TX}_{\mathit{peg\text{-}out}}\}$, $\mathsf{internalState}$\} \textbf{from} I/O:

\begin{enumerate}[itemsep=0.5em]

\item Parse $\mathit{TX}_{\mathit{peg\text{-}out}} = (\mathit{ContractAddr}, \mathit{pid}_{\mathit{client}}, s_{\mathit{settle}}, \epsilon)$. Let $s_{\mathit{settle}}$ be the latest state from $\mathsf{stateList}$. Check:
\begin{itemize}
    \item $s_{\mathit{settle}} = s_{\mathit{settle}}$; \cmt{Settlement carries the latest valid state}
    \item $\mathit{pid}_{\mathit{client}} \in \mathsf{identities}$; \cmt{Settling client is registered}
\end{itemize}

\item \textbf{send} \{ReadL1\} \textbf{to} $(\pcur, \scur, \mathcal{F}_{\text{ledger}}: \text{client}_{\text{L1}})$,
\textbf{wait for} \{ReadL1, $\mathit{output} = \{\{\mathit{TX}\}_{\text{L1}}, \mathit{State}_{\text{L1}}, \{\mathit{pid}\}\}$\};

\item Check the following conditions on $\mathit{output}$:
\begin{itemize}
    \item $\mathit{TX}_{\mathit{peg\text{-}out}} \in \{\mathit{TX}\}_{\text{L1}}$; \cmt{Peg-out committed on L1}
    \item $s_{\mathit{settle}} \in \mathit{State}_{\text{L1}}$; \cmt{Latest L2 state recorded on L1}
\end{itemize}

\item \textbf{if} all checks pass: \textbf{reply} \{Settlement, true, $s_{\mathit{settle}}$\};

\end{enumerate}
}\end{functionality}\vspace{1em}

\begin{lemma}\label{lem:SettleC}
    The subroutine $\mathcal{F}^{\text{FRoll}}_{\text{settlement}}$ guarantees \emph{correct L2 settlement}.
\end{lemma}

\begin{proof}
    $\mathcal{F}^{\text{FRoll}}_{\text{settlement}}$ outputs success only after verifying that the peg-out carries the latest state (Step~1), the settling client is registered (Step~1), the peg-out is committed on L1 (Step~3), and the latest state is recorded on L1 (Step~3). Under honest verifier and L1 security assumptions, no mismatch can occur.
\end{proof}

\subsubsection{Update Round Functionality}

The update round subroutine checks that each pending self-submitted regular transaction, self-submitted fast-finality transaction, and escape-hatch settlement request from an honest client is processed within its respective L1-confirmation bound. For self-submitted regular transactions, the bound is $T_{L_1} + T_{\text{challenge}}$, accounting for both L1 confirmation and the subsequent fraud-proof challenge period; for self-submitted fast-finality transactions, the bound is simply $T_{L_1}$, since fast finality bypasses the challenge period and relies instead on the watchtower quorum certificate posted on L1; and for escape-hatch settlements, the bound is $T_{L_1} + T_{\text{challenge}}$, reflecting the additional challenge period before the peg-out is committed.

\vspace{1em}\begin{functionality}{Description of subroutine 
  $\mathcal{F}^{\text{FRoll}}_{\text{updRnd}} = 
  (\text{updRnd})$}{

\textbf{Participating roles:} \{updRnd\}

\noindent\textbf{Corruption model:} incorruptible

\noindent \textbf{Protocol parameters:}
\begin{itemize}
    \item $T_{L_1}$: L1 inclusion delay bound for regular transaction
    \item $T_{\text{challenge}}$: fraud-proof challenge period
\end{itemize}

}\end{functionality}\vspace{.5em}

\begin{functionality}{Description of 
  $\mathcal{M}^{\text{XRoll}}_{\text{updRnd}}$}{

\textbf{Implemented role(s):} \{updRnd\}

\noindent \textbf{CheckID}(\emph{pid}, \emph{sid}, 
\emph{role}): Accept all messages with the same \emph{sid}.

\vspace{0.5em}
\noindent \textbf{Main:}

\vspace{0.5em}

\textbf{recv} \{UpdateRound, $\mathsf{internalState}$\}
\textbf{from} I/O:
\begin{enumerate}[itemsep=0.5em]
\item Parse $\mathsf{round}$, $\mathsf{requestQueue}$,
and $\mathsf{CorruptionSet}$ from $\mathsf{internalState}$.

\item \textbf{send} \{ReadL1\} \textbf{to}
$(\pcur, \scur, \mathcal{F}_{\text{ledger}}:
\text{client}_{\text{L1}})$,
\textbf{wait for} \{ReadL1, $\mathit{L1ReadResult} =
\{\{\mathit{TX}\}_{\text{L1}}, \mathit{State}_{\text{L1}},
\{\mathit{pid}\}\}$\};

\item \textbf{// Check 1: Self-submitted regular transaction liveness}\\
For each $q \in \mathsf{requestQueue}$ proposed by an honest party $\mathit{pid} \notin \mathsf{CorruptionSet}$ at time $t_{\mathrm{sub}}$, with $\texttt{getType}(q) = \texttt{selfsubmit}$:
\begin{itemize}
    \item \textbf{if} $\mathsf{round} + 1 >
    t_{\mathrm{sub}} + T_{L_1} + T_{\text{challenge}}$:
    \textbf{reply} \{UpdateRound, false\};
    \hfill\cmt{Self-submitted tx liveness violated}
\end{itemize}

\item \textbf{// Check 2: Fast-finality self-submitted transaction liveness}\\
For each $q \in \mathsf{requestQueue}$ proposed by an honest party $\mathit{pid} \notin \mathsf{CorruptionSet}$ at time $t_{\mathrm{sub}}$, with $\texttt{getType}(q) = \texttt{selfsubmit\text{-}fast}$:
\begin{itemize}
    \item \textbf{if} $\mathsf{round} + 1 >
    t_{\mathrm{sub}} + T_{L_1}$:
    \textbf{reply} \{UpdateRound, false\};
    \hfill\cmt{Fast-finality tx liveness violated; no challenge period required}
\end{itemize}

\item \textbf{// Check 3: Escape-hatch settlement liveness}\\
For each $q \in \mathsf{requestQueue}$ proposed by an honest party $\mathit{pid} \notin \mathsf{CorruptionSet}$ at time $t_{\mathrm{sub}}$, with $\texttt{getType}(q) = \texttt{escape\text{-}hatch}$:
\begin{itemize}
    \item \textbf{if} $\mathsf{round} + 1 >
    t_{\mathrm{sub}} + T_{L_1} + T_{\text{challenge}}$:
    \textbf{reply} \{UpdateRound, false\};
    \hfill\cmt{Escape-hatch settlement liveness violated}
\end{itemize}

\item \textbf{reply} \{UpdateRound, true\};
\hfill
\end{enumerate}

}\end{functionality}\vspace{1em}

\begin{lemma}
\label{lem:LiveX}
The ideal functionality 
$\mathcal{F}^{\text{FRoll}}_{\text{layer2}}$ guarantees 
the following liveness properties:
\begin{enumerate}
    \item \textbf{(Protocol joining.)}
    $(f_{L_2} + f_{L_1},\, T_{L_2} + T_{L_1} + T_{\text{challenge}})$-liveness: under $f_{L_2}$ and $f_{L_1}$ corruption, an accepted Join request results in output $\{\text{Join}, s_{\mathit{init}}\}$ within $T_{L_2} + T_{L_1} + T_{\text{challenge}}$. The off-chain latency $T_{L_2}$ is not enforced by $\mathcal{F}^{\text{FRoll}}_{\text{updRnd}}$ due to asynchronous communication; the L1 latency $T_{L_1}$ and $T_{\text{challenge}}$ is inherited from $\mathcal{F}_{\text{ledger}}$.

    \item \textbf{(Regular state update.)} The functionality offers two state update liveness.
    \begin{itemize}
        \item \emph{Through operator.} $(f_{L_2} + f_{L_1},\, T_{L_2} + T_{L_1} + T_{\text{challenge}})$-liveness: under $f_{L_2}$ and $f_{L_1}$ corruption, an accepted regular update transaction request results in being added to $\mathsf{executedRequest}$ and changes in $\mathsf{stateList}$ within $T_{L_2} + T_{L_1} + T_{\text{challenge}}$.

        \item \emph{Self-submit.} $(f_{L_1},\, T_{L_1} + T_{\text{challenge}})$-liveness: under $f_{L_1}$ corruption, an accepted self-submit regular update transaction request results in being added to $\mathsf{executedRequest}$ and changes in $\mathsf{stateList}$ within $T_{L_1} + T_{\text{challenge}}$.
    \end{itemize}
    The off-chain latency $T_{L_2}$ is not enforced by $\mathcal{F}^{\text{FRoll}}_{\text{updRnd}}$ due to asynchronous communication; the L1 latency $T_{L_1}$ and $T_{\text{challenge}}$ is inherited from $\mathcal{F}_{\text{ledger}}$.

    \item \textbf{(Fast-finality state update.)} The functionality offers two fast-finality state update liveness.
    \begin{itemize}
        \item \emph{Through operator.} $(f_{L_2} + f_{L_1},\, T_{L_2} + T_{L_1})$-liveness: under $f_{L_2}$ and $f_{L_1}$ corruption, an accepted fast-finality update transaction request results in being added to $\mathsf{executedRequest}$ and changes in $\mathsf{stateList}$ within $T_{L_2} + T_{L_1}$.

        \item \emph{Self-submit.} $(f_{L_2} + f_{L_1},\, T_{L_1})$-liveness: under $f_{L_1}$, an accepted self-submit fast-finality update transaction request results in being added to $\mathsf{executedRequest}$ and changes in $\mathsf{stateList}$ within $T_{L_1}$.
    \end{itemize}

    \item \textbf{(Settlement.)} The functionality offers two settlement liveness.
    \begin{itemize}
        \item \emph{Collaborative.} $(f_{L_2} + f_{L_1},\, T_{L_2} + T_{L_1} + T_{\text{challenge}})$-liveness: under $f_{L_2}$ and $f_{L_1}$ corruption, an accepted settlement request results in $\mathsf{onchainState}$ reflecting $s_{\mathit{settle}}$ and\\ $\{\text{Settlement}, s_{\mathit{settle}}\}$ output within $T_{L_2} + T_{L_1} + T_{\text{challenge}}$.

        \item \emph{Escape-hatch.} $(f_{L_1},\, T_{L_1} + T_{\text{challenge}})$-liveness: under $f_{L_1}$ corruption, an accepted escape-hatch settlement request results in $\mathsf{onchainState}$ reflecting $s_{\mathit{settle}}$ and\\ $\{\text{Settlement}, s_{\mathit{settle}}\}$ output within $T_{L_1} + T_{\text{challenge}}$.
    \end{itemize}
\end{enumerate}

\end{lemma}

\begin{proof}
We argue each property separately.

\medskip
\noindent\textbf{(1) Protocol joining.}
Suppose liveness for join requests is violated. A join 
request involves three phases: (i)~the client deposits 
funds on L1, (ii)~the operator publishes the 
corresponding peg-in transaction on L1, and (iii)~the 
peg-in survives the challenge period 
$T_{\text{challenge}}$ without fraud proof. 
$\mathcal{F}^{\text{FRoll}}_{\text{join}}$ outputs 
success only if the peg-in is consistent with 
$s_{\mathit{init}}$ and both the deposit and peg-in are 
committed on L1 without fraud proof during 
$T_{\text{challenge}}$.

As long as the L1 blockchain guarantees liveness, the 
deposit is included within $T_{L_1}$. If at least one 
honest operator exists, the operator publishes the 
peg-in within $T_{L_2}$, which is then included on L1 
within another $T_{L_1}$. The honest client who works as verifier ensures 
that any incorrect publication is invalidated by a 
fraud proof during $T_{\text{challenge}}$, and a 
correct publication survives. The total bound is 
$T_{L_2} + T_{L_1} + T_{\text{challenge}}$, yielding 
$(f_{L_2} + f_{L_1},\, 
T_{L_2} + T_{L_1} + T_{\text{challenge}})$-liveness and 
contradicting the assumption.

\medskip
\noindent\textbf{(2) Regular state update.}
Suppose liveness for regular state update requests is 
violated. A regular update involves: (i)~the client 
submits a transaction to the operator or directly to L1 
via self-submit, (ii)~the transaction is included in a 
batch and published to L1 (if routed via the operator) 
or included directly on L1 (if self-submitted), and 
(iii)~the batch or self-submitted transaction survives 
$T_{\text{challenge}}$ without fraud proof. 
$\mathcal{F}^{\text{FRoll}}_{\text{update}}$ outputs 
success only if the new state is correctly computed, no 
conflict exists with previously executed requests, each 
honest sender's transaction appears in 
$\mathsf{requestQueue}$, and the transaction is 
committed on L1 without fraud proof during 
$T_{\text{challenge}}$.

We distinguish two cases based on the path taken by the 
honest client submitting the request at time~$t$.

\emph{Collaborative path.} The collaborative path is reflected fully by $\mathcal{F}^{\text{FRoll}}_{\text{update}}$ that waits for the triggering from the simulator since there is no restriction in $\mathcal{F}^{\text{FRoll}}_{\text{updRnd}}$ to simulate the off-chain communication latency, represented with $T_{L_1}$. Although the $\mathcal{F}^{\text{FRoll}}_{\text{update}}$ checks the committed transaction on L1 blockchain with regular commitment latency $T_{L_1}$. However, the commitment rule of collaborative path further relies on the fraud-proof scheme with a challenge window $T_{\text{challenge}}$, 
yielding $(f_{L_2} + f_{L_1},\, 
T_{L_2} + T_{L_1} + T_{\text{challenge}})$-liveness.

\emph{Self-submit path.} $\mathcal{F}^{\text{FRoll}}_{\text{updRnd}}$ captures the requirement for round update when the client bypasses the operator and 
submits the transaction directly to L1 via the 
self-submit mechanism. By L1 liveness, the transaction 
is committed on L1 within $T_{L_1}$. The honest 
verifier ensures only correct state transitions survive 
$T_{\text{challenge}}$. The request is therefore 
accessible via a read request within 
$T_{L_1} + T_{\text{challenge}}$, the adversary can not influence through delaying off-chain communication. Eventually, it yields 
$(f_{L_1},\, 
T_{L_1} + T_{\text{challenge}})$-liveness and 
depending only on L1 liveness and the challenge period.

In both cases, the corresponding time bound is 
achieved, contradicting the assumption.

\medskip
\noindent\textbf{(3) Fast-finality state update.} Suppose liveness for fast-finality update requests is violated. Unlike regular updates, fast-finality requests bypass the challenge period: $\mathcal{F}^{\text{FRoll}}_{\text{update}}$ outputs success as soon as at least $f{+}1$ out of $2f{+}1$ watchtower attestations on the fast-finality transaction are committed on L1, without requiring $T_{\text{challenge}}$ to elapse. We distinguish two cases based on the path taken by the honest client submitting the request at time~$t$.

\emph{Collaborative path.} The collaborative path is reflected fully by $\mathcal{F}^{\text{FRoll}}_{\text{update}}$ waiting for the simulator's trigger, since $\mathcal{F}^{\text{FRoll}}_{\text{updRnd}}$ imposes no restriction on the off-chain communication latency $T_{L_2}$. The check in $\mathcal{F}^{\text{FRoll}}_{\text{update}}$ requires the fast-finality transaction to be committed on L1 within the regular L1 latency $T_{L_1}$, alongside the watchtower-quorum attestations. Under the honest-majority watchtower committee assumption (at most $f$ out of $2f{+}1$ corrupted), at least $f{+}1$ honest watchtowers observe the fast-finality transaction on L1 and publish their attestations within the same $T_{L_1}$ bound. The challenge window is bypassed once the watchtower quorum lands, yielding $(f_{L_2} + f_{L_1},\, T_{L_2} + T_{L_1})$-liveness.

\emph{Self-submit path.} $\mathcal{F}^{\text{FRoll}}_{\text{updRnd}}$ captures the requirement for round update when the client bypasses the operator and submits the fast-finality transaction directly to L1. By L1 liveness, the transaction is committed on L1 within $T_{L_1}$. The honest-majority watchtower committee then accumulates at least $f{+}1$ attestations on L1 within the same $T_{L_1}$ bound, again bypassing the challenge window. The request is therefore accessible via a read request within $T_{L_1}$, and the adversary cannot influence the bound through delaying off-chain communication. This yields $(f_{L_1},\, T_{L_1})$-liveness, depending only on L1 liveness and the watchtower quorum.

In both cases, the corresponding time bound is achieved, contradicting the assumption. The difference from the regular-update bound is that the L1 commitment latency no longer compounds with the challenge window: fast-finality replaces $T_{L_1} + T_{\text{challenge}}$ with $T_{L_1}$ alone, since the watchtower quorum delivers finality within a single L1 inclusion delay.

\medskip
\noindent\textbf{(4) Settlement.}
Suppose liveness for settlement requests is violated. 
Settlement can proceed via two paths: collaborative 
(through the operator) or escape-hatch (directly to 
L1). $\mathcal{F}^{\text{FRoll}}_{\text{settlement}}$ 
outputs success only if the peg-out transaction is 
consistent with the latest state and is committed on 
L1 without a fraud proof during $T_{\text{challenge}}$.

\emph{Collaborative path.} The collaborative path is fully captured by $\mathcal{F}^{\text{FRoll}}_{\text{update}}$, which waits for the simulator's trigger, since $\mathcal{F}^{\text{FRoll}}_{\text{updRnd}}$ imposes no constraint on the off-chain latency $T_{L_2}$ controlled by the adversary. The L1-commitment check resolves within latency $T_{L_1}$ plus the challenge window $T_{\text{challenge}}$. The settlement is therefore accessible via a read request within $T_{L_2} + T_{L_1} + T_{\text{challenge}}$, yielding $(f_{L_2} + f_{L_1},\, T_{L_2} + T_{L_1} + T_{\text{challenge}})$-liveness.

\emph{Escape-hatch path.} $\mathcal{F}^{\text{FRoll}}_{\text{updRnd}}$ captures the requirement for round update when the client bypasses the operator and publishes the peg-out directly to L1, so settlement liveness holds even if every operator is corrupted. By L1 liveness, the peg-out is committed on L1 within $T_{L_1}$ plus $T_{\text{challenge}}$. The settlement is therefore accessible via a read request within $T_{L_1} + T_{\text{challenge}}$, and the adversary cannot influence the bound through delaying off-chain communication. This yields $(f_{L_1},\, T_{L_1} + T_{\text{challenge}})$-liveness, depending only on L1 liveness and the challenge period.

In both cases, the corresponding time bound is 
achieved, contradicting the assumption.

\end{proof}


\subsection{Security Proof}
\label{apd:XRollProof}

After proposing the ideal functionality and real-world implementation, we now show the security of the \xr protocol. To start with we first show the ideal functionality captures all the security properties:

\ThmidealCross*

\begin{proof}
    According to Lemma~\ref{lem:OpenC}--\ref{lem:LiveX}, the ideal functionality $\mathcal{F}^{\text{FRoll}}_{\text{layer2}}$ guarantees all security properties.
\end{proof}

After defining the ideal functionality $\mathcal{F}^{\text{FRoll}}_{\text{layer2}}$, we prove that the real \xr protocol iUC-realizes it. The proof is done in 7 steps of successive game replacement. We first define a simulator $\mathcal{S}_{\text{FRoll}}$ that internally simulates a full run of $\mathcal{P}^{\text{FRoll}}$, and a dummy functionality $\mathcal{F}^{\text{FRoll}}_{\text{dummy}}$ that relays messages between $\mathcal{E}$ and $\mathcal{S}_{\text{FRoll}}$. This base ideal execution yields the same distribution of messages to $\mathcal{E}$ as the real execution. We use the execution ensemble $\mathsf{EXEC}$ to denote the messages observed by $\mathcal{E}$, including the output for the input request and the leakage to adversary, when interacting with adversary $\mathcal{A}$, real protocol $\mathcal{P}$, ideal functionality $\mathcal{F}$ and simulator $\mathcal{S}$ in the proofs that follow.

In each subsequent step, we incrementally add interaction between the simulator and the ideal functionality and extend the functionality, thereby forming the corresponding subroutines, while keeping the changes transparent to both $\mathcal{E}$ and $\mathcal{A}$. We continue until we obtain the target functionality $\mathcal{F}^{\text{FRoll}}_{\text{layer2}}$ defined by our framework. At every step, the simulator is adjusted so that the new ideal execution is indistinguishable from the previous one. For each transition, we discuss the differences relative to the prior step and prove that, given the same inputs from $\mathcal{E}$ and $\mathcal{A}$, the resulting outputs remain the same up to computational indistinguishability under any adversarial influence strategy.

We begin by defining the dummy ideal functionality
$\mathcal{F}^{\text{FRoll}}_{\text{dummy}}=(\mathcal{F}_{\text{client-dummy}}, \mathcal{F}_{\text{ledger}} \mid \perp)$
and the simulator $\mathcal{S}_{\text{FRoll}}$ as follows. The dummy functionality forwards every request from $\mathcal{E}$ to the simulator and returns the simulator's response unchanged. The simulator $\mathcal{S}_{\text{FRoll}}$ runs $\mathcal{P'}^{\text{FRoll}}$ internally and produces identical outputs.

\vspace{1em}\begin{functionality}{Description of $\mathcal{M}_{\text{client-dummy}}$ of $\mathcal{F}^{\text{FRoll}}_{\text{dummy}}$}{

\textbf{Implemented role(s):} \{client-dummy\}

\noindent \textbf{Main:}

\textbf{recv} any request \textbf{from} I/O:
\begin{enumerate}[itemsep=0.5em]
    \item Forward request to $\mathcal{S}$ through NET;
\end{enumerate}

\hrule\vspace{0.5em}

\textbf{recv} any message \textbf{from} NET:
\begin{enumerate}
    \item Output the message to $\mathcal{E}$ through I/O;
\end{enumerate}

}\end{functionality}\vspace{.5em}

\begin{functionality}{Description of simulator $\mathcal{S}_{\text{FRoll}}$}{

$\mathcal{S}_{\text{FRoll}}$ internally simulates $\mathcal{P'}^{\text{FRoll}}$, a copy of the real protocol $\mathcal{P}^{\text{FRoll}}$ as defined in Section~\ref{apd:XRollReal}, including both roles of the client machine (client and verifier), the operator, and the watchtower committee.

\vspace{0.5em}
\textbf{Real protocol simulation:}
\begin{itemize}[itemsep=0.3em]
    \item $\mathcal{S}_{\text{FRoll}}$ simulates honest entities for all roles: client, verifier, operator, and watchtower, according to the real protocol. The verifier role monitors L1 and publishes fraud proofs.
    \item If participants are corrupted, $\mathcal{S}_{\text{FRoll}}$ leaks the corresponding messages sent to corrupted entities to the adversary $\mathcal{A}$ and continues simulating honest parties based on $\mathcal{A}$'s instructions.
\end{itemize}

\textbf{Network communication from/to the environment:}
\begin{itemize}[itemsep=0.3em]
    \item Messages that $\mathcal{S}_{\text{FRoll}}$ receives on the network interface (from $\mathcal{E}$/$\mathcal{A}$) are forwarded to $\mathcal{P'}^{\text{FRoll}}$.
    \item Messages sent by $\mathcal{P'}^{\text{FRoll}}$ on its network interface (to $\mathcal{E}$/$\mathcal{A}$) are forwarded to the environment.
\end{itemize}

\textbf{Input requests and outputs:}
\begin{itemize}[itemsep=0.3em]
    \item Unlike $\mathcal{P}^{\text{FRoll}}$, which receives inputs directly from $\mathcal{E}$, the simulation $\mathcal{P'}^{\text{FRoll}}$ receives requests forwarded from $\mathcal{F}^{\text{FRoll}}_{\text{layer2}}$. Instead of sending outputs directly to $\mathcal{E}$, $\mathcal{S}_{\text{FRoll}}$ sends them to $\mathcal{F}^{\text{FRoll}}_{\text{layer2}}$.
\end{itemize}

\textbf{Message delivery:}
\begin{itemize}[itemsep=0.3em]
    \item The \xr protocol assumes asynchronous communication via $\mathcal{F}^{\text{FRoll}}_{\text{com}}$. The simulator bookkeeps all messages in $\mathcal{P'}^{\text{FRoll}}$ and triggers delivery according to the adversary's scheduling decisions, mirroring the real-world protocol.
\end{itemize}

\textbf{Corruption handling:}
\begin{itemize}[itemsep=0.3em]
    \item $\mathcal{S}_{\text{FRoll}}$ keeps the corruption status of entities in $\mathcal{P}^{\text{FRoll}}$, $\mathcal{P'}^{\text{FRoll}}$ and $\mathcal{F}^{\text{FRoll}}_{\text{layer2}}$ synchronized. When an entity in $\mathcal{P'}^{\text{FRoll}}$ becomes corrupted, $\mathcal{S}_{\text{FRoll}}$ corrupts the corresponding entity in $\mathcal{F}^{\text{FRoll}}_{\text{layer2}}$ before continuing.
    \item Adversarial commands for corrupted participants (e.g., publishing on $\mathcal{F}_{\text{ledger}}$, sending via $\mathcal{F}^{\text{FRoll}}_{\text{com}}$) are forwarded to $\mathcal{P'}^{\text{FRoll}}$.
    \item When a corrupted participant in $\mathcal{P'}^{\text{FRoll}}$ wants to output to $\mathcal{E}$, $\mathcal{S}_{\text{FRoll}}$ instructs the corresponding entity in $\mathcal{F}^{\text{FRoll}}_{\text{layer2}}$ to output.
\end{itemize}
}\end{functionality}\vspace{1em}

\begin{lemma}
\label{lem:XRoll1}
    For all PPT adversaries $\mathcal{A}$, there exists a PPT simulator $\mathcal{S}_{\text{FRoll}}$ such that for all PPT environments $\mathcal{E}$ and all security parameters $k\in\mathbb{N}$,
    $\mathsf{EXEC}^{\mathcal{P}^{\text{FRoll}}}_{\mathcal{A},\mathcal{E}}(k)\ \stackrel{c}{\approx}\
    \mathsf{EXEC}^{\mathcal{F}^{\text{FRoll}}_{\text{dummy}}}_{\mathcal{S}_{\text{FRoll}},\mathcal{E}}(k)$,
    where $\stackrel{c}{\approx}$ denotes computational indistinguishability.
\end{lemma}

\begin{proof}
    Fix an arbitrary PPT environment $\mathcal{E}$ and adversary $\mathcal{A}$. We argue that the execution ensembles in the real and ideal worlds are computationally indistinguishable by analyzing the three components observable by $\mathcal{E}$.

    \medskip
    \noindent\textbf{Observable components.} In the real world, the execution ensemble $\mathsf{EXEC}^{\mathcal{P}^{\text{FRoll}}}_{\mathcal{A},\mathcal{E}}(k)$ consists of:
    \begin{enumerate}
        \item \emph{I/O outputs} delivered to $\mathcal{E}$ by $\mathcal{P}^{\text{FRoll}}_{\text{client}}$:
        $\{\text{Join}, s_{\mathit{init}}\}$,
        $\{\text{Settlement},\\ \mathsf{onchainState}\}$,
        $\{\text{Read}, \mathit{ReadResult}\}$,
        $\{\text{GetCurRound}, \mathit{round}\}$.
        \item \emph{On-chain transactions} committed on $\mathcal{F}_{\text{ledger}}$ during execution:
        $\mathit{TX}_{\mathit{deposit}}$, $\mathit{TX}_{\mathit{peg\text{-}in}}$, $\mathit{TX}_{\mathit{peg\text{-}out}}$, operator batch publications, $\mathit{TX}_{\mathit{fraud}}$, and watchtower certificate transactions $\mathit{TX}_{\mathit{WT}}$.
        \item \emph{Adversarial leakage}, comprising messages received by corrupted parties during the protocol execution, and the corrupted parties' internal state.
    \end{enumerate}

    In the ideal world, the simulator $\mathcal{S}_{\text{FRoll}}$ internally runs $\mathcal{P'}^{\text{FRoll}}$ and interacts with the dummy functionality $\mathcal{F}^{\text{FRoll}}_{\text{dummy}}$, which by definition forwards every request from $\mathcal{E}$ to $\mathcal{S}_{\text{FRoll}}$ unchanged and relays the simulator's responses back to $\mathcal{E}$. We show that each component is computationally indistinguishable across the two worlds.

    \medskip
    \noindent\textbf{(1) I/O outputs.} Since $\mathcal{F}^{\text{FRoll}}_{\text{dummy}}$ acts as a transparent relay, $\mathcal{S}_{\text{FRoll}}$ receives exactly the same sequence of requests as $\mathcal{P}^{\text{FRoll}}_{\text{client}}$ would in the real world. By construction, $\mathcal{S}_{\text{FRoll}}$ executes the same client, verifier, operator, and watchtower logic inside $\mathcal{P'}^{\text{FRoll}}$ under the same adversarial scheduling, producing the same I/O outputs. The only potential difference arises from the randomness of $\mathcal{F}_{\text{sig}}$: signature strings carried by client transactions, fraud proofs, and watchtower attestations may differ between the two worlds because fresh randomness is sampled independently. However, since $\mathcal{F}_{\text{sig}}$ realizes EUF-CMA security, signatures generated on the same messages are computationally indistinguishable. Hence the I/O outputs are computationally indistinguishable.

    \medskip
    \noindent\textbf{(2) On-chain transactions.} Since $\mathcal{F}^{\text{FRoll}}_{\text{dummy}}$ forwards all requests to $\mathcal{S}_{\text{FRoll}}$, the simulated protocol $\mathcal{P'}^{\text{FRoll}}$ generates and publishes the same set of transactions to $\mathcal{F}_{\text{ledger}}$ as $\mathcal{P}^{\text{FRoll}}$ would in the real world, including operator checkpoint transactions, verifier-published fraud proofs, and watchtower certificate transactions. Transactions may contain different signature values due to independent randomness in $\mathcal{F}_{\text{sig}}$, but by the EUF-CMA security of the signature scheme, the transaction distributions are computationally indistinguishable.

    \medskip
    \noindent\textbf{(3) Adversarial leakage.} By the definition of $\mathcal{S}_{\text{FRoll}}$, the corruption status of all entities (client, verifier, operator, watchtower) is kept synchronized between $\mathcal{P'}^{\text{FRoll}}$ and $\mathcal{F}^{\text{FRoll}}_{\text{dummy}}$. Since $\mathcal{F}^{\text{FRoll}}_{\text{dummy}}$ forwards all $\mathcal{E}$-requests to the simulator, $\mathcal{S}_{\text{FRoll}}$ can reconstruct the same internal state and message history for corrupted parties as in the real execution. Consequently, The leakage delivered to $\mathcal{A}$ is computationally indistinguishable up to signature randomness, which is again computationally indistinguishable by EUF-CMA security.

    \medskip
    Conclusively, we have:
    $\mathsf{EXEC}^{\mathcal{P}^{\text{FRoll}}}_{\mathcal{A}, \mathcal{E}}(k)\stackrel{c}{\approx}\mathsf{EXEC}^{\mathcal{F}^{\text{FRoll}}_{\text{dummy}}}_{\mathcal{S}_{\text{FRoll}}, \mathcal{E}}(k)$.
\end{proof}

Next, we extend the dummy functionality with the submission subroutine, yielding
$\mathcal{F}^{\text{FRoll}}_{\text{layer2-submit}} = (\mathcal{F}_{\text{client-submit}}, \mathcal{F}_{\text{ledger}} \mid \mathcal{F}^{\text{FRoll}}_{\text{submit}})$.
The subroutine $\mathcal{F}^{\text{FRoll}}_{\text{submit}}$ is defined in Appendix~\ref{apd:XRollsubmit}. The intermediate client functionality adds the Submit request routing through $\mathcal{F}^{\text{FRoll}}_{\text{submit}}$; all other requests are forwarded to $\mathcal{S}$ unchanged.

\vspace{1em}\begin{functionality}{Description of $\mathcal{M}_{\text{client-submit}}$ of $\mathcal{F}^{\text{FRoll}}_{\text{layer2-submit}}$}{

\textbf{Implemented role(s):} \{client-submit\}

\noindent\textbf{Main:}\vspace{0.5em}

\textbf{recv} \{Submit, $\mathit{request}$\} \textbf{from} I/O:
\begin{enumerate}[itemsep=0.5em]
\item \textbf{send} \{Submit, $\mathit{request}$, $\mathsf{internalState}$\} \textbf{to} $(\pcur, \scur, \mathcal{F}^{\text{FRoll}}_{\text{submit}}:\text{submit})$,
\textbf{wait for} \{Submit, $\mathit{response}$\} s.t. $\mathit{response} \in \{\text{true}, \text{false}\}$;
\item \textbf{if} $\mathit{response} =$ true: $\mathsf{requestQueue}.\text{add}(\mathit{request})$;
\textbf{send} $\mathit{request}$ \textbf{to} $\mathcal{S}$ via NET;
\end{enumerate}

\hrule\vspace{0.5em}

\textbf{recv} \{ReadL1\} \textbf{from} I/O or NET:
\begin{enumerate}[itemsep=0.5em]
    \item \textbf{send} \{Read\} \textbf{to} $(\pcur, \scur, \mathcal{F}_{\text{ledger}}: \text{client}_{\text{L1}})$; \textbf{wait for} $\mathit{L1ReadResult}$;
    \item \textbf{reply} \{ReadL1, $\mathit{L1ReadResult}$\} via I/O;
\end{enumerate}

\hrule\vspace{0.5em}

\textbf{recv} any other message \textbf{from} I/O:
\begin{enumerate}
    \item Forward to $\mathcal{S}$ through NET;
\end{enumerate}

\hrule\vspace{0.5em}

\textbf{recv} any message \textbf{from} NET:
\begin{enumerate}
    \item Output the message to $\mathcal{E}$ through I/O;
\end{enumerate}

}\end{functionality}\vspace{.5em}

\begin{functionality}{Description of simulator $\mathcal{S}_{\text{XRoll-submit}}$}{
The simulator $\mathcal{S}_{\text{XRoll-submit}}$ behaves the same as $\mathcal{S}_{\text{FRoll}}$.
}\end{functionality}\vspace{1em}

\begin{lemma}
\label{lem:XRoll2}
    For all PPT adversaries $\mathcal{A}$, there exists a PPT simulator $\mathcal{S}_{\text{XRoll-submit}}$ such that for all PPT environments $\mathcal{E}$ and all security parameters $k\in\mathbb{N}$,
    \[
    \mathsf{EXEC}^{\mathcal{F}^{\text{FRoll}}_{\text{dummy}}}_{\mathcal{S}_{\text{FRoll}},\mathcal{E}}(k)
    \ \stackrel{c}{\approx}\
    \mathsf{EXEC}^{\mathcal{F}^{\text{FRoll}}_{\text{layer2-submit}}}_{\mathcal{S}_{\text{XRoll-submit}},\mathcal{E}}(k).
    \]
\end{lemma}

\begin{proof}
    Fix an arbitrary PPT environment $\mathcal{E}$. The simulator logic is identical in both executions; the only difference is the ideal functionality through which requests are forwarded to the simulator. We analyze the three observable components.

    \medskip
    \noindent\textbf{Difference between the two executions.} In $\mathcal{F}^{\text{FRoll}}_{\text{dummy}}$, every request from $\mathcal{E}$ is forwarded directly to $\mathcal{S}_{\text{FRoll}}$. In $\mathcal{F}^{\text{FRoll}}_{\text{layer2-submit}}$, the subroutine $\mathcal{F}^{\text{FRoll}}_{\text{submit}}$ intercepts each request and checks semantic validity before forwarding. Requests that fail these checks are rejected and never reach the simulator.

    \medskip
    \noindent\textbf{(1) I/O outputs.} In the real protocol $\mathcal{P}^{\text{FRoll}}$ (and hence in $\mathcal{P'}^{\text{FRoll}}$), the client machine already enforces the same validity checks for all four request types (Join, Submit, SubmitFast, Settlement): only predefined request types are processed, and malformed or out-of-phase requests produce no output. Therefore, any request rejected by $\mathcal{F}^{\text{FRoll}}_{\text{submit}}$ would also produce no output inside the simulation. For accepted requests, both simulators execute the same client logic and produce the same I/O outputs. Hence, the I/O outputs are identical.

    \medskip
    \noindent\textbf{(2) On-chain transactions.} Since only accepted requests trigger protocol actions in $\mathcal{P'}^{\text{FRoll}}$, and the set of accepted requests is the same in both executions, the transactions published to $\mathcal{F}_{\text{ledger}}$ are identical except for signature values. By EUF-CMA security, the on-chain transactions are computationally indistinguishable.

    \medskip
    \noindent\textbf{(3) Adversarial leakage.} Since invalid requests are rejected by the ideal functionality and ignored by honest parties during simulation, no leakage is generated toward the adversary in either world. Moreover, the corruption status remains synchronized across both worlds. Consequently, the leakage delivered to $\mathcal{A}$ is computationally indistinguishable in both executions.

    \medskip
    Conclusively,
    $\mathsf{EXEC}^{\mathcal{F}^{\text{FRoll}}_{\text{dummy}}}_{\mathcal{S}_{\text{FRoll}}, \mathcal{E}}(k) \stackrel{c}{\approx} \mathsf{EXEC}^{\mathcal{F}^{\text{FRoll}}_{\text{layer2-submit}}}_{\mathcal{S}_{\text{FRoll-submit}}, \mathcal{E}}(k)$.
\end{proof}

Next, we extend $\mathcal{F}^{\text{FRoll}}_{\text{layer2-submit}}$ with the join subroutine, yielding $\mathcal{F}^{\text{FRoll}}_{\text{layer2-join}} = (\mathcal{F}_{\text{client-join}}, \mathcal{F}_{\text{ledger}} \mid \mathcal{F}^{\text{FRoll}}_{\text{submit}}, \mathcal{F}^{\text{FRoll}}_{\text{join}})$. The intermediate client functionality additionally routes the Join request (from NET) through $\mathcal{F}^{\text{FRoll}}_{\text{join}}$.

\vspace{1em}\begin{functionality}{Description of $\mathcal{M}_{\text{client-join}}$ of $\mathcal{F}^{\text{FRoll}}_{\text{layer2-join}}$}{

\textbf{Implemented role(s):} \{client-join\}

\vspace{0.5em}\noindent \noindent \textbf{Main:}

\textbf{recv} \{Submit, $\mathit{request}$\} \textbf{from} I/O:

\begin{enumerate}[itemsep=0.5em]
\item \textbf{send} \{Submit, $\mathit{request}$, $\mathsf{internalState}$\} \textbf{to} $(\pcur, \scur, \mathcal{F}^{\text{FRoll}}_{\text{submit}}:\text{submit})$,
\textbf{wait for} \{Submit, $\mathit{response}$\} s.t. $\mathit{response} \in \{\text{true}, \text{false}\}$;
\item \textbf{if} $\mathit{response} =$ true: $\mathsf{requestQueue}.\text{add}(\mathit{request})$;
\textbf{send} $\mathit{request}$ \textbf{to} $\mathcal{S}$ via NET;
\end{enumerate}

\hrule
\vspace{0.5em}

\textbf{recv} \{Join, $\mathit{Attachment}$\} \textbf{from} NET:

\begin{enumerate}[itemsep=0.5em]
\item \textbf{send} \{Join, $\mathit{Attachment}$, $\mathsf{internalState}$\} \textbf{to} $(\pcur, \scur, \mathcal{F}^{\text{FRoll}}_{\text{join}}:\text{join})$,
\textbf{wait for} \{Join, $\mathit{response}$\} s.t. $\mathit{response} \in \{\text{true}, \text{false}\}$;
\item \textbf{if} $\mathit{response} =$ true: update $\mathsf{internalState}$ according to $\mathit{Attachment}$;
\textbf{reply} \{Join, $s_{\mathit{init}}$\} via I/O;
\end{enumerate}

\hrule
\vspace{0.5em}

\textbf{recv} any request \textbf{from} I/O:
\begin{enumerate}[itemsep=0.5em]
    \item Forward request to $\mathcal{S}$ through NET;
\end{enumerate}

\hrule
\vspace{0.5em}

\textbf{recv} any message \textbf{from} NET:

\begin{enumerate}
    \item Output the message to $\mathcal{E}$ through I/O;
\end{enumerate}

}\end{functionality}\vspace{.5em}

\begin{functionality}{Description of simulator $\mathcal{S}_{\text{XRoll-join}}$}{

The simulator $\mathcal{S}_{\text{XRoll-join}}$ behaves identically to $\mathcal{S}_{\text{XRoll-submit}}$, except for the following additional behavior upon detecting a completed join:

\vspace{0.5em}
\textbf{Join interaction with $\mathcal{F}^{\text{FRoll}}_{\text{layer2-join}}$:}

\begin{enumerate}[itemsep=0.5em]
\item $\mathcal{S}_{\text{XRoll-join}}$ monitors $\mathcal{P'}^{\text{FRoll}}$. When it detects that a client entity is about to produce $\{\text{Join}, s_{\mathit{init}}\}$, it intercepts and prepares $\mathit{Attachment}$ by extracting $s_{\mathit{init}}$ and $\mathit{TX}_{\mathit{peg\text{-}in}}$ from the simulation state.
\item $\mathcal{S}_{\text{XRoll-join}}$ sends $\{\text{Join}, \mathit{Attachment}\}$ to $\mathcal{F}^{\text{FRoll}}_{\text{layer2-join}}$ via NET.
\end{enumerate}

}\end{functionality}\vspace{1em}

\begin{lemma}
\label{lem:XRoll3}
    For all PPT adversaries $\mathcal{A}$, there exists a PPT simulator $\mathcal{S}_{\text{XRoll-join}}$ such that for all PPT environments $\mathcal{E}$ and all security parameters $k\in\mathbb{N}$,
    \[
    \mathsf{EXEC}^{\mathcal{F}^{\text{FRoll}}_{\text{layer2-submit}}}_{\mathcal{S}_{\text{XRoll-submit}},\mathcal{E}}(k)
    \ \stackrel{c}{\approx}\
    \mathsf{EXEC}^{\mathcal{F}^{\text{FRoll}}_{\text{layer2-join}}}_{\mathcal{S}_{\text{XRoll-join}},\mathcal{E}}(k).
    \]
\end{lemma}

\begin{proof}
    Fix an arbitrary PPT environment $\mathcal{E}$. The only difference between the two executions is that $\mathcal{F}^{\text{FRoll}}_{\text{layer2-join}}$ routes the Join request (from NET) through $\mathcal{F}^{\text{FRoll}}_{\text{join}}$, whereas in $\mathcal{F}^{\text{FRoll}}_{\text{layer2-submit}}$ the join output is produced entirely by the simulator. We analyze the three observable components.

    \medskip
    \noindent\textbf{(1) I/O outputs.} The only I/O output affected is $\{\text{Join}, s_{\mathit{init}}\}$. We consider two cases.

    \emph{Case~1: Successful join.} In $\mathsf{EXEC}^{\mathcal{F}^{\text{FRoll}}_{\text{layer2-submit}}}_{\mathcal{S}_{\text{FRoll-submit}},\mathcal{E}}(k)$, the simulator produces a join output when the simulated $\mathcal{P'}^{\text{FRoll}}$ completes the joining procedure: the deposit transaction is committed on L1, the operator publishes the peg-in inside a checkpoint transaction, and the checkpoint either survives $T_{\text{challenge}}$ without a valid fraud proof or accumulates $\geq f{+}1$ watchtower attestations on L1. In $\mathsf{EXEC}^{\mathcal{F}^{\text{FRoll}}_{\text{layer2-join}}}_{\mathcal{S}_{\text{FRoll-join}},\mathcal{E}}(k)$, the simulator instead extracts $\mathit{Attachment} = \{s_{\mathit{init}}, \mathit{TX}_{\mathit{peg\text{-}in}}\}$ from the simulation state and sends it to $\mathcal{F}^{\text{FRoll}}_{\text{join}}$, which outputs to $\mathcal{E}$ only if all checks pass. A successful join in $\mathcal{P'}^{\text{FRoll}}$ implies that $\mathit{TX}_{\mathit{peg\text{-}in}}$ is consistent with $s_{\mathit{init}}$, both the deposit and peg-in transactions are committed on L1, and $s_{\mathit{init}}$ is recorded on L1. These are exactly the checks in $\mathcal{F}^{\text{FRoll}}_{\text{join}}$ (Steps~1--3). Under the assumption of at least one honest verifier (who publishes a fraud proof against any invalid checkpoint within $T_{\text{challenge}}$), an honest-majority watchtower committee (at most $f$ out of $2f{+}1$ corrupted), and the EUF-CMA security of $\mathcal{F}_{\text{sig}}$ (which prevents the adversary from forging signatures on a malformed peg-in), $\mathcal{F}^{\text{FRoll}}_{\text{join}}$ accepts whenever $\mathcal{P'}^{\text{FRoll}}$ completes.

    \emph{Case~2: Failed join.} If the operator does not publish the peg-in checkpoint, or a fraud proof invalidates the checkpoint within $T_{\text{challenge}}$, or fewer than $f{+}1$ watchtower attestations are committed on L1, $\mathcal{P'}^{\text{FRoll}}$ does not complete the join. Correspondingly, $\mathcal{S}_{\text{FRoll-join}}$ does not send the Join request to $\mathcal{F}^{\text{FRoll}}_{\text{join}}$, so no output is produced.

    In all cases, the I/O outputs are identical.

    \medskip
    \noindent\textbf{(2) On-chain transactions.} The set of transactions published to $\mathcal{F}_{\text{ledger}}$ is fully determined by $\mathcal{P'}^{\text{FRoll}}$, which runs the real-world protocol on the input requests forwarded by the same $\mathcal{F}^{\text{FRoll}}_{\text{submit}}$ in both executions. Transactions published by honest participants: including the client deposit, operator checkpoint, watchtower attestations, and any verifier-published fraud proofs, therefore differ across the two worlds only in their signature values, which are computationally indistinguishable by the EUF-CMA security of $\mathcal{F}_{\text{sig}}$. When participants are corrupted, the transactions they publish on L1 are likewise computationally indistinguishable, since the simulator forwards all corrupted-party messages to $\mathcal{P'}^{\text{FRoll}}$ and signature randomness is the only source of variation.

    \medskip
    \noindent\textbf{(3) Adversarial leakage.} In both executions, the simulator extracts the plaintext leakage (under our plaintext-leakage assumption) from the ideal functionality and uses it to drive $\mathcal{P'}^{\text{FRoll}}$, generating leakage to the adversary according to the corruption status. Since the two executions feed $\mathcal{P'}^{\text{FRoll}}$ the same inputs (those accepted by $\mathcal{F}^{\text{FRoll}}_{\text{submit}}$) and run the same real-world protocol, $\mathsf{internalState}$ at every honest participant, including client, verifier, operator, and watchtower roles, changes identically across the two executions. The interposition of $\mathcal{F}^{\text{FRoll}}_{\text{join}}$ affects only how the join I/O output is produced (via the ideal functionality's checks rather than directly by the simulator); it does not change any network-side message exchange or any honest party's internal computation, so the leakage delivered to corrupted parties is unaffected. Consequently, the leaked $\mathsf{internalState}$ and received messages delivered to corrupted parties are computationally indistinguishable across the two executions.

    \medskip
    Conclusively,
    $\mathsf{EXEC}^{\mathcal{F}^{\text{FRoll}}_{\text{layer2-submit}}}_{\mathcal{S}_{\text{FRoll-submit}}, \mathcal{E}}(k)\stackrel{c}{\approx}\mathsf{EXEC}^{\mathcal{F}^{\text{FRoll}}_{\text{layer2-join}}}_{\mathcal{S}_{\text{FRoll-join}}, \mathcal{E}}(k)$.
\end{proof}

Next, we extend $\mathcal{F}^{\text{FRoll}}_{\text{layer2-join}}$ with the update subroutine, yielding $\mathcal{F}^{\text{FRoll}}_{\text{layer2-update}} = (\mathcal{F}_{\text{client-update}}, \mathcal{F}_{\text{ledger}} \mid \mathcal{F}^{\text{FRoll}}_{\text{submit}}, \mathcal{F}^{\text{FRoll}}_{\text{join}}, \mathcal{F}^{\text{FRoll}}_{\text{update}})$.

\vspace{1em}\begin{functionality}{Description of $\mathcal{M}_{\text{client-update}}$ of $\mathcal{F}^{\text{FRoll}}_{\text{layer2-update}}$}{

\textbf{Implemented role(s):} \{client-update\}

\vspace{0.5em}\noindent \noindent \textbf{Main:}

\textbf{recv} \{Submit, $\mathit{request}$\} \textbf{from} I/O:

\begin{enumerate}[itemsep=0.5em]
\item \textbf{send} \{Submit, $\mathit{request}$, $\mathsf{internalState}$\} \textbf{to} $(\pcur, \scur, \mathcal{F}^{\text{FRoll}}_{\text{submit}}:\text{submit})$,
\textbf{wait for} \{Submit, $\mathit{response}$\} s.t. $\mathit{response} \in \{\text{true}, \text{false}\}$;
\item \textbf{if} $\mathit{response} =$ true: $\mathsf{requestQueue}.\text{add}(\mathit{request})$;
\textbf{send} $\mathit{request}$ \textbf{to} $\mathcal{S}$ via NET;
\end{enumerate}

\hrule
\vspace{0.5em}

\textbf{recv} \{Join, $\mathit{Attachment}$\} \textbf{from} NET:

\begin{enumerate}[itemsep=0.5em]
\item \textbf{send} \{Join, $\mathit{Attachment}$, $\mathsf{internalState}$\} \textbf{to} $(\pcur, \scur, \mathcal{F}^{\text{FRoll}}_{\text{join}}:\text{join})$,
\textbf{wait for} \{Join, $\mathit{response}$\} s.t. $\mathit{response} \in \{\text{true}, \text{false}\}$;
\item \textbf{if} $\mathit{response} =$ true: update $\mathsf{internalState}$ according to $\mathit{Attachment}$;
\textbf{reply} \{Join, $s_{\mathit{init}}$\} via I/O;
\end{enumerate}

\hrule
\vspace{0.5em}

\textbf{recv} \{Update, $\mathit{Attachment}$\} \textbf{from} NET:

\begin{enumerate}[itemsep=0.5em]
\item \textbf{send} \{Update, $\mathit{Attachment}$, $\mathsf{internalState}$\} \textbf{to} $(\pcur, \scur, \mathcal{F}^{\text{FRoll}}_{\text{update}}:\text{update})$,
\textbf{wait for} \{Update, $\mathit{response}$, $\mathit{newState}$, $\mathit{executedReq}$\};
\item \textbf{if} $\mathit{response} =$ true: update $\mathsf{internalState}$ with $\mathit{newState}$ and $\mathit{executedReq}$;
\end{enumerate}

\hrule
\vspace{0.5em}

\textbf{recv} any request \textbf{from} I/O:
\begin{enumerate}[itemsep=0.5em]
    \item Forward request to $\mathcal{S}$ through NET;
\end{enumerate}

\hrule
\vspace{0.5em}

\textbf{recv} any message \textbf{from} NET:

\begin{enumerate}
    \item Output the message to $\mathcal{E}$ through I/O;
\end{enumerate}

}\end{functionality}\vspace{.5em}

\begin{functionality}{Description of simulator $\mathcal{S}_{\text{XRoll-update}}$}{

The simulator $\mathcal{S}_{\text{XRoll-update}}$ behaves identically to $\mathcal{S}_{\text{XRoll-join}}$, except for the following additional behavior upon detecting a confirmed state update:

\vspace{0.5em}
\textbf{State update interaction with $\mathcal{F}^{\text{FRoll}}_{\text{layer2-update}}$:}

\begin{enumerate}[itemsep=0.5em]
\item $\mathcal{S}_{\text{XRoll-update}}$ monitors $\mathcal{P'}^{\text{FRoll}}$ for confirmed state updates. For regular requests, it detects when a batch is committed on L1. For fast-finality requests, it detects when $f{+}1$ watchtower signatures are committed on L1.
\item Upon detecting a confirmed update, $\mathcal{S}_{\text{XRoll-update}}$ prepares $\mathit{Attachment}$ by extracting $\mathit{requestBatch}$ and $\mathit{newState}$ from the simulation state.
\item $\mathcal{S}_{\text{XRoll-update}}$ sends $\{\text{Update}, \mathit{Attachment}\}$ to $\mathcal{F}^{\text{FRoll}}_{\text{layer2-update}}$ via NET.
\end{enumerate}

}\end{functionality}\vspace{1em}

\begin{lemma}
\label{lem:XRoll4}
    For all PPT adversaries $\mathcal{A}$, there exists a PPT simulator $\mathcal{S}_{\text{XRoll-update}}$ such that for all PPT environments $\mathcal{E}$ and all security parameters $k\in\mathbb{N}$,
    \[
    \mathsf{EXEC}^{\mathcal{F}^{\text{FRoll}}_{\text{layer2-join}}}_{\mathcal{S}_{\text{XRoll-join}},\mathcal{E}}(k)
    \ \stackrel{c}{\approx}\
    \mathsf{EXEC}^{\mathcal{F}^{\text{FRoll}}_{\text{layer2-update}}}_{\mathcal{S}_{\text{XRoll-update}},\mathcal{E}}(k).
    \]
\end{lemma}

\begin{proof}
    Fix an arbitrary PPT environment $\mathcal{E}$. The only difference between the two games is that $\mathcal{F}^{\text{FRoll}}_{\text{layer2-update}}$ routes the Update request (from NET) through the subroutine $\mathcal{F}^{\text{FRoll}}_{\text{update}}$. We analyze the three observable components.

    \medskip
    \noindent As defined in both ideal functionalities, the Update request does not directly produce I/O outputs to $\mathcal{E}$; it only modifies $\mathsf{internalState}$. However, like Arbitrum, FRoll requires posting transaction batches, watchtower attestations, and fraud proofs to L1, which are observable by $\mathcal{E}$. We must argue that the same on-chain view and the same evolution of $\mathsf{internalState}$ arise in both executions.

    \medskip
    \noindent\textbf{(1) I/O outputs.} Although the Update request itself produces no I/O output, a divergence in $\mathsf{internalState}$ would cause observable differences in subsequent Read or Settlement outputs. In $\mathcal{P'}^{\text{FRoll}}$, a state update is accepted along one of two paths.

    \emph{Optimistic path.} A regular checkpoint is accepted as a valid state update only after: (i)~the operator publishes the checkpoint transaction $\mathit{TX}_{\mathit{batch}} = (\pcur, \mathit{ContractAddr}, \epsilon, \{\mathit{batch}, \mathit{resultState}\})$ on L1, (ii)~$\mathit{resultState}$ is the correct output of sequentially executing $\mathit{batch}$ from the previous confirmed state, (iii)~no fraud proof is committed within $T_{\text{challenge}}$, and (iv)~each honest sender's transaction in $\mathit{batch}$ has a corresponding entry in $\mathsf{requestQueue}$. If the checkpoint contains an incorrect transition, the honest client acting as verifier publishes a fraud proof, and the checkpoint is invalidated.

    \emph{Fast-finality path.} A fast-finality transaction $\mathit{TX}_{\mathit{fast}}$ is accepted once more than $f{+}1$ valid watchtower certificate transactions $\mathit{TX}_{\mathit{WT}}^{(j)} = (\mathit{TX}_{\mathit{fast}}, \sigma, \sigma_{W_j})$ are committed on L1 for the same $\mathit{TX}_{\mathit{fast}}$. Under the honest-majority watchtower assumption (at most $f$ out of $2f{+}1$ corrupted), at least $f{+}1$ honest watchtowers are available, and the EUF-CMA security of $\mathcal{F}_{\text{sig}}$ prevents the adversary from forging watchtower signatures, so any quorum reaching the threshold contains at least one honest signature on a well-formed $\mathit{TX}_{\mathit{fast}}$.

    These two paths are precisely the alternative checks enforced by $\mathcal{F}^{\text{FRoll}}_{\text{update}}$ (Steps~1--5). Since $\mathcal{S}_{\text{FRoll-update}}$ sends the Update message precisely when an update is confirmed in $\mathcal{P'}^{\text{FRoll}}$ along either path, $\mathcal{F}^{\text{FRoll}}_{\text{update}}$ accepts if and only if the corresponding update was accepted in $\mathcal{P'}^{\text{FRoll}}$. The existence of the challenge period and the unforgeability guaranteed by EUF-CMA prevent $\mathcal{P'}^{\text{FRoll}}$ from advancing any honest client's $\mathsf{internalState}$ on an incorrect transition or a forged transaction; therefore $\mathsf{internalState}$ changes identically across the two executions.

    \medskip
    \noindent\textbf{(2) On-chain transactions.} The operator checkpoint transactions, watchtower certificate transactions, and verifier fraud proofs are generated by $\mathcal{P'}^{\text{FRoll}}$, which runs identically in both simulators. The game hop only interposes $\mathcal{F}^{\text{FRoll}}_{\text{update}}$ as a filter on which updates are reflected in $\mathsf{internalState}$, not on which transactions are published to $\mathcal{F}_{\text{ledger}}$. Since the filter accepts exactly those updates that $\mathcal{P'}^{\text{FRoll}}$ would accept, and the ideal signature functionality guarantees signatures on the same messages are computationally indistinguishable, the on-chain transactions are computationally indistinguishable.

    \medskip
    \noindent\textbf{(3) Adversarial leakage.} In both executions, the simulator extracts the plaintext leakage (under our plaintext-leakage assumption) from the ideal functionality and uses it to drive $\mathcal{P'}^{\text{FRoll}}$, generating leakage to the adversary according to the corruption status. Since the two executions feed $\mathcal{P'}^{\text{FRoll}}$ the same accepted inputs and run the same real-world protocol, $\mathsf{internalState}$ at every honest participant -- including client, verifier, operator, and watchtower roles -- evolves identically across the two executions. The interposition of $\mathcal{F}^{\text{FRoll}}_{\text{update}}$ affects only how update acceptance is gated within the ideal functionality; it does not alter any network-side message exchange or any honest party's internal computation. Consequently, the leaked $\mathsf{internalState}$ and received messages delivered to corrupted parties are computationally indistinguishable across the two executions.

    \medskip
    Conclusively,
    $\mathsf{EXEC}^{\mathcal{F}^{\text{FRoll}}_{\text{layer2-join}}}_{\mathcal{S}_{\text{FRoll-join}}, \mathcal{E}}(k)
    \stackrel{c}{\approx}
    \mathsf{EXEC}^{\mathcal{F}^{\text{FRoll}}_{\text{layer2-update}}}_{\mathcal{S}_{\text{FRoll-update}}, \mathcal{E}}(k)$.
\end{proof}

Next, we extend $\mathcal{F}^{\text{FRoll}}_{\text{layer2-update}}$ with the read subroutine, yielding $\mathcal{F}^{\text{FRoll}}_{\text{layer2-read}} = (\mathcal{F}_{\text{client-read}}, \mathcal{F}_{\text{ledger}} \mid \mathcal{F}^{\text{FRoll}}_{\text{submit}}, \mathcal{F}^{\text{FRoll}}_{\text{join}}, \mathcal{F}^{\text{FRoll}}_{\text{update}}, \mathcal{F}^{\text{FRoll}}_{\text{read}})$.

\vspace{1em}\begin{functionality}{Description of $\mathcal{M}_{\text{client-read}}$ of $\mathcal{F}^{\text{FRoll}}_{\text{layer2-read}}$}{

\textbf{Implemented role(s):} \{client-read\}

\vspace{0.5em}\noindent \noindent \textbf{Main:}

\textbf{recv} \{Submit, $\mathit{request}$\} \textbf{from} I/O:

\begin{enumerate}[itemsep=0.5em]
\item \textbf{send} \{Submit, $\mathit{request}$, $\mathsf{internalState}$\} \textbf{to} $(\pcur, \scur, \mathcal{F}^{\text{FRoll}}_{\text{submit}}:\text{submit})$,
\textbf{wait for} \{Submit, $\mathit{response}$\} s.t. $\mathit{response} \in \{\text{true}, \text{false}\}$;
\item \textbf{if} $\mathit{response} =$ true: $\mathsf{requestQueue}.\text{add}(\mathit{request})$;
\textbf{send} $\mathit{request}$ \textbf{to} $\mathcal{S}$ via NET;
\end{enumerate}

\hrule
\vspace{0.5em}

\textbf{recv} \{Join, $\mathit{Attachment}$\} \textbf{from} NET:

\begin{enumerate}[itemsep=0.5em]
\item \textbf{send} \{Join, $\mathit{Attachment}$, $\mathsf{internalState}$\} \textbf{to} $(\pcur, \scur, \mathcal{F}^{\text{FRoll}}_{\text{join}}:\text{join})$,
\textbf{wait for} \{Join, $\mathit{response}$\} s.t. $\mathit{response} \in \{\text{true}, \text{false}\}$;
\item \textbf{if} $\mathit{response} =$ true: update $\mathsf{internalState}$ according to $\mathit{Attachment}$;
\textbf{reply} \{Join, $s_{\mathit{init}}$\} via I/O;
\end{enumerate}

\hrule
\vspace{0.5em}

\textbf{recv} \{Update, $\mathit{Attachment}$\} \textbf{from} NET:

\begin{enumerate}[itemsep=0.5em]
\item \textbf{send} \{Update, $\mathit{Attachment}$, $\mathsf{internalState}$\} \textbf{to} $(\pcur, \scur, \mathcal{F}^{\text{FRoll}}_{\text{update}}:\text{update})$,
\textbf{wait for} \{Update, $\mathit{response}$, $\mathit{newState}$, $\mathit{executedReq}$\};
\item \textbf{if} $\mathit{response} =$ true: update $\mathsf{internalState}$ with $\mathit{newState}$ and $\mathit{executedReq}$;
\end{enumerate}

\hrule
\vspace{0.5em}

\textbf{recv} \{Read\} \textbf{from} I/O:

\begin{enumerate}[itemsep=0.5em]
\item \textbf{send} \{Read, $\mathsf{internalState}$\} \textbf{to} $(\pcur, \scur, \mathcal{F}^{\text{FRoll}}_{\text{read}}:\text{read})$,
\textbf{wait for} $\mathit{ReadResult}$;
\item \textbf{if} $\mathit{ReadResult} \neq \bot$: \textbf{reply} \{Read, $\mathit{ReadResult}$\} via I/O;
\end{enumerate}

\hrule
\vspace{0.5em}

\textbf{recv} any request \textbf{from} I/O:
\begin{enumerate}[itemsep=0.5em]
    \item Forward request to $\mathcal{S}$ through NET;
\end{enumerate}

\hrule
\vspace{0.5em}

\textbf{recv} any message \textbf{from} NET:

\begin{enumerate}
    \item Output the message to $\mathcal{E}$ through I/O;
\end{enumerate}

}\end{functionality}\vspace{.5em}

\begin{functionality}{Description of simulator $\mathcal{S}_{\text{XRoll-read}}$}{
The simulator $\mathcal{S}_{\text{XRoll-read}}$ behaves the same as $\mathcal{S}_{\text{XRoll-update}}$.
}\end{functionality}\vspace{1em}

\begin{lemma}
\label{lem:XRoll5}
    For all PPT adversaries $\mathcal{A}$, there exists a PPT simulator $\mathcal{S}_{\text{XRoll-read}}$ such that for all PPT environments $\mathcal{E}$ and all security parameters $k\in\mathbb{N}$,
    \[
    \mathsf{EXEC}^{\mathcal{F}^{\text{FRoll}}_{\text{layer2-update}}}_{\mathcal{S}_{\text{XRoll-update}},\mathcal{E}}(k)
    \ \stackrel{c}{\approx}\
    \mathsf{EXEC}^{\mathcal{F}^{\text{FRoll}}_{\text{layer2-read}}}_{\mathcal{S}_{\text{XRoll-read}},\mathcal{E}}(k).
    \]
\end{lemma}

\begin{proof}
    Fix an arbitrary PPT environment $\mathcal{E}$. The only difference is that in $\mathsf{EXEC}^{\mathcal{F}^{\text{FRoll}}_{\text{layer2-read}}}_{\mathcal{S}_{\text{FRoll-read}},\mathcal{E}}(k)$ the read result is decided by $\mathcal{F}^{\text{FRoll}}_{\text{read}}$ based on $\mathsf{internalState}$, whereas in $\mathsf{EXEC}^{\mathcal{F}^{\text{FRoll}}_{\text{layer2-update}}}_{\mathcal{S}_{\text{FRoll-update}},\mathcal{E}}(k)$ read results are produced entirely by the simulator. We analyze the three observable components.

    \medskip
    \noindent\textbf{(1) I/O outputs.} By Lemma~\ref{lem:XRoll4}, the L1 ledger state and $\mathsf{internalState}$ are identical in both executions. In both executions, read results are reconstructed from $\mathcal{F}_{\text{ledger}}$ using the same reconstruction rule: the intersection of $\mathit{L1ReadResult}$ and $\mathsf{internalState}$, with the L2 view derived from the latest L1-final checkpoint surviving $T_{\text{challenge}}$ and any fast-finality transaction backed by $\geq f{+}1$ watchtower attestations on L1. Since both executions rely on the same ledger and the same $\mathsf{internalState}$, the read outputs delivered to $\mathcal{E}$ are distributionally identical.

    \medskip
    \noindent\textbf{(2) On-chain transactions.} Read requests do not publish any transactions to $\mathcal{F}_{\text{ledger}}$. Hence on-chain transactions are identical in both games.

    \medskip
    \noindent\textbf{(3) Adversarial leakage.} The game hop only interposes $\mathcal{F}^{\text{FRoll}}_{\text{read}}$ on the Read request, which is an I/O-facing operation based on ideal fcuntionalitie's $\mathsf{internalState}$ that produces no network messages to corrupted parties. Both simulators maintain synchronized corruption status across all role types (client, verifier, operator, watchtower) and run the same protocol logic inside $\mathcal{P'}^{\text{FRoll}}$. The leakage is computationally indistinguishable.

    \medskip
    Conclusively,
    $\mathsf{EXEC}^{\mathcal{F}^{\text{FRoll}}_{\text{layer2-update}}}_{\mathcal{S}_{\text{FRoll-update}}, \mathcal{E}}(k)
    \stackrel{c}{\approx}
    \mathsf{EXEC}^{\mathcal{F}^{\text{FRoll}}_{\text{layer2-read}}}_{\mathcal{S}_{\text{FRoll-read}}, \mathcal{E}}(k)$.
\end{proof}

Next, we extend $\mathcal{F}^{\text{FRoll}}_{\text{layer2-read}}$ with the settlement subroutine, yielding $\mathcal{F}^{\text{FRoll}}_{\text{layer2-settlement}} = (\mathcal{F}_{\text{client-settlement}}, \mathcal{F}_{\text{ledger}} \mid \mathcal{F}^{\text{FRoll}}_{\text{submit}}, \mathcal{F}^{\text{FRoll}}_{\text{join}},\\ \mathcal{F}^{\text{FRoll}}_{\text{update}}, \mathcal{F}^{\text{FRoll}}_{\text{read}}, \mathcal{F}^{\text{FRoll}}_{\text{settlement}})$.

\vspace{1em}\begin{functionality}{Description of $\mathcal{M}_{\text{client-settlement}}$ of $\mathcal{F}^{\text{FRoll}}_{\text{layer2-settlement}}$}{

\textbf{Implemented role(s):} \{client-settlement\}

\vspace{0.5em}\noindent \noindent \textbf{Main:}

\textbf{recv} \{Submit, $\mathit{request}$\} \textbf{from} I/O:

\begin{enumerate}[itemsep=0.5em]
\item \textbf{send} \{Submit, $\mathit{request}$, $\mathsf{internalState}$\} \textbf{to} $(\pcur, \scur, \mathcal{F}^{\text{FRoll}}_{\text{submit}}:\text{submit})$,
\textbf{wait for} \{Submit, $\mathit{response}$\} s.t. $\mathit{response} \in \{\text{true}, \text{false}\}$;
\item \textbf{if} $\mathit{response} =$ true: $\mathsf{requestQueue}.\text{add}(\mathit{request})$;
\textbf{send} $\mathit{request}$ \textbf{to} $\mathcal{S}$ via NET;
\end{enumerate}

\hrule
\vspace{0.5em}

\textbf{recv} \{Join, $\mathit{Attachment}$\} \textbf{from} NET:

\begin{enumerate}[itemsep=0.5em]
\item \textbf{send} \{Join, $\mathit{Attachment}$, $\mathsf{internalState}$\} \textbf{to} $(\pcur, \scur, \mathcal{F}^{\text{FRoll}}_{\text{join}}:\text{join})$,
\textbf{wait for} \{Join, $\mathit{response}$\} s.t. $\mathit{response} \in \{\text{true}, \text{false}\}$;
\item \textbf{if} $\mathit{response} =$ true: update $\mathsf{internalState}$ according to $\mathit{Attachment}$;
\textbf{reply} \{Join, $s_{\mathit{init}}$\} via I/O;
\end{enumerate}

\hrule
\vspace{0.5em}

\textbf{recv} \{Update, $\mathit{Attachment}$\} \textbf{from} NET:

\begin{enumerate}[itemsep=0.5em]
\item \textbf{send} \{Update, $\mathit{Attachment}$, $\mathsf{internalState}$\} \textbf{to} $(\pcur, \scur, \mathcal{F}^{\text{FRoll}}_{\text{update}}:\text{update})$,
\textbf{wait for} \{Update, $\mathit{response}$, $\mathit{newState}$, $\mathit{executedReq}$\};
\item \textbf{if} $\mathit{response} =$ true: update $\mathsf{internalState}$ with $\mathit{newState}$ and $\mathit{executedReq}$;
\end{enumerate}

\hrule
\vspace{0.5em}

\textbf{recv} \{Read\} \textbf{from} I/O:

\begin{enumerate}[itemsep=0.5em]
\item \textbf{send} \{Read, $\mathsf{internalState}$\} \textbf{to} $(\pcur, \scur, \mathcal{F}^{\text{FRoll}}_{\text{read}}:\text{read})$,
\textbf{wait for} $\mathit{ReadResult}$;
\item \textbf{if} $\mathit{ReadResult} \neq \bot$: \textbf{reply} \{Read, $\mathit{ReadResult}$\} via I/O;
\end{enumerate}

\hrule
\vspace{0.5em}

\textbf{recv} \{Settlement, $\mathit{Attachment}$\} \textbf{from} NET:

\begin{enumerate}[itemsep=0.5em]
\item \textbf{send} \{Settlement, $\mathit{Attachment}$, $\mathsf{internalState}$\} \textbf{to} $(\pcur, \scur, \mathcal{F}^{\text{FRoll}}_{\text{settlement}}:\text{settlement})$,
\textbf{wait for} \{Settlement, $\mathit{response}$, $s_{\mathit{settle}}$\} s.t. $\mathit{response} \in \{\text{true}, \text{false}\}$;
\item \textbf{if} $\mathit{response} =$ true: update $\mathsf{internalState}$;
\textbf{reply} \{Settlement, $s_{\mathit{settle}}$\} via I/O;
\end{enumerate}

\hrule
\vspace{0.5em}

\textbf{recv} any request \textbf{from} I/O:
\begin{enumerate}[itemsep=0.5em]
    \item Forward request to $\mathcal{S}$ through NET;
\end{enumerate}

\hrule
\vspace{0.5em}

\textbf{recv} any message \textbf{from} NET:

\begin{enumerate}
    \item Output the message to $\mathcal{E}$ through I/O;
\end{enumerate}

}\end{functionality}\vspace{.5em}

\begin{functionality}{Description of simulator $\mathcal{S}_{\text{FRoll-settlement}}$}{

The simulator $\mathcal{S}_{\text{FRoll-settlement}}$ behaves identically to $\mathcal{S}_{\text{FRoll-read}}$, except for the following additional behavior upon detecting a completed settlement:

\vspace{0.5em}
\textbf{Settlement interaction with $\mathcal{F}^{\text{FRoll}}_{\text{layer2-settlement}}$:}

\begin{enumerate}[itemsep=0.5em]

\item $\mathcal{S}_{\text{FRoll-settlement}}$ monitors $\mathcal{P'}^{\text{FRoll}}$. When it detects that a client entity is about to produce a successful settlement output inside the simulation (either via \texttt{collaborate} or \texttt{escape-hatch}), it intercepts this output.

\item $\mathcal{S}_{\text{FRoll-settlement}}$ prepares $\mathit{Attachment}$ by extracting:
\begin{itemize}
    \item $\mathit{TX}_{\mathit{peg\text{-}out}}$: the peg-out transaction committed on $\mathcal{F}_{\text{ledger}}$ in $\mathcal{P'}^{\text{FRoll}}$;
\end{itemize}

\item $\mathcal{S}_{\text{FRoll-settlement}}$ sends $\{\text{Settlement}, \mathit{Attachment}\}$ to $\mathcal{F}^{\text{FRoll}}_{\text{layer2-settlement}}$ via NET.

\end{enumerate}

}\end{functionality}\vspace{1em}

\begin{lemma}
\label{lem:XRoll6}
    For all PPT adversaries $\mathcal{A}$, there exists a PPT simulator $\mathcal{S}_{\text{XRoll-settlement}}$ such that for all PPT environments $\mathcal{E}$ and all security parameters $k\in\mathbb{N}$,
    \[
    \mathsf{EXEC}^{\mathcal{F}^{\text{FRoll}}_{\text{layer2-read}}}_{\mathcal{S}_{\text{XRoll-read}},\mathcal{E}}(k)
    \ \stackrel{c}{\approx}\
    \mathsf{EXEC}^{\mathcal{F}^{\text{FRoll}}_{\text{layer2-settlement}}}_{\mathcal{S}_{\text{XRoll-settlement}},\mathcal{E}}(k).
    \]
\end{lemma}

\begin{proof}
    Fix an arbitrary PPT environment $\mathcal{E}$. The only difference between the two executions is that $\mathcal{F}^{\text{FRoll}}_{\text{layer2-settlement}}$ routes the Settlement request (from NET) through $\mathcal{F}^{\text{FRoll}}_{\text{settlement}}$. We analyze the three observable components, noting that settlement in FRoll can proceed via two paths: \texttt{collaborate} (through the operator, with the peg-out enclosed in an operator checkpoint) or \texttt{escape-hatch} (with the peg-out submitted directly to L1 by the client).

    \medskip
    \noindent\textbf{(1) I/O outputs.} The I/O output affected is $\{\text{Settlement}, \mathsf{onchainState}\}$.

    \emph{Case~1: Successful settlement (either path).} In $\mathcal{P'}^{\text{FRoll}}$, a settlement succeeds when $\mathit{TX}_{\mathit{peg\text{-}out}}$ is committed on L1 without a fraud proof during $T_{\text{challenge}}$. Along the \texttt{collaborate} path, $\mathit{TX}_{\mathit{peg\text{-}out}}$ is enclosed in an operator checkpoint that itself survives $T_{\text{challenge}}$; along the \texttt{escape-hatch} path, $\mathit{TX}_{\mathit{peg\text{-}out}}$ is published standalone by the client and survives $T_{\text{challenge}}$ in its own right. The simulator $\mathcal{S}_{\text{FRoll-settlement}}$ triggers $\mathcal{F}^{\text{FRoll}}_{\text{settlement}}$ precisely when this occurs along either path. The checks in $\mathcal{F}^{\text{FRoll}}_{\text{settlement}}$: (1)~$\mathit{TX}_{\mathit{peg\text{-}out}}$ carries the latest state from $\mathsf{stateList}$, (2)~the settling client is registered in $\mathsf{identities}$, (3)~$\mathit{TX}_{\mathit{peg\text{-}out}}$ is committed on L1, and (4)~the post-settlement state is recorded on L1, are exactly the conditions that hold when the simulated settlement succeeds. Under the assumption of at least one honest verifier (who publishes a fraud proof against any peg-out inconsistent with the latest state within $T_{\text{challenge}}$), L1 liveness and safety, and the EUF-CMA security of $\mathcal{F}_{\text{sig}}$ (which prevents the adversary from forging the client's signature on a peg-out), $\mathcal{F}^{\text{FRoll}}_{\text{settlement}}$ accepts whenever $\mathcal{P'}^{\text{FRoll}}$ completes. Hence the output is identical.

    \emph{Case~2: Failed settlement.} If the operator does not publish the peg-out (under \texttt{collaborate}), or the client fails to publish $\mathit{TX}_{\mathit{peg\text{-}out}}$ to L1 (under \texttt{escape-hatch}), or a fraud proof invalidates the peg-out within $T_{\text{challenge}}$, $\mathcal{P'}^{\text{FRoll}}$ does not complete the settlement. Correspondingly, $\mathcal{S}_{\text{FRoll-settlement}}$ does not trigger $\mathcal{F}^{\text{FRoll}}_{\text{settlement}}$, so no output is produced in either execution.

    \medskip
    \noindent\textbf{(2) On-chain transactions.} The peg-out transaction published to $\mathcal{F}_{\text{ledger}}$ (whether enclosed in an operator checkpoint along the \texttt{collaborate} path or published standalone along the \texttt{escape-hatch} path) is generated by $\mathcal{P'}^{\text{FRoll}}$, which runs identically in both simulators. The game hop only affects when the ideal functionality generates the I/O output, not which transactions are published. By the EUF-CMA security of $\mathcal{F}_{\text{sig}}$, the distributions of on-chain transactions are computationally indistinguishable.

    \medskip
    \noindent\textbf{(3) Adversarial leakage.} Both simulators maintain synchronized corruption status across all role types (client, verifier, operator, watchtower) and execute the same protocol logic inside $\mathcal{P'}^{\text{FRoll}}$. The game hop interposes $\mathcal{F}^{\text{FRoll}}_{\text{settlement}}$ between the simulator and the I/O output but does not alter the simulator's internal execution or its interaction with corrupted parties. The leakage is identical.

    \medskip
    Conclusively,
    $\mathsf{EXEC}^{\mathcal{F}^{\text{FRoll}}_{\text{layer2-read}}}_{\mathcal{S}_{\text{FRoll-read}}, \mathcal{E}}(k)\stackrel{c}{\approx}\mathsf{EXEC}^{\mathcal{F}^{\text{FRoll}}_{\text{layer2-settlement}}}_{\mathcal{S}_{\text{FRoll-settlement}}, \mathcal{E}}(k)$.
\end{proof}

As a final step, we extend $\mathcal{F}^{\text{FRoll}}_{\text{layer2-settlement}}$ with the subroutine $\mathcal{F}^{\text{FRoll}}_{\text{updRnd}}$ to reach $\mathcal{F}^{\text{FRoll}}_{\text{layer2}} = (\mathcal{F}_{\text{client}}, \mathcal{F}_{\text{ledger}} \mid \mathcal{F}^{\text{FRoll}}_{\text{submit}}, \mathcal{F}^{\text{FRoll}}_{\text{join}}, \mathcal{F}^{\text{FRoll}}_{\text{update}},\\ \mathcal{F}^{\text{FRoll}}_{\text{read}}, \mathcal{F}^{\text{FRoll}}_{\text{settlement}}, \mathcal{F}^{\text{FRoll}}_{\text{updRnd}})$.

\vspace{1em}\begin{functionality}{Description of $\mathcal{M}_{\text{client}}$ of $\mathcal{F}^{\text{FRoll}}_{\text{layer2}}$}{

\textbf{Implemented role(s):} \{client\}

\noindent \textbf{Main:}

\vspace{0.5em}

\textbf{recv} \{Submit, $\mathit{request}$\} \textbf{from} I/O:
\begin{enumerate}[itemsep=0.5em]
\item \textbf{send} \{Submit, $\mathit{request}$, $\mathsf{internalState}$\} \textbf{to} $(\pcur, \scur, \mathcal{F}^{\text{FRoll}}_{\text{submit}}:\text{submit})$,
\textbf{wait for} \{Submit, $\mathit{response}$\} s.t. $\mathit{response} \in \{\text{true}, \text{false}\}$;
\item \textbf{if} $\mathit{response} =$ true: $\mathsf{requestQueue}.\text{add}(\mathit{request})$;
\textbf{send} $\mathit{request}$ \textbf{to} $\mathcal{S}$ via NET;
\end{enumerate}

\hrule\vspace{0.5em}

\textbf{recv} \{Join, $\mathit{Attachment}$\} \textbf{from} NET:
\begin{enumerate}[itemsep=0.5em]
\item \textbf{send} \{Join, $\mathit{Attachment}$, $\mathsf{internalState}$\} \textbf{to} $(\pcur, \scur, \mathcal{F}^{\text{FRoll}}_{\text{join}}:\text{join})$,
\textbf{wait for} \{Join, $\mathit{response}$\} s.t. $\mathit{response} \in \{\text{true}, \text{false}\}$;
\item \textbf{if} $\mathit{response} =$ true: update $\mathsf{internalState}$ according to $\mathit{Attachment}$;
\textbf{reply} \{Join, $s_{\mathit{init}}$\} via I/O;
\end{enumerate}

\hrule\vspace{0.5em}

\textbf{recv} \{Update, $\mathit{Attachment}$\} \textbf{from} NET:
\begin{enumerate}[itemsep=0.5em]
\item \textbf{send} \{Update, $\mathit{Attachment}$, $\mathsf{internalState}$\} \textbf{to} $(\pcur, \scur, \mathcal{F}^{\text{FRoll}}_{\text{update}}:\text{update})$,
\textbf{wait for} \{Update, $\mathit{response}$, $\mathit{newState}$, $\mathit{executedReq}$\};
\item \textbf{if} $\mathit{response} =$ true: update $\mathsf{internalState}$ with $\mathit{newState}$ and $\mathit{executedReq}$;
\end{enumerate}

\hrule\vspace{0.5em}

\textbf{recv} \{Read\} \textbf{from} I/O:
\begin{enumerate}[itemsep=0.5em]
\item \textbf{send} \{Read, $\mathsf{internalState}$\} \textbf{to} $(\pcur, \scur, \mathcal{F}^{\text{FRoll}}_{\text{read}}:\text{read})$,
\textbf{wait for} $\mathit{ReadResult}$;
\item \textbf{if} $\mathit{ReadResult} \neq \bot$: \textbf{reply} \{Read, $\mathit{ReadResult}$\} via I/O;
\end{enumerate}

\hrule\vspace{0.5em}

\textbf{recv} \{Settlement, $\mathit{Attachment}$\} \textbf{from} NET:
\begin{enumerate}[itemsep=0.5em]
\item \textbf{send} \{Settlement, $\mathit{Attachment}$, $\mathsf{internalState}$\} \textbf{to} $(\pcur, \scur, \mathcal{F}^{\text{FRoll}}_{\text{settlement}}:\\\text{settlement})$,
\textbf{wait for} \{Settlement, $\mathit{response}$, $s_{\mathit{settle}}$\} s.t. $\mathit{response} \in \{\text{true}, \text{false}\}$;
\item \textbf{if} $\mathit{response} =$ true: update $\mathsf{internalState}$;
\textbf{reply} \{Settlement, $s_{\mathit{settle}}$\} via I/O;
\end{enumerate}

\hrule\vspace{0.5em}

\textbf{recv} \{UpdateRound\} \textbf{from} NET:
\begin{enumerate}[itemsep=0.5em]
\item \textbf{send} \{UpdateRound, $\mathsf{internalState}$\} \textbf{to} $(\pcur, \scur, \mathcal{F}^{\text{FRoll}}_{\text{updRnd}}:\text{updRnd})$,
\textbf{wait for} \{UpdateRound, $\mathit{response}$\} s.t. $\mathit{response} \in \{\text{true}, \text{false}\}$;
\item \textbf{if} $\mathit{response} =$ true: $\mathsf{round} \leftarrow \mathsf{round} + 1$;
\textbf{reply} \{UpdateRound, $\mathit{response}$\} via NET;
\end{enumerate}

\hrule\vspace{0.5em}

\textbf{recv} \{GetCurRound\} \textbf{from} I/O or NET:
\begin{enumerate}[itemsep=0.5em]
    \item \textbf{reply} \{GetCurRound, $\mathsf{round}$\};
\end{enumerate}

\hrule\vspace{0.5em}

\textbf{recv} \{ReadL1\} \textbf{from} I/O or NET:
\begin{enumerate}[itemsep=0.5em]
    \item \textbf{send} \{Read\} \textbf{to} $(\pcur, \scur, \mathcal{F}_{\text{ledger}}: \text{client}_{\text{L1}})$;
    \textbf{wait for} $\mathit{L1ReadResult}$;
    \item \textbf{reply} \{ReadL1, $\mathit{L1ReadResult}$\} via I/O;
\end{enumerate}

}\end{functionality}\vspace{.5em}

\begin{functionality}{Description of simulator $\mathcal{S}_{\text{XRoll-updRnd}}$}{

The simulator $\mathcal{S}_{\text{XRoll-updRnd}}$ behaves identically to $\mathcal{S}_{\text{XRoll-settlement}}$, except for the following additional behavior when handling round-update requests:

\vspace{0.5em}
\textbf{Round update interaction with $\mathcal{F}^{\text{FRoll}}_{\text{layer2}}$:}

\begin{enumerate}[itemsep=0.5em]
\item Whenever $\mathcal{S}_{\text{XRoll-updRnd}}$ receives a round-update instruction from $\mathcal{A}$ (i.e., $\mathcal{A}$ advances the clock in the simulated $\mathcal{P'}^{\text{FRoll}}$), $\mathcal{S}_{\text{XRoll-updRnd}}$ simultaneously sends $\{\text{UpdateRound}\}$ to $\mathcal{F}^{\text{FRoll}}_{\text{layer2}}$ via NET.
\end{enumerate}

}\end{functionality}\vspace{1em}

\begin{lemma}
\label{lem:XRoll7}
    For all PPT adversaries $\mathcal{A}$, there exists a PPT simulator $\mathcal{S}_{\text{XRoll-updRnd}}$ such that for all PPT environments $\mathcal{E}$ and all security parameters $k\in\mathbb{N}$,
    \[
    \mathsf{EXEC}^{\mathcal{F}^{\text{FRoll}}_{\text{layer2-settlement}}}_{\mathcal{S}_{\text{XRoll-settlement}},\mathcal{E}}(k)
    \ \stackrel{c}{\approx}\
    \mathsf{EXEC}^{\mathcal{F}^{\text{FRoll}}_{\text{layer2}}}_{\mathcal{S}_{\text{XRoll-updRnd}},\mathcal{E}}(k).
    \]
\end{lemma}

\begin{proof}
    Fix an arbitrary PPT environment $\mathcal{E}$. The only difference between the two games is that $\mathcal{F}^{\text{FRoll}}_{\text{layer2}}$ routes the UpdateRound request (from NET) through $\mathcal{F}^{\text{FRoll}}_{\text{updRnd}}$. We analyze the three observable components.

    \medskip
    \noindent\textbf{(1) I/O outputs.} The only I/O output affected is $\{\text{GetCurRound}, \mathsf{round}\}$. In $\mathsf{EXEC}^{\mathcal{F}^{\text{FRoll}}_{\text{layer2-settlement}}}_{\mathcal{S}_{\text{FRoll-settlement}},\mathcal{E}}(k)$, $\mathcal{S}_{\text{FRoll-settlement}}$ maintains the round counter inside the simulated $\mathcal{P'}^{\text{FRoll}}$, whose clock is advanced by $\mathcal{A}$. In $\mathsf{EXEC}^{\mathcal{F}^{\text{FRoll}}_{\text{layer2}}}_{\mathcal{S}_{\text{FRoll-updRnd}},\mathcal{E}}(k)$, $\mathcal{S}_{\text{FRoll-updRnd}}$ additionally forwards each round-update request from $\mathcal{A}$ to $\mathcal{F}^{\text{FRoll}}_{\text{updRnd}}$, which enforces three liveness predicates over $\{\mathit{TX}\}_{\text{L1}}$ before accepting: each pending self-submitted regular transaction from an honest client must be L1-final within $T_{L_1} + T_{\text{challenge}}$, each pending self-submitted fast-finality transaction must be L1-final within $T_{L_1}$ (bypassing the challenge period via the watchtower quorum), and each pending escape-hatch settlement must be L1-final within $T_{L_1} + T_{\text{challenge}}$. In $\mathcal{P'}^{\text{FRoll}}$, the operator is triggered by $\mathcal{F}^{\text{FRoll}}_{\text{com}}$ to publish a checkpoint transaction periodically via the \{UpdateRequest\} request, and watchtowers publish their attestations on observing fast-finality transactions on L1. As long as $\mathcal{F}_{\text{ledger}}$ guarantees L1 liveness, the honest verifier publishes fraud proofs against invalid checkpoints within $T_{\text{challenge}}$, and the honest-majority watchtower committee (at most $f$ out of $2f{+}1$ corrupted) accumulates at least $f{+}1$ attestations on each fast-finality transaction within $T_{L_1}$, all three liveness predicates hold at the same rounds in which $\mathcal{P'}^{\text{FRoll}}$ would advance its clock. Therefore the round counters evolve identically across the two executions. All other I/O outputs (Join, Read, Settlement) are unaffected by this game hop.

    \medskip
    \noindent\textbf{(2) On-chain transactions.} The UpdateRound request itself does not publish any new transactions to $\mathcal{F}_{\text{ledger}}$; it only reads L1 to verify the three liveness predicates. The checkpoint publications, watchtower attestation transactions, and any verifier fraud proofs are generated by the operator's \{UpdateRequest\} request, watchtowers' \{CheckFastL1\} request, and the verifier's \{UpdateCheck\} request in $\mathcal{P'}^{\text{FRoll}}$ respectively, all of which run identically in both simulators. Hence on-chain transactions are identical in both games.

    \medskip
    \noindent\textbf{(3) Adversarial leakage.} The game hop only interposes $\mathcal{F}^{\text{FRoll}}_{\text{updRnd}}$ on the UpdateRound request. Since $\mathcal{F}^{\text{FRoll}}_{\text{updRnd}}$ only reads L1 and does not generate any network messages to corrupted parties, the simulator's internal execution of $\mathcal{P'}^{\text{FRoll}}$ and its interaction with corrupted parties (across all role types: client, verifier, operator, watchtower) are unaffected. The leakage delivered to $\mathcal{A}$ is computationally indistinguishable.

    \medskip
    Conclusively,
    $\mathsf{EXEC}^{\mathcal{F}^{\text{FRoll}}_{\text{layer2-settlement}}}_{\mathcal{S}_{\text{FRoll-settlement}}, \mathcal{E}}(k)\stackrel{c}{\approx}\mathsf{EXEC}^{\mathcal{F}^{\text{FRoll}}_{\text{layer2}}}_{\mathcal{S}_{\text{FRoll-updRnd}}, \mathcal{E}}(k)$.
\end{proof}

\ThmrealizeCross*
\begin{proof}
    Let $\mathcal{A}$ be any PPT adversary and let $\mathcal{E}$ be any PPT environment. By Lemmas~\ref{lem:XRoll1}--\ref{lem:XRoll7}, the sequence of game hops yields a chain of computationally indistinguishable execution ensembles, each adjacent pair differing only in how one subroutine is implemented:
    \[
    \mathsf{EXEC}^{\mathcal{P}^{\text{FRoll}}}_{\mathcal{A},\mathcal{E}}
    \ \stackrel{c}{\approx}\ 
    \mathsf{EXEC}^{\mathcal{F}^{\text{FRoll}}_{\text{dummy}}}_{\mathcal{S}_{\text{FRoll}},\mathcal{E}}
    \ \stackrel{c}{\approx}\ 
    \cdots
    \ \stackrel{c}{\approx}\ 
    \mathsf{EXEC}^{\mathcal{F}^{\text{FRoll}}_{\text{layer2}}}_{\mathcal{S}_{\text{XRoll-updRnd}},\mathcal{E}}.
    \]
    By transitivity of computational indistinguishability, we obtain
    $
    \mathsf{EXEC}^{\mathcal{P}^{\text{FRoll}}}_{\mathcal{A},\mathcal{E}}
    \stackrel{c}{\approx}
    \mathsf{EXEC}^{\mathcal{F}^{\text{FRoll}}_{\text{layer2}}}_{\mathcal{S}_{\text{XRoll-updRnd}},\mathcal{E}}$, which proves that $\mathcal{P}^{\text{FRoll}}$ iUC-realizes $\mathcal{F}^{\text{FRoll}}_{\text{layer2}}$.
\end{proof}
\end{document}